\documentclass{memo-l}

\usepackage[all]{xy} 
\usepackage{amssymb}
\usepackage{graphicx}
\usepackage{lscape}
\usepackage{bbm}

\def\l{{\mathbf{l}}}

\newtheorem{theorem}{Theorem}[chapter]
\newtheorem{lemma}[theorem]{Lemma}
\newtheorem{lemme}[theorem]{Lemma}
\newtheorem{exemple}[theorem]{Example}
\newtheorem{proposition}[theorem]{Proposition}
\newtheorem{corollary}[theorem]{Corollary}

\theoremstyle{definition}
\newtheorem{definition}[theorem]{Definition}

\newtheorem{notation}[theorem]{Notation}

\theoremstyle{remark}
\newtheorem{remarque}[theorem]{Remark}
\newtheorem{conjecture}[theorem]{Conjecture}

\numberwithin{section}{chapter}
\numberwithin{equation}{chapter}

\makeindex

\begin{document}

\frontmatter

\title{Planar Markovian Holonomy Fields}


\author{Franck Gabriel}
\address{UPMC, 4 Place Jussieu, 75005 Paris (France)}
\thanks{Supported by the ``Contrats Doctoraux du minist\`{e}re fran\c{c}ais de la recherche"  and the ERC grant, ``Behaviour near criticality,Ó held by M. Hairer.}
\curraddr{Mathematics Institute, University Of Warwick, Gibbet Hill Rd, Coventry CV4 7AL, (United-Kingdom)}
\email{franck.gabriel@normalesup.org}


\date{10/24/2016}

\subjclass[2010]{81T13, 81T40, 60G60, 60G51, 60G09, 20F36, 57M20, 81T27, 28C20, 58D19, 60B15}

\keywords{random field, random holonomy, Yang-Mills measure, L\'{e}vy process on compact Lie groups, braid group, de Finetti's theorem, planar graphs, continuous limit}

\dedicatory{}

\maketitle

\tableofcontents

\begin{abstract}
This text defines and studies planar Markovian holonomy fields which are processes  indexed by paths on the plane which takes their values in a compact Lie group. These processes behave well under the concatenation and orientation-reversing operations on paths. Besides, they satisfy some independence and invariance by area-preserving homeomorphisms properties. A symmetry arises in the study of planar Markovian holonomy fields: the invariance by braids. For finite and infinite random sequences the notion of invariance by braids is defined and we prove a new version of the de-Finetti's theorem. This allows us to construct a family of planar Markovian holonomy fields called the planar Yang-Mills fields. We prove that any regular planar Markovian holonomy field is a planar Yang-Mills field. Planar Yang-Mills fields can be partitioned into three categories according to their degree of symmetry: we study some equivalent conditions in order to classify them. Finally, we recall the notion of (non planar) Markovian holonomy fields defined by Thierry L\'{e}vy. Using the results previously proved, we compute the spherical part of any regular Markovian holonomy field. 
\end{abstract}


\chapter*{Introduction}

Yang-Mills theory is a theory of random connections on a principal bundle, the law of which satisfies some local symmetry: the gauge symmetry. It was introduced in the work of Yang and Mills, in 1954, in \cite{YM54}. Since then, mathematicians have tried to formulate a proper quantum Yang-Mills theory. The construction on a four dimensional manifold for any compact Lie group is still a challenge: we will focus in this article on the $2$-dimensional quantum Yang-Mills theory. On a formal level, a Yang-Mills measure is a measure on the space of connections which looks like: 
 \begin{align*}
 e^{-\frac{1}{2}S_{{\sf YM}}(A)} DA, 
 \end{align*}
 where $S_{{\sf YM}}(A)$ is the Yang-Mills action of the connection $A$, which is the $L^2$ norm of the curvature, and $DA$ is a translation invariant measure on the space of connections. Yet, many problems arise with this formulation, the main of which is that the space of connections can not be endowed with a translation invariant measure. It took some time to understand which space could be endowed by a well-defined measure. 

One possibility to handle this difficulty in a probabilistic way is to consider holonomies of the random connections along some finite set of paths: thus, after the works of Gross \cite{Gro85}, \cite{Gro88}, Driver \cite{Dri89}, \cite{Dri91} and Sengupta \cite{Sen92}, \cite{Sen97a} who constructed the Yang-Mills field for a small class of paths but for any surface, it was well understood that the Yang-Mills measure was a process indexed by some nice paths. Their construction uses the fact that the holonomy process under the Yang-Mills measure should satisfy a stochastic differential equation driven by a Brownian white-noise curvature. The Yang-Mills measure has to be constructed on the multiplicative functions from the set of paths to a Lie group, that is the set of functions which have a good behavior under concatenation and orientation-inversion of paths. This idea was already present in the precursory work of Albeverio, H{\o}egh-Krohn and Holden (\cite{Albquestion}, \cite{Alb}, \cite{AHKH88}, \cite{AHKH86}). 
  
In \cite{L1}, \cite{Levythese} and \cite{Levy}, L\'{e}vy gave a new construction. This construction allowed him to consider any compact Lie groups, any surfaces and any rectifiable paths. Besides, it allowed him to generalize the definition of Yang-Mills measure to the setting where, in some sense, the curvature of the random connection is a conditioned L\'{e}vy noise. The idea was to establish the rigorous discrete construction, as proposed by E. Witten in \cite{Wit91} and \cite{Wit92} and to show that one could take a continuous limit. 

The discrete construction was defined by considering a perturbation of a uniform measure, the Ashtekar-Lewandowski measure, by a density. The continuous limit was established using the general Theorem $3.3.1$ in \cite{Levy}. This theorem must be understood as a two-dimensional Kolmogorov's continuity theorem and one should consider it as one of the most important theorem in the theory of two-dimensional holonomy fields. In the article \cite{FGA}, G. C\'{e}bron, A. Dahlqvist and the author show how to use this theorem in order to construct generalizations of the master field constructed in \cite{SenguptaA} and \cite{Levymaster}.  

In the seminal book \cite{Levy}, L\'{e}vy defined also Markovian holonomy fields. This is the axiomatic point of view on Yang-Mills measures, seen as families of measures, indexed by surfaces which have a good behavior under chirurgical operations on surfaces and are invariant under area-preserving homeomorphisms. The importance of this notion is that Yang-Mills measures are Markovian holonomy fields. It is still unknown if any regular Markovian holonomy field is a Yang-Mills measure but this work is a first step in order to prove so. 

The axiomatic formulation of the Markovian holonomy fields allows us to understand L\'{e}vy processes as one-dimensional planar Markovian holonomy fields. 

\section*{L\'{e}vy processes and planar Markovian holonomy fields}
\label{levyprocplanar}
Let $G$ be a compact Lie group. If ${\sf dim}(G)\geq 1$, we endow the group $G$ with a bi-invariant Riemannian distance $d_G$. If $G$ is a finite group, we endow it with the distance $d_{G}(x,y) = \delta_{x,y}$. There exist two notions of L\'{e}vy processes depending on the definitions of the increments: left increments $Y_{t}Y_s^{-1}$ or right increments $Y_s^{-1}Y_{t}$. We will fix the following convention: in this article, a L\'{e}vy process on $G$ is a c\`{a}dl\`{a}g process with independent and stationary right increments which begins at the neutral element. In fact one can use a weaker definition and forgot about the c\`{a}dl\`{a}g property and define a L\'{e}vy process as a continuous in probability family of random variables $(Y_t)_{t \in \mathbb{R}^{+}}$ such that for any $t>s\geq 0$:
\begin{itemize}
\item $Y_s^{-1}Y_t$ has same law as $Y_{t-s}$, 
\item $Y_s^{-1}Y_t$ is independent of $\sigma(Y_u, u<s)$, 
\item $Y_0 = e$ a.s. 
\end{itemize}

Let $Y$ be a L\'{e}vy process on $G$. Let us denote by $\mathcal{D}(\mathbb{R})$ the set of integrable smooth densities on $\mathbb{R}$. For any $vol \in \mathcal{D}(\mathbb{R})$, one can define a measure $\mathbb{E}_{vol}$ on $G^{\mathbb{R}}$ such that, under $\mathbb{E}_{vol}$, the canonical projection process $\left(X_t\right)_{t \in \mathbb{R}}$ has the law of $\left(Y_{vol(]-\infty,t])}\right)_{t \in \mathbb{R}}$. The family $\big(\mathbb{E}_{vol}\big)_{vol \in \mathcal{D}}$ satisfies three properties: 
\begin{description}
\item[-Area-preserving increasing homeomorphism invariance] Let us consider $\psi$, an increasing homeomorphism of $\mathbb{R}$. Let $vol$ and $vol'$ be two smooth densities in $\mathcal{D}(\mathbb{R})$. Let us suppose that $\psi$ sends $vol$ on $vol'$. The mapping $\psi$ induces a measurable mapping from $G^{\mathbb{R}}$ to itself which we will denote also by $\psi$ and which is defined by: $$\psi((x_t)_{t \in \mathbb{R}}) = \left(x_{\psi(t)}\right)_{t \in \mathbb{R}}.$$
It is then easy to see that $ \mathbb{E}_{vol} = \mathbb{E}_{vol'} \circ \psi^{-1}$. For example, for any real $t \in \mathbb{R}$ and any bounded function $f$ on $G$: 
\begin{align*}
\mathbb{E}_{vol'}\big[f(X_{\psi(t)})\big] &= \mathbb{E}\big[f (Y_{vol'(]-\infty,\psi(t)])})\big] = \mathbb{E}\big[f (Y_{vol(]-\infty,t])})\big] = \mathbb{E}_{vol}\big[f(X_t)\big].
\end{align*}
\item[-Independence] Let $vol$ be a smooth density in $\mathcal{D}(\mathbb{R})$. Let $[s_0,t_0]$ and $[s_1,t_1]$ be two disjoint intervals. Under $\mathbb{E}_{vol}$, $\sigma\big((X_s^{-1}X_t), s_0\leq s<t\leq t_0 \big)$ is independent of $\sigma\big((X_s^{-1}X_t), s_1\leq s<t\leq t_1 \big)$. 
\item[-Locality property] Let $vol$ and $vol'$ be two smooth densities in $\mathcal{D}(\mathbb{R})$. Let  $t_0$  be a real such that $vol_{\mid]-\infty,t_0]} = vol'_{\mid]-\infty,t_0]}$. The law of $(X_t)_{t \leq t_0}$ is the same under $\mathbb{E}_{vol}$ as under $\mathbb{E}_{vol'}$. 
\end{description}

Let us consider a family of measures $\left(\mathbb{E}_{vol}\right)_{vol \in \mathcal{D}(\mathbb{R})}$ on $G^{\mathbb{R}}$; we say that it is stochastically continuous if, for any ${vol \in \mathcal{D}(\mathbb{R})}$, for any sequence $(t_n)_{n \in \mathbb{N}}$, if $t_n$ converges to $t \in \mathbb{R} \cup \{-\infty\}$, $\mathbb{E}_{vol}\left(d_G(X_{t_n}, X_t)\right) \underset{n \to \infty}{\longrightarrow} 0$, where we recall that $(X_t)_{t \in \mathbb{R}}$ is the canonical projection process and where, by convention, $X_{-\infty}$ is the constant function equal to the neutral element $e$. If  $\left(\mathbb{E}_{vol}\right)_{vol \in \mathcal{D}(\mathbb{R})}$ is stochastically continuous and satisfies the three axioms stated above then there exists a L\'{e}vy process $(Y_t)_{t \in \mathbb{R}^{+}}$ such that, for any smooth density $vol$ in $\mathcal{D}(\mathbb{R})$, the canonical projection process $\left(X_t\right)_{t \in \mathbb{R}}$ has the law of $\left(Y_{vol(]-\infty,t])}\right)_{t \in \mathbb{R}}$. 

With these axioms in mind, looking in Section \ref{secdefplanarHF} at the definitions of planar Markovian holonomy fields, the reader can understand why we can consider L\'{e}vy processes as one-dimensional planar Markovian holonomy fields. The surprising fact that we will prove in this paper is that the family of regular two-dimensional planar Markovian holonomy fields is not bigger than the set of one-dimensional planar Markovian holonomy fields.

\section*{Braids}
The most innovative idea of this paper is to introduce for the very first time the braid group in the study of two-dimensional Yang-Mills theory. This is also one of the main ingredient in the article \cite{FGA}.

The braid group is an object which possesses different facets: a combinatorial, a geometric and an algebraic one. One can introduce the braid group using geometric braids: this construction allows us to have a graphical and combinatorial framework to work with. Since it is the most intuitive construction, we quickly present it so that the reader will be familiar with these objects. 

\begin{proposition}
For any $n \geq 2$, let the conÞguration space $C_{n} (\mathbb{R}^{2})$ of n indistinguishable points in the plane be $\big((\mathbb{R}^{2})^{n}\setminus\Delta \big) / \mathfrak{S}_{n}$ where $\Delta$ is the union of the hyperplanes $\{ \mathbf{x} \in (\mathbb{R}^{2})^{n},\ \mathbf{x}_{i} = \mathbf{x}_{j}\}$.
The fundamental group of the configuration space $C_{n} (\mathbb{R}^{2})$ is the braid group with $n$ strands $\mathcal{B}_{n}$: 
\begin{align*}
\mathcal{B}_{n} = \pi_{1}\big(C_{n}(\mathbb{R}^{2})\big).
\end{align*}
\end{proposition}

Every continuous loop ${\gamma}$ in $C_{n}(\mathbb{R}^{2})$ parametrized by $[0,1]$ and based at the point $\big((1,0), ... , (n,0)\big)$ can be seen as $n$ continuous functions $\gamma_{j} \in \mathcal{C}\big([0,1], \mathbb{R}^{2}\big)$ such that, if we set $\sigma:j \mapsto \gamma_{j}(1)$ for any $j\in \{1,É, n\}$, the following conditions hold:
\begin{align*}
&1\text{-\ } \forall j \in \{1,..., n\}, \gamma_j(0) = (j,0), \\
&2\text{-\ } \sigma \in \mathfrak{S}_{n},\\
&3\text{-\ } \forall\ t \in [0,1], \forall j\neq j', \gamma_{j}(t) \neq \gamma_{j'}(t). 
\end{align*}
The function $\gamma_j$ is given by the image of $\gamma$ by the projection $\pi_j: \left(\mathbb{R}^{2}\right)^{n}\!\to~\mathbb{R}^{2}.$ We call $\gamma$ a geometric braid since if we draw the $(\gamma_{j})_{j=1}^{n}$ in $\mathbb{R}^{3}$, we obtain a physical braid. One can look at Figure \ref{tressephy} to have an illustration of this fact. 
\begin{figure}
 \centering
  \includegraphics[width=165pt]{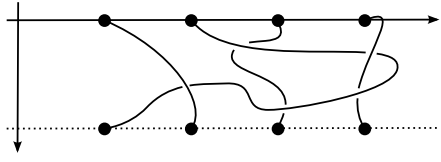}
 \caption{A physical braid $\beta$.}
 \label{tressephy}
\end{figure}

With this point of view, the composition of two braids is just obtained by gluing two geometric braids, taking then the equivalence class by isotopy of the new braid as shown in Figure \ref{multifig}. In this paper, we will take the convention that, in order to compute $\beta_1\beta_2$, one has to put the braid $\beta_2$ above the braid $\beta_1$.

\begin{figure}
 \centering
  \includegraphics[width=165pt]{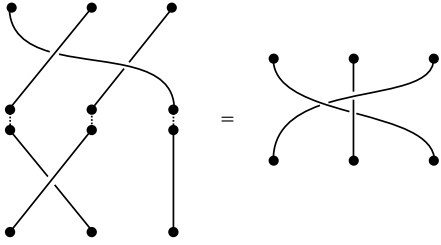}
 \caption{The multiplication of two braids.}
 \label{multifig}
\end{figure}

As we see in Figure \ref{diagramfig}, one can represent a braid by a two dimensional diagram (or, to be correct, classes of equivalence of two-dimensional diagrams) that we call $n$-diagrams. This representation can remind the reader the representation of any permutation by a diagram, yet, in this representation of braids, one remembers which string is above an other at each crossing. It is a well-known result that any $n$-diagram represents a unique braid with $n$-strands.

\begin{figure}
 \centering
  \includegraphics[width=110pt]{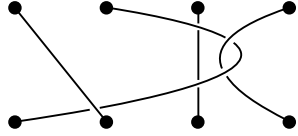}
 \caption{A two dimensional diagram representation of $\beta$.}
 \label{diagramfig}
\end{figure}

 Thus, in order to construct a braid, we only have to construct a $n$-diagram. Besides, every computation can be done with the $n$-diagrams.

For any $i \in \{1,É, n-1\}$, let $\beta_{i}$ be the equivalence class of $(\gamma_{j}^{i})_{j=1}^{n}$ defined by: 
\begin{align*}
&\forall k \in \{1,É, n\}\setminus \{i, i+1\},\ \forall t \in [0,1],\ \gamma^{i}_{k}(t) = (k,0),\\
&\forall t \in [0,1],\ \gamma_{i}^{i}(t) \ \ \ \!= \left(i+\frac{1}{2}\right)-\frac{1}{2}e^{i\pi t},\\
&\forall t \in [0,1],\ \gamma_{i+1}^{i}(t)= \left(i+\frac{1}{2}\right)+\frac{1}{2}e^{i\pi t}, 
\end{align*}
with the usual convention $\mathbb{R}^{2} \simeq \mathbb{C}$. As any braid can be obtained by braiding two adjacent strands, the family $\left(\beta_{i}\right)_{i=1}^{n-1}$ generates $\mathcal{B}_{n}$. 

\begin{figure}
 \centering
  \includegraphics[width=165pt]{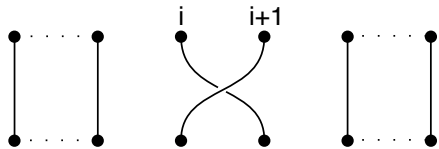}
 \caption{The elementary braid $\beta_{i}$.}
  \label{elbraid}
\end{figure}

\section*{Layout of the article}

Since the theory of Markovian holonomy fields is a newborn theory which mixes geometry, representation, probabilities, we recall all the tools we need and try to make this paper accessible to any people from any domain of mathematics. This paper is in the same time an introduction and a sequel to \cite{Levy}. The reader shouldn't be surprised that we copy some of the definitions of \cite{Levy} as any reformulation wouldn't have been as good as L\'{e}vy's formulation. 

In Section \ref{sectionpathetautre}, we recall the classical notions: paths, multiplicative functions, É Besides, we supply a lack in \cite{Levy}: we decided to develop the notion of random holonomy fields, as it might be possible, in the future, that some general random holonomy fields of interest would not be Markovian holonomy fields. Thus, any proposition in \cite{Levy} that could be applied to random holonomy fields is stated in this setting. We study the projection of random holonomy fields on the set of gauge-invariant random holonomy fields and, in the gauge-invariant setting, we explain how to restrict and extend the structure group. At last, we develop the loop paradigm which, in particular, implies the new Proposition \ref{exten4}.

The Section \ref{Sectiongraph} is devoted to the theory of planar graphs and the notion of $\mathbb{G}-\mathbb{G}'$ piecewise diffeomorphisms. One of the main results is Corollary \ref{injectiondansn} which states that any generic planar graph can be seen, via such diffeomorphism,  as a sub-graph of the $\mathbb{N}^{2}$-graph.

Using the previous sections, we define in Section \ref{planarmarkovianholonomyfields}, four different notions of planar Markovian holonomy fields. Under some regularity condition, it will be proved in the paper that the four notions are essentially equivalent. These objects are processes, indexed by paths drawn on the plane, which are gauge-invariant, invariant under area-preserving homeomorphisms, which satisfy a weak independence property and a locality property. We consider the questions of restriction and extension of the structure group for planar Markovian holonomy fields.

The equivalence between the notions of weak discrete and weak continuous planar Markovian holonomy field is then proved in Section \ref{weakconstructibility}, using a theorem of Moser and Dacorogna. 

In Section \ref{Sectiongroupofreducedloop} we define the group of reduced group, as L\'{e}vy did in \cite{Levy}, and obtain a generalization of L\'{e}vy's work in the planar case. This allows us to exhibit general families of loops which generate the group of reduced loops of any planar graph.  

Two sections are devoted to the link between braids and probabilities: Sections \ref{Braids} and \ref{Braids2}. In the first one, after explaining an algebraic definition of the braid group, we show how the Artin's theorem can be applied on the group of reduced loops. We define the notion of invariance by braid for finite sequences of random variables. Section \ref{Braids2} is devoted to the geometric point of view on braids and to a de-Finetti-Ryll-Nardzewski's theorem for random infinite sequences which are invariant under the action of the braid groups. Under an assumption of independence of the diagonal-conjugacy classes, one can characterize the invariant by diagonal conjugation braidable sequences which are sequences of i.i.d. random variables. In the end of the section, we apply these results to processes.

In Sections \ref{planaryangmills},  \ref{carsec} and \ref{Charac}, the reader can find the main results about planar Markovian holonomy fields. Section \ref{Braids}, on finite braid-invariant sequences of random variables, allow us in Section \ref{planaryangmills} to construct, for any L\'{e}vy process which is self-invariant by conjugation, a planar Yang-Mills field associated with it. This construction differs from all the previous constructions since it uses neither the notion of uniform or Ashtekar-Lewandowski measure nor the notion of stochastic differential equations. This allows us to consider any self-invariant by conjugation L\'{e}vy processes, where before, one had to consider L\'{e}vy processes with density with respect to the Haar measure and which were invariant by conjugation by the structure group $G$. In Section \ref{carsec} and \ref{Charac}, using the results of Section \ref{Braids2}, we prove that any regular planar Markovian holonomy field is a planar Yang-Mills field. Besides, we show that one can characterize their degree of symmetry according to the law of the holonomy associated to simple loops.

Since any regular planar Markovian holonomy field is a planar Yang-Mills field, is it possible to show that any Markovian holonomy field is a Yang-Mills field? In Section \ref{Markov}, we answer partly to this question. First we recall the notion of Markovian holonomy fields. The free boundary condition expectation is constructed and allows us to make a bridge between Markovian holonomy fields and planar Markovian holonomy fields. The results shown previously allows us to prove Theorem \ref{theo}: the spherical part of a regular Markovian holonomy field is equal to the spherical part of a Yang-Mills field.

In order to get a more accurate idea of the results shown in this article and the different notions defined in it, one can refer to the diagram Page \pageref{diagram}. \\

{\em Acknowledgements.} The author would like to gratefully thank A. Dahlqvist and G. C\'{e}bron for the useful discussions, his PhD advisor Pr. T. L\'{e}vy for his helpful comments and his postdoctoral supervisor, Pr. M. Hairer, for giving him the time to finalize this article. Also, he wishes to acknowledge the help provided by A. Bouthier who suggested using Jordan's theorem (Theorem 8.8), and by P. and M.-F. Gabriel for proof-reading the English.

The first version of this work has been made during the PhD of the author at the university Paris 6 UPMC. This final version of the paper was completed during his postdoctoral position at the University of Warwick where the author is supported by the ERC grant, ÒBehaviour near criticalityÓ, held by Pr. M. Hairer.

\mainmatter

\part{Basic Notions}

\chapter[Background]{Backgrounds: Paths, Random Multiplicative Functions on Paths}

Let $M$ be either a smooth compact surface (possibly with boundary) or the plane $\mathbb{R}^{2}$. A {\em{measure of area}} on $M$ is a smooth non-vanishing density on $M$, that is, a Borel measure which has a smooth positive density with respect to the Lebesgue measure in any coordinate chart. It will often be denoted by $vol$. We call $(M,vol)$ a {\em measured surface}. We endow $M$ with a Riemannian metric $\gamma$ and we will denote by $\gamma_0$ the standard Riemannian metric on $\mathbb{R}^{2}$.

\label{sectionpathetautre}
\section{Paths}

The notion of paths that we will use in this paper is given in the following definition.

\begin{definition}
 A parametrized path on $M$ is a continuous curve $c: [0,1] \to M$ which is either constant or Lipschitz continuous with speed bounded below by a positive constant. 
\end{definition}
Two parametrized paths can give the same drawing on $M$ but with different speed and we will only consider equivalence classes of paths. 
\begin{definition}Two parametrized paths on $M$ are equivalent if they differ by an increasing bi-Lipschitz homeomorphism of $[0,1]$. An equivalence class of parametrized paths is called a path and the set of paths on $M$ is denoted by $P(M)$. 
\end{definition}

Actually, the notion of path does not depend on $\gamma$ since the distances defined by two different Riemannian metric are equivalent. Two parametrized paths $pp_{1}$ and $pp_{2}$ which represent the same path $p$ share the same endpoints. It is thus possible to define the endpoints of $p$ as the endpoints of any representative of $p$. If $p$ is a path, by $\underline{p}$ (resp. $\overline{p}$) we denote the starting point (resp. the arrival point) of $p$. From now on, we will not make any difference between a path $p$ and any parametrized path $pp\in p$. 
\begin{definition}
A path is simple either if it is injective on $[0,1]$ or if it is injective on $[0,1[$ and $\underline{p} = \overline{p}$. 
\end{definition}

Later, we will need the following subset of paths. 
\begin{definition}
\label{defensemble}
We define ${\sf Aff}_\gamma(M)$ to be the set of paths on $M$ which are piecewise geodesic paths with respect to $\gamma$.
\end{definition}
The set of paths ${\sf Aff}_{\gamma_0}(\mathbb{R}^{2})$ will be simply denoted by ${\sf Aff}(\mathbb{R}^{2})$. An other set of paths will be very important for our study: the set of loops.  
\begin{definition}
A loop $l$ is a path such that $\underline{l} = \overline{l}$. A smooth loop is a loop whose image is an oriented smooth $1$-dimensional submanifold of $M$. The set of loops is denoted by $L(M)$. Let $m$ be a point of $M$. A loop $l$ is based at $m$ if $\underline{l} = m$. The set of loops based at $m$ is denoted by $L_{m}(M)$. 
\end{definition}

We can define the inverse and concatenation operations on paths. Let $p_1$ and $p_2$ be paths, let $pp_{1}$ and $pp_2$ be representatives of these paths and let us suppose that $\overline{p_{1}} = \underline{p_{2}}$. The inverse of $p_1$, denoted by $p_1^{-1}$, is the equivalence class of the parametrized path $t \mapsto pp_{1}(1-t).$ The concatenation of $p_{1}$ and $p_{2}$ denoted by $p_{1}p_{2}$ is the equivalence class of the parametrized path: 
$$
 pp_{1}.pp_{2}: t  \mapsto \left\{
    \begin{array}{ll}
      pp_{1}(2t) & \mbox{if } t \leq 1/2,  \\
      
      pp_{2}(2t-1) & \mbox{if } t > 1/2. 
    \end{array}
\right.
$$
\begin{definition}
A set of paths $P$ is connected if any couple of endpoints of elements of $P$ can be joined by a concatenation of elements of $P$. 
\end{definition}

Using concatenation we can introduce a relation on the set of loops. 

\begin{definition}
\label{equiv}
Let $P$ be a set of paths. Two loops $l$ and $l'$ are elementarily equivalent in $P$ if there exist three paths, $a, b, c \in P$ such that $\{l, l' \} = \{ab, acc^{-1}b\}$. We say that $l$ and $l'$ are equivalent in $P$ if there exists a finite sequence $l=l_{0},  ..., l_{n}=l'$ such that $l_{i}$ is elementarily equivalent to $l_{i+1}$ for any $i \in \{0,  ..., n-1\}.$ We will write it $l \simeq_P l'$. 
\end{definition}

\begin{definition}
A lasso is a loop $l$ such that one can find a simple loop $m$, the meander, and a path $s$, the spoke, such that $l = s m s^{-1}$.  
\end{definition}

A loop has a well-defined origin and orientation. A cycle is a loop in which one forgets about the endpoint. In a non-oriented cycle, the endpoint and the orientation are forgotten.

\begin{definition}
\label{cycle}
We say that two loops $l_{1}$ and $l_{2}$ are related if and only if they can be decomposed as: $l_{1} = cd$, $l_{2}=dc$, with $c$ and $d$ two paths. The set of equivalence classes for the relation defined on $L(M)$ is the set of cycles. The operation of inversion is compatible with this equivalence. A non-oriented cycle is a pair $\{l, l^{-1} \}$ where $l$ is a cycle. Besides, a cycle is simple if any loop which represents it is simple and it is said smooth if any loop which represents it is smooth. 
\end{definition}

We need a notion of convergence of paths in order to define the continuity of random holonomy fields. The definition makes use of the Riemannian metric $\gamma$, yet the notion of convergence with fixed endpoints will not depend on the choice of the Riemannian metric. We denote by $d_{\gamma}$ the distance on $M$ which is associated with~$\gamma$. 

\begin{definition}
Let $p_{1}$ and $p_{2}$ be two paths of $M$. Let $\ell(p_{1})$ (resp. $\ell(p_{2})$) be the length of the path $p_{1}$ (resp. $p_{2}$). We define the distance between $p_{1}$ and $p_{2}$ as: 
$$d_{l}(p_{1}, p_{2}) = \inf_{pp_1\in p_1, pp_2 \in p_2} \sup_{t \in [0,1]} \left[d_{\gamma}(pp_{1}(t), pp_{2}(t)) \right] + | \ell(p_{1})-\ell(p_{2})|.$$
The topology induced by $d_{l}$ does not depend on the choice of $\gamma$. 

Let $(p_{n})_{n \geq 0}$ be a sequence of paths on $M$. Let $p$ be a path on $M$. The sequence $(p_{n})_{n \geq 0}$ converges to $p$ with fixed endpoints if and only if: 
\begin{itemize}
\item $d_{l}(p_{n}, p) \to 0$ as $n\to +\infty$, 
\item $\forall n \geq 0$, $\overline{p_{n}} = \overline{p}$ and $\underline{p_{n}} = \underline{p}$.
\end{itemize}
\end{definition}

We will see that the convergence with fixed endpoints behaves well when one considers images of paths by bi-Lipschitz homeomorphisms. Let us consider $\psi$ a locally bi-Lipschitz homeomorphism from $\mathbb{R}^{2}$ to itself. 
\begin{lemme}
\label{lem:bilip}
Let $p$ be a path on the plane, the image of $p$, $\psi(p)$, is also a path. 
\end{lemme}

\begin{proof}
We only have to prove that $\psi(p)$ has a finite length. This is a consequence of the fact that if $c: [0,1] \to \mathbb{R}^{2}$ is a continuous function, the length of $c$ is given by $\sup_{\{t_1, ...,t_n | t_1<...<t_n\} \subset [0,1]} \sum_{i=1}^{n-1} |c(t_i)-c(t_{i+1})|$.
\end{proof}

\begin{lemme}
\label{lem:bilip2}
Let $(p_n)_{n\in \mathbb{N}}$ be a sequence of paths which converges to a path $p$ with fixed endpoints. The sequence $\left( \psi(p_n)\right)_{n \in \mathbb{N}}$ converges to $\psi(p)$ with fixed endpoints. 
\end{lemme}

\begin{proof}
Let us consider $(p_n)_{n\in \mathbb{N}}$ and $p$ which satisfy the conditions of the lemma. L\'{e}vy proved in Lemma $1.2.17$ of \cite{Levy}, that $p_n$ converges to $p$ uniformly when the paths are parametrized at constant speed. Let us denote by $(p_n(t))_{t \in [0,1]}$ and $(p(t))_{t \in [0,1]}$ these parametrized paths. Since $\psi$ is locally Lipschitz, $\left(\psi(p_n(t))\right)_{t \in [0,1]}$ converges uniformly to $\left(\psi(p(t))\right)_{t \in [0,1]}$ when $n$ goes to infinity and there exists a real $R$ such that the speed of $\psi(p_n)$ for any integer $n$ and the speed of $\psi(p)$ are bounded by $R$. An application of Lemma $1.2.18$ of \cite{Levy} and the triangular inequality allows us to assert that the length of $\psi(p_n)$ converges to the length of $\psi(p)$. This proves that $d_l(\psi(p_n), \psi(p))$ converges to zero when $n$ goes to infinity.  
\end{proof}

Using this notion of convergence, one can define a notion of density. The following lemma was proved by L\'{e}vy in Proposition $1.2.12$ of \cite{Levy}.

\begin{lemme}
\label{dense}
The set of paths ${\sf Aff}_\gamma(M)$ is dense in $P(M)$ for the convergence with fixed endpoints. 
\end{lemme}

 One has to be careful when working with the convergence with fixed endpoints. For example, the set of paths whose images are concatenation of horizontal and vertical segments is {\em not} dense in $P(\mathbb{R}^{2})$. Indeed, one condition in order to have the convergence with fixed endpoints is that the length of the paths converges to the length of the limit path. But, for any path $p$ which can be written as a concatenation of horizontal and vertical segments, $\ell( p) \geq || \overline{p}- \underline{p}||_{1},$ where $|| . ||_{1}$ is the usual $L^{1}$ norm on $\mathbb{R}^{2}$, yet this inequality does not hold for a general path $p$.

\section{Measures on the set of multiplicative functions}
\label{multi}
In this section, the presentation differs from the one of \cite{Levy}: new definitions and new results already appear in this section. 

From now on, except if specified, $G$ is a compact Lie group, with the usual convention that a compact Lie group of dimension $0$ is a finite group. The neutral element will be denoted by $e$. We endow $G$ with a bi-invariant Riemannian distance $d_G$. If $G$ is a finite group, we endow it with the distance $d_G(x,y) = \delta_{x,y}$. We denote by $\mathcal{M}(G)$ the space of finite Borel positive measures on $G$. 

\subsection{Definitions}
Let $P$ be a subset of $P(M)$ and let $L$ be a set of loops in~$P$. 

\begin{definition}
A function $h$ from $P$ to $G$ is multiplicative if and only if: 
\begin{itemize}
\item $h(c^{-1}) = h(c)^{-1}$ for any path $c$ in $P$ such that $c^{-1} \in P$, 
\item $h(c_{1}c_{2}) = h(c_{2}) h(c_{1})$ for any paths $c_{1}$ and $c_{2}$ in $P$ which can be concatenated and such that $c_{1}c_{2} \in P$.
\end{itemize}
We denote by $\mathcal{M}ult(P, G)$ the set of multiplicative functions from $P$ to $G$. A function from $L$ to $G$ is pre-multiplicative over $P$ if and only if:
\begin{itemize}
\item it is multiplicative, 
\item for any $l$ and $l'$ in $L$ which are equivalent in $P$, we have: $h(l)=h(l')$.
\end{itemize}
We denote by $\mathcal{M}ult_P(L, G)$ the set of pre-multiplicative functions over $P$. 
\end{definition}

We will often make the following slight abuse of notation. 
\begin{notation}
\label{notation}
Let $c$ be a path in $P$. If a multiplicative function $h$ is not specified in a formula, $h(c)$ will stand for the function on $\mathcal{M}ult(P, G)$: 
\begin{align*}
h(c): \mathcal{M}ult(P, G) &\to G \\
h\ \ \ \ \ &\mapsto h(c).
\end{align*}
\end{notation}

The notion of equivalence of loops, as stated in Definition \ref{equiv}, is important due to the following remark. 

\begin{remarque}
\label{equalityequiv}
Let $h$ be in $\mathcal{M}ult(P, G)$ and let $l, l'$ be loops in $P$. A simple induction and the multiplicative property of $h$ imply that if $l \simeq_P l'$ then $h(l) = h(l')$. 
\end{remarque}

Let $P$ be a set of paths and let $Q$ be a freely generating subset of $P$ in the sense that: 
\begin{itemize}
\item any path in $P$ is a finite concatenation of elements of $Q$, 
\item no element of $Q$ can be written as a non-trivial finite concatenation of paths in $Q \cup Q^{-1}$, 
\item $Q\cap Q^{-1} = \emptyset$. 
\end{itemize}
Then we have the identification: 
\begin{align}
\label{lequationedgeparadigme}
\mathcal{M}ult(P, G) \simeq G^{Q}.
\end{align}
This is the {\em edge paradigm} for multiplicative functions. The novelty of the approach we have in this paper is to put the emphasis on the {\em loop paradigm} for gauge-invariant random holonomy fields. The first paradigm is interesting for general random holonomy fields on surfaces, yet the second seems to be more appropriate for gauge-invariant random holonomy fields on the plane.

\begin{remarque}
\label{extendG} All the following definitions and propositions deal with multiplicative functions on a set of paths $P$. All of them extend to $G^{T}$, with $T \subset \mathbb{R}$. Indeed, if $P=\cup_{r \in T}\{c_r, c_r^{-1}\}$, with $c_r$ being the path on the plane based at $0$ and going clockwise once around the circle of center $(0,r)$ and radius $r$, then $\mathcal{M}ult(P, G) \simeq G^{T}.$ 
\end{remarque}

We will now endow the space of multiplicative functions with a $\sigma$-field in order to be able to speak about measures on $\mathcal{M}ult(P,G)$. 
\begin{definition}
The Borel $\sigma$-field $\mathcal{B}$ on $\mathcal{M}ult(P, G)$ is the smallest $\sigma$-field such that for any paths $c_{1},  ..., c_{n}$ and any continuous function $f: G^{n} \to \mathbb{R}$, the mapping $h \mapsto f\left(h(c_{1}),  ..., h(c_{n})\right)$ is measurable. 
\end{definition}

\begin{definition}
A random holonomy field $\mu$ on the set $P$ is a measure on $\left(\mathcal{M}ult(P, G),\mathcal{B}\right)$.  If $P = P(M)$, we call it a random holonomy field on $M$.\end{definition}

Let $\mu$ be a random holonomy field on $P$: the weight of $\mu$ is $\mu({\mathbbm{1}})$. One can define a regularity notion for random holonomy fields. 
\begin{definition}
\label{cont}
A random holonomy field $\mu$ on $P$ is stochastically continuous if for any sequence $(p_{n})_{n \geq 0}$ of elements of $P$ which converges with fixed endpoints to $p \in P$, 
\begin{align}
\label{conv0}
\int_{\mathcal{M}ult(P, G)} d_{G}\left(h(p_{n}), h(p)\right) \mu(dh) \underset{n \to \infty}{\longrightarrow} 0. 
\end{align}
The measure $\mu$ is locally stochastically $\frac{1}{2}$-H\"{o}lder continuous if for any compact set $S\subset M$, for any measure of area $vol$ on $M$, there exists $K>0$ such that for any simple loop $l \in P$ bounding a disk $D$ such that $l\subset S$: 
\begin{align}
\label{Hold}
\int_{\mathcal{M}ult(P, G)} d_{G}\left(e,h(l)\right) \mu (dh) \leq K \sqrt{vol(D)}, 
\end{align}
where $e$ is the neutral element of $G$.

A family of random holonomy fields, $(\mu)_{\mu\in \mathcal{F}}$, with each $\mu$ defined on some set $P_{\mu}$, is uniformly locally stochastically $\frac{1}{2}$-H\"{o}lder continuous if the constant $K$ in equation (\ref{Hold}) is independent of the random holonomy field in $\mathcal{F}$. 
\end{definition}

\subsection{Construction of random holonomy fields I }
Let us review the two mains results on which the construction of random holonomy fields is based.

\begin{notation}
\label{restriction}
Let $J$ and $K$ be two subsets of $P(M)$ such that $J \subset K$. The restriction function from $\mathcal{M}ult(K, G)$ to $\mathcal{M}ult(J, G)$ will be denoted by $\rho_{J, K}$. If $M \subset M'$ are two surfaces, we denote by $\rho_{M, M'}$ the restriction function $\rho_{P(M),P(M')}$. The notation is set such that for any $J \subset K \subset L \subset P(M)$, $\rho_{J, K} \circ \rho_{K, L} = \rho_{J, L}.$ 
\end{notation}

The fact that $G$ is a compact group allows us to construct measures on the set of multiplicative functions by taking projective limits of random holonomy fields on finite subsets of paths. This behavior is very different from what can be observed for Gaussian measures on Banach spaces. Indeed, in \cite{Levy}, Proposition 2.2.3, L\'{e}vy proved, when $\mathcal{F}$ is a collection of finite subsets of $P$, the next proposition using an application of Carath\'{e}odory's extension theorem. We give a proof based on the Riesz-Markov's theorem, proof which shows clearly why we only consider compact groups. 

\begin{proposition}
\label{limitproj}
Let $\mathcal{F}$ be a collection of subsets of paths on $M$. We denote by $P$ their union. Suppose that, when ordered by the inclusion, $\mathcal{F}$ is directed: for any $J_1$ and $J_2$ in $\mathcal{F}$, there exists $J_3 \in \mathcal{F}$ such that $J_1 \cup J_2 \subset J_3$. For any $J \in \mathcal{F}$, let $m_J$ be a probability measure on $\left(\mathcal{M}ult(J, G),\mathcal{B}\right)$. Assume that the probability spaces $\left(\mathcal{M}ult(J, G),\mathcal{B},m_J\right)$ endowed with the restriction mappings $\rho_{J, K}$ for $J \subset K$ form a projective system. This means that for any $J_1$ and $J_2$ in $\mathcal{F}$ such that $J_1 \subset J_2$, one has $$m_{J_1} = m_{J_2} \circ \rho_{J_1,J_2}^{-1}.$$
Then there exists a unique probability measure $m$ on $\left(\mathcal{M}ult(P, G),\mathcal{B}\right)$ such that for any $J \in \mathcal{F}$, $$m_J = m \circ \rho_{J, P}^{-1}.$$
\end{proposition}

\begin{proof}
We endow $G^{P}$ with the product topology. As an application of Tychonoff's theorem it is a compact space. A consequence of this is that $\mathcal{M}ult(P,G)$, endowed with the restricted topology, is also a compact space as it is closed in $G^{P}$. Besides, the $\sigma$-field $\mathcal{B}$ is the Borel $\sigma$-field on $\mathcal{M}ult(P,G)$. 
Let us consider $A$ the set of cylinder continuous functions, that is the set of functions $\mathbf{f}: \mathcal{M}ult(P,G) \to \mathbb{R}^{+}$ of the form: 
\begin{align*}
\mathbf{f}: h \mapsto f\left(h(p_1), ..., h(p_n)\right), 
\end{align*}
for some $ n \in \mathbb{N}$, some $p_1,...,p_n \in P$ and some continuous function $f: G^{n} \to \mathbb{R}$.

The set $A$ is a subalgebra of the algebra $C\left(\mathcal{M}ult(P,G),\mathbb{R}\right)$ of real-valued continuous functions on $\mathcal{M}ult(P,G)$. This subalgebra separates the points of $\mathcal{M}ult(P,G)$ and contains a non-zero constant function. Due to the Stone-Weierstrass's theorem, $A$ is dense in $C\left(\mathcal{M}ult(P,G),\mathbb{R}\right)$. Any function $\mathbf{f}$ in $A$ depends only on a finite number of paths, so that there exists some $J \in \mathcal{F}$ such that $\mathbf{f}$ can be seen as a continuous function on $\mathcal{M}ult(J,P)$. We define: 
\begin{align*}
m(\mathbf{f}) = m_J(\mathbf{f}),
\end{align*}
which does not depend on the chosen $J \in \mathcal{F}$ thanks to the projectivity and multiplicative properties. 

We have defined a positive linear functional $m$ on $A$, the norm of which is bounded by the total weight of any of the measures $(m_J)_{J \in \mathcal{F}}$. Thus $m$ can be extended on $C\left(\mathcal{M}ult(P,G),\mathbb{R}\right)$ and an application of the Riesz-Markov's theorem allows us to consider $m$ as a measure on $\left(\mathcal{M}ult(P,G), \mathcal{B}\right)$. This is the projective limit of $(m_J)_{J \in \mathcal{F}}$. 
\end{proof}

The notion of locally stochastically $\frac{1}{2}$-H\"{o}lder continuity allows us to have an extension theorem from some subsets of paths to their closure, as shown in the proof of Corollary 3.3.2 of \cite{Levy}. 

\begin{theorem}
\label{exten2}
Let $\mu_{{\sf Aff}_\gamma(M)}$ be a random holonomy field on ${\sf Aff}_\gamma(M)$. If it is locally stochastically $\frac{1}{2}$-H\"{o}lder continuous then there exists a unique stochastically continuous random holonomy field $\mu$ on $M$ such that: $$\mu_{{\sf Aff}_\gamma(M)} = \mu \circ \rho_{{\sf Aff}_\gamma(M), P(M)}^{-1}.$$ 
\end{theorem}

\subsection{Gauge-invariance }
For any subset $P$ of $P(M)$, a natural group acts on $\mathcal{M}ult(P, G)$, the {\em gauge group}, that we are going to describe. Let us fix a subset $P$ of $P(M)$ which will stay fixed until the end of the chapter. 

\begin{definition}
\label{actiongauge}
Let $V = \{x\in M, \exists\ \! p \in P, x = \underline{p}\text{ or }x = \overline{p}\}$ be the set of endpoints of $P$. We define the partial gauge group associated with $P$ by setting $J_{P}~= G^{V}$. If $P=P(M)$, this group is called the gauge group of $M$. 
The group $J_P$ acts by gauge transformations on the space $\mathcal{M}ult\left(P,G\right)$: if $j \in J_P$, the action of $j$ on $h \in \mathcal{M}ult(P, G)$ is given by: 
\begin{align*}
\forall c \in P, (j \bullet h)(c) = j_{\overline{c}}^{-1} h(c) j_{\underline{c}}. 
\end{align*}
\end{definition}

Let $Q=\{c_{1},  ..., c_{n}\}$ be a finite set of paths on $M$. Looking only at the evaluation on $c_{i}$ for $i\in \{1,  ..., n \}$, we have the inclusion: $\mathcal{M}ult(Q, G) \subset G^{n}$. The gauge action of $J_Q$ on $\mathcal{M}ult(Q,G)$ extends naturally to an action on $G^{n}$ by $(j \bullet \mathbf{g})_{i} = j_{\overline{c_{i}}}^{-1} \mathbf{g}_{i} j_{\underline{c}_{i}}$ for any $i\in \{1, ...,n\}$. 

\begin{remarque}
\label{diagconj}
If $l_{1},  ..., l_{n}$ are loops based at a point $m$, the partial gauge group is nothing but $G$ and the corresponding action on $G^{n}$ is the diagonal conjugation: $$j \bullet (g_{1},  ..., g_{n}) = (j^{-1} g_{1}j ,  ..., j^{-1} g_{n} j).$$
We denote by $\big[(g_1, ..., g_n)\big] $ the equivalence class of $(g_1, ..., g_n)$ in $G^{n}$ under the diagonal conjugation action. 
\end{remarque}

We now define a sub-$\sigma$-field of $\mathcal{B}$, the {\em invariant $\sigma$-field.}
\begin{definition}
On $\mathcal{M}ult(P, G)$, the invariant $\sigma$-field, denoted by $\mathcal{I}$, is the smallest $\sigma$-field such that for any paths $c_{1},  ..., c_{n}$ in $P$ and any continuous function $f: G^{n} \to \mathbb{R}$ invariant under the action of $J_{\{c_{1},  ..., c_{n}\}}$ on $G^{n}$ defined after Definition \ref{actiongauge}, the mapping $h \mapsto f\left(h(c_{1}),  ..., h(c_{n})\right)$ is measurable. 
\end{definition} 

Let us remark that if $M$ is the disjoint union of two smooth compact surfaces, $M = M_{1} \sqcup M_{2}$ then $\mathcal{M}ult\left(P(M),G\right) \simeq \mathcal{M}ult(P(M_{1}),G) \times \mathcal{M}ult(P(M_{2}),G)$. Besides, let $\mathcal{I}$ (respectively $\mathcal{I}_{1}$, $\mathcal{I}_{2}$) be the invariant $\sigma$-field on $\mathcal{M}ult\left(P(M),G\right)$ (respectively $\mathcal{M}ult\left(P(M_{1}),G\right)$, $\mathcal{M}ult\left(P(M_{2}),G\right)$). We have $\mathcal{I} \simeq \mathcal{I}_{1} \otimes \mathcal{I}_{2}.$

Locally bi-Lipschitz homeomorphisms between surfaces give rise to some examples of functions which are measurable with respect to the Borel and the invariant $\sigma$-fields. Given $M$ and $M'$ two smooth compact surfaces, suppose that we are given a locally bi-Lipschitz homeomorphism $\psi$ from $M$ to $M'$, we can construct, for any $h$ in $\mathcal{M}ult\left(P(M'),G\right)$, a natural multiplicative function $\psi^{*}h$ on $M$: 
\begin{align*}
\left(\psi^{*}h\right) ( p) = h\left( \psi( p)\right), \forall p \in P(M).
\end{align*}
This defines a function $\psi^{*}: \mathcal{M}ult\left(P(M'),G\right) \to \mathcal{M}ult\left(P(M),G\right).$ The function $\psi^{*}$ is measurable for the Borel and the invariant $\sigma$-fields. From now on, we denote also by $\psi$ the application $\psi^{*}$.

On the invariant $\sigma$-field on $\mathcal{M}ult(P,G)$, any measure is of course invariant by the gauge transformations. Explicitly, for any measure $\mu$ on $\left(\mathcal{M}ult(P, G), \mathcal{I}\right)$, for any measurable continuous function $f$ from $\left(\mathcal{M}ult(P, G), \mathcal{I}\right)$ to $\mathbb{R}$ and for any $j \in J_P$:
\begin{align}
\label{inv2}
\int_{\mathcal{M}ult(P, G)} f(j \bullet h) d\mu(h) = \int_{\mathcal{M}ult(P, G)} f( h) d\mu(h). 
\end{align}
The following definition is less trivial as the following class of gauge-invariant measures is not equal to the collection of all measures. 

\begin{definition}
Let $\mu$ be a random holonomy field on $P$. We say that $\mu$ is invariant under gauge transformations if and only if the Equality (\ref{inv2}) holds for any continuous function $f$ from $\left(\mathcal{M}ult(P, G),\mathcal{B}\right)$ to $\mathbb{R}$ and any $j \in J_P$.
\end{definition}

\begin{remarque}
\label{haarchemin}
Let $\mu$ be a gauge-invariant random holonomy field on $P$. Let $p$ a path in $P$ which is not a loop: $\underline{p} \neq \overline{p}$. Then under $\frac{\mu}{\mu({\mathbbm{1}})}$, $h( p)$ has the law of a Haar random variable. Indeed, applying the gauge transformation which is equal to $1$ everywhere except at $\underline{p}$ or $\overline{p}$, where its value is set to be an arbitrary element of $G$, we see that the law of $h( p)$ is invariant by left- and right-multiplication. 
\end{remarque}

There exists a one-to-one correspondence between measures on $\left(\mathcal{M}ult( P, G),\mathcal{I}\right)$ and gauge-invariant measures on $\left(\mathcal{M}ult( P, G),\mathcal{B}\right)$. The next proposition is similar to the results of \cite{Baez}. For any positive integer $n$, for any continuous function $f$ on $G^{n}$ and any set of paths $\{c_1, ..., c_n\}$ in $P$, we define the function $\hat{f}_{J_{c_1, ..., c_n}}$ such that, for any $g_1, ..., g_n$ in $G$: 
\begin{align}
\label{moyennisee}
\hat{f}_{J_{c_1, ..., c_n}}(g_1,...,g_n) = \int_{J_{c_1, ..., c_n}} f\left(j \bullet (g_1, ..., g_n)\right)dj, 
\end{align}
where $dj$ is the Haar measure on $J_{c_1, ..., c_n}$. 

\begin{notation}
If $\mu$ is a finite measure on a measurable space $(\Omega, \mathcal{A})$ and if $\mathcal{B} \subset \mathcal{A}$ is a sub-$\sigma$-field, by $\mu_{\mid \mathcal{B}}$, we denote the image of $\mu$ by the identity map: $(\Omega, \mathcal{A}) \to (\Omega, \mathcal{B})$. 
\end{notation}

\begin{proposition}
\label{exten}
For any measure $\mu$ on $\left(\mathcal{M}ult( P, G),\mathcal{I}\right)$, there exists a unique gauge-invariant random holonomy field on $P$ which will be denoted either by $\hat{\mu}$ or $\mu^{\widehat { }}$, such that $\hat{\mu}_{\mid \mathcal{I}} = \mu.$
\end{proposition}

\begin{proof}
The uniqueness of $\hat{\mu}$ follows from the upcoming Proposition \ref{unicite1}. Let us prove its existence. 
We will define $\hat{\mu}$ by the fact that for any measurable function $f: G^{n} \to \mathbb{R}^{+}$ and any $n$-tuple $c_{1}$, ..., $c_{n}$ of elements of $P$: 
\begin{align*}
\hat{\mu}\left( f \left(h(c_{1}),  ..., h(c_{n}) \right) \right) = \mu \left( \hat{f}_{J_{c_{1},  ..., c_{n}}} \left(h(c_{1}),  ..., h(c_{n})\right) \right).
\end{align*}

Let us consider a finite set of paths in $P$, $P_1 = \{ c_{1},  ..., c_{n}\}$. Let us consider the natural inclusion $\iota: \mathcal{M}ult\left(P_1, G\right) \subset G^{n}$ given by the evaluations on $c_1, ..., c_n$. The equalities $\hat{\mu}_{P_1}(f) = \mu \left( \hat{f}_{J_{c_{1},  ..., c_{n}}} \left(h(c_{1}),  ..., h(c_{n})\right) \right)$ for any continuous function on $G^{n}$ define a linear positive functional on $C(G^{n})$. By compactness of $G^{n}$, applying the theorem of Riesz-Markov, it gives a measure $\hat{\mu}_{P_1}$ on $G^{n}$, the support of which is easily seen to be a subset of $\iota\left(\mathcal{M}ult\left(P_1, G\right)\right)$. We can thus look at the induced measure on $\mathcal{M}ult\left(P_1, G\right)$ denoted by $\hat{\mu}_{\mid\mathcal{M}ult(P_1, G)}$. The family of measures $\left(\hat{\mu}_{\mid\mathcal{M}ult(P_1, G)}\right)_{P_1\subset P, \#P_1 <\infty}$ forms a projective family of measures for the inclusion of sets. Thus, by Proposition \ref{limitproj}, it defines a measure on $\left(\mathcal{M}ult\left(P, G\right),\mathcal{B}\right)$.
\end{proof}

Let us introduce a notion which will be important in the definition of planar markovian holonomy fields. Let $\mu$ be a random holonomy field on $P$. 
\begin{definition}
\label{Iindependence}
Let $P_1$ and $P_2$ be two families of paths in $P$. We will say that $\left(h( p)\right)_{p \in P_1}$ and $\left(h( p)\right)_{p \in P_2}$ are $\mathcal{I}$-independent if and only if, for any finite family $(p_{1}^{i})_{i=1}^{n}$ in $P_1$, any finite family $(p_{2}^{i})_{i=1}^{m}$ in $P_2$ and any continuous function $f: G^{n} \to \mathbb{R}$ (resp. $g: G^{m} \to \mathbb{R}$) invariant under the action of $J_{\{p_{1}^{1}, ...,p_{1}^{n}\}}$ (resp. $J_{\{p_{2}^{1}, ...,p_{2}^{m}\}}$), the following equality holds: 
\begin{align}
\label{indep}
\mu\left[f\left((h(p_{1}^{i}))_{i=1}^{n}\right)g\left((h(p_{2}^{j}))_{j=1}^{m}\right)\right] \!=\! \mu\left[f\left((h(p_{1}^{i}))_{i=1}^{n}\right)\right] \mu \left[g\left((h(p_{2}^{j}))_{j=1}^{m}\right)\right].
\end{align}
This is equivalent to say that under $\mu$, the two $\sigma$-fields $\sigma\left(h( p): p \in P_1\right) \cap \mathcal{I}$ and $\sigma\left(h( p): p \in P_2\right) \cap \mathcal{I}$ are independent. 
\end{definition}

Let us remark that the invariant $\sigma$-field on $G^{2}$ which we denote by $\mathcal{I}_{(2)}$ is, in general, different from the product $\mathcal{I} \otimes \mathcal{I}$ where $\mathcal{I}$ is the invariant $\sigma$-field of $G$. When $G = \mathfrak{S}_3$, this fact is implied by the following assertion: $$11 = \#\left\{\left[(\sigma,\sigma')\right], (\sigma,\sigma') \in \mathfrak{S}_3^{2} \right\}\neq \left(\#\left\{\left[\sigma\right], \sigma \in \mathfrak{S}_3 \right\}\right)^{2}=9.$$ In particular, if $(X, Y)$ is a random vector such that $X$ is $\mathcal{I}$-independent of $Y$, the knowledge of the laws of the random conjugacy classes $[X]$ and $[Y]$ {\em does not allow us} to reconstruct the law of the random diagonal conjugacy class $[(X, Y)]$. In the following remark, we will see that in some very special cases, the $\mathcal{I}$-independence is equivalent to the independence. 

\begin{remarque}
\label{I-indep-et-indep-normale}
Let $P_1$ and $P_2$ be two sets of paths such that their sets of endpoints $V_{P_1}$ and $V_{P_2}$ are disjoint. The two families $\left(h( p)\right)_{p \in P_1}$ and $\left(h( p)\right)_{p \in P_2}$, defined on $\left(\mathcal{M}ult(P,G), \mathcal{B}, \mu\right)$, are $\mathcal{I}$-independent if and only if they are independent. 

We only have to prove that the $\mathcal{I}$-independence implies the independence. Let us suppose that they are $\mathcal{I}$-independent. If $f$ and $g$ are real-valued continuous functions on $G^{n}$ and $G^{m}$ respectively, we denote by $f \otimes g$ the function from $G^{n} \times G^{m}$ to $\mathbb{R}$ defined by: $$f \otimes g(x_{1}, ..., x_{n},x_{n+1}, ...,x_{n+m}) = f(x_{1},  ..., x_{n}) g(x_{n+1},...,x_{n+m}).$$ With this notation and the notation (\ref{moyennisee}), since the two families $P_1$ and $P_2$ have disjoint sets of endpoints, $$\widehat{\left(f\otimes g\right)}_{J_{P_1 \cup P_2}} = \widehat{f}_{J_{P_1}} \otimes \widehat{g}_{J_{P_2}},$$ 
where the partial gauge group was defined in Definition \ref{actiongauge}. 
Thus, using the gauge-invariance of $\mu$, 
\begin{align*}
\mu\left[ f \left((h(p ))_{p \in P_1}\right) g \left((h(p ))_{p \in P_2}\right) \right] 
&= \mu\left[\left(f\otimes g\right)\left((h(p ))_{p \in P_1},(h( p))_{p \in P_2}\right) \right] \\
&= \mu\left[\widehat{\left(f\otimes g\right)}_{J_{P_1 \cup P_2}} \left((h(p ))_{p \in P_1},(h( p))_{p \in P_2} \right) \right] \\
&= \mu\left[\widehat{f}_{J_{P_1}} \otimes \widehat{g}_{J_{P_2}} \left((h(p ))_{p \in P_1},(h( p))_{p \in P_2} \right) \right]\\
&= \mu\left[\widehat{f}_{J_{P_1}} \left((h(p ))_{p \in P_1} \right)\right] \mu \left[ \widehat{g}_{J_{P_2}} \left((h(p ))_{p \in P_2} \right)\right] \\
&= \mu\left[f \left((h(p))_{p \in P_1} \right)\right] \mu \left[g \left((h( p))_{p \in P_2} \right)\right].
\end{align*}
This proves that the two families $\left(h( p)\right)_{p \in P_1}$ and $\left(h( p)\right)_{p \in P_2}$ are independent. 
\end{remarque}

  Let us introduce the main ingredient in order to construct gauge-invariant random holonomy fields: the {\em loop paradigm} for multiplicative functions. From now on, $P$ will be connected, stable by concatenation and inversion, $m$ is an endpoint of $P$ and we recall that $L_m$ is the set of loops in $P$ based at $m$. 
  \begin{lemme}
  \label{multpara}
The {\em loop paradigm} for the multiplicative functions is: 
\begin{align}
\label{looppara}
\mathcal{M}ult(P,G)/J_{P} \simeq \mathcal{M}ult_{P}(L_m,G)/J_{L_m}.
\end{align}
\end{lemme}

\begin{proof}
There exists a natural restriction function: 
\begin{align*}
r: \mathcal{M}ult(P,G)/J_{P} \to  \mathcal{M}ult_{P}(L_m,G)/J_{L_m}. 
\end{align*}
Let us show that there exists an application 
\begin{align*}
\iota: \mathcal{M}ult_{P}(L_m,G)/J_{L_m} \to \mathcal{M}ult(P,G)/J_{P}, 
\end{align*}
such that $r \circ \iota = id$ and $\iota \circ r=id$.  

The proof uses the ideas used in order to prove Lemma 2.1.5 of \cite{Levy}. For any endpoint $v$ of $P$, let $q_{v}$ be a path in $P$ joining $m$ to $v$. This is possible since we supposed that $P$ was connected. We set $q_{m}$ to be the trivial path. Then, for any path $p$ in $P$ we define $l( p) = q_{\underline{p}}p q_{\overline{p}}^{-1}$. One can look at the Figure \ref{Construcfig} to have a better understanding of $l( p)$. For any $h$ in $\mathcal{M}ult_P(L_m, G)$, we define for any path $p$, $$\iota(h)( p) = h( l ( p ) ). $$
This is a multiplicative function. Let us show, for example, that it is compatible with the concatenation operation. For any $h \in \mathcal{M}ult_P(L_m, G)$ and any paths $p$ and $p'$ in $P$ such that $\overline{p} = \underline{p'}$, the following sequence of equalities holds: 
\begin{align*}
\iota(h)( p p') &= h\left(l( p p' )\right) = h( q_{\underline{p}}p p' q_{\overline{p'}}^{-1}) = h( q_{\underline{p}}p q_{\overline{p}}^{-1} q_{\underline{p'}} p' q_{\overline{p'}}^{-1}) \\&= h( q_{\underline{p'}} p' q_{\overline{p'}}^{-1}) h( q_{\underline{p}}p q_{\overline{p}}^{-1}) =h\left(l(p')\right) h\left(l(p )\right)= \iota(h) (p') \iota(h)(p ), 
\end{align*}
where in the third equality we used the fact that $h$ is an element of $\mathcal{M}ult_{P}(L_m,G)$ and not only in $\mathcal{M}ult(L_m,G)$. 

Thus $\iota$ is an application from $\mathcal{M}ult_P(L_m,G)$ to $\mathcal{M}ult(P,G)$. This application $\iota$ defines a function, that we will also call $\iota$ from $\mathcal{M}ult_{P}(L_m,G)/J_{L_m}$ to $\mathcal{M}ult(P,G)/J_{P}$. Indeed, if $j \in J_{L_m}\simeq G$, $h \in \mathcal{M}ult_P(L_m,G)$ and $p \in P$: 
\begin{align*}
\iota(j \bullet h) (p ) = j \bullet h(l( p)) = j(m)^{-1} h(l( p)) j(m) = \tilde{j} \bullet \iota(h) (p ), 
\end{align*} 
where $\tilde{j}$ is the constant function which is equal to $j$. 
Let us show that $\iota \circ r = id$: for any $h\in \mathcal{M}ult(P,G)$,
\begin{align*}
\left(\iota(r(h))( p ) \right)_{p \in P}= r(h)(l( p)) &=\left( h( q_{\underline{p}} p q_{\overline{p}}^{-1}) \right)_{p \in P}\\&= \left( h(q_{\overline{p}})^{-1} h( p) h(q_{\underline{p}}) \right)_{p \in P},
\end{align*}
thus, in $\mathcal{M}ult(P,G)/J_P$, we have the equality $\left(\iota(r(h))( p ) \right)_{p \in P} = (h(p))_{p \in P}$. The equality $r \circ \iota = id$ is even easier. 
\end{proof}

\begin{figure}
 \centering
  \includegraphics[width=130pt]{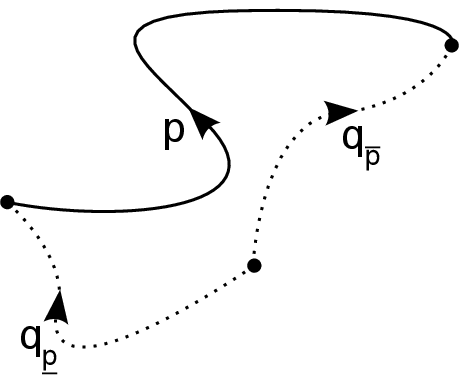}
 \caption{Construction of the loop $l(p )$.}
 \label{Construcfig}
\end{figure}
 
From the proof of Lemma \ref{multpara}, one also gets the following lemma. 
\begin{lemme}
\label{inv3}
There exists an application: 
\begin{align*}
\iota: \mathcal{M}ult_P(L_m, G) \to \mathcal{M}ult(P, G)
\end{align*}
which is measurable for the Borel $\sigma$-field and such that, for any loop $l\in L_m$, the following diagram is commutative: 
\begin{align*}
 \xymatrix{
  {\mathcal{M}ult_P(L_m, G)} \ar[rr]^\iota \ar[dr]_{h(l)} && \mathcal{M}ult(P, G) \ar[dl]^{h(l)} \\
  & \ G &&
 }
 \end{align*}
\end{lemme}

An other consequence of Lemma \ref{multpara} is Lemma $2.1.5$ in \cite{Levy} given below. 

\begin{lemme}
\label{inv1}
For any paths $c_1, ..., c_n$ in $P$ and any measurable function $f: G^{n} \to \mathbb{R}$ invariant under the action of $J_{c_1, ..., c_n}$ on $G^{n}$, there exist  $n$ loops $l_1, ..., l_n$ in $P$ based at $m$ and a measurable function $\tilde{f}: G^{n} \to \mathbb{R}$ invariant under the diagonal action of $G$ such that:
\begin{align*}
f\left(h(c_1), ..., h(c_n)\right) = \tilde{f}\left(h(l_1), ..., h(l_n)\right).
\end{align*}
\end{lemme}

Lemma \ref{inv1} allows us to reduce the family of variables that $\mathcal{I}$ has to make measurable: we only have to look at finite collections of loops based at the same point. This leads us to the Definition $2.1.6$ of \cite{Levy}, which, in our case, is a result and not a definition.

\begin{proposition}
The invariant $\sigma$-field $\mathcal{I}$ on $\mathcal{M}ult(P,G)$ is the smallest $\sigma$-field such that for any positive integer $n$, any loops $l_{1},  ..., l_{n}$ based at $m$ and any continuous function $f: G^{n} \to \mathbb{R}$ invariant under the diagonal action of $G$, the mapping $h \mapsto f(h(l_{1}),  ..., h(l_{n}))$ is measurable. 
\end{proposition}

Another consequence of Lemma \ref{inv1} is the following proposition.

\begin{proposition}
\label{unicite1}
If $\mu$ and $\nu$ are two stochastically continuous gauge-invariant random holonomy fields on $P$, the two following assertions are equivalent: 
\begin{enumerate}
\item $\mu$ and $\nu$ are equal, 
\item there exist an endpoint $m$ of $P$ and $A_m$ a dense subset of $L_m$ for the convergence with fixed endpoints, such that for any integer $n$, any $n$-tuple of loops $l_{1},  ..., l_{n}$ in $A_m$ and any continuous function $f: G^{n} \to \mathbb{R}$ invariant under the diagonal action of $G$,  
\begin{align*}
\int_{\mathcal{M}ult(P, G)}f\left(h(l_{1}),  ..., h(l_{n})\right) d\mu(h) = \int_{\mathcal{M}ult(P, G)}f\left(h(l_{1}),  ..., h(l_{n})\right) d\nu(h). 
\end{align*}
\end{enumerate}
\end{proposition}

If the random holonomy fields are not stochastically continuous, the proposition still holds if one replaces $A_m$ by $L_m$.

\begin{remarque}
\label{changementdepoint}
The first consequence of this proposition is the {\em change of base point invariance property} of gauge-invariant random holonomy fields. For the sake of simplicity, let us consider $\mu$ a gauge-invariant random holonomy field on $M$. Let us consider a bijection $\psi: M \to M$ and let us consider for any point $x$ of $M$, $p_x$ a path from $\psi(x)$ to~$x$. Then the random holonomy field which has the law of: 
\begin{align*}
\left(h(p_{x})\right)_{x \in M} \bullet \left[ h( p)\right]_{p \in P}
\end{align*}
under $\mu$, is still gauge-invariant and the last proposition shows that $\mu$ and the new random holonomy field are equal. Thus, for any paths $p_1, ... p_n $, we have the equality in law: 
\begin{align*}
\left(\left[\left(h(p_{x})\right)_{x \in M} \bullet \left[ h( p)\right]_{p \in P}\right] (p_i) \right)_{i=1}^{n}= \left(h(p_i)\right)_{i=1}^{n}\end{align*}
under $\mu$. For example, if $l_1, ..., l_n$ are $n$ loops based at $m$ and if $s$ is a path from $m'$ to $m$, under a gauge-invariant measure $\mu$, $\left(h(sl_1s^{-1}), ..., h(sl_ns^{-1})\right)$ has the same law as $\left(h(l_1), ..., h(l_n)\right)$. 
\end{remarque}

\subsection{Construction of random holonomy fields II: the gauge-invariant case }
\label{subsectionconstruction}
In this section, for the sake of simplicity, we will suppose that $M$ is connected. However, all results could be easily extended to the non-connected case. Thanks to Lemma \ref{multpara} and Proposition \ref{exten}, constructing a gauge-invariant random holonomy field $\mu$ becomes easier. We recall that $P$ is a connected set of paths, stable by concatenation and inversion. 

\begin{proposition}
\label{crea1}
Let $m$ be an endpoint of $P$. Suppose that for any finite subset $L$ of loops in $P$ based at $m$, we are given a gauge-invariant measure $\mu_L$ on $\mathcal{M}ult_P(L, G)$ such that, when endowed with the natural restriction functions, $\left((\mathcal{M}ult_P(L, G),\mathcal{B}), \mu_L\right)$ is a projective family. Then there exists a unique gauge-invariant random holonomy field $\mu$ on $P$ such that for any finite subset  $L$  of loops in $P$ based at $m$, one has: $$\mu_L = \mu \circ \rho_{L, P}^{-1}.$$
\end{proposition}

\begin{proof}
The uniqueness of such a measure comes from a direct application of Proposition \ref{unicite1}. 

Let us prove the existence of the measure $\mu$. Let $L_m$ be the set of loops in $P$ based at $m$.  Using a slight modification of Proposition \ref{limitproj}, we can consider the projective limit $\mu_{L_m}$ of $(\mu_{L_f})_{L_f \subset L_m, \#L_f<\infty}$, defined on $\left(\mathcal{M}ult_P(L_m, G),\mathcal{B}\right)$ and which is gauge-invariant. The set $P$ satisfies the assumptions of Lemma \ref{inv3}: let us consider a measurable application $\iota$ from $\mathcal{M}ult_P(L_m, G)$ to $\mathcal{M}ult(P, G)$ given by this lemma. We define the measure: 
\begin{align*}
\mu = \left((\mu_{L_m}\circ\iota^{-1})_{\mid\mathcal{I}}\right)^{\widehat{}}, 
\end{align*}
where we remind the reader that $(\ )^{\ \!\widehat{}}\ $ is the notation for the extension of measures from the invariant $\sigma$-field to the Borel $\sigma$-field given by Proposition \ref{exten}. By definition, it is defined on the Borel $\sigma$-field on $\mathcal{M}ult(P, G)$ and it is gauge-invariant. 

If $L$ is a finite subset of loops in $P$ based at $m$, thanks to the definitions of $\iota$ and $\mu_{L_m}$, $(\mu_{L_m} \circ \iota^{-1}) \circ \rho^{-1}_{L, P} = \mu_{L}$. The gauge-invariance of $\mu_L$ implies that $\left({(\mu_L)}_{\mid \mathcal{I}}\right)^{\widehat{}}=\mu_L$. This leads us to the conclusion: $\mu_L = \mu \circ \rho_{L, P}^{-1}$.
\end{proof}

In particular, if we combine this proposition with Theorem~\ref{exten2}, we get the following result. 
\begin{proposition}
\label{exten4}
Let $\gamma$ be a Riemannian metric on $M$, let $m$ be a point of $M$. Suppose that for any finite subset $L$ of loops in ${\sf Aff}_\gamma(M)$ based at $m$, we are given a gauge-invariant measure $\mu_L$ on $\mathcal{M}ult_P(L, G)$ such that $\left(\left(\mathcal{M}ult_P(L, G),\mathcal{B}\right), \mu_L\right)$ is a projective family of uniformly locally stochastically $\frac{1}{2}$-H\"{o}lder continuous random holonomy fields. Then there exists a unique stochastically continuous gauge-invariant random holonomy field $\mu$ on $M$ such that for any finite subset $L$ of loops in $P(M)$ based at $m$, one has: $$\mu_L = \mu \circ \rho_{L, P(M)}^{-1}.$$
\end{proposition}

\subsection{Restriction and extension of the structure group}
\label{restricexten1}
Let $H$ be a closed subgroup of $G$. There exists a natural injection $i_{P}: \left(\mathcal{M}ult(P, H),\mathcal{B}\right) \to \left(\mathcal{M}ult(P, G),\mathcal{B}\right).$ Thus, we can always push forward any $H$-valued random holonomy field by $i_{P}$ in order to define a $G$-valued random holonomy field. Of course, if a $G$-valued random holonomy field on $P$, say $\mu$, is such that there exists a closed group $H\subset G$ such that for any path $p \in P$, one has $h( p) \in H$ $\mu$ a.s., then we can restrict the group to $H$: for any finite $P_f \subset P$ it defines a measure on $\mathcal{M}ult(P_f, H)$ and we can take the projective limit thanks to Proposition $\ref{limitproj}$.

In the gauge-invariant setting, what can be done ? First of all, if $\mu$ is a $H$-valued gauge-invariant random holonomy field, $\mu \circ i_P^{-1}$ is not in general a $G$-valued gauge-invariant random holonomy field. The simplest counterexample is to consider $P$ to be reduced to a single loop: a $G$-valued random variable can be $H$-invariant but not $G$-invariant by conjugation. Thus, in order to extend the structure group from $H$ to $G$ of a $H$-gauge-invariant random holonomy field $\mu$, one has to consider: 
\begin{align}
\label{eqexten}
\left((\mu \circ i_P^{-1})_{\mid \mathcal{I}}\right)^{\widehat { }}
\end{align}
the gauge-invariant extension (see Proposition \ref{exten}) to $\mathcal{B}$ of the restriction on the invariant $\sigma$-field $\mathcal{I}$ of $\mu \circ i_P^{-1}$. Thus, the natural injection is replaced by the following map: 
\begin{align*}
\mu \mapsto \left((\mu \circ i_P^{-1})_{\mid \mathcal{I}}\right)^{\widehat { }}.
\end{align*}

\begin{notation}
\label{rest}
In the following, we will denote $\mu \circ \hat{i}_{P}^{-1}:=\left((\mu \circ i_P^{-1})_{\mid \mathcal{I}}\right)^{\widehat { }}$.
\end{notation}

Now, let us consider the problem of restricting a gauge-invariant random holonomy field $\mu$. Thanks to  Lemma \ref{inv1}, we know that the only important objects are loops based at $m$. Hence the question: what can be done with a $G$-valued random holonomy field such that for any loop or for any simple loop $l \in L_m$, $\mu$ a.s. $h(l) \in H$? An important remark is that it does not imply that for any path $p \in P$, $\mu$ a.s. $h( p) \in H$. Indeed, as we have seen in Remark \ref{haarchemin}, for any $p$ such that $\overline{p}\ne \underline{p}$, under $\mu/\mu({\mathbbm{1}})$, $h(p)$ has the law of a Haar random variable on $G$. Nevertheless, the following result is true. 

\begin{proposition}
\label{restrictionholo}
Let $\mu$ be a $G$-valued gauge-invariant random holonomy field such that for any loop $l \in L_m$, $h(l) \in H$, $\mu$ a.s. Then there exists an $H$-valued gauge-invariant random holonomy field $\mu_H$ such that: 
\begin{align*}
\mu = \mu_H \circ \hat{i}_P^{-1} . 
\end{align*}
Let $M$ be a smooth connected compact surface and let us suppose that $P$ is $P(M)$ and that $\mu$ is stochastically continuous. The result is still true if for any lasso $l$ based at $m$, $h(l) \in H$, $\mu$ a.s.
\end{proposition}

\begin{remarque}
\label{notunicity}
An important remark is that $\mu_H$ is {\em not} unique. Besides, using the group of reduced loops (Section 2.4 of \cite{Levy} and the forthcoming Section \ref{sec:generRL}), one can show in the last case that it is enough that $h(l) \in H$, $\mu$ a.s., for any simple loop $l$ based at $m$. This is due to the fact that for any graph, there exists a family of generators of the group of reduced loops which can be approximated, for the convergence with fixed endpoints, by simple loops.
\end{remarque}

We give below the loop-erasure lemma, taken from Proposition 1.4.9 in \cite{Levy}, that we will use in the proof of Proposition \ref{restrictionholo}.

\begin{lemme}
\label{looper}
Let $(M,\gamma)$ be a Riemannian compact surface and let $c$ be a loop in ${\sf Aff}_{\gamma} (M)$. There exists in ${\sf Aff}_{\gamma} (M )$ a finite sequence of lassos $l_1$,...,$l_p$ and a simple loop~$d$ with the same endpoints as $c$ such that: $$c \simeq l_1...l_p d.$$
\end{lemme}

\begin{proof}[Proof of Proposition \ref{restrictionholo}]
Let us prove the second case, when $P = P(M)$ and $\mu$ is stochastically continuous. The first assertion is easier and can be proved using the second part of the proof. 

Let us suppose that for any simple lasso $l \in L_m(M)$, $h(l) \in H$, $\mu$ a.s. Let us consider $\gamma$ a Riemannian metric on $M$. As a consequence of Lemma \ref{looper}, for any loop $l \in {\sf Aff}_\gamma(M)$ based at $m$, $h(l) \in H$, $\mu$ a.s. Thus, by Lemma \ref{dense}, using the stochastic continuity of $\mu$ and the fact that $H$ is closed, for any $l \in L_m(M)$, $h(l) \in H$, $\mu$ a.s.

By restricting the measure $\mu$, one can define, for any finite subset ${L_f}$ of $L_m(M)$, a gauge-invariant measure $\mu_{L_f}$ on $\mathcal{M}ult_{P(M)}(L_{f},H)$. As a consequence of Proposition \ref{crea1}, there exists a unique $H$-valued gauge-invariant random holonomy field $\mu_{H}$ on $M$ such that for any finite subset ${L_f}$ of $L_m(M)$, $\mu_{L_f} = \mu_{H} \circ \rho_{L_f, P(M)}^{-1}.$ This $H$-valued gauge-invariant random holonomy field $\mu_{H}$ satisfies the equality: $\mu = \mu_{H} \circ \hat{i}_{P(M)}^{-1}$.
\end{proof}

\chapter{Graphs}

\label{Sectiongraph}
\section{Definitions and simple facts}

The construction of special random fields, the planar Markovian holonomy fields, uses the notion of graphs. The graphs we consider are not only combinatorial ones: we insist that the faces are homeomorphic to an open disk of $\mathbb{R}^{2}$. Let $M$ be a either a smooth compact surface with boundary or the plane $\mathbb{R}^{2}$. 

\begin{definition}
\label{ref:pregraph}
A pre-graph on $M$ is a triple $\mathbb{G} = (\mathbb{V}, \mathbb{E}, \mathbb{F})$ such that: 
\begin{itemize}
\item $\mathbb{E}$, the set of edges, is a non-empty finite set of simple paths on $M$, stable by inversion, such that two edges which are not each other's inverse meet, if at all, only at some of their endpoints, 
\item $\mathbb{V}$, the set of vertices, is the finite subset of $M$ given by $\bigcup\limits_{e \in \mathbb{E}} \{\overline{e}, \underline{e}\}$,
\item $\mathbb{F}$, the set of faces, is the set of the connected components of $M \setminus \bigcup\limits_{e \in \mathbb{E}} e\big([0,1]\big)$.
\end{itemize}
Any pre-graph $\mathbb{G}=(\mathbb{V}, \mathbb{E}, \mathbb{F})$ whose {\em bounded} faces $F \in \mathbb{F}$ are homeomorphic to an open disk of $\mathbb{R}^{2}$ is called a graph on $M$. 
\end{definition}

\begin{remarque}
By Proposition 1.3.10 in \cite{Levy}, if $\mathbb{G}$ is a graph on $M$ then $\partial M$ can be represented by a concatenation of edges in $\mathbb{E}$. 
\end{remarque}

Due to the last definition, any pre-graph $\mathbb{G} = (\mathbb{V}, \mathbb{E}, \mathbb{F})$ is characterized by its set of edges $\mathbb{E}$. Thus, in order to construct a pre-graph, we will only define its set of edges. We will often use the following graph. 
 \begin{exemple}
 \label{G(l)}
Let $l$ be a simple loop on $\mathbb{R}^{2}$. We denote by $\mathbb{G}(l)$ the graph on $\mathbb{R}^{2}$ composed of $l$ and $l^{-1}$ as unique edges. 
 \end{exemple}
When $M$ is homeomorphic to a sphere, we will consider that $\left(\{m\}, \emptyset, M\setminus\{m\}\right)$ is a graph for any $m \in M$. 

\begin{definition} A graph is connected if and only if any two points of $\mathbb{V}$ are the endpoints of the same path in $P(\mathbb{G})$. A connected graph on $\mathbb{R}^{2}$ will be also called a {\em finite planar graph}; its set of faces is composed of one unbounded face denoted by $F_{\infty}$ and a set $\mathbb{F}^{b}$ of bounded faces. 
\end{definition}

\begin{definition}
Let $\mathbb{G}$ be a graph on M, $P(\mathbb{G})$ is the set of paths obtained by concatenating edges of $\mathbb{G}$. The set of loops in $P(\mathbb{G})$ is denoted by $L(\mathbb{G})$ and if $v\in \mathbb{V}$, $L_{v}(\mathbb{G})$ is the set of loops in $L(\mathbb{G})$ based at $v$. 
\end{definition}

For any smooth connected compact surface with boundary $M$ embedded in $\mathbb{R}^{2}$, a graph on $M$ can be considered as a finite planar graph. This kind of graphs, of interest later, will be called embedded graphs on $\mathbb{R}^{2}$. 

\begin{definition}
An embedded graph on $\mathbb{R}^{2}$ is a graph on a smooth connected compact surface with boundary $M$ embedded in $\mathbb{R}^{2}$.
\end{definition}

The two definitions of graphs on $\mathbb{R}^{2}$ seen here are in fact almost equivalent. An embedded graph is obviously a graph on $\mathbb{R}^{2}$ and a direct consequence of Propositions $1.3.24$ and $1.3.26$ of \cite{Levy} is the following result. 

\begin{proposition}
\label{embfinit}
Every finite planar graph on $\mathbb{R}^{2}$ is a subgraph of an embedded graph.
\end{proposition}

The intersection of a graph $\mathbb{G} = (\mathbb{V}, \mathbb{E}, \mathbb{F})$ with a subset $A$ of $\mathbb{R}^{2}$ is the pre-graph $(\mathbb{V}', \mathbb{E}', \mathbb{F}')$, denoted by $\mathbb{G}\cap A$, such that $\mathbb{E}' = \{e \in \mathbb{E}, e \cap A \nsubseteq \{\underline{e}, \overline{e}\}\}$. This allows us to  define the notion of a planar graph. For any positive real $r$, let $\mathbb{D}(0,r)$ be the closed ball of center $(0,0)$ and radius $r$ in $\mathbb{R}^{2}$. 

\begin{definition}
\label{infiniteplanar}
A planar graph $\mathbb{G} = (\mathbb{V}, \mathbb{E}, \mathbb{F})$ is a triple of sets which represent the vertices, the edges and the faces which are linked by the same relations as in Definition \ref{ref:pregraph} and for which there exists an increasing unbounded sequence of positive reals $(r_n)_{n \in \mathbb{N}}$ such that for each integer $n$, $\mathbb{G}\cap \mathbb{D}(0,r_n)$ is a finite planar graph. 
\end{definition}

\begin{exemple}
\label{N2}
We consider $\mathbb{N}^{2}$ as a planar graph, the edges being the vertical and horizontal segments between nearest neighbors. 
\end{exemple}

\begin{notation}
Sometimes, one wants to consider connected graphs whose edges are in a given subset $A$ of $P(M)$. We denote by $\mathcal{G}(A)$ the set of connected graphs $\mathbb{G}=\{\mathbb{V}, \mathbb{E}, \mathbb{F} \}$ such that $\mathbb{E} \subset A$. 
\end{notation}

In the notions of graph exposed above, the edges are non-oriented, which means that there is no preference between $e$ and $e^{-1}$ for any edge $e$. 
\begin{definition}
An orientation on a graph $\mathbb{G}$ is the data of a subset $\mathbb{E}^{+}$ of $\mathbb{E}$ such that $\mathbb{E}^{+}\cap \left(\mathbb{E}^{+}\right)^{-1} = \emptyset$ and $\mathbb{E}^{+} \cup (\mathbb{E}^{+})^{-1} = \mathbb{E}.$ Given an orientation $\mathbb{E}^{+}$ on $\mathbb{G}$, for each subset $J$ of $\mathbb{E}$, we denote by $J^{+}$ the set $J\cap \mathbb{E}^{+}$. 
\end{definition}

\section{Graphs and homeomorphisms}

In the following we will need to understand the action of orientation-preserving homeomorphisms on the set of graphs. 

\begin{definition}
Let $\mathbb{G}$ and $\mathbb{G}'$ be two finite planar graphs. They are homeomorphic if there exists an orientation-preserving homeomorphism $\psi$ which sends $\mathbb{G}$ on $\mathbb{G}'$. We will denote it by $\psi(\mathbb{G}) = \mathbb{G}'$ and by definition, this means that $\psi$ induces a bijection $\mathcal{S}^{\psi}_{\mathbb{G}}$ from the set $\mathbb{V}$ of vertices of $\mathbb{G}$ to the set  $\mathbb{V}'$ of vertices of $\mathbb{G}'$ and a bijection $\mathcal{E}^{\psi}_{\mathbb{G}}$ from the set $\mathbb{E}$ of edges of $\mathbb{G}$ to the set  $\mathbb{E}'$ of edges of $\mathbb{G}'$. These bijections are defined by: 
\begin{align*}
\mathcal{S}^{\psi}_{\mathbb{G}}(v) = \psi(v) \text{, for any } v\in \mathbb{V},\\
\mathcal{E}^{\psi}_{\mathbb{G}} (e) = \psi(e) \text{, for any }e \in \mathbb{E}. 
\end{align*}
\end{definition}

\begin{definition}
\label{equ-home}
Let $\psi$ and $\psi'$ be two  orientation-preserving homeomorphisms of $\mathbb{R}^{2}$ which send $\mathbb{G}$ on $\mathbb{G}'$. The homeomorphisms $\psi$ and $\psi'$ are equivalent on $\mathbb{G}$ if and only if $\mathcal{S}^{\psi}_{\mathbb{G}} = \mathcal{S}^{\psi'}_{\mathbb{G}}$ and $\mathcal{E}^{\psi}_{\mathbb{G}}  = \mathcal{E}^{\psi'}_{\mathbb{G}}$. 
\end{definition}

We would like to have an easy way to know if two finite planar graphs are homeomorphic. For that, an important notion is the cyclic order of the outgoing edges at a vertex. 

\begin{definition}
Let $\mathbb{G}=(\mathbb{V}, \mathbb{E}, \mathbb{F})$ be a finite planar graph. Let $v$ be a vertex and let $\mathbb{E}_v$ be the set of edges $e\in \mathbb{E}$ such that $\underline{e} = v$. For any $e \in \mathbb{E}_v$, let $e_p$ be a parametrized path which represents $e$. We define: $$r_0 = \min \left\{ \bigg|\bigg| v - e_p\left(\frac{1}{2}\right)\bigg|\bigg|, e \in \mathbb{E}_v\right\}.$$ 
Let $r \in ]0,r_0[$. For each $e \in \mathbb{E}_v$, we define $s_{e}( r) \in \left[0,\frac{1}{2}\right]$ as the first time $e_p$ hits the boundary of $\mathbb{D}(0,r)$: 
\begin{align*}
s_e( r) = \inf \left\{ t \in \left[0,\frac{1}{2}\right],  || v - e_p\left(t\right) ||= r\right\}.
\end{align*}
The cyclic permutation of $\mathbb{E}_v$, corresponding to the cyclic order of the points $\left\{ e_p(s_e(r )), e \in~\mathbb{E}_v\right\}$ on the circle $\partial \mathbb{D}(0,r)$ oriented anti-clockwise, does not depend on the chosen $r \in ]0,r_0[$: it is the cyclic order of the edges at the vertex $v$ denoted by $\sigma_v$.\end{definition}

A consequence of Jordan-Sch\"{o}nfliess theorem is the Heffter-Edmonds-Ringel rotation principle, stated in Theorem $3.2.4$ of \cite{MoharTommassen}.

\begin{theorem}
\label{homeograph}
Let $\mathbb{G} = (\mathbb{V}, \mathbb{E}, \mathbb{F})$ and $\mathbb{G}' =(\mathbb{V}', \mathbb{E}', \mathbb{F}')$ be two finite planar graphs such that the following assertions hold: 
\begin{enumerate}
\item there exists a bijection $\mathcal{S}: \mathbb{V} \to \mathbb{V}'$, 
\item there exists a bijection $\mathcal{E}: \mathbb{E} \to \mathbb{E}'$ such that for any $e \in \mathbb{E}$, $\mathcal{E}\left(e^{-1}\right)= \mathcal{E}(e)^{-1}$,
\item for any edge $e \in \mathbb{E}$, $\mathcal{S} \big(\underline{e}\big) = \underline{\mathcal{E}(e)}$,
\item for any vertex $v \in \mathbb{V}$, $\sigma_{\mathcal{S}(v)} = \mathcal{E} \circ \sigma_{v} \circ \mathcal{E}^{-1}$.
\end{enumerate}
Then there exists an orientation-preserving homeomorphism $\psi: \mathbb{R}^{2} \to \mathbb{R}^{2}$ such that $\psi(\mathbb{G}) = \mathbb{G}'$ and $\psi$ induces the two bijections $\mathcal{S}$ and $\mathcal{E}$. 
\end{theorem}

If one considers only piecewise affine edges, the theorem can be applied to pre-graph with affine edges. 

Later we will need the notion of diffeomorphisms at infinity. The motivation will appear in Lemma \ref{area-preserv-exp} where we show that the free boundary condition expectation on the plane associated with a Markovian holonomy field is a planar Markovian holonomy field.  In the following definition, $\mathbb{D}(0,R)^{c}$ is the complement set of the closed disk centered at $0$ and of radius $R$. 
  
\begin{definition}
\label{diffeoatinf}
A homeomorphism $\psi$ of $\mathbb{R}^{2}$ is a diffeomorphism at infinity if there exists a real $R$ such that $\psi_{\mid \mathbb{D}(0,R)^{c}}$ is a diffeomorphism. 
\end{definition}

Each time we consider a homeomorphism from a close domain delimited by a Jordan curve to an other domain delimited by an other Jordan curve, we can extend it as a diffeomorphism at infinity. Indeed, using the Carath\'{e}odory's theorem for Jordan curves, we can suppose that both domains are the unit disk. In this case, the result follows from the following lemma. 

\begin{lemme}
\label{quasi-diff}
Let $\mathbb{D}$ be the closed disk of center $0$ and radius $1$. Let $\psi: \partial \mathbb{D} \to \partial \mathbb{D}$ be a homeomorphism. There exists a diffeomorphism $\Psi: \mathbb{D}^{c} \to \mathbb{D}^{c}$ such that for any $x \in \partial \mathbb{D}$, $$\underset{y \to x}{\lim} \Psi(y) = \psi(x).$$
Besides, if $\psi$ preserves the orientation, $\Psi$ will also preserve the orientation. 
\end{lemme}

\begin{proof}
Let $\eta$ be a smooth even positive function supported on $[-1,1]$. Let us consider for any real $r>1$ the function $\eta_{r}(.) = (r-1)^{-1}\eta((r-1)^{-1}.)$. The family $(\eta_{r})_{r > 1}$ is a smooth even approximation to the identity when $r$ goes to $1$. 

There is a natural bijection $\Phi$ between the set of homeomorphisms of $\partial \mathbb{D}$ and the set $Hom^{\mathbb{R}}_{\partial\mathbb{D}}$ of strictly increasing or decreasing continuous functions $f$ from $\mathbb{R}$ to $\mathbb{R}$ such that $f-{\sf Id}$ is $1$-periodic. Let $\psi: \partial \mathbb{D} \to \partial \mathbb{D}$ be a homeomorphism of the circle. 
We define the smooth function $\Psi$ by: 
\begin{align*}
\Psi:\ \ \ \  \mathbb{D}^{c} &\to \mathbb{D}^{c}\\
re^{2i\pi\theta} &\mapsto r e^{2i \pi (\Phi(\psi)*\eta_r) (\theta)}.
\end{align*}

Since $\psi$ is continuous on the disk, the function $\Phi(\psi)$ is uniformly continuous. Thus $\Phi(\psi) * \eta_r$ converges uniformly to $\Phi(\psi)$ as $r$ tends to $1$. This implies that for any $x \in \partial \mathbb{D}$, $\lim\limits_{y \to x} \Psi(y) = \psi(x)$. Besides, for any real $r > 1$, the convolution with $\eta_{r}$ sends $Hom_{\partial \mathbb{D}}^{\mathbb{R}}$ on itself: this  implies that $\Psi$ is bijective. Since $\Psi$ is differentiable, it remains to show that the Jacobian of $\Psi$ is strictly positive. Yet, for any $x \in \mathbb{D}^{c}$, only the module of $x$ is involved in the calculation of the module of $\Psi(x)$: the Jacobian matrix is triangular. Since $\eta_r$ is even for any $r > 1$ and $\Phi(\psi)$ is strictly increasing (or decreasing), the derivative of $\Phi(\psi) *\eta_r$ is strictly positive (or negative). These two facts imply that the Jacobian matrix of $\Psi$ is invertible, thus the function $\Psi$ is a diffeomorphism. The last assertion about the orientation-preserving property is straightforward. 
\end{proof}

\section{Graphs and partial order}
The graphs with  piecewise affine edges are interesting when one considers a special partial order on graphs studied in \cite{Levy}. 
\begin{definition}
Let $\mathbb{G}$ and $\mathbb{G}'$ be two planar graphs. We say that $\mathbb{G}'$ is finer than $\mathbb{G}$ if $P(\mathbb{G}) \subset P(\mathbb{G}')$. We denote it by $\mathbb{G} \preccurlyeq \mathbb{G}'$. 
\end{definition}

In fact, in Lemma 1.4.6 of \cite{Levy}, L\'{e}vy showed that this partial order is not directed. Yet, one can, by restricting it to a dense subspace of graphs, make it directed: for this, the edges of the graphs which we consider must be in a good subspace as defined below. 

\begin{definition}
Let $P$ be a subset of $P(M)$. A good subspace $A$ of $P$ is a dense subset of $P$ for the convergence with fixed endpoints such that for any finite subset $\{c_1, ..., c_n\}$ of $A$ there exists a graph $\mathbb{G}$ such that $\{c_1, ..., c_n\} \subset P(\mathbb{G})$.
\end{definition}

If $A$ is a good subspace, $\mathcal{G}(A)$ endowed with $\preccurlyeq$ is directed. The following lemma is a reformulation of Proposition $1.4.8$ of \cite{Levy}. 

\begin{lemme}
\label{goodspace}
For any Riemannian metric $\gamma$ on $M$, the set ${\sf Aff}_\gamma(M)$ is a good subspace for $P(M)$. 
\end{lemme}

There are other natural examples of good subspaces of $P(M)$. For example, Baez in \cite{Baez} used the good subspace of piecewise real-analytic paths in $P(\mathbb{R}^{2})$ in order to define the Ashtekar and Lewandowski uniform measure. Another example of good subspace is used in the articles \cite{Sen92} and \cite{Sen97a}.

By definition, any path in $M$ can be approximated by a sequence of paths in $A$ if $A$ is a good subspace. But ${\sf Aff}_{\gamma}(M)$ satisfies the stronger property which roughly asserts that $\mathcal{G}({\sf Aff}_{\gamma}(M))$ is ``dense'' for a certain notion in the set of planar graphs. The next theorem is a direct consequence of Proposition $1.4.10.$ in \cite{Levy}. It has to be noticed that, in the proof of Proposition $1.4.10.$ in \cite{Levy}, the measure of area does not have to be the measure of area associated with the chosen Riemannian metric. For the next theorem, let us suppose that $M$ is an oriented compact surface with boundary. 

\begin{theorem}
\label{approxgraph}
Let $\mathbb{G}=(\mathbb{V}, \mathbb{E}, \mathbb{F})$ be a graph on $M$. Let $\gamma$ be a Riemannian metric on $M$ and let $vol$ be a measure of area on $M$. There exists a sequence of finite planar graphs $\big(\mathbb{G}_n=(\mathbb{V}_n, \mathbb{E}_n, \mathbb{F}_n)\big)_{n \in \mathbb{N}}$ in $\mathcal{G}\big({\sf Aff}_{\gamma}\left(M\right)\big)$ such that: 
\begin{enumerate}
\item for any integer $n$, there exists $\psi_n$ an orientation-preserving homeomorphism of $M$ such that $\psi_n(\mathbb{G}) = \mathbb{G}_n$. 
\item $\mathbb{V}_n = \mathbb{V}$,
\item for any edge $e \in \mathbb{E}$, $\psi_n(e)$ converges to $e$ for the convergence with fixed endpoints, 
\item for any face $F \in \mathbb{F}$, $vol(\psi_n(F)) \underset{n \to \infty}{\longrightarrow} vol(F)$. 
\end{enumerate}
\end{theorem}

Another interesting property of $\mathcal{G}({\sf Aff}(\mathbb{R}^{2}))$ is the fact that any generic finite planar graph with piecewise affine edges can be sent by a piecewise smooth application on a subgraph of the $\mathbb{N}^{2}$ planar graph. We will prove this in Section \ref{Univ}, but before, we need to gather a few facts about graphs and triangulations.

\section{Graphs and piecewise diffeomorphisms} 

\begin{definition}
Let $\mathbb{G}$ be a finite planar graph in $\mathcal{G}\big({\sf Aff}(\mathbb{R}^{2})\big)$. It is simple if the boundary of any face of $\mathbb{G}$ is a simple loop. It is a triangulation if any bounded face is a non degenerate triangle. 
\end{definition}

\begin{definition}
Let $\mathbb{G}$ be a finite planar graph in $\mathcal{G}\big({\sf Aff}(\mathbb{R}^{2})\big)$. 
A mesh of $\mathbb{G}$ is a simple graph $\mathbb{G}'$ in $\mathcal{G}\big({\sf Aff}(\mathbb{R}^{2})\big)$ such that $\mathbb{G} \preccurlyeq \mathbb{G}'$. 
A triangulation of $\mathbb{G}$ is a triangulation $\mathbb{T}$ such that $\mathbb{G} \preccurlyeq \mathbb{T}$ and the unbounded face of $\mathbb{T}$ is the unbounded face of~$\mathbb{G}$.  
\end{definition}

Two triangulations are homeomorphic if they are homeomorphic as finite planar graphs. 

\begin{definition}
Let $\mathbb{G}$ and $\mathbb{G'}$ be two finite planar graphs in $\mathcal{G}\big({\sf Aff}(\mathbb{R}^{2})\big)$. A homeomorphism $\phi: \mathbb{R}^{2} \to \mathbb{R}^{2}$ is a $\mathbb{G}-\mathbb{G}'$ piecewise diffeomorphism if the three following assertions hold: 
\begin{enumerate}
\item $\phi(\mathbb{G})= \mathbb{G}'$, 
\item there exists a mesh $\mathbb{G}_0$ of $\mathbb{G}$ (resp $\mathbb{G}_0'$ of $\mathbb{G}'$) such that $\phi(\mathbb{G}_0) = \mathbb{G}'_0$ and for any bounded face $F$ of $\mathbb{G}_0$, $\phi_{\mid F}: F \to \phi(F)$ is a diffeomorphism whose Jacobian determinant is bounded below and above by some strictly positive real numbers and can also be extended on the boundary of $F$, 
\item let $F_{\infty}$ be the unbounded face of $\mathbb{G}_0$. The application $\phi_{\mid F_{\infty}}: F_{\infty} \to \phi(F_{\infty})$ is a diffeomorphism. 
\end{enumerate}
We will say that $\mathbb{G}_0$ is a good mesh for $\phi$. 
\end{definition}

The piecewise diffeomorphisms we will construct will always be of the following form: they will be the extension (using Lemma \ref{quasi-diff} and the discussion before) of a piecewise affine homeomorphism from the interior of a piecewise affine Jordan curve to itself. Recall the definition of equivalence defined in Definition \ref{equ-home}.

\begin{proposition}
\label{triang}
Let $\mathbb{G}_1$ and $\mathbb{G}_2$ be two homeomorphic simple finite planar graphs with piecewise affine edges. Let us choose an  orientation-preserving homeomorphism $\phi: \mathbb{R}^{2} \to \mathbb{R}^{2}$ such that $\phi(\mathbb{G}_{1}) = \mathbb{G}_2$. There exist two triangulations, $\mathbb{T}_1$ of $\mathbb{G}_1$, $\mathbb{T}_2$ of $\mathbb{G}_2$ and an  orientation-preserving $\mathbb{G}_1-\mathbb{G}_2$ piecewise-diffeomorphism $\psi$ such that: 
\begin{enumerate}
\item $\mathbb{T}_1$ is a good mesh for $\psi$, 
\item $\psi$ and $\phi$ are equivalent on $\mathbb{G}_1$, 
\item $\psi(\mathbb{T}_1) = \mathbb{T}_2$. 
\end{enumerate}
Consequently, the set of  orientation-preserving $\mathbb{G}_1-\mathbb{G}_2$ piecewise diffeomorphisms is not empty. 
\end{proposition}

In order to prove this proposition, we will need the following result proved in the paper of Aronov-Seidel-Souvaine (\cite{Aronov}).

\begin{theorem}
\label{Aronov}
Let $Q_1$ and $Q_2$ be two simple $n$-gons, seen as planar graphs with $n$ vertices. Let us choose an  orientation-preserving homeomorphism $\psi$ which sends $Q_1$ on $Q_2$. Let $T_1$ (resp. $T_2$) be a triangulation of $Q_1$ (resp. $Q_2$). There exists $\hat{T}_1$ (resp. $\hat{T}_2$) a triangulation of $Q_1$ (resp. $Q_2$), finer than $T_1$ (resp. $T_2$) and an  orientation-preserving homeomorphism $\psi'$ such that $\psi$ and $\psi'$ are equivalent on $Q_1$ and $\psi'(\hat{T}_1)=\hat{T}_2$. 
\end{theorem}

Let $\mathbb{G}_1$, respectively $\mathbb{G}_2$, be a simple graph in $\mathcal{G}\big({\sf Aff}(\mathbb{R}^{2})\big)$ with only one face denoted $F_1$, respectively $F_2$. Let $\psi$ be an orientation-preserving homeomorphism which sends $\mathbb{G}_1$ on $\mathbb{G}_2$. Then there exists a positive integer $n$ such that  $\partial F_1$ and $\partial F_2$ can be seen as two $n$-gons such that, when one considers these $n$-gons as graphs, $\psi$ sends $\partial F_1$ on $\partial F_2$: in order to do so, it is enough to add some vertices on the boundaries of $F_1$ and $F_2$. This remark will allow us to apply Theorem \ref{Aronov} to the faces of simple planar graphs with piecewise affine edges. Let us remark also that the homeomorphism $\psi$ between $\hat{T_1}$ and $\hat{T_2}$ in Theorem $\ref{Aronov}$ can be chosen so that it is affine on each bounded face of $\hat{T_1}$. 

\begin{lemme}
\label{trianhom}
Let $T_1$ and $T_2$ be two triangulations in the plane. If they are homeomorphic, there exists a function $\psi$ defined on the union of the bounded faces of $T_1$ and affine on each bounded face of $T_1$ such that $\psi$ is an  orientation-preserving homeomorphism which sends $T_1$ on~$T_2$. 
\end{lemme}

\begin{proof}
Let $T_1$ and $T_2$ be two homeomorphic triangulations and let $\phi$ be an orientation-preserving homeomorphism of $\mathbb{R}^{2}$ which sends $T_1$ on $T_2$. For any bounded face $F$ of $T_1$, we can find an orientation-preserving  affine map $\psi_{\mid F}$, defined on $F$, such that $\psi_{\mid F}$ and $\phi$ are equivalent on the border $\partial F$, seen as a graph with $3$ vertices. This map is actually unique.

Let us remark that for any triangle $T$, any $x \in T$ and any affine map $F$, $F(x)$ depends only on the image by $F$ of the edge which contains $x$. This allows us to glue the affine maps $\big(\psi_{\mid F}\big)_{F}$ and to get the desired $\psi$.
\end{proof}

We can now prove Proposition \ref{triang}. 

\begin{proof}[Proof of Proposition \ref{triang}]
Let $\mathbb{G}_1$ and $\mathbb{G}_2$ be two simple homeomorphic finite planar graphs with piecewise affine edges. Let $\phi$ be an orientation-preserving homeomorphism such that $\phi(\mathbb{G}_{1}) = \mathbb{G}_2$. For any bounded face $F$ of $\mathbb{G}_1$, $F$ and $\phi(F)$ are simple polygons. 
As any polygon can be triangulated, one consequence of Theorem \ref{Aronov} and Lemma \ref{trianhom} is that there exists $\mathbb{T}_{1,F}$ (resp. $\mathbb{T}_{2,F}$) a triangulation of $F$ (resp. $\phi(F)$) and $\psi_{\mid F}$ a function defined on $F$, affine on each bounded face of $\mathbb{T}_{1,F}$, such that $\psi_{\mid F}$ is an orientation-preserving homeomorphism between $\mathbb{T}_{1,F}$ and $\mathbb{T}_{2,F}$ and such that $\psi_{\mid F}$ and $\phi$ are equivalent on $\partial F$. 
We define $\mathbb{T}_1$ (resp. $\mathbb{T}_2$) as the triangulation obtained by taking the union of all the triangulations $\left(\mathbb{T}_{1,F}\right)_{F}$ ($\left( \mathbb{T}_{2,F} \right)_{F}$). 
As in the proof of Lemma \ref{trianhom}, we can glue the $\psi_{\mid F}$ together: this gives a function $\psi_{\mid F_{\infty}^{c}}$ defined on the complementary of the unbounded face $F_{\infty}$ of $\mathbb{G}_1$. As $\mathbb{G}_1$ is simple, the boundary of $F_{\infty}$ is a Jordan curve. Thus, according to the discussion we had before Lemma \ref{quasi-diff}, we can extend $\psi_{\mid F_{\infty}^{c}}$ on $F_{\infty}$ and the resulting homeomorphism, denoted by $\psi$, is such that $\psi_{\mid F_{\infty}}$ is a diffeomorphism. By construction, $\psi$ is an orientation-preserving  $\mathbb{G}-\mathbb{G}'$ piecewise diffeomorphism, $\phi$ and $\psi$ are equivalent on $\mathbb{G}_1$ and $\psi(\mathbb{T}_1) = \mathbb{T}_2$. 
\end{proof}

\section{The $\mathbb{N}^{2}$ planar graph}\label{Univ}

We have seen after Lemma \ref{dense} that the set of piecewise horizontal or vertical paths is not dense in $P(\mathbb{R}^{2})$ for the convergence with fixed endpoints. In the following, we show that, in some sense, we can always inject any graph in the $\mathbb{N}^{2}$ graph defined in Exemple \ref{N2}. This property is crucial in the study of planar Markovian holonomy fields. Let $\mathbb{G}=(\mathbb{V}, \mathbb{E}, \mathbb{F})$ be a finite planar graph in $\mathcal{G}\big({\sf Aff}\left(\mathbb{R}^{2}\right)\!\big)$. We remind the reader that for any $v \in \mathbb{V}$, $\mathbb{E}_{v}$ is the set of edges $e\in \mathbb{E}$ such that $\underline{e} = v$. 

\begin{definition}
The graph $\mathbb{G}$ is generic if for any vertex $v \in \mathbb{V}$, $\# \mathbb{E}_v \leq 4$.
\end{definition}

It is worth noticing that any finite planar graph can be approximated by a generic graph. This is illustrated in Figure \ref{generaiquefig}.

\begin{lemme}
\label{gene}
Let $v_0$ be a vertex of $\mathbb{G}$. There exists a sequence of generic graphs $\mathbb{G}_n=(\mathbb{V}_n, \mathbb{E}_n, \mathbb{F}_n)$ in $\mathcal{G}\big({\sf Aff}\left(\mathbb{R}^{2}\right)\!\big)$ such that for any $n \geq 0$:
\begin{enumerate}
\item $v_0 \in \mathbb{G}_n$, 
\item there exists an injective function $\mathcal{L}_n: L_v(\mathbb{G}) \to L_v(\mathbb{G}_n)$ such that for any loop $l \in L_{v}(\mathbb{G})$, $\mathcal{L}_n(l)$ converges with fixed endpoints to $l$. 
\end{enumerate}
\end{lemme}

\begin{figure}
 \centering
  \includegraphics[width=200pt]{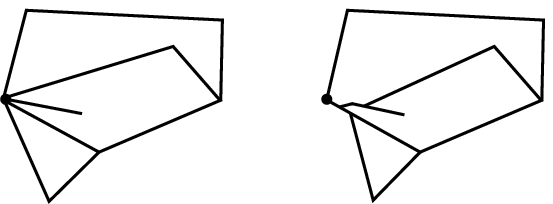}
 \caption{An approximation by a generic graph.}
 \label{generaiquefig}
\end{figure}

The notion of generic graphs was defined so that one could send any of such graph in the $\mathbb{N}^{2}$ planar graph. Let us suppose, until the end of the chapter, that $\mathbb{G}$ is a generic finite planar graph in $\mathcal{G}\big({\sf Aff}\left(\mathbb{R}^{2}\right)\!\big)$.

\begin{proposition}
\label{injection1}
There exists an orientation-preserving homeomorphism of $\mathbb{R}^{2}$, denoted by $\psi$, such that $\psi \left( \mathbb{G} \right)$ is a subgraph of the $\mathbb{N}^{2}$ planar graph. 
\end{proposition}

\begin{proof}
For each $v \in \mathbb{V}$, we choose a point $\tilde{v}$ of $\mathbb{N}^{2}$ such that the points $(\tilde{v})_{v \in \mathbb{V}}$ are all distinct. For each of these points, we choose a subset $\mathbb{E}_{\tilde{v}}$ of edges in $\frac{1}{3}\mathbb{N}^{2}$ going out of $\tilde{v}$ such that $\#\mathbb{E}_{\tilde{v}} = \#\mathbb{E}_v$.
We consider the two pre-graphs: 
\begin{enumerate}
\item $\mathbb{G}_p$ with set of edges $\mathbb{E}_p$, such that $\mathbb{E}_p = \{ e_p([0,\frac{1}{3}]), e_p \text{ represents } e \in \mathbb{E}\}, $
\item $\mathbb{G}'_p$ such that the set of edges $\mathbb{E}'_p$ is equal to $ \cup_{v \in \mathbb{V}} \mathbb{E}_{\tilde{v}}$. 
\end{enumerate}
Let us define the application $\mathcal{S}: v \mapsto \tilde{v}$. Because $\#\mathbb{E}_{\tilde{v}} = \#\mathbb{E}_v $ and thanks to the shape of the graphs, we can choose a bijection $\mathcal{E}: \mathbb{E}_p \to \mathbb{E}'_p$ such that the conditions of Theorem $\ref{homeograph}$ hold. Using this theorem, there exists an orientation-preserving homeomorphism $\psi$ such that $\psi(\mathbb{G}_p) = \mathbb{G}'_p$ and $\psi$ induces the two bijections $\mathcal{S}$ and $\mathcal{E}$. Let us define $\mathbb{G}' = \psi(\mathbb{G})$: we approximate $\mathbb{G}'$ in a $\frac{1}{3^k} \mathbb{N}^{2}$. For $k$ big enough this approximation defines a graph $\tilde{\mathbb{G}}$ without new vertices. By construction, the assumptions 1. to 4. of Theorem \ref{homeograph} hold for $\mathbb{G}$ and $\tilde{\mathbb{G}}$. Using a dilatation we can suppose that $k=1$. Using Theorem \ref{homeograph}, there exists an orientation-preserving homeomorphism $\psi$ which sends $\mathbb{G}$ to the subgraph $\tilde{\mathbb{G}}$ of the $\mathbb{N}^{2}$ planar graph. 
\end{proof}

\begin{corollary}
\label{injectiondansn}
There exists a subgraph  $\mathbb{G}'$ of the $\mathbb{N}^{2}$ graph such that the set of orientation-preserving $\mathbb{G}-\mathbb{G}'$ piecewise diffeomorphisms is not empty. 
\end{corollary}

\begin{proof}
Due to Proposition \ref{injection1} there exists a subgraph $\mathbb{G}'$ of the $\mathbb{N}^{2}$ planar graph such that $\mathbb{G}$ and $\mathbb{G}'$ are homeomorphic: the graph $\mathbb{G}'$ is also simple. By Proposition \ref{triang} the set of orientation-preserving $\mathbb{G}-\mathbb{G}'$ piecewise diffeomorphisms is not empty.
\end{proof}

\chapter{Planar Markovian Holonomy Fields}

\label{planarmarkovianholonomyfields}
In this chapter, we define the continuous and discrete planar Markovian holonomy fields: these are families of random holonomy fields on subsets of $P(\mathbb{R}^{2})$ satisfying an area-preserving homeomorphism invariance and an independence property. 
\section{Definitions}
\label{secdefplanarHF}
First, we define the strong and weak notions of (continuous) planar Markovian holonomy fields. We will use the following notation: if $l$ is a simple loop in $\mathbb{R}^{2}$, ${\sf Int(}l)$ will stand for the bounded connected component of $\mathbb{R}^{2}\setminus l$.

\begin{definition}
\label{planarHF}
A $G$-valued strong (continuous) planar Markovian holonomy field is the data, for each measure of area $vol$ on $\mathbb{R}^{2}$ of a gauge-invariant random holonomy field $\mathbb{E}_{ vol}$ on $\mathbb{R}^{2}$ of weight $\mathbb{E}_{vol}(\mathbbm{1}) = 1$, such that the three following axioms hold: 
\begin{description}
\item[$\mathbf{P_{1}}$ ] Let $vol$ and $vol'$ be two measures of area on $\mathbb{R}^{2}$. Let $\psi: \mathbb{R}^{2} \to \mathbb{R}^{2}$ be a locally bi-Lipschitz homeomorphism which preserves the orientation and which sends $vol$ on $vol'$ (i.e. $vol' = vol \circ \psi^{-1}$). The mapping from $\mathcal{M}ult(P(\mathbb{R}^{2}),G)$ to itself induced by $\psi$, denoted also by $\psi$, satisfies: 
$$\mathbb{E}_{vol'} \circ \psi^{-1} = \mathbb{E}_{vol}.$$
Moreover, let $\mathbb{G}$ and $\mathbb{G}'$  be two finite planar graphs, let $\phi: \mathbb{R}^{2} \to \mathbb{R}^{2}$ be a homeomorphism which preserves the orientation, which sends $vol$ on $vol'$ and which sends $\mathbb{G}$ on $\mathbb{G'}$. The mapping from $\mathcal{M}ult(P(\mathbb{G}'),G)$ to $\mathcal{M}ult\big(P(\mathbb{G}),G\big) $ induced by $\phi$, denoted also by $\phi$, satisfies: 
$$\left(\mathbb{E}_{vol'}\right)_{| \mathcal{M}ult(P(\mathbb{G}'), G)} \circ \phi^{-1} = \left(\mathbb{E}_{vol}\right)_{| \mathcal{M}ult(P(\mathbb{G}), G)}.$$

\item[$\mathbf{P_{2}}$ ] For any measure of area $vol$ on $\mathbb{R}^{2}$, for any simple loops $l_1$ and $l_2$ such that ${{\sf Int(}l_1)}$ and ${{\sf Int(}l_2)}$ are disjoint, under $\mathbb{E}_{vol}$, the two families: 
\begin{align*}
&\Big\{h( p), p \in P\big(\overline{{\sf Int(}l_1)}\big)\Big\} \text{ and }\Big\{h( p), p \in P\big(\overline{{\sf Int(}l_2)}\big)\Big\}
\end{align*}
are $\mathcal{I}$-independent. 

\item[$\mathbf{P_{3}}$ ] For any measures of area on $\mathbb{R}^{2}$, $vol$ and $vol'$, if $l$ is a simple loop such that $vol$ and $vol'$ are equal when restricted to the interior of $l$, the following equality holds: 
\begin{align*}
\big(\mathbb{E}_{vol}\big)_{\mid \mathcal{M}ult(P(\overline{{\sf Int(}l)}),G)} = \big(\mathbb{E}_{vol'}\big)_{\mid \mathcal{M}ult(P(\overline{{\sf Int(}l)}),G)}. 
\end{align*}
 \end{description}
\end{definition}

In the study of Markovian holonomy fields, it will be convenient to have the notion of weak (continuous) planar Markovian holonomy fields.

\begin{definition}
\label{weakplanarHF}
A $G$-valued weak (continuous) planar Markovian holonomy field is the data, for each measure of area $vol$ on $\mathbb{R}^{2}$ of a gauge-invariant random holonomy field $\mathbb{E}_{ vol}$ on ${\sf Aff}\left(\mathbb{R}^{2}\right)$ of weight $\mathbb{E}_{vol}(\mathbbm{1}) = 1$, such that the three following axioms hold: 
\begin{description}
\item[$\mathbf{wP_{1}}$] Let $vol$ and $vol'$ be two measures of area on $\mathbb{R}^{2}$. Let $\psi: \mathbb{R}^{2} \to \mathbb{R}^{2}$ be a diffeomorphism at infinity which preserves the orientation and which sends $vol$ on $vol'$ (i.e. $vol' = vol \circ \psi^{-1}$). Let $p_1, ..., p_n$ be paths in ${\sf Aff}\left(\mathbb{R}^{2}\right)$ such that for any $i \in \{1,...,n\}$, $p'_i=\psi(p_i)$ is in ${\sf Aff}\left(\mathbb{R}^{2}\right)$. Then for any continuous function $f: G^{n} \to \mathbb{R}$, 
\begin{align*}
\mathbb{E}_{ vol} \Big[f\big(h(p_1), ..., h(p_n)\big) \Big] = \mathbb{E}_{vol'} \Big[f \big(h(p'_1), ..., h(p'_n) \big)\Big]. 
\end{align*}
\item[$\mathbf{wP_{2}}$ ] For any measure of area $vol$ on $\mathbb{R}^{2}$, for any simple loops $l_1$ and $l_2$ in ${\sf Aff}\left(\mathbb{R}^{2}\right)$ such that $\overline{{\sf Int(}l_1)}$ and $\overline{{\sf Int(}l_2)}$ are disjoint, under $\mathbb{E}_{vol}$, the two families: 
\begin{align*}
&\Big\{h( p), p \in {\sf Aff}\left(\mathbb{R}^{2}\right)\cap P\big(\overline{{\sf Int(}l_1)}\big)\Big\} \text{ and }\Big\{h( p), p \in {\sf Aff}\left(\mathbb{R}^{2}\right) \cap P\big(\overline{{\sf Int(}l_2)}\big)\Big\}
\end{align*}
are independent. 
\item[$\mathbf{wP_{3}}$ ] For any measures of area on $\mathbb{R}^{2}$, $vol$ and $vol'$, if $l$ is a simple loop such that $vol$ and $vol'$ are equal when restricted to the interior of $l$, the following equality holds: 
\begin{align*}
\big(\mathbb{E}_{vol}\big)_{\mid \mathcal{M}ult({\sf Aff}(\overline{{\sf Int(}l)}),G)} = \big(\mathbb{E}_{vol'}\big)_{\mid \mathcal{M}ult({\sf Aff}(\overline{{\sf Int(}l)}),G)}. 
\end{align*}
 \end{description}

\end{definition}

It can seem strange that we replaced the $\mathcal{I}$-independence by the usual independence in $\mathbf{wP_2}$, but this was precisely the point of Remark \ref{I-indep-et-indep-normale}. As a consequence, any strong planar Markovian holonomy field defines, by restriction, a weak planar Markovian holonomy field. We will see later that the two notions are equivalent when we restrict them to stochastically continuous objects. By $G$-valued (continuous) planar Markovian holonomy fields, we will denote the family of $G$-valued strong or weak (continuous) planar Markovian holonomy fields. 

\begin{definition}
A $G$-valued planar Markovian holonomy field $\big(\mathbb{E}_{vol}\big)_{vol}$ is stochastically continuous if, for any measure of area $vol$ on $\mathbb{R}^{2}$, $\mathbb{E}_{vol}$ is stochastically continuous. 
\end{definition}

A discrete counterpart exists for strong planar Markovian holonomy fields.

\begin{definition}
\label{discreteplanarHF}
A $G$-valued strong discrete planar Markovian holonomy field is the data, for each measure of area $vol$, for each finite planar graph $\mathbb{G}$, of a gauge-invariant random holonomy field $\mathbb{E}_{vol}^{\mathbb{G}}$ on $P(\mathbb{G})$ of weight $\mathbb{E}_{vol}^{\mathbb{G}}(\mathbbm{1})=1$, such that the four following axioms hold: 
\begin{description}
\item[$\mathbf{DP_{1}}$ ] Let $vol$ and $vol'$ be two measures of area on $\mathbb{R}^{2}$, let $\mathbb{G}$ and $\mathbb{G}'$ be two finite planar graphs. Let $\psi$ be a homeomorphism  which preserves the orientation, satisfies $\psi(\mathbb{G}) = \mathbb{G}'$ and such that for any $F \in \mathbb{F}^{b}$, $vol(F)= vol'(\psi(F))$. The mapping from $\mathcal{M}ult(P(\mathbb{G}'),G)$ to $\mathcal{M}ult\big(P(\mathbb{G}),G\big) $ induced by $\psi$, denoted also by $\psi$, satisfies: 
$$\mathbb{E}_{vol'}^{\mathbb{G}'} \circ \psi^{-1} = \mathbb{E}_{vol}^{\mathbb{G}}.$$
\item[$\mathbf{DP_{2}}$ ] For any measure of area $vol$ on $\mathbb{R}^{2}$, for any finite planar graph $\mathbb{G}$, for any simple loops $l_1$ and $l_2$ in $P(\mathbb{G})$, such that ${\sf Int(}l_1) \cap {\sf Int(}l_2) = \emptyset$, under $\mathbb{E}_{vol}^{\mathbb{G}}$, the two families: 
\begin{align*}
&\Big\{h(p ), p \in P(\mathbb{G})\cap P\big(\overline{{\sf Int(}l_1)}\big)\Big\} \text{ and }\Big\{h(p ), p \in P(\mathbb{G})\cap P\big(\overline{{\sf Int(}l_2)}\big)\Big\}
\end{align*}
are $\mathcal{I}$-independent. 
\item[$\mathbf{DP_{3}}$ ] For any measures of area on $\mathbb{R}^{2}$, $vol$ and $vol'$, if $l$ is a simple loop such that $vol$ and $vol'$ are equal when restricted to the interior of $l$, if $\mathbb{G}$ is included in $\overline{{\sf Int(}l)}$, then the following equality holds: 
\begin{align*}
\mathbb{E}_{vol}^{\mathbb{G}} = \mathbb{E}_{vol'}^{\mathbb{G}}. 
\end{align*}
\item[$\mathbf{DP_{4}}$] For any measure of area  $vol$ on $\mathbb{R}^{2}$, for any finite planar graphs $\mathbb{G}_{1}$ and $\mathbb{G}_{2}$, such that $\mathbb{G}_{1} \preccurlyeq \mathbb{G}_{2}$: 
$$\mathbb{E}^{\mathbb{G}_{2}}_{vol} \circ \rho_{P(\mathbb{G}_1), P(\mathbb{G}_2)}^{-1} = \mathbb{E}^{\mathbb{G}_{1}}_{vol}, $$
where we remind the reader that $$\rho_{P(\mathbb{G}_1), P(\mathbb{G}_2)}: \mathcal{M}ult\big(P(\mathbb{G}_2),G\big) \to \mathcal{M}ult\big(P(\mathbb{G}_1),G\big)$$ is the restriction map. 
\end{description}
\end{definition}

We will use also the following weak version of discrete planar Markovian holonomy fields. 
\begin{definition}
A $G$-valued weak discrete planar Markovian holonomy field is the data, for each measure of area $vol$, for each finite graph $\mathbb{G}$ in $\mathcal{G}\big({\sf Aff}\left(\mathbb{R}^{2}\right)\!\big)$, of a gauge-invariant random holonomy field $\mathbb{E}_{vol}^{\mathbb{G}}$ on $P(\mathbb{G})$ of weight $\mathbb{E}_{vol}^{\mathbb{G}}(\mathbbm{1})=1$, such that the four following axioms hold: 

\begin{description}
\item[$\mathbf{wDP_1}$] Let $vol$ and $vol'$ be two measures of area on $\mathbb{R}^{2}$, let $\mathbb{G}$ and $\mathbb{G}'$ be two {\em simple} finite planar graphs in $\mathcal{G}\big({\sf Aff}\left(\mathbb{R}^{2}\right)\!\big)$. Let $\psi$ be a $\mathbb{G}-\mathbb{G}'$ piecewise diffeomorphism  which preserves the orientation. Suppose that for any bounded face $F$ of $\mathbb{G}$, $vol(F)= vol'(\psi(F))$. Then the mapping from $\mathcal{M}ult(P(\mathbb{G}'),G)$ to $\mathcal{M}ult\big(P(\mathbb{G}),G\big) $ induced by $\psi$ satisfies: 
$$\mathbb{E}_{vol'}^{\mathbb{G}'} \circ \psi^{-1} = \mathbb{E}_{vol}^{\mathbb{G}}.$$
\item[ $\mathbf{wDP_{2}}$ ] For any measure of area $vol$ on $\mathbb{R}^{2}$, for any finite graph $\mathbb{G}$ in $\mathcal{G}\big({\sf Aff}\left(\mathbb{R}^{2}\right)\!\big)$, for any simple loops $l_1$ and $l_2$ in $P(\mathbb{G})$, such that $\overline{{\sf Int(}l_1)}$ and $\overline{{\sf Int(}l_2)}$ are disjoint, under $\mathbb{E}_{vol}^{\mathbb{G}}$, the two families: 
\begin{align*}
&\Big\{h(p ), p \in P(\mathbb{G})\cap P\big(\overline{{\sf Int(}l_1)}\big)\Big\} \text{ and }\Big\{h(p ), p \in P(\mathbb{G})\cap P\big(\overline{{\sf Int(}l_2)}\big)\Big\}
\end{align*}
are independent. 
\item[$\mathbf{wDP_{3}}$ ] For any measures of area on $\mathbb{R}^{2}$, $vol$ and $vol'$, if $l$ is a simple loop such that $vol$ and $vol'$ are equal when restricted to the interior of $l$, if $\mathbb{G}$ is included in $\overline{\sf Int}(l)$, then the following equality holds: 
\begin{align*}
\mathbb{E}_{vol}^{\mathbb{G}} = \mathbb{E}_{vol'}^{\mathbb{G}}. 
\end{align*}
\item[ $\mathbf{wDP_{4}}$] For any measure of area  $vol$ on $\mathbb{R}^{2}$, for any finite planar graphs $\mathbb{G}_{1}$ and $\mathbb{G}_{2}$ in $\mathcal{G}\big({\sf Aff}\left(\mathbb{R}^{2}\right)\!\big)$, such that $\mathbb{G}_{1} \preccurlyeq \mathbb{G}_{2}$: 
$$\mathbb{E}^{\mathbb{G}_{2}}_{vol} \circ \rho_{P(\mathbb{G}_1), P(\mathbb{G}_2)}^{-1} = \mathbb{E}^{\mathbb{G}_{1}}_{vol}, $$
where again $\rho_{P(\mathbb{G}_1), P(\mathbb{G}_2)}: \mathcal{M}ult\big(P(\mathbb{G}_2),G\big) \to \mathcal{M}ult\big(P(\mathbb{G}_1),G\big)$ is the restriction map. 
\end{description}
\end{definition}

Let us remark that the Axioms $\mathbf{DP_3}$ and $\mathbf{wDP_3}$ can be directly deduced respectively from $\mathbf{DP_1}$ and $\mathbf{wDP_1}$ by considering the identity function of the plane. Yet, in order to have a similar formulation for continuous and discrete objects we preferred to keep them in the definitions.

As for the continuous objects, any $G$-valued strong discrete planar Markovian holonomy field defines, by restriction, a weak discrete planar Markovian holonomy field. 
By $G$-valued discrete planar Markovian holonomy fields, we will denote the family of $G$-valued strong or weak discrete planar Markovian holonomy fields. In any assertion about $G$-valued discrete planar Markovian holonomy fields, the reader will have to understand that, in the case we are working with a weak discrete planar Markovian holonomy field, all the graphs must be in $\mathcal{G}\left({\sf Aff}\left(\mathbb{R}^{2}\right)\!\right)$. From now on, if not specified, all the planar Markovian holonomy fields will be $G$-valued, thus we will omit to specify it.

\begin{remarque}
 \label{graphaff}
 Let $\big(\mathbb{E}^{\mathbb{G}}_{vol}\big)_{\mathbb{G}, vol}$ be a discrete planar Markovian holonomy field. Using Proposition \ref{limitproj}, the Axiom $\mathbf{DP_4}$ or $\mathbf{wDP_4}$ allows us to define for any measure of area $vol$ and any possibly infinite planar graph $\mathbb{G}$, a unique gauge-invariant random holonomy field $\mathbb{E}^{\mathbb{G}}_{vol}$ on $P(\mathbb{G})$ whose weight $\mathbb{E}^{\mathbb{G}}_{vol}(\mathbbm{1})$ is equal to $1$, such that, for any finite planar graph $\mathbb{G}_{f}\preccurlyeq \mathbb{G}$, $\mathbb{E}^{\mathbb{G}}_{vol} \circ \rho_{P(\mathbb{G}_f),P(\mathbb{G})}^{-1} = \mathbb{E}^{\mathbb{G}_f}_{vol}.$

Besides, the family $$\bigg\{\Big(\mathcal{M}ult\big(P\left(\mathbb{G}\right),G\big) , \mathcal{B}, \mathbb{E}^{\mathbb{G}}_{vol}\Big)_{\mathbb{G}\in \mathcal{G}({\sf Aff}(\mathbb{R}^2))}, \big(\rho_{P(\mathbb{G}),P(\mathbb{G}')}\big)_{\mathbb{G},\mathbb{G'}\in \mathcal{G}({\sf Aff}(\mathbb{R}^{2})), \mathbb{G} \preccurlyeq \mathbb{G}'}\bigg\}$$
is a projective family. There exists a unique gauge-invariant random holonomy field on ${\sf Aff}\left(\mathbb{R}^{2}\right)$, whose weight is equal to $1$, which we denote by $\mathbb{E}_{vol}^{{\sf Aff}}$, such that for any finite planar graph $\mathbb{G}_{f} \in \mathcal{G}\big({\sf Aff}(\mathbb{R}^{2})\big)$, $\mathbb{E}_{vol}^{{\sf Aff}} \circ \rho_{P(\mathbb{G}_f),{\sf Aff}\left(\mathbb{R}^{2}\right)}^{-1} = \mathbb{E}^{\mathbb{G}_f}_{vol}.$
\end{remarque}

The notions of area-dependent continuity and locally stochastically $\frac{1}{2}$-H\"older continuity that we are going to define are similar to the notions explained in Definition $3.2.8$ of \cite{Levy} that L\'{e}vy used for discrete Markovian holonomy fields. Let $\big(\mathbb{E}^{\mathbb{G}}_{vol}\big)_{\mathbb{G},vol}$ be a family of random holonomy fields such that for any measure of area $vol$ and any finite planar graph $\mathbb{G}$, $\mathbb{E}^{\mathbb{G}}_{vol}$ is a random holonomy field on $P(\mathbb{G})$.

\begin{definition}
The family $\big(\mathbb{E}^{\mathbb{G}}_{vol}\big)_{\mathbb{G},vol}$ is {\em locally stochastically $\frac{1}{2}$-H\"{o}lder continuous} if for any measure of area $vol$ on $\mathbb{R}^{2}$, $\big(\mathbb{E}^{\mathbb{G}}_{vol}\big)_{\mathbb{G}}$ is a uniformly locally $\frac{1}{2}$-H\"{o}lder continuous family of random holonomy fields. 

It is {\em continuously area-dependent} if, for any sequence of finite planar graphs $\mathbb{G}_n$ which are the images of a common graph $\mathbb{G}$ by a sequence of area-preserving homeomorphisms $\psi_n$ ($\psi_n(\mathbb{G}) =\mathbb{G}_n$) and such that $vol(\psi_n(F))$ tends to $vol (F)$ as $n$ tends to infinity for any bounded face $F$ of $\mathbb{G}$, the following convergence holds: 
 $$\mathbb{E}^{\mathbb{G}_n}_{vol} \circ \psi_n^{-1} \underset{n \to \infty}{\longrightarrow} \mathbb{E}^{\mathbb{G}}_{vol},$$
 where we denote by $\psi_n$ the induced map from $\mathcal{M}ult\big(P(\mathbb{G}_n),G\big)$ to $\mathcal{M}ult\big(P(\mathbb{G}), G\big)$. 
 
 It is {\em regular} if it is locally stochastically $\frac{1}{2}$-H\"{o}lder continuous and continuously area-dependent. 
 \end{definition}

The new notion of stochastic continuity in law is defined as follow. 

\begin{definition}
Let $\left(\mathbb{E}^{\mathbb{G}}_{vol}\right)_{\mathbb{G},vol}$ be a family of random holonomy fields such that for any measure of area $vol$ and any finite planar graph $\mathbb{G}$, $\mathbb{E}^{\mathbb{G}}_{vol}$ is a random holonomy field on $P(\mathbb{G})$. The family $\left(\mathbb{E}_{ vol}^{\mathbb{G}}\right)_{\mathbb{G},vol}$ is stochastically continuous in law if for any measure of area $vol$, for any integer $m$, any finite planar graph $\mathbb{G}$, any sequence of finite planar graphs $(\mathbb{G}_{n})_{n \geq 0}$ and any sequence of $m$-tuples of loops in $\mathbb{G}_n$, $\left((l^{n}_k)_{k=1}^{m}\right)_{n \in \mathbb{N}}$, if there exists a $m$-tuple of loops in $\mathbb{G}$, $(l_k)_{k =1}^{m}$ such that for any $i \in \{1,..,k\}$, $l^{n}_i$ converges with fixed endpoints to $l_i$ when $n$ goes to infinity, then the law of $\big(h(l_k^{n})\big)_{k=1}^{m}$ under $\mathbb{E}^{\mathbb{G}_n}_{vol}$ converges to the law of $\big(h(l_k)\big)_{k=1}^{m}$ under $\mathbb{E}^{\mathbb{G}}_{vol}$ when $n$ goes to infinity. 
\end{definition}

Let us also remark that the Axioms $\mathbf{DP_1}$ and $\mathbf{wDP_1}$ are not discrete versions of $\mathbf{P_1}$ and $\mathbf{wP_1}$ since in $\mathbf{DP_1}$ and $\mathbf{wDP_1}$ we do not require that $vol'$ is the image of $vol$ by $\psi$. Thus, it is not obvious that any planar Markovian holonomy field, when restricted to graphs, defines a discrete planar Markovian holonomy field. For now, we define the notion of constructibility but later we will show that, under some regularity conditions, any planar Markovian holonomy field is constructible.

\begin{definition}
\label{constructible1}
Let $\big(\mathbb{E}_{vol}\big)_{vol}$ be a weak (resp. strong) planar Markovian holonomy field. It is constructible if the family of measures $\left(\left(\mathbb{E}_{vol}\right)_{\mid \mathcal{M}ult(P(\mathbb{G}),G) }\right)_{ \mathbb{G},vol}$ is a weak (resp. strong) discrete planar Markovian holonomy field. 
\end{definition} 

\begin{remarque}
\label{restric}
If $\left(\mathbb{E}_{vol}\right)_{vol}$ is a constructible stochastically continuous planar Markovian holonomy field, its restriction to graphs defines a stochastically continuous in law discrete planar Markovian holonomy field. 
\end{remarque}

We have seen, in Remark \ref{graphaff}, that given a strong discrete planar Markovian holonomy field $\left(\mathbb{E}^{\mathbb{G}}_{vol}\right)_{\mathbb{G}, vol}$, we could define a family of probability measures $\left(\mathbb{E}_{vol}^{\sf Aff}\right)_{vol}$. If $\left(\mathbb{E}^{\mathbb{G}}_{vol}\right)_{\mathbb{G}, vol}$ is locally stochastically $\frac{1}{2}$-H\"{o}lder continuous, so is $\mathbb{E}_{vol}^{\sf Aff}$ for any measure of area $vol$. By Theorem $\ref{exten2}$, one can extend $\mathbb{E}_{vol}^{\sf Aff}$ as a stochastically continuous random holonomy field on $\mathbb{R}^{2}$, denoted by $\mathbb{E}_{vol}$. We have thus defined $\big(\mathbb{E}_{vol}\big)_{vol}$ a family of stochastically continuous random holonomy fields on $\mathbb{R}^{2}$. 

Using a slight modification of Theorem $3.2.9$ in \cite{Levy}, if $\big(\mathbb{E}_{vol}^{\mathbb{G}}\big)_{\mathbb{G}, vol}$ is continuously area-dependent then the family $\big(\mathbb{E}_{vol}\big)_{vol}$ is a stochastically continuous strong planar Markovian holonomy field. Let us explain the only difficult part of this assertion which is to prove that the axiom $\mathbf{P_{1}}$ is valid for $\big(\mathbb{E}_{vol}\big)_{vol}$. 

Using the same arguments as L\'{e}vy used in Proposition $3.4.1$ of \cite{Levy}, if $\big(\mathbb{E}_{vol}^{\mathbb{G}}\big)_{\mathbb{G}, vol}$ is continuously area-dependent then for any finite planar graph $\mathbb{G}$, for any measure of area $vol$, $\mathbb{E}_{vol}^{\mathbb{G}} = \left(\mathbb{E}_{vol}\right)_{| \mathcal{M}ult(P(\mathbb{G}), G)}.$ Let us remark that it is important, in order to prove this assertion, that we consider all the homeomorphisms in the Axiom $\mathbf{DP_1}$. As a consequence, $\big(\mathbb{E}_{vol}\big)_{vol}$ satisfies the second assertion in Axiom $\mathbf{P_1}$. 

It remains to prove the first assertion in Axiom $\mathbf{P_1}$. Let $vol$, $vol'$ and $\psi$ which satisfy the conditions of this first assertion. Let $p_1$,  ..., $p_n$ be paths on the plane and let $f$ be a continuous function on $G^{n}$. We need to prove that: 
\begin{align}
\label{eq:amontrer}
\mathbb{E}_{vol}\left[f(h(p_1),  ..., h(p_n))\right] = \mathbb{E}_{vol'}\left[f(h(\psi(p_1)),  ..., h(\psi(p_n)))\right].
\end{align} 
Let us consider, for any $i \in \{1, ...,n\}$, a sequence of piecewise affine paths $(p_{i}^{j})_{j \in \mathbb{N}}$ which converges with fixed endpoints to $p_i$ when $j$ goes to infinity. Using Lemmas \ref{lem:bilip} and \ref{lem:bilip2}, for any  $i \in \{1, ...,n\}$, $\psi(p_i^{j})$ converges with fixed endpoints to $\psi(p_i)$ when $j$ goes to infinity. Since $\mathbb{E}_{vol}$ is stochastically continuous, it is enough to prove Equation (\ref{eq:amontrer}) when $p_1$,  ..., $p_n$ are piecewise affine paths. But in this case, there exists a graph $\mathbb{G}$ such that  $\{ p_1,  ..., p_n\} \subset P(\mathbb{G})$ and $\psi(\mathbb{G})$ is also a planar graph: Equality (\ref{eq:amontrer}) is a consequence of the already proven second assertion in Axiom $\mathbf{P_1}$. 

In a nutshell, we just proved the following theorem.

\begin{theorem}
\label{exten3}
Let $\big(\mathbb{E}_{vol}^{\mathbb{G}}\big)_{\mathbb{G}, vol}$ be a strong discrete planar Markovian holonomy field. If it is regular then there exists a unique stochastically continuous strong planar Markovian holonomy field $\big(\mathbb{E}_{vol}\big)_{vol}$ such that, for any finite planar graph $\mathbb{G}$ and any measure of area $vol$, $\big(\mathbb{E}_{ vol}^{\mathbb{G}}\big)_{\mathbb{G}, vol}$ is the restriction to $\mathcal{M}ult(P(\mathbb{G}), G)$ of $\mathbb{E}_{vol}$: $\mathbb{E}_{vol} \circ \rho^{-1}_{P(\mathbb{G}), P(M)} = \mathbb{E}^{\mathbb{G}}_{vol}. $
\end{theorem}

The unicity is a consequence of Proposition \ref{unicite1} and Lemma \ref{dense}. 

\begin{remarque}
\label{rem:erreur}
The proof of L\'{e}vy of the axiom $\mathbf{A_{4}}$ page $123$ of \cite{Levy} in the proof of Theorem $3.2.9$ is based on a wrong statement, namely that $\gamma'$ is well defined and is a Riemannian metric. In this article, we corrected this proof by considering the modified Axiom $\mathbf{P_{1}}$ in Definition \ref{planarHF} and Axiom $\mathbf{A_{4}}$ in Definition \ref{Markovcont}. 
\end{remarque}

The next assertion is a consequence of Theorem \ref{exten3} and Remark \ref{restric} when one considers strong discrete planar Markovian holonomy fields. It is a direct consequence of Theorem $\ref{exten2}$ and Remark \ref{restric} for weak discrete planar Markovian holonomy fields. 

\begin{corollary}
\label{coroll:regularity}
Any strong discrete planar Markovian holonomy field which is regular is stochastically continuous in law. 

Any weak discrete planar Markovian holonomy field which is locally stochastically $\frac{1}{2}$-H\"{o}lder continuous is stochastically continuous in law. 
\end{corollary}

In the rest of the paper, we will mostly work with stochastically continuous in law discrete planar Markovian holonomy fields. 

\subsection{Example: the index field}
For any parametrized loop $l$, for any $x$ in $\mathbb{R}^{2}\setminus l\big([0,1]\big)$, the index of $l$ with respect to $x$ is defined as the integer: 
\begin{align*}
\mathfrak{n}_l(x)=\frac{1}{2i\pi} \oint_{l} \frac{dz}{z-x}.
\end{align*}
Actually, one needs to approximate uniformly $l$ by piecewise smooth loops and take the limit. An other way to define the index field is by first constructing the $L^{2}$-functions valued non-random holonomy field on ${\sf Aff}(\mathbb{R}^{2})$ which sends $l$ on $\mathfrak{n}_l$: this can be defined as a combinatorial object. Using the $L^{2}$ norm on the set of $\mathbb{N}$-valued functions on the plane and considering the Lebesgue measure on the plane, it is easy to see that this holonomy field is locally stochastically $\frac{1}{2}$-H\"{o}lder continuous. Using Theorem \ref{exten2}, we can extend it in order to get a stochastically continuous non-random planar holonomy field on the plane. This construction shows that $\mathfrak{n}_l$ is square-integrable for any rectifiable loop. This can be obtained also by using the Banchoff-Pohl's inequality proved in \cite{Vogt}. Since $\mathfrak{n}_l$ takes values in $\mathbb{N}$ and since any loop is bounded, $\mathfrak{n}_{l}$ is integrable against any measure of area $vol$. Besides if $l_1$ and $l_2$ are based at the same point, using the additivity of the curve integral, we get: 
\begin{align}
\label{additivitedelindex}
\mathfrak{n}_{l_1l_2} = \mathfrak{n}_{l_1}+\mathfrak{n}_{l_2}. 
\end{align}
\begin{figure}
 \centering
  \includegraphics[width=100pt]{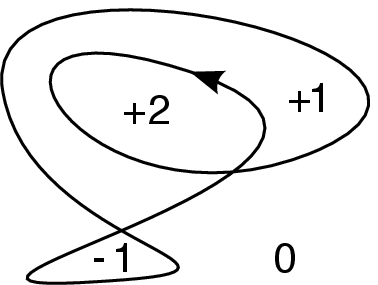}
 \caption{The index of a curve.}
 \label{Indexfig}
\end{figure}
Let $D$ be an element of the Lie algebra $\mathfrak{g}$ of $G$. We can now define the index field driven by $D$. 
\begin{definition}
The index field driven by $D$ is the only planar Markovian holonomy field $\big(\mathbb{E}_{vol}\big)_{vol}$ such that for any measure of area $vol$, any loops $l_1, ..., l_n$ based at the same point and any continuous function $f$ from $G^{n}$ to $\mathbb{R}$ invariant by diagonal conjugation, we have: 
\begin{align*}
\mathbb{E}_{vol}\Big[f\big(h(l_1),..., h(l_n)\big)\Big] = f\left(e^{D \int_{\mathbb{R}^{2}} \mathfrak{n}_{l_1}(x) vol(dx)}, ..., e^{D \int_{\mathbb{R}^{2}} \mathfrak{n}_{l_n}(x) vol(dx)}\right). 
\end{align*}
\end{definition}

The existence of such a planar Markovian holonomy field is due to the fact that one can consider for any finite family of loops $(l_1, ...,l_n)$ based at $0$, the random holonomy field on $(l_1, ...,l_n)$  such that $(h(l_1), ...,h(l_n))$ has the law of $\left(U e^{D \int_{\mathbb{R}^{2}} \mathfrak{n}_{l_1}(x) vol(dx)} U^{-1},  ..., Ue^{D \int_{\mathbb{R}^{2}} \mathfrak{n}_{l_n}(x) vol(dx)}U^{-1}\right)$, where $U$ is a Haar random variable on $G$. It is a gauge-invariant random holonomy field due to the Equation (\ref{additivitedelindex}) and it is actually a measure on $\mathcal{M}ult_{P(\mathbb{R}^{2})}(\{l_1, ...,l_n\},G)$. An application of Proposition \ref{crea1} allows us to conclude. An interesting fact with this planar Markovian holonomy field is that it is stochastically continuous and constructible. It can also be used in order to add a drift to any holonomy field on the plane. Indeed, if $\mu$ is a random holonomy field on the plane, $vol$ be a measure of area and $D$ an element of the center of $\mathfrak{g}$, there exists a planar holonomy field $\mu^{D,vol}$ such that for any loops $l_1, ..., l_n$ based at the same point and any continuous function $f$ from $G^{n}$ to $\mathbb{R}$ invariant by diagonal conjugation, we have: 
\begin{align*}
\mu^{D,vol}\left[f\left(h(l_1),..., h(l_n)\right)\right]\!=\!\mu\!\left[f\!\left(\!e^{D \int_{\mathbb{R}^{2}} \mathfrak{n}_{l_1}(x) vol(dx)} h(l_1), ..., e^{D \int_{\mathbb{R}^{2}} \mathfrak{n}_{l_n}(x) vol(dx) }h(l_n)\!\right)\!\right]\!. 
\end{align*}
Any regularity which holds for $\mu$ holds for $\mu^{D,vol}$. Besides, if $(\mathbb{E}_{vol})_{vol}$ is a planar Markovian holonomy field, so is $(\mathbb{E}^{D,vol}_{vol})_{vol}$.

\section{Restriction and extension of the structure group}

We have seen in Section \ref{restricexten1} how to restrict and extend the structure group of a gauge-invariant random holonomy field: we would like to do the same for planar Markovian holonomy fields. We will work in the setting of discrete planar Markovian holonomy fields, but the upcoming Propositions \ref{extensiondiscret} and \ref{restrictiondiscret} can also be applied to (continuous) planar Markovian holonomy fields. Let $H$ be a closed subgroup of $G$. 

\subsection{Extension}
Let $\big(\mathbb{E}_{vol}^{\mathbb{G}}\big)_{\mathbb{G}, vol}$ be a $H$-valued discrete planar Markovian holonomy field. Recall the notation $\mu \circ \hat{i}_{P}^{-1}$ defined in Notation \ref{rest}. Following Section \ref{restricexten1}, for any finite planar graph $\mathbb{G}$ and any measure of area $vol$, we can see $\mathbb{E}^{\mathbb{G}}_{vol}$ as a $G$-valued gauge-invariant random holonomy field on $P(\mathbb{G})$ by considering $\mathbb{E}^{\mathbb{G}}_{vol} \circ \hat{i}_{P(\mathbb{G})}^{-1}$. It is not difficult to verify next proposition. 

\begin{proposition}
\label{extensiondiscret}
The family $\big(\mathbb{E}^{\mathbb{G}}_{vol} \circ \hat{i}_{P(\mathbb{G})}^{-1}\big)_{\mathbb{G}, vol}$ is a $G$-valued discrete planar Markovian holonomy field. The regularities are the same for the $H$-valued and the $G$-valued random holonomy fields. 
\end{proposition}

\subsection{Restriction}

\begin{proposition}
\label{restrictiondiscret} 
Let $\big(\mathbb{E}_{vol}^{\mathbb{G}}\big)_{\mathbb{G}, vol}$ be a $G$-valued stochastically continuous in law discrete planar Markovian holonomy field. Let us suppose that for any finite planar graph $\mathbb{G}$, any vertex $v$ of $\mathbb{G}$, any measure of area $vol$ and any simple loop $l \in L_v(\mathbb{G})$, $$h(l) \in H, \ \ \mathbb{E}_{vol}^{\mathbb{G}} \text{ a.s., }$$ then there exists a $H$-valued stochastically continuous in law discrete planar Markovian holonomy field $\big(\tilde{\mathbb{E}}_{vol}^{\mathbb{G}}\big)_{\mathbb{G}, vol} $ such that: 
\begin{align*}
\mathbb{E}_{vol}^{\mathbb{G}} = \tilde{\mathbb{E}}_{vol}^{\mathbb{G}} \circ \hat{i}_{P(\mathbb{G})}^{-1}, 
\end{align*}
for any finite planar graph $\mathbb{G}$ and any measure of area $vol$. 
\end{proposition}
 
The proof of Proposition \ref{restrictiondiscret} relies heavily on a theorem which will be proved later, namely Theorem \ref{caracterisation}, thus it will be given page \pageref{restrictiondiscretpage}. It is more difficult than one might think to prove this proposition because of the non-unicity of the random holonomy field $\mu_H$ in Proposition \ref{restrictionholo}. In fact, one can show in general that the natural choice we made in Proposition \ref{restrictionholo} does not allow one to define a $H$-valued discrete Markovian holonomy field: we will give an exemple page \pageref{ex:restriction} which illustrates this fact. Let us explain the problem which appears when one wants to restrict the structure group of a discrete Markovian holonomy field. Let $\big(\mathbb{E}_{vol}^{\mathbb{G}}\big)_{\mathbb{G}, vol}$ be a $G$-valued stochastically continuous in law discrete planar Markovian holonomy field which satisfies the condition of Proposition~\ref{restrictiondiscret}. Let $\mathbb{G}$ be a finite planar graph and let $vol$ be a measure of area on the plane. It is natural to set: 
\begin{align}
\label{restrict}\tilde{\mathbb{E}}^{\mathbb{G}}_{vol} = \widehat{\left(\left((\mathbb{E}_{vol}^{\mathbb{G}})_{\mid L_v(\mathbb{G})} \circ \iota^{-1}\right)_{\mid \mathcal{I}_{H}}\right)},
\end{align}
where $v$ is any vertex of $\mathbb{G}$, $\iota: \mathcal{M}ult(L_v(\mathbb{G}),H) \to \mathcal{M}ult(P(\mathbb{G}), H)$ is any map given by Lemma \ref{inv3}, $\mathcal{I}_{H}$ is the $H$-invariant $\sigma$-field and  $\ \widehat{}\ $ is the gauge-invariant extension (where the gauge group is now built on $H$) given by Proposition \ref{exten}.

Let $l$ and $l'$ be two simple loops in $P(\mathbb{G})$, with $\underline{l} = v$, such that $\overline{{\sf Int(}l)} \cap \overline{{\sf Int(}l')} =~\emptyset$ as shown in the Figure \ref{restrictionfig}. If the family of measures $\big(\tilde{\mathbb{E}}^{\mathbb{G}}_{vol}\big)_{\mathbb{G}, vol}$ just defined above was a discrete planar Markovian holonomy field, then $h(l)$ and $h(l')$ would be independent under $\tilde{\mathbb{E}}_{vol}^{\mathbb{G}}$. But if $p$ is the path from $v$ to $\underline{l'}$ used to define $\iota$ and if $f_1, f_2$ are two continuous functions on $H$ invariant by conjugation by $H$, we have: 
\begin{align}
\label{eq:egalitenoninde}
\tilde{\mathbb{E}}^{\mathbb{G}}_{vol} \Big(f_1\big(h\left(l\right)\big) f_2\big(h\left(l'\right)\big)\Big)&= \mathbb{E}^{\mathbb{G}}_{vol}\Big({f_1}\big(h(l)\big){f_2}\big(h(pl'p^{-1})\big) \Big).
\end{align}In the r.h.s. appear the two loops $l$ and $pl'p^{-1}$ which are not of null intersection (as they share at least $v$) and only appear functions invariant by conjugation by $H$ and not by $G$. This does not allow us to split the expectation into a product.

\begin{figure}
 \centering
  \includegraphics[width=150pt]{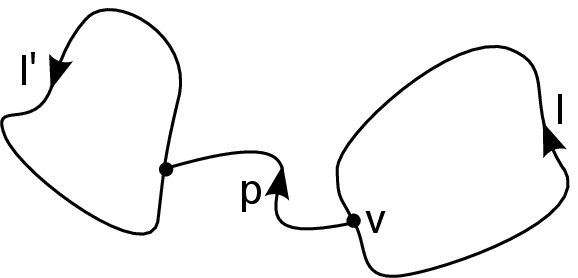}
 \caption{Two simple loops $l$, $l'$ with a path $p$ joining $\underline{l}$ to $\underline{l'}$.}
 \label{restrictionfig}
\end{figure}

\chapter{Weak Constructibility}

\label{weakconstructibility}
In this section, any weak continuous planar Markovian holonomy field is shown to be constructible. Let us state a proposition which is a direct consequence of an important theorem of Moser and Dacorogna in \cite{Moser}. 
Let ${\sf Leb}$ be the Lesbegue measure on $\mathbb{R}^{2}$. 

\begin{proposition}
\label{MoserDac}
Let $Q$ be an open simple $n$-gon in $\mathbb{R}^{2}$. Let $f$ and $g$ be two strictly positive functions on $Q$ which are in $C^{1}(Q) \cap C^{0}(\overline{Q})$. Suppose that: 
\begin{align*}
\int_Q f d{\sf Leb} = \int_Q g d{\sf Leb}. 
\end{align*}
Then there exists $\phi \in {\sf Diff}^{1}(Q) \cap {\sf Diff}^{0}(\overline{Q})$, a homeomorphism of $\overline{Q}$ which restricts to a diffeomorphism of $Q$, such that: 
\begin{align*}
 g. {\sf Leb}_{\mid Q} = \big(f.{\sf Leb}_{\mid Q}\big) \circ \phi^{-1}. 
\end{align*}
and $\phi(x) = x$ for any $x \in \partial Q$. 
\end{proposition}

\begin{proof}
In \cite{Moser}, page $15$ the authors define for any positive integer $k$, a property $(H_k)$ for open subsets of $\mathbb{R}^{n}$. They show in Theorem $7$ of the same paper, that for any positive integer $k$, any open domain $\Omega$ which satisfies $(H_k)$, any positive functions $f$ and $g$ in $C^{k}(\Omega)$ with $f+\frac{1}{f}$ and $g+\frac{1}{g}$ bounded and which satisfy: 
\begin{align*}
\int_\Omega fd{\sf Leb} = \int_\Omega gd{\sf Leb}, 
\end{align*}
there exists $\phi \in {\sf Diff}^{1}(\Omega) \cap {\sf Diff}^{0}(\overline{\Omega})$ with $\phi(x) = x$ on $\partial \Omega$ such that: 
\begin{align*}
 g. d{\sf Leb}_{\mid \Omega} = \big(f. d{\sf Leb}_{\mid \Omega}\big) \circ \phi^{-1}. 
\end{align*}
Besides, Proposition $A.2$ of the same paper asserts that any domain with Lipschitz boundary satisfies $(H_k)$ for every $k \geq 1$. The proposition follows from this discussion. 
\end{proof}

\begin{theorem}
\label{weakconst}
Any weak planar Markovian holonomy field is constructible. 
\end{theorem}

\begin{proof}
Let $\left(\mathbb{E}_{vol}\right)_{vol}$ be a weak planar Markovian holonomy field. Let us consider $\big(\mathbb{E}_{vol}^{\mathbb{G}}\big)_{\mathbb{G}, vol}$, the family of random holonomy fields that we get by restricting $\left(\mathbb{E}_{vol}\right)_{vol}$ on $\mathcal{M}ult(P(\mathbb{G}),G)$ for any finite planar graph $\mathbb{G}$ in $\mathcal{G}\left({\sf Aff}\left(\mathbb{R}^{2}\right)\right)$. As explained before Definition \ref{constructible1}, we only have to check that the Axiom $\mathbf{wDP_1}$ is satisfied by $\big(\mathbb{E}_{vol}^{\mathbb{G}}\big)_{\mathbb{G}, vol}$. 

Let $vol$ and $vol'$ be two measures of area on $\mathbb{R}^{2}$. Consider $\mathbb{G}$ and $\mathbb{G}'$ two simple finite planar graphs in $\mathcal{G}\left({\sf Aff}(\mathbb{R}^{2})\right)$. Let $\psi$ be an orientation-preserving $\mathbb{G}-\mathbb{G}'$ piecewise diffeomorphism. Let us suppose that for any bounded face $F$ of $\mathbb{G}$, $vol(F) = vol'(\psi(F))$. 

Let $F'_{\infty}$ be the unbounded face of $\mathbb{G}'$. Let us suppose that we managed to construct an orientation-preserving diffeomorphism at infinity $\Phi$ on $\mathbb{R}^{2}$ such that: 
\begin{align}
\label{condition1}
vol'_{\mid (F'_{\infty})^{c}}&= \left(vol \circ (\Phi \circ \psi)^{-1}\right)_{\mid (F'_{\infty})^{c}}, \\
\label{condition2}
\Phi_{\mid \mathbb{G}'} &= Id_{\mid \mathbb{G}'}.   
\end{align}
As $\mathbb{G}'$ is a simple graph, the boundary of $F'_{\infty}$ is a simple loop. Applying the Axiom $\mathbf{wP_3}$ and using the condition (\ref{condition1}): 
\begin{align*}
\mathbb{E}_{vol'}^{\mathbb{G}'}=\mathbb{E}_{vol\circ (\Phi \circ \psi)^{-1}}^{\mathbb{G}'}. 
\end{align*} 
Yet by condition (\ref{condition2}),  $\mathbb{G}' = \Phi(\mathbb{G}') =  \Phi \circ \psi (\mathbb{G})$. Thus, as an application of Axiom $\mathbf{wP_1}$, we get: 
\begin{align*}
\mathbb{E}_{vol'}^{\mathbb{G}'} =\mathbb{E}_{vol\circ (\Phi \circ \psi)^{-1}}^{\mathbb{G}'}= \mathbb{E}^{\Phi \circ \psi (\mathbb{G})}_{vol \circ (\Phi \circ \psi)^{-1}} = \mathbb{E}^{\mathbb{G}}_{vol}, 
\end{align*}
which is the desired equality.

It remains to construct an orientation-preserving diffeomorphism at infinity $\Phi$ on $\mathbb{R}^{2}$ satisfying the two conditions (\ref{condition1}) and (\ref{condition2}). This will be done by applying twice the Proposition \ref{MoserDac}. Let $\mathbb{G}_0$ be a good mesh for $\psi$. First we regularize the measure $vol \circ \psi^{-1}$, which does not have a smooth density, by applying the proposition for each face of the mesh $\mathbb{G}_0$. Then we transport the resulting measure of area on $vol'$ by applying again Proposition \ref{MoserDac} for each bounded face of $\mathbb{G}'$. 

Let us fix a measure of area $vol''$ such that, for any bounded face $F$ of $\mathbb{G}_0$, $vol''(\psi(F)) = vol(F)$. Let us consider any bounded face $F_0$ of $\mathbb{G}_0$ which is included in a bounded face of $\mathbb{G}$. By definition of a $\mathbb{G}-\mathbb{G}'$ piecewise diffeomorphism and the definition of a good mesh for $\psi$, $\psi_{\mid F_0}$ is a diffeomorphism from $F_0$ to $\psi(F_0)$ which are two simple $n$-gons. Thus, $\widetilde{vol}_{\mid \psi(F_0)} = vol_{\mid F_0} \circ (\psi_{\mid F_0})^{-1}$ defines a measure with strictly positive smooth density on $\psi(F_0)$. Using the condition on the Jacobian determinant of $\psi$, this smooth density can be extended as a strictly positive continuous function on $\overline{\psi(F_0)}$. Using Proposition \ref{MoserDac}, we can consider $\phi_{\mid \psi(F_0)} \in {\sf Diff}^{1}(\psi(F_0)) \cap {\sf Diff}^{0}\left(\overline{\psi(F_0)}\right)$ with $\phi_{\mid \psi(F_0)}(x) = x$ for any $x \in \partial \psi(F_0)$ such that: 
\begin{align*}
 vol''_{\mid \psi(F_0)} = \widetilde{vol}_{\mid \psi(F_0)} \circ (\phi_{\mid \psi(F_0)})^{-1}. 
\end{align*}

Let us finally set $\phi_{\mid \psi(F_{\infty})} = Id_{\mid \psi(F_{\infty})}$, where $F_{\infty}$ is the unbounded face of $\mathbb{G}$. Thanks to  the boundary condition on $\phi_{\mid \psi(F)}$ for any face $F$ of $\mathbb{G}_0$, we can glue together all the homeomorphisms $\phi_{\mid \psi(F)}$ constructed for each face $F$ of $\mathbb{G}_{0}$. It defines an orientation-preserving diffeomorphism at infinity $\phi_1$ on $\mathbb{R}^{2}$ such that $vol'' _{\mid \psi(F_{\infty})^{c}}= (vol \circ (\phi_1 \circ \psi)^{-1})_{\mid \psi(F_{\infty})^{c}}$ and $(\phi_1)_{\mid \mathbb{G}'} = Id_{\mid \mathbb{G}'}$. 

For any bounded face $F$ of $\mathbb{G}'$, we have: 
\begin{align*}
vol''(F) = vol(\psi^{-1}(F)) = vol'(F). 
\end{align*} 
Besides, $\mathbb{G}'$ is a simple graph: we can apply Proposition \ref{MoserDac} for any bounded face $F$ of $\mathbb{G}'$ in order to transport $vol''_{\mid F}$ on $vol'_{\mid F}$. Applying the same arguments (gluing the homeomorphisms as we just did) allows us to construct an orientation-preserving homeomorphism $\phi_2$ such that: 
\begin{align*}
(vol')_{\mid (F'_{\infty})^{c}} &= (vol'' \circ \phi_2^{-1})_{\mid (F'_{\infty})^{c}}\\
(\phi_2)_{\mid \mathbb{G}'} &= Id_{\mid \mathbb{G}'}, 
\end{align*}
where we recall that $F'_{\infty} = \psi(F_{\infty})$ is the unbounded face of $\mathbb{G}'$. The orientation-preserving diffeomorphism at infinity $\Phi = \phi_2 \circ \phi_1$ satisfies the two conditions (\ref{condition1}) and (\ref{condition2}). \end{proof}

\chapter{Group of Reduced Loops}

 \label{Sectiongroupofreducedloop}
In order to construct planar Markovian holonomy fields, we need to study the group of reduced loops. In order to construct a gauge-invariant random holonomy field on $P$, it is enough to construct a measure on $\mathcal{M}ult_P(L,G)$ for any set $L$ of loops of $P$: this was the loop paradigm explained in Lemma \ref{multpara}. Let $\mathbb{G}$ be a finite planar graph, let $v$ be a vertex of $\mathbb{G}$, let $L$ be the set of loops $L_v(\mathbb{G})$ and let $P$ be equal to $P(\mathbb{G})$, then: 
\begin{align*}
\mathcal{M}ult_P(L,G) = {\sf Hom} (\pi_1(\mathbb{G}, v), G^{\vee}), 
\end{align*}
where $G^{\vee}$ is the group based on the same set at $G$, endowed with the multiplication $._{\vee}$ such that $x._{\vee}y = yx$ for any $x,y \in G$ and $\pi_1(\mathbb{G}, v)$ is the fundamental group of $\mathbb{G}$ based at $v$. In this chapter we study the group $\pi_1(\mathbb{G}, v)$. 

\section{Definition and facts}
\label{sec:red}
Let us fix, until the end on the section, a finite planar graph $\mathbb{G} = (\mathbb{V}, \mathbb{E}, \mathbb{F})$ and a vertex $v$ of $\mathbb{G}$. The group of based reduced loops $RL_{v}(\mathbb{G})$ is the fundamental group of $\mathbb{G}$ based at $v$: $RL_{v}(\mathbb{G}) = \pi_{1}(\mathbb{G},v)$. For convenience we define it using a combinatorial point of view, as  L\'{e}vy does in Section 1.3.4 of \cite{Levy}.

Let $l$ be a loop in $P(\mathbb{G})$. Recall the definition of equivalence of paths explained in Definition \ref{equiv}. The equivalence class of $l$ in $P(\mathbb{G})$, denoted by $[l]_{\simeq}$, contains a unique element of shortest combinatorial length, which is said to be reduced. Besides, if $l_{1}$ and $l_{2}$ are two loops in $P( \mathbb{G})$ based at $v$, $[l_{1} l_{2}]_{\simeq}$ depends only on $[l_{1}]_{\simeq}$ and $[l_{2}]_{\simeq}$. Thus, it is equivalent to speak about equivalence classes or about reduced paths and the set of reduced paths is endowed with an internal operation. 

\begin{definition}
The set of reduced loops in $P(\mathbb{G})$ based at $v$ will be denoted by $RL_{v}(\mathbb{G})$. Let $l_{1}$ and $l_{2}$ be two loops in $RL_{v}( \mathbb{G})$, we define $l_{1} \times l_{2} = [l_{1} l_{2}]_{\simeq}$.
\end{definition}
 Endowed with this operation, $RL_{v}( \mathbb{G})$ is a group. The existence of the inverse of a loop $l$ based at $v$ is due to the fact that $[l l^{-1}]_{\simeq} = [1_v]_{\simeq}$, where $1_v$ is the trivial path constant to $v$. 
In the following, we will denote the reduced product of $l_1$ with $l_2$ by $l_1l_2$ rather than $l_1\times l_2$. In the following we will study $RL_{v}( \mathbb{G})$ using families of lassos. Let us state a simple, yet crucial lemma about lassos.

\begin{lemme}
\label{conjug}
Two lassos in $L_{v}(\mathbb{G})$ whose meanders represent the same cycle are conjugated in $RL_v(\mathbb{G})$.
\end{lemme}

\begin{proof}
Let $l$ and $l'$ be two lassos based at $v$. They can be written as $l=sms^{-1}$ and $l'=s'm's'^{-1}$, where $s$ and $s'$ are respectively the spoke of $l$ and $l'$. As the loops $m$ and $m'$ are related, there exist $c$ and $d$ two paths such that $m=cd$ and $m'=dc$. Let us denote by $p$ the loop $s'c^{-1}s^{-1}$, then $l'=plp^{-1}$.
\end{proof}

\begin{definition}
A loop in $\mathbb{G}$ is called a facial lasso if it is a lasso and its meander represents a non-oriented facial cycle of $\mathbb{G}$.
\end{definition}

The exact definition of facial cycle, an oriented or non-oriented cycle which represents the boundary of a face, is given in \cite{Levy}, in Definition 1.3.13. For any face $F$ of $\mathbb{G}$, we will denote by $\partial F$ both the non-oriented and oriented facial cycles associated with $F$. If we specify that $\partial F$ is oriented, we will consider the anti-clockwise orientation. 

 In the following, we address the problem of creating families of lassos which generate the whole group $RL_v(\mathbb{G})$. The well-known Lemma \ref{generate1} provides a solution of this problem which is not adapted to our context, but will nevertheless be the departure point of our discussion. In order to state it, we need the definition of a spanning tree. 

\begin{definition}
A spanning tree $T$ of $\mathbb{G}$ is a subset of $\mathbb{E}$ such that: 
\begin{itemize}
\item if an edge $e$ is in $T$, $e^{-1}$ is also in $T$, 
\item the set of non degenerate loops in $T$ is empty, 
\item $\mathbb{V} =\big\{ \underline{e}, e \in \mathbb{E}\big\}$.
\end{itemize}
If $\mathbb{G}$ is composed of a unique edge $e$ which is a loop, the set $\{\underline{e}\}$ is considered as a spanning tree of $\mathbb{G}$. A rooted spanning tree is the data of a spanning tree $T$ and a vertex $v$ of $\mathbb{G}$. 
\end{definition}

Let $T$ be a spanning tree of $\mathbb{G}$ rooted at $v$. The restriction of $T$ to a subgraph $\mathbb{G}' = (\mathbb{V}', \mathbb{E}', \mathbb{F}')$ of $\mathbb{G}$ is defined as following: if $\mathbb{E} = \{l, l^{-1}\}$ with $l$ a loop, then $T'=\{\underline{l}\}$ and in the other cases we consider $T'= T \cap  \mathbb{E}'$.

\begin{definition}
\label{def:lassosedge}
Let $u$ and $w$ be two vertices of $\mathbb{G}$. The path $[u, w]_{T}$ is the unique injective path in $T$ joining $u$ to $w$. For any edge $e \in \mathbb{E}$, we set $l_{e, T} = [v,\underline{e}]_{T} e [\overline{e},v]_{T}$. 
\end{definition}

\begin{lemme}
\label{generate1}
Let $\mathbb{E}^{+}$ be an orientation of $\mathbb{G}$. The group $RL_{v}(\mathbb{G})$ is freely generated by the loops $\big\{l_{e, T}: e \in (\mathbb{E} \setminus T)^{+}\big\}$. 
\end{lemme}

\begin{proof}
We only have to prove that $\left( \big\{l_{e, T}: e \in (\mathbb{E} \setminus T)^{+}\big\}, RL_{v}(\mathbb{G})\right)$ satisfies the universal property of free groups: given any function $f$ from $\big\{l_{e, T}: e \in (\mathbb{E} \setminus T)^{+}\big\}$ to a group $G$ there exists a homomorphism $\phi: RL_{v}(\mathbb{G}) \to G$ such that $\phi(l_{e, T}) = f(l_{e, T}),$ for any $e \in (\mathbb{E} \setminus T)^{+}$.

Let $G$ be any group and let $1$ be its neutral element. We recall the Equation (\ref{lequationedgeparadigme}), in Subsection \ref{multi}, which shows that one can construct a multiplicative function from $P(\mathbb{G})$ to $G$ by specifying the value on $\mathbb{E}^{+}$. In the definition of multiplicative functions, we asked that the function reverses the order of multiplication. Only for this proof, we will suppose that it preserves the order. This means that if $g \in \mathcal{M}ult(P(\mathbb{G}),G)$, then for any path $p_1$ and $p_2$ in $P$ which can be concatenated, $g(p_1 p_2) = g(p_1) g(p_2)$. Let $f$ be a function from $\big\{l_{e, T}: e \in (\mathbb{E} \setminus T)^{+}\big\}$ to $G$. We define the element $\phi$ in $G^{\mathbb{E}^{+}}$ by: 
\begin{align*}
\phi(e) = &\left\{
    \begin{array}{ll}
       f(l_{e,T}), & \text{if } e \in (\mathbb{E} \setminus T)^{+}, \\ 
           1, & \mbox{otherwise.}
                 \end{array}
\right.
\end{align*}
This defines an element of $\mathcal{M}ult(P(\mathbb{G}), G)$, called also $\phi$, which restriction on $L_{v}(\mathbb{G})$ induces a homeomorphism from $RL_v(\mathbb{G})$ to $G$. Beside, for any path $p$ in $T$, $\phi( p) = 1$. Let $e$ be any element of $(\mathbb{E} \setminus T)^{+}$: 
\begin{align*}
\phi(l_{e,T}) = \phi([v,\underline{e}]_Te[\overline{e}, v]_T) = \phi([v,\underline{e}]_T) \phi(e) \phi([\overline{e}, v]_T)) = f(l_{e,T}).
\end{align*}
The universal property of free groups holds: $RL_{v}(\mathbb{G})$ is the free group generated by $ \big\{l_{e, T}: e \in (\mathbb{E} \setminus T)^{+}\big\}$. 
\end{proof}

\begin{remarque}
\label{rq:nbfree}
The loops $l_{e, T}$ defined above are lassos and since $\mathbb{G}$ is a finite planar graph, $\# (\mathbb{E} \setminus T)^{+} = \# \mathbb{E}^{+} - \# T =  \# \mathbb{E}^{+} - \# \mathbb{V} + 1 = \# \mathbb{F}^{b}$. 
\end{remarque}

\section{The example of $RL_0({\mathbb{N}^{2}})$}
We define in this section a family of facial lassos in $\mathbb{N}^{2}$ which will be important in Section \ref{carsec}. Even if this family can be studied with the help of the upcoming Proposition \ref{generate}, we give an elementary proof that it generates $RL_0({\mathbb{N}^{2}})$. 

\begin{notation}
\label{not:er}
Let $(i, j)$ and $(k, l)$ be couples of reals such that $i=k$ or $j=l$. We denote by $(i, j)\to (k, l)$ the straight line from $(i, j)$ to $(k, l)$ in $\mathbb{R}^{2}$. If $j=l$ and $k=i+1$ it will also be denoted by $e^{r}_{i, j}$; if $i=k$ and $l=j+1$ it will also be denoted by $e^{u}_{i, j}$.
\end{notation}

\begin{definition}
\label{base}
Let $i$, $j$ be two non negative integers. Let $\partial c_{i, j}$ be the loop in $L( \mathbb{N}^{2})$ defined by: 
\begin{align*}
\partial c_{i, j} &=  (i, j) \to (i+1,j) \to (i+1,j+1) \to (i, j+1) \to (i,j) \\&= e^{r}_{i, j} e^{u}_{i+1,j} (e^{r}_{i, j+1})^{-1}(e^{u}_{i, j})^{-1}.
\end{align*}
Let $p_{i, j}$ be the path in $P( \mathbb{N}^{2})$ defined by: 
$$p_{i, j}= (0,0) \to (i,0) \to (i, j) = e^{r}_{0,0} ... e^{r}_{i-1,0} e^{u}_{i,0} ... e^{u}_{i,j-1}.$$
Let $L_{i, j}$ be the reduced loop based at $0$:
$$L_{i, j}= \big[p_{i, j} \partial c_{i, j} p_{i, j}^{-1}\big]_{\simeq}.$$
\end{definition}
One can refer to Figure \ref{Lij} to have a clear representation of the lasso $L_{i,j}$. 

\begin{figure}
 \centering
  \includegraphics[width=180pt]{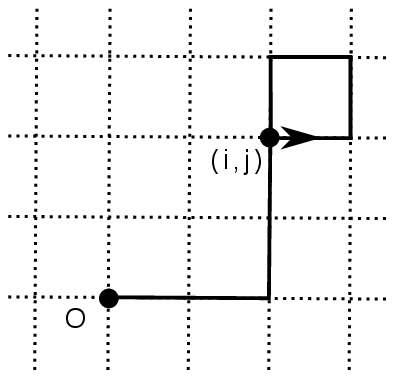}
 \caption{The lasso $L_{i,j}$.}
 \label{Lij}
\end{figure}

\begin{lemme}
\label{generativeexemple}
The family $\big(L_{i, j}\big)_{(i, j) \in \mathbb{N}^{2}}$ is a freely generating subset of $RL_{0}( \mathbb{N}^{2})$.
\end{lemme}

\begin{proof} 
We only have to work with the finite planar graph: $$\mathbb{G} = \mathbb{N}^{2} \cap \big\{ (x, y), x \leq k, y \leq k'\big\},$$ where $k$ and $k'$ are any positive integers. We remind the reader that the intersection of a graph with a set was defined before Definition \ref{infiniteplanar}. Lemma \ref{generate1} implies that $RL_{0}(\mathbb{G})$ is a free group of rank $k \times k'$. 
Let $l$ be a loop in $RL_{0}(\mathbb{G})$. We endow the graph $\mathbb{G}$ with the following orientation: from bottom to top, from left to right. Let $T$ be the tree defined by: 
$$T~= \Big \{ \big(e^{u}_{i, j}\big)^{\pm1}, i\in \{0,...,k\}, j \in \{0,..., k'-1\}\Big\} \cup \Big\{ \big(e^{r}_{i,0}\big)^{\pm1}, i\in \{0,..., k-1\}\Big\}.$$ 
The root of $T$ will be chosen to be $(0,0)$. One can look at Figure \ref{arbre} to have a better idea of the graph and the tree we have just constructed. Applying Lemma \ref{generate1} to this situation, $l$ can be written as the reduced concatenation of some elements of $\big\{l_{e, T}^{\pm1}\big\}_{e \in (\mathbb{E}\setminus T)^{+}}$, where $\mathbb{E}$ is the set of edges of ~$\mathbb{G}.$ Moreover $(\mathbb{E} \setminus T)^{+}$ is equal to $\big\{ e_{i, j}^{r}, i \in \{0,..., k-1\}, j \in \{1,..., k'\}\big\}$. Since $$(l_{e^{r}_{i, j},T})^{-1} = L_{i,0} L_{i,1}... L_{i, j-1},$$ the family $\big(L_{i, j}\big)_{(i, j)\in\mathbb{N}^{2}}$ is a generating subset of $RL_{0}(\mathbb{G})$ whose cardinal is $k \times k'$:  it is a freely generating subset of $RL_{0}(\mathbb{G})$. 
\end{proof}

\begin{figure}
 \centering
  \includegraphics[width=180pt]{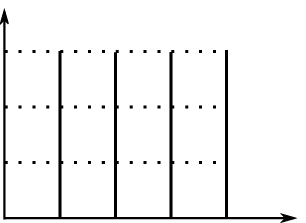}
 \caption{The spanning tree $T$.}
 \label{arbre}
\end{figure}

\section{Family of generators of $RL_{v}(\mathbb{G})$}
 \label{sec:generRL}
In the setting of planar graphs, this section is a generalization of Section $2.9$ in \cite{Levy} about tame generators. The proofs explained here do not use the ideas in \cite{Levy} but rather uses a recursive decomposition of graphs. Let $\mathbb{G}= (\mathbb{V}, \mathbb{E}, \mathbb{F})$ be a finite planar graph, let $v \in \mathbb{V}$ and let $T$ be a spanning tree of $\mathbb{G}$ rooted at $v$. Next definition follows Definition 2.4.6. of \cite{Levy}. Let $(l_{e, T})_{e \in \mathbb{E}}$ be the loops defined in Definition \ref{def:lassosedge}.

\begin{definition}
\label{def:lassogenerat}
Let $F$ be a bounded face of $\mathbb{G}$ and let $c_{F}$ be a simple loop representing the facial non-oriented cycle associated with $F$: it can be written as $c_{F}=e_{1}... e_{n}$. We define the reduced path $\l_{c_{F},T} = l_{e_{1},T}...l_{e_{n},T}$ in $RL_{v}(\mathbb{G})$.
\end{definition}

Let $(c_{F})_{F \in \mathbb{F}^{b}}$ be such that $c_{F}$ is a representative of the non-oriented facial cycle associated with $F$: it is called a {\em family of facial loops of $\mathbb{G}$}. We have defined a new family of loops $\big(\mathbf{l}_{c_{F},T}\big)_{F \in \mathbb{F}^{b}}.$ The difference with Definition $2.4.7$ in \cite{Levy} is that the choice of $c_{F}$ is not given by the choice of $T$, there is freedom to choose the base point of $c_F$.

A remark that we will often use is that, when one changes the root of $T$ from $v$ to an other vertex $v'$, this has the effect to conjugate the family $\big(\mathbf{l}_{c_{F},T}\big)_{F \in \mathbb{F}^{b}}$ by $[v',v]_{T}$. This comes from the fact that, for any spanning tree $T$, any vertices $v$, $v'$ and $v''$, we have the equality in the set of reduced paths, $[v,v'']_{T} = [v,v']_{T} [v',v'']_{T}$. 
Proposition \ref{generate} and Lemma \ref{generate2} give the two most important properties of these families of reduced loops. 
\begin{figure}
 \centering
  \includegraphics[width=270pt]{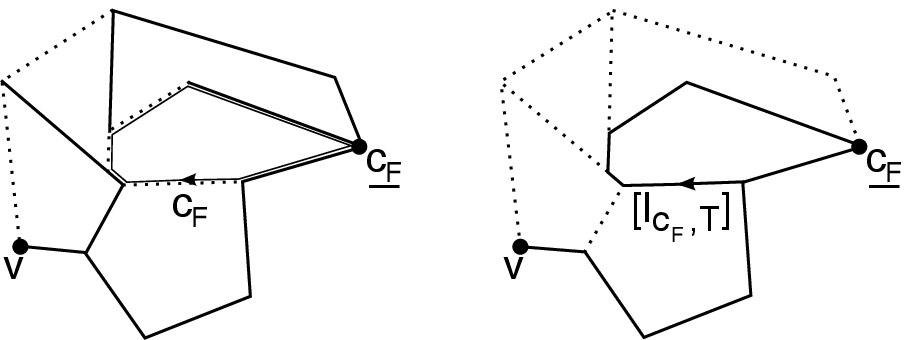}
 \caption{A graph, a spanning tree, a facial cycle: the associated reduced facial lasso. }
\end{figure}

\begin{proposition}
\label{generate}
For any bounded face $F$, $\l_{c_{F},T}$ is a facial lasso based at $v$ whose meander represents the non-oriented facial cycle $\partial F$. Besides, $(\l_{c_{F},T})_{F \in \mathbb{F}^{b}}$ freely generates $RL_{v}(\mathbb{G}).$
\end{proposition}

\begin{proof}
The equality $\l_{c_{F},T} = [v,\underline{c}_{F}]_{T} c_{F} [v,\underline{c}_{F}]_{T}^{-1}$ allows us to see that $\l_{c_{F},T}$ is a lasso of meander $c_{F}$ and spoke $[v,\underline{c}_{F}]_{T}$. 

As seen in Lemma \ref{generate1} and Remark \ref{rq:nbfree}, $RL_{v}(\mathbb{G})$ is a free group of rank $\# \mathbb{F}^{b}$. Thus, it remains to show that $(\l_{c_{F},T})_{F \in \mathbb{F}^{b}}$ generates $RL_{v}(\mathbb{G})$. Using Lemma \ref{generate1}, it is enough to show that for every $e \in \mathbb{E}\setminus T$, $l_{e, T}$ is a product of elements of the form $\l_{c_{F},T}^{\pm 1}$. Let $e$ be an edge which is not in $T$. As $T$ is a tree, there exist $c$, $p$ and $p'$ three simple paths in $T$ which do not intersect, except at the point $\overline{c}=\underline{p}=\underline{p'}$, such that: 
\begin{itemize}
\item $[v,\underline{e}]_{T} = c \ \!p,$ 
\item $[v,\overline{e}]_{T} = c\ \!p', $
\item the meander $m$ of $l_{e, T}$ is $pep'^{-1}$. 
\end{itemize}
Let $v'$ be any point of $\mathbb{G}$ inside the meander of $l_{e,T}$. Since $T$ is a tree, $[v,v']_{T}$ must begin with the path $c$. If not, it would create a non degenerate loop in $T$. Define $\mathbb{G}'$ (resp. $T'$) the restriction of $\mathbb{G}$ (resp. $T$) to the closure of the inside of the meander $m$ of $l_{e,T}$. An example is given in Figure  \ref{restric1}. Let $\overline{c}$ be the root of $T'$. For every bounded face $F$ of $\mathbb{G}$ inside $m$, $\l_{c_{F},T}=c\ \!\l_{c_{F},T'} c^{-1}$ where $\l_{c_{F},T'}$ is the facial lasso based at $\overline{c}$ defined in $\mathbb{G}'$ thanks to $T'$. 

Applying the upcoming Lemma \ref{generate2} to $\mathbb{G}'$ endowed with $T'$, $m$ can be written as a product of lassos of the form $\l_{c_{F},T'}^{\pm 1}$, thus $l_{e, T}$ can be written as a product of lassos of the form $\l_{c_{F},T}^{\pm 1}$. 
\end{proof}

\begin{figure}
 \centering
  \includegraphics[width=270pt]{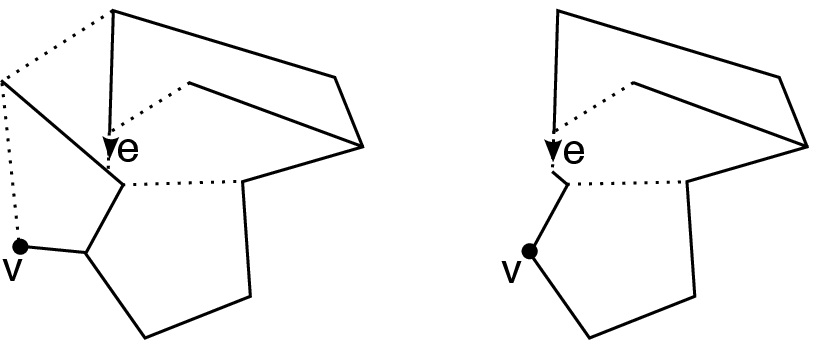}
 \caption{The restriction of $\mathbb{G}$ used in Proposition \ref{generate}.}
 \label{restric1}
\end{figure}

\begin{lemme}
\label{generate2}
Let us suppose that $v$ is actually a vertex on the boundary of the unbounded face. Let $(F_{i})_{i=1}^{ \# \mathbb{F}^{b}}$ be an enumeration of the bounded faces. Let $l_{\infty}$ be the only loop in $P(\mathbb{G})$ based at $v$ with anti-clockwise orientation which represents the non-oriented facial cycle of the unbounded face $F_{\infty}$. There exists a permutation $\sigma$ of $\{1,  ..., \# \mathbb{F}^{b}\}$ and an application $\epsilon: \{1, ...,n \} \to \{-1,1\}$ such that the equality  
\begin{align*}
\l_{c_{F_{\sigma(n)}},T}^{\epsilon(n)} \l_{c_{F_{\sigma(n-1)}},T}^{\epsilon(n-1)} ... \l_{c_{F_{\sigma(1)}},T}^{\epsilon(1)} = l_{\infty} 
\end{align*}
holds in $RL_{v}(\mathbb{G})$. Besides, for any integer $k\in \{1,...,n\}$, $\epsilon(k)$ is equal to $1$ if and only if $c_{F_{\sigma(k)}}$ is oriented anti-clockwise. 
\end{lemme}

\begin{proof}
In this proof, all the equalities will hold in $RL_{v}(\mathbb{G})$: from now on we will omit to specify this. The last assertion comes from a topological index argument. Let us suppose that there exists a permutation $\sigma$ of $\{1,  ..., \# \mathbb{F}^{b}\}$ and an application $\epsilon: \{1, ...,n \} \to \{-1,1\}$ such that $\l_{c_{F_{\sigma(n)}},T}^{\epsilon(n)} \l_{c_{F_{\sigma(n-1)}},T}^{\epsilon(n-1)} ... \l_{c_{F_{\sigma(1)}},T}^{\epsilon(1)} = l_{\infty}.$ We can compute the index of $l_{\infty}$: $\mathfrak{n}_{l_{\infty}} = \mathfrak{n}_{\l_{c_{F_{\sigma(n)}},T}}{\epsilon(n)}+ ...+\mathfrak{n}_{\l_{c_{F_{\sigma(1)}},T}}{\epsilon(1)}. $ For any  bounded face of $\mathbb{G}$, namely $F$, we can evaluate the last equality for any $x\in F$. This implies that for any $ i \in \{1,...,n\}$, $1 = \mathfrak{n}_{\l_{c_{F_{\sigma(i)}},T}}\epsilon(i),$ hence the second assertion.

Let us show the first part of Lemma \ref{generate2}. The proof goes by induction on the number $\# \mathbb{F}^{b}$ of bounded faces. For a graph with only one bounded face $F$ the result is true since $l_{\infty} = l_{c_{F}}^{\epsilon}$, with $\epsilon$ being $-1$ of $1$, depending on the orientation of~$c_{F}$. Let us suppose that $\# \mathbb{F}^{b} >1$. There exists a unique way to write $l_{\infty}$ as $p_{1}e_{1}p_2 e_2... e_n p_{n}$ with $p_{i}$ a path in $T$ (which can be constant) and $e_{i}$ an edge in $\mathbb{E}\setminus T$ bounding $F_{\infty}$ for any $i$ in $\{1, ..., n\}$. Let us decompose the graph $\mathbb{G}$ in $n$ subgraphs. The $i$-th subgraph $\mathbb{G}_{i}$ is the part of $\mathbb{G}$ which is inside the meander $m_{i}$ of $l_{e_{i},T}$. The vertex $v_{i}=\underline{m}_{i}$ will be the chosen point on the boundary of $\mathbb{G}_i$, then: 
\begin{itemize}
\item the restriction $T_{i}$ of $T$ to $\mathbb{G}_{i}$ is still a spanning tree of $\mathbb{G}_{i}$, 
\item for any bounded face $F$ in $\mathbb{G}_i$, $\l_{c_F, T} = [v, v_{i}]_{T} \l_{c_F, T_i}[v, v_{i}]_{T}^{-1}$, where $\l_{c_F, T_i}$ is the facial lasso based at $v_i$ defined in the graph $\mathbb{G}_i$.
\end{itemize}

If $n>1$, each of the graphs $\mathbb{G}_{i}$ has strictly less than $\# \mathbb{F}^{b}$ bounded faces. An example is drawn in Figure \ref{Decoupe2}. By induction the result holds for $\mathbb{G}_{i}$, based at $v_{i}$ and endowed with $T_{i}$. It follows that $l_{e_{i},T}$, which is equal to $[v,v_{i}]_{T} m_{i}[v,v_{i}]_{T} ^{-1}$, is an ordered product of all the facial lasso (or their inverse) associated with the faces $F$ in $\mathbb{G}_{i}$. But the family $(\mathbb{G}_{i})_{i}$ induces a partition of the set of bounded faces of $\mathbb{G}$. As $l_{\infty}=l_{e_{1},T}...l_{e_{n},T}$ it is now clear that the result holds. 
\begin{figure}
 \centering
  \includegraphics[width= 270pt]{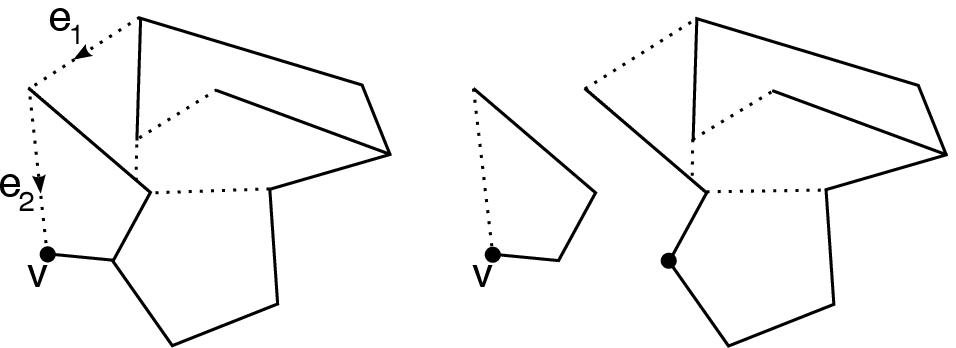}
 \caption{Decomposition when $n>1$.}
 \label{Decoupe2}
\end{figure}

It remains the case where $n=1$. In this case, $l_{\infty} = p e p'$, with $p$ and $p'$ two simple paths in $T$ and $e$ an edge in $\mathbb{E} \setminus T$ bounding $F_{\infty}$. We have to find a new way of decomposing $\mathbb{G}$ in order to apply the induction hypothesis. Let $F$ be the only bounded face which is surrounded by $e$. We can suppose $c_{F}$ turning clockwise thus it can be decomposed as $c_{F}=ae^{-1}b$. Consider the loop: $l =[v,\underline{c}_{F}]_{T} a e^{-1} [\underline{e},v]_{T}$. This is a lasso and as before we consider $\mathbb{G}_{l}$, which is the graph $\mathbb{G}$ restricted to the closure of the interior of the meander $m_{l}$ of $l$. We base this graph at $v_{l}=\underline{m}_{l}$. 

First of all, if $\mathbb{G}_l$ has the same number of faces than $\mathbb{G}$, as in Figure \ref{Decoupe4}, then the equality $l = l_\infty^{-1}$ must hold and thus one has $l_\infty = ([v,\underline{e}]_{T} b [\underline{c}_{F},v]_{T} )\l_{c_F}^{-1}$. The path $\tilde{l} = [v,\underline{e}]_{T} b [\underline{c}_{F},v]_{T}$ is a loop based at $v$ which represents the non-oriented facial cycle of the unbounded face of the graph obtained when one removes $e$ to $\mathbb{G}$. On this graph, $T$ is still a spanning tree and this graph has one less bounded face. The induction hypothesis allows us to conclude. 

\begin{figure}
 \centering
  \includegraphics[width= 270pt]{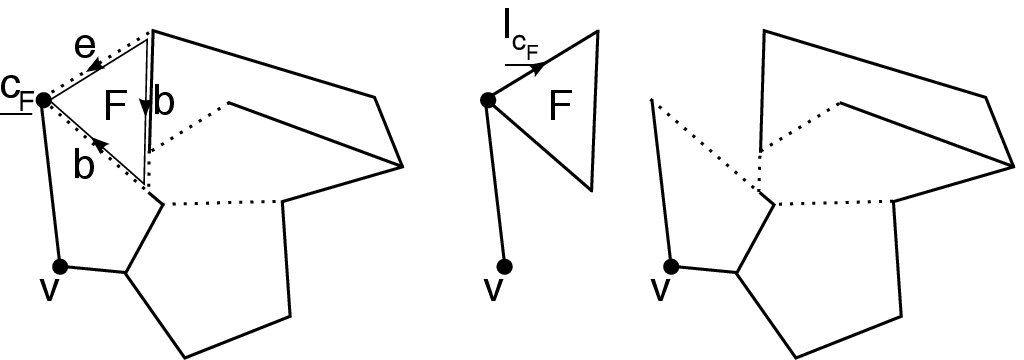}
 \caption{Decomposition when $\mathbb{G}_l$ has as many faces as $\mathbb{G}$.}
 \label{Decoupe4}
\end{figure}

In the case where $\mathbb{G}_l$ has less faces than $\mathbb{G}$, as in Figure \ref{Decoupe3}, the restriction of $T$ in $m_{l}$ is not a spanning tree. We will define $T_{l}$ to be the restriction of $T$ in $\overline{{\sf Int}(m_{l})}$ to which one adds all the edges in the path $a$ and we set its root equal to $v_{l}$. With these modifications, $T_{l}$ is a spanning tree of $\mathbb{G}_{l}$ and for any bounded face $F$ of $\mathbb{G}_l$, we have $\l_{c_F, T} = s_l \l_{c_F, T_l}s_l^{-1}$, where $\l_{c_F, T_l}$ is the facial lasso based at $v_l$ defined in $\mathbb{G}_l$ thanks to $T_l$ and $s_{l}$ is the spoke of $l$. We define also $l' = [v,\underline{c}_{F}]_{T} a [\overline{e},v]_{T}$ and $\mathbb{G}_{l'}$ the part of $\mathbb{G}$ inside $\overline{{\sf Int}(m_{l'})}$ of $l'$. The restriction of $T$ to $\mathbb{G}_{l'}$ is denoted $T_{l'}$. In this case $T_{l'}$ is a spanning tree of $\mathbb{G}_{l'}$. Besides, for any bounded face $F$ in $\mathbb{G}_{l'}$, $\l_{c_F, T} = s_{l'} \l_{c_F, T_{l'}}s_{l'}^{-1}$, where $\l_{c_F, T_{l'}}$ is the facial lasso based at $v_{l'}= \underline{m}_{l'}$ defined in $\mathbb{G}_{l'}$ thanks to $T_{l'}$ and $s_{l'}$ is the spoke of $l'$. In the case we are studying, $\mathbb{G}_{l}$ and $\mathbb{G}_{l'}$ have strictly less bounded faces than $\mathbb{G}$, thus we can apply the induction hypothesis. Using the link between facial lassos in $\mathbb{G}_{l}$ (resp. in $\mathbb{G}_{l'}$) and in $\mathbb{G}$, there exists an ordering on the bounded faces of $\mathbb{G}_{l}$ (resp. $\mathbb{G}_{l'}$) such that $l$ (resp. $l'$) is the ordered product of the facial lassos $(\l_{c_{F},T}^{\pm1})_{F \in \mathbb{F}^{b}_{l}}$ (resp. $(\l_{c_{F},T}^{\pm1})_{F\in\mathbb{F}^{b}_{l'}}$), where $\mathbb{F}_l^{b}$ (resp. $\mathbb{F}_l^{b'}$) is the set of bounded faces of $\mathbb{G}_l$ (resp. $\mathbb{G}_{l'}$).  Since $l_{\infty} = [v,\underline{e}]_{T} e [\overline{e},v]_{T} = l^{-1}l'$ and as any bounded face $F$ of $\mathbb{G}$ is either a bounded face of $\mathbb{G}_l$ or $\mathbb{G}_{l'}$, we can conclude that there exists a permutation $\sigma$ of $\{1,  ..., \# \mathbb{F}^{b}\}$ and an application $\epsilon: \{1, ...,n \} \to \{-1,1\}$ such that: 
\begin{align*}
\l_{c_{F_{\sigma(n)}},T}^{\epsilon(n)} \l_{c_{F_{\sigma(n-1)}},T}^{\epsilon(n-1)} ... \l_{c_{F_{\sigma(1)}},T}^{\epsilon(1)} = l_{\infty}.
\end{align*}
\begin{figure}
 \centering
  \includegraphics[width= 270pt]{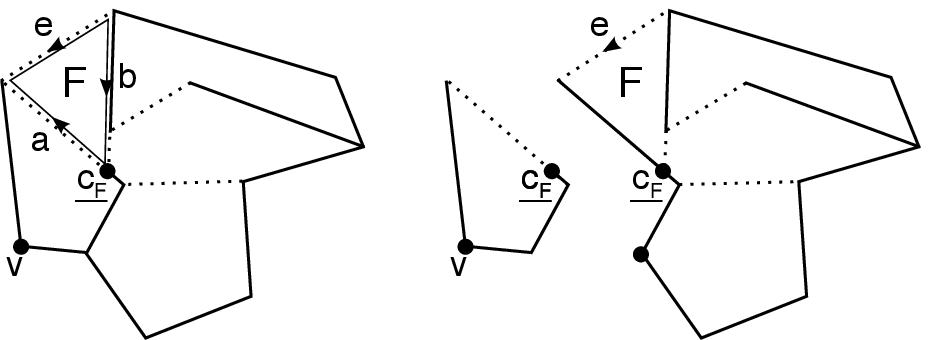}
 \caption{Decomposition when $\mathbb{G}_l$ has less faces than $\mathbb{G}$.}
 \label{Decoupe3}
\end{figure}
This allows us to conclude.
\end{proof}

Let us finish with a proposition which will be needed in order to prove Proposition \ref{creaYMpure2}. 

\begin{proposition}
\label{combi1}
Let $l_1$ and $l_2$ be two simple loops in $\mathbb{G}$ such that $\overline{{\sf Int(}l_1)}$ and $\overline{{\sf Int(}l_2)}$ are disjoint. There exists a spanning tree $T$, rooted at $v$, such that for any family of facial loops $\left(c_F\right)_{F \in \mathbb{F}^{b}}$ the following assertions hold: 
\begin{enumerate}
\item for every loop $l$ in $P(\mathbb{G})$ included in $\overline{{\sf Int(}l_1)}$, $\Big[[v,\underline{l}]_T l [v,\underline{l}]_T^{-1}\Big]_{\simeq}$ is a product in $RL_{v}(\mathbb{G})$ of elements of $ \left\{\l_{c_F,T}^{\pm 1} ; F \in \mathbb{F}^{b}, F \subset \overline{{\sf Int(}l_1)} \right\}$, 
\item for every loop $l$ in $P(\mathbb{G})$ included in $\overline{{\sf Int(}l_2)}$, $\Big[[v,\underline{l}]_T l [v,\underline{l}]_T^{-1}\Big]_{\simeq}$ is a product in $RL_{v}(\mathbb{G})$ of elements of $ \left\{\l_{c_F,T}^{\pm 1} ; F \in \mathbb{F}^{b}, F \subset \overline{{\sf Int(}l_2)}\right\}$.
\end{enumerate}
\end{proposition}

\begin{proof}
We can decompose $l_1$ and $l_2$ as a concatenation of edges of $\mathbb{G}$: $l_1 =e_1^{1}...e_1^{n}$ and $l_2 =e_2^{1}... e_2^{m}.$ The set $\left\{e_1^{1}, ... , e_1^{n-1}, e_2^{1}, ..., e_2^{m-1}\right\}$ can be extended as a spanning tree $T$ of the graph $\mathbb{G}$, rooted at $v$. Thanks to  the construction, the restriction  $T_1$ of $T$ to $\overline{{\sf Int(}l_1)}$ is a spanning tree of the restriction $\mathbb{G}_1$ of $\mathbb{G}$ to $\overline{{\sf Int(}l_1)}$. We set $v_1$ to be equal to $\underline{e_1^{1}}$: this is the root of $T_1$. Applying Proposition \ref{generate}, for any loop $l$ inside $l_1$, $\Big[[v_1,\underline{l}]_{T_1} l [v_1,\underline{l}]_{T_1}^{-1}\Big]_{\simeq}$ is a product of elements of $\left\{ \l_{c_F, T_1}^{\pm 1}; F \in \mathbb{F}^{b}, F \subset \overline{{\sf Int(}l_1)}\right\}$. 
For any vertex $w$ in $\mathbb{G}_1$, $[v,w]_T = [v,v_1]_T[v_1,w]_{T_1}$ in $RL_{v}(\mathbb{G})$. Thus for any face $F \in \mathbb{F}^{b}$ such that $F$ is included in $\overline{{\sf Int(}l_1)}$, $\l_{c_F, T} = [v,v_1]_{T} \l_{c_F,T_1}[v,v_1]_{T}^{-1}$ in $RL_{v}(\mathbb{G})$ and for any loop $l$ in $P(\mathbb{G})$ included in $\overline{{\sf Int(}l_1)}$, $[v,\underline{l}]_{T} l [v_1, \underline{l}]_{T}^{-1}$ is equal in $RL_{v}(\mathbb{G})$ to $[v,v_1]_{T}[v_1,\underline{l}]_{T_1} l [v_1,\underline{l}]_{T_1}^{-1}[v,v_1]_{T}^{-1}.$
Thus the loop $[v,\underline{l}]_{T} l [v_1, \underline{l}]_{T}^{-1}$ is a product of elements of $ \left\{\l_{c_F,T}^{\pm 1} ; F \in \mathbb{F}^{b}, F \subset \overline{{\sf Int(}l_1)}\right\}$. The same holds for the loop $l_2$. 
\end{proof}

\section[Unicity and construction]{Random holonomy fields and the group of reduced loops}

Let us explain two applications of the group of reduced loops which concern the uniqueness and construction of random holonomy fields on the plane.

\begin{proposition}
\label{unic}
Let $\mu$ and $\nu$ be two stochastically continuous measures on $\big(\mathcal{M}ult(P(\mathbb{R}^{2}), G), \mathcal{B}\big)$ which are invariant by gauge transformations. The two assertions are equivalent: 
\begin{enumerate}
\item $\mu$ and $\nu$ are equal, 
\item there exist  $v\in \mathbb{R}^{2}$ and $A_v$ a good subspace of $L_v(\mathbb{R}^{2})$, such that for any finite planar graph $\mathbb{G}$ in $\mathcal{G}\big(A_v\big)$ which has $v$ as a vertex, there exist  a rooted spanning tree $T$ and a family of facial loops $\left(c_F\right)_{F \in \mathbb{F}^{b}}$ of $\mathbb{G}$ such that the law of $(h(\l_{c_F, T}))_{F \in \mathbb{F}^{b}}$ is the same under $\mu$ and under $\nu$.
\end{enumerate}
\end{proposition}

\begin{proof}
It is an easy application of the multiplicative property of random holonomy fields, Proposition \ref{unicite1} and Proposition \ref{generate}. 
\end{proof}

\begin{proposition}
\label{constructfinal}
Suppose that for any finite planar graph $\mathbb{G}$ in $\mathcal{G}\big({\sf Aff}(\mathbb{R}^{2})\big)$, we are given a diagonal conjugation-invariant measure $\mu_\mathbb{G}$ on $G^{\#\mathbb{F}^{b}}$, a rooted spanning tree $T$ and a family of facial loops $(c_F)_{F \in \mathbb{F}^{b}}$. For any finite planar graph $\mathbb{G}$, there is only one possibility to extend $\mu_\mathbb{G}$ as a gauge-invariant random field on $P(\mathbb{G})$, which will be also denoted by $\mu_\mathbb{G}$, such that the law of $\left(h(\l_{c_F,T})\right)_{F \in \mathbb{F}^{b}}$ under $\mu_\mathbb{G}$ is the same as the law of the canonical projections under the measure $\mu_\mathbb{G}$ on $G^{\#\mathbb{F}^{b}}$. If $\left(\mu_\mathbb{G}\right)_{\mathbb{G} \in \mathcal{G}({\sf Aff}(\mathbb{R}^{2}))}$ is uniformly locally stochastically $\frac{1}{2}$-H\"{o}lder continuous, if for any finite planar graphs $\mathbb{G}$ and $\mathbb{G}'$ in $\mathcal{G}\big({\sf Aff}(\mathbb{R}^{2})\big)$ such that $\mathbb{G} \preccurlyeq \mathbb{G}'$ and for any family of facial loops $(c_F)_{F \in \mathbb{F}^{b}}$ of $\mathbb{G}$, $\left(h(\mathbf{l}_{c_F, T})\right)_{F \in \mathbb{F}^{b}}$ has the same law under $\mu_{\mathbb{G}}$ as under $\mu_{\mathbb{G}'}$, then there exists a unique stochastically continuous random holonomy field $\mu$ on the plane such that for any finite planar graph $\mathbb{G}$ in $\mathcal{G}\big({\sf Aff}(\mathbb{R}^{2})\big)$, for any rooted spanning tree $T$ and for any family of facial loops $(c_F)_{F \in \mathbb{F}^{b}}$ of $\mathbb{G}$, the law of $\left(h(\mathbf{l}_{c_F, T})\right)_{F \in \mathbb{F}^{b}}$ is the same under $\mu$ as under $\mu_\mathbb{G}$. 
\end{proposition}

\begin{proof}
For any finite planar graph $\mathbb{G}$ in $\mathcal{G}\big({\sf Aff}(\mathbb{R}^{2})\big)$ and any vertex $v$ of $\mathbb{G}$, there exists a natural measurable function from ${\sf Hom}\big(RL_{v}(\mathbb{G}),G^{\vee}\big)$ to the multiplicative functions $\mathcal{M}ult_{P(\mathbb{G})}(L_{v}(\mathbb{G}),G)$: we can transport any measure from the first space to the second. Using the freeness of the generating families $\mathbf{l}_{c_F, T}$, the multiplicity property of random holonomy fields and Proposition \ref{crea1}, we can extend $\mu_\mathbb{G}$ as a gauge-invariant random field on $P(\mathbb{G})$. This gauge-invariant random field does not depend on the choice of $v$. An application of Proposition \ref{exten4} and Lemma \ref{goodspace} allows us to construct the desired $\mu$. The uniqueness of $\mu$ is a consequence of Proposition \ref{unic}.  
\end{proof}

\part{Construction of Planar Markovian Holonomy Fields}

\chapter[Braids and Probabilities I]{Braids and Probabilities I: an Algebraic Point of View and Ginite Random Sequences}

For any finite planar graph $\mathbb{G}$, we have constructed in the last section a set of generating family of facial lassos of $\mathbb{G}$: what is the transformation which sends one generating family to an other? It has to be noticed that, as soon as the root of the spanning tree is chosen, for any generating family of lassos constructed in the last section, their product, up to some suitable permutation, is always equal to the same loop. This remark and Artin's theorem, Theorem \ref{artin}, motivate the study of the group of braids. 

\label{Braids}

\section{Generators, relations, actions}

A geometric definition of the braid group was given in the introduction. One can also define the braid group using a generator-relation presentation. Let $n$ be an integer greater than 2.
\begin{definition}
The braid group with $n$ strands $\mathcal{B}_{n}$ is the group with the following presentation: 
\begin{align*}
\left<\big(\beta_{i}\big)_{i=1}^{n-1} \text{\huge{\textbar} \normalsize} \forall i, j \in \{1,..., n-1\}, \begin{matrix} | i - j | = 1\Longrightarrow \beta_{i}\beta_{j}\beta_{i}= \beta_{j}\beta_{i}\beta_{j}\\\!\!\!\!\!\!\!\!\!\!\!| i-j | > 1\Longrightarrow \beta_{i} \beta_{j} = \beta_{j} \beta_{i} \end{matrix} \right>.
\end{align*}
\end{definition}

The elements $(\beta_{i})_{i=1}^{n-1}$ we defined in the introduction satisfy the braid group relations. An example of the first relation between $\beta_{i}$ and $\beta_{j}$ when $| i-j| = 1$ is given in Figure \ref{fig:braidrel}.

\begin{figure}
 \centering
  \includegraphics[width=150pt]{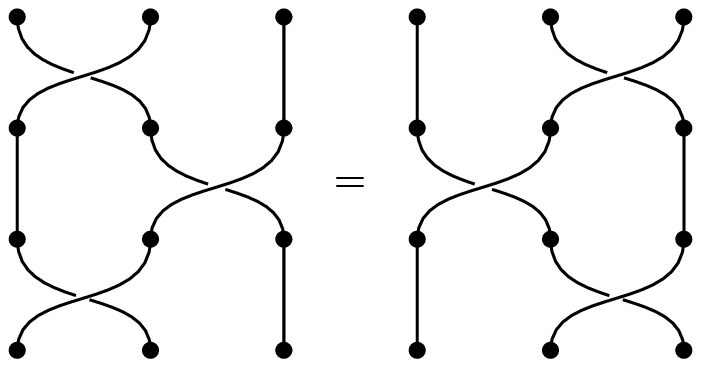}
 \caption{The braid relation}
 \label{fig:braidrel}
\end{figure}

This presentation of the braid group is not intuitive, yet it allows us to recall some natural actions of the braid group $\mathcal{B}_{n}$: one on the free group of rank $n$ and one on $G^{n}$. Let $\mathbb{F}_{n}$ be the free group of rank $n$ generated by $e_{1},   ..., e_{n}$ and let $G$ be any group. 

\begin{definition}
\label{actionlibre}
The natural action of $\mathcal{B}_{n}$ on $\mathbb{F}_{n}$ is given by: 
\begin{align*}
\beta_{i} e_{i} &= e_{i+1}, \\
\beta_{i} e_{i+1} &= e_{i+1} e_{i} e_{i+1}^{-1}, \\
\beta_{i} e_{j} &= e_{j}, \text{ for any j } \notin \{i,i+1\}. 
\end{align*} 
\end{definition}
There exists a diagrammatic way to compute the action: one puts $e_1$, ..., $e_n$ at the bottom of a diagram representing $\beta$, then propagates these $e_1, ..., e_n$ in the diagram from the bottom to the top with the rule that, at each crossing, the value on the string which is behind does not change and the value on the upper string is conjugated by the value of the other so that the product from right to left remains unchanged. At the end one gets a $n$-uple $(f_1,...,f_n)$ at the top of the diagram: the braid sends $e_i$ on $f_i$.

\begin{definition}
The natural action of $\mathcal{B}_{n}$ on $G^{n}$ is given by: 
\begin{align}
\label{action}
\beta_{i} \bullet (x_{1},  ...,x_{i-1},x_{i},x_{i+1},  ...,x_{n}) = (x_{1},  ...,x_{i-1},x_{i} x_{i+1} x_{i}^{-1}, x_{i},...,x_{n}), 
\end{align}
for any integer $i\in\{1,..., n-1\}$ and $n$-tuple $(x_{i})_{i=1}^{n}$ in $G^{n}$.
\end{definition}
There exists also a diagrammatic way to compute the action: one puts $x_1$, ..., $x_n$ at the upper part of a diagram representing $\beta$, then propagates these $x_1$, ..., $x_n$ in the diagram from the top to the bottom with the rule that, at each crossing, the value on the string which is behind does not change and the value on the upper string is conjugated by the value of the other so that the product from left to right remains unchanged. At the end one gets a $n$-uple $(y_1,...,y_n)$ at the top of the diagram: the braid sends $(x_1,...,x_n)$ on $(y_1,...,y_n)$.

Let $h$ be a $G$-valued multiplicative function on the free group. This means that $h\left(x^{-1}\right)=h(x)^{-1}$ and $h(xy) = h(y)h(x)$ for any $x$ and $y$ in $\mathbb{F}_{n}$. For any $n$-uple $(f_1, ..., f_n)$ of elements of $\mathbb{F}_{n}$, we define $h(f_1,...,f_n) = (h(f_1), ..., h(f_n))$. The following lemma shows how both actions are linked: it is a consequence of the diagrammatic formulation of both actions. 

\begin{lemme}
\label{actionmieux}
For any braid $\beta \in \mathcal{B}_n$, $h\left(\beta \bullet (e_1,...,e_n) \right) = \beta^{-1} \bullet h(e_1,...,e_n).$
\end{lemme}

With the $n$-diagrams picture in mind, it is obvious that the application, which sends a braid on the permutation obtained by erasing the information at each crossing, is a homomorphism: it is the one which sends $\beta_{i}$ on the transposition $(i, i+1)$ for any integer ~$i \in \{1,  ...,n-1\}$. 

\begin{lemme}
\label{braidperm}
The operation of erasing the information at each crossing induces a natural homomorphism from $\mathcal{B}_{n}$ to $\mathfrak{S}_{n}$. We will denote the image of $\beta$ by $\sigma_{\beta}.$
\end{lemme}

\section{Artin theorem and the group of reduced loops}
For any braid $\beta$ with $n$ strands, the action $a_{\beta}$ of $\beta$ on $\mathbb{F}_{n}$ is an automorphism of $\mathbb{F}_{n}$: there exists a morphism from $\mathcal{B}_{n}$ in $\mathbb{A}ut(\mathbb{F}_{n})$ which is moreover injective. In \cite{Artin} and \cite{Artin2}, Artin gave a sufficient and necessary condition for an automorphism of $\mathbb{F}_{n}$ to be the induced action of a braid in $\mathcal{B}_{n}$.
\begin{theorem}
\label{artin}
An automorphism $\mathbf{a}$ of $\mathbb{F}_{n}$ is the induced action of a braid in $\mathcal{B}_{n}$ if and only if the two following conditions hold: 
\begin{description}
\item[Conjugacy property] for any $i$ in $\{1,  ..., n\}$, $\mathbf{a}(e_{i})$ is in the same conjugacy class as one of the elements of $(e_{j})_{j=1}^{n}$. 
\item[Product invariance] $\mathbf{a}(e_{n}\ ...\ e_{1}) = e_{n}\ ...\ e_{1}$. 
\end{description}
\end{theorem}

\begin{remarque}
\label{braidconj}
Let $\beta$ be a braid in $\mathcal{B}_{n}$ and $a_{\beta}$ the induced action on $\mathbb{F}_{n}$. For each $i$ in $\{1,   ..., n\}$, $a_{\beta}(e_{i})$ is conjugated to $e_{\sigma_{\beta}(i)}$ and this property characterizes~$\sigma_{\beta}$.
\end{remarque}

Let $\mathbb{G}=(\mathbb{V}, \mathbb{E}, \mathbb{F})$ be a finite planar graph, $v$ be a vertex of $\mathbb{G}$, $T$ and $T'$ be two spanning trees of $\mathbb{G}$ rooted at $v$, $\left(c_F\right)_{F \in \mathbb{F}^{b}}$ and $(c'_F)_{F \in \mathbb{F}^{b}}$ be two families of oriented facial loops oriented anti-clockwise. There exist two freely generating families of the group of reduced loops of $\mathbb{G}$ associated with $\left(c_F\right)_{F \in \mathbb{F}^{b}}$ and $(c'_F)_{F \in \mathbb{F}^{b}}$. Using Artin theorem, we can characterize the transformation which sends one family on the other. 

\begin{proposition}
\label{tresseetgene}
There exists an enumeration of the bounded faces $(F_i)_{i=1}^{\# \mathbb{F}^{b}}$ and a braid $\beta$ in $\mathcal{B}_{\#\mathbb{F}^{b}}$ such that: 
\begin{align*}
\beta \bullet \big(\l_{c_{F_i}, T}\big)_{i=1}^{\#\mathbb{F}^{b}} = \big(\l_{c'_{F_{\sigma(i)}}, T'}\big)_{i=1}^{\#\mathbb{F}^{b}}, 
\end{align*}
where $\sigma = \sigma_\beta$ and where $\beta$ is seen as acting on the free group generated by $\big(\l_{c_{F_i}, T}\big)_{i=1}^{\#\mathbb{F}^{b}}$.
\end{proposition}

\begin{proof}
For any bounded face $F$ of $\mathbb{G}$, the first part of Proposition \ref{generate} asserts that $\l_{c_{F},T}$ and $\l_{c'_{F},T'}$ are facial lassos based at $v$ whose meanders represent the facial cycle~$\partial F$ oriented anti-clockwise. By Lemma \ref{conjug}, we deduce that $\l_{c'_{F},T'}$ is conjugated to $\l_{c_{F},T}$ in $RL_{v}(\mathbb{G})$. 
Besides, thanks to Lemma \ref{generate2}, we can find an enumeration of the bounded faces $(F_{i})_{i=1}^{\# \mathbb{F}^{b}}$ and a permutation $\sigma$ of $\{1,  ..., \# \mathbb{F}^{b}\}$ such that: 
\begin{enumerate}
\item $\l_{c_{F_{n}},T}\ \l_{c_{F_{n-1}},T}...\ \l_{c_{F_{1}},T} = l_{\infty}, $
\item $\l_{c'_{F_{\sigma(n)}},T'}\ \l_{c'_{F_{\sigma(n-1)}},T'}...\ \l_{c'_{F_{\sigma(1)}},T'} = l_{\infty},$
\end{enumerate}
in $RL_v(\mathbb{G})$, where $l_{\infty}$ is the facial loop based at $v$, turning anti-clockwise, representing the non-oriented facial cycle $\partial F_{\infty}$. Besides, Proposition \ref{generate} tells us that both $(\l_{c_{F}})_{F \in \mathbb{F}^{b}}$ and $(\l_{c'_{F}})_{F \in \mathbb{F}^{b}}$ are free families of generators of the free group $RL_{v}(\mathbb{G})$. A natural automorphism of $RL_{v}(\mathbb{G})$ is defined by:
\begin{align*}
\forall i\in \{1,   ..., \# \mathbb{F}^{b} \},\ \mathbf{a}(\l_{c_{F_{i}}, T}) = \l_{c'_{F_{\sigma(i)}},T'}. 
\end{align*}
This automorphism of free group satisfies the conditions of Artin's theorem given in  Theorem~\ref{artin}. There exists a braid $\beta$ such that $\mathbf{a}$ is equal to $\mathbf{a}_{\beta}$, the action induced by $\beta$ on the free group $RL_{v}(\mathbb{G})$ with free generators $(\l_{c_{F},T})_{F \in \mathbb{F}^{b}}$. Using Remark \ref{braidconj}, it is straightforward to see that $\sigma$ is equal to $\sigma_{\beta}$. \end{proof}

\section{Braids and finite sequence of random variables}
In the last section, the transformations between families of loops of the form $\left(\l_{c_F,T}\right)_{F\in\mathbb{F}^{b}}$ have been characterized. In the context of random holonomy fields, a random variable is associated with any loop: it is natural to study the action of the braid groups on finite sequence of random variables. When one has to deal with non-commutative random variables (i.e. random variables in a non-commutative group), this action is in some sense more appropriate than the symmetrical group action which is often studied in the mathematical literature. This leads to a theory of braidability which is more efficient than the exchangeability concept for sequences of random variables in a non-commutative group. Let $n$ be an integer strictly greater than $1$ and $G$ be an arbitrary topological group.

\begin{definition}
The braid group $\mathcal{B}_n$ acts on the set of $n$-tuple of $G$-valued random variables according to the formula:
\begin{align}
\label{action1}
\beta_{i} \bullet\! (X_{1},  ..., X_{i-1}, X_{i}, X_{i+1},  ..., X_{n})\! =\! (X_{1},  ..., X_{i-1}, X_{i}X_{i+1}X_{i}^{-1}, X_i,..., X_{n})
\end{align}
for any $i \in \{1,  ...,n-1 \}$.
\end{definition} 

Recall the notation $\sigma_{\beta}$ which was defined in Lemma \ref{braidperm}.

\begin{definition}
\label{def:purelyinv}
Let $(X_{1},   ..., X_{n})$ be a finite sequence of $G$-valued random variables. It is {\em purely invariant by braids} if for any braid $\beta \in \mathcal{B}_{n}$ one has the equality in law: 
\begin{align*}
\beta \bullet (X_{1},   ..., X_{n}) = \sigma_{\beta} \bullet (X_{1},  ..., X_{n}), 
\end{align*}
where $\sigma \bullet (X_{1},  ..., X_{n}) = \left(X_{\sigma^{-1}(1)},  ..., X_{\sigma^{-1}(n)}\right)$ for any permutation $\sigma \in \mathfrak{S}_n$.

It is {\em invariant by braids} if for any braid $\beta \in \mathcal{B}_n$, one has the equality in law: 
\begin{align*}
\beta \bullet (X_{1},  ..., X_{n}) = (X_{1},  ..., X_{n}). 
\end{align*}
 \end{definition}

Let us recall that if $m$ is a probability measure on $G$, the support of $m$, denoted by ${\sf Supp}(m)$, is the smallest closed subset of $G$ of measure $1$ for $m$. The closure of the subgroup generated by the support of $m$ is denoted by $H_{m}$. If $X$ is a $G$-valued random variable and $m$ its law, we define ${\sf Supp}(X) = {\sf Supp}(m)$ and $H_X= H_m$. Let $T$ be a finite index set such that $\# T \geq 2$.

\begin{definition}
Let $\left(X_t\right)_{t \in T}$ be a sequence of $G$-valued random variables. We say that $\left(X_t\right)_{t \in T}$ is {\em auto-invariant by conjugation} if for any different elements $i$ and $j$ in $T$ and for any $g \in~{\sf Supp}(X_j)$, we have the equality in law: 
\begin{align} \label{inv}
g X_i g^{-1} = X_i. 
\end{align}
This definition can be extended to collections of measures on $G$. 
\end{definition}

The first result on random sequences which are purely invariant by braids is the following proposition. 
\begin{proposition}
\label{quasi-inv}
A finite sequence of independent $G$-valued random variables is auto-invariant by conjugation if and only if it is purely invariant by braids. 
\end{proposition}

\begin{proof} 
Let $(X_1, ..., X_n)$ be a finite sequence of $G$-valued random variables which are independent. Let us suppose that $(X_1, ..., X_n)$ is auto-invariant by conjugation. Since $\beta \mapsto \sigma_{\beta}$ is a morphism and using the independence of the variables, we just have to show that $(X_{2}, X_{2}^{-1}X_{1}X_{2})$ and $(X_{2}, X_{1})$ have the same law. This result follows from the independence of the variables $X_1$ and $X_2$ and from the invariance by conjugation of the law of $X_{1}$ by any element $g$ in the support of $X_2$. 

Now, let us suppose instead that $(X_1, ..., X_n)$  is purely invariant by braids. Let $i<j$ be two integers in $\{1,..., n\}$. Let $\beta_{(i,j)}$ be the braid defined by: 
\begin{align*}
\beta_{(i,j)} = \beta_{i}^{-1}...\beta_{j-2}^{-1}\beta_{j-1}...\beta_{i  }.
\end{align*}
An example of such a braid is shown in Figure \ref{Quasi-invariance}. By considering only the $i^{th}$ and $j^{th}$ positions in the equality in law $\beta_{(i,j)} \bullet \left(X_1, ...,X_n\right) = \sigma_{\beta_{i,j}} \bullet  \left(X_1, ...,X_n\right)$, we get the following equality in law: $$(X_i,X_j) = (X_i, X_i X_j X_i^{-1}).$$ By disintegration and using the independence of the variables, one gets the desired result.
\end{proof}

\begin{figure}
 \centering
  \includegraphics[width=250pt]{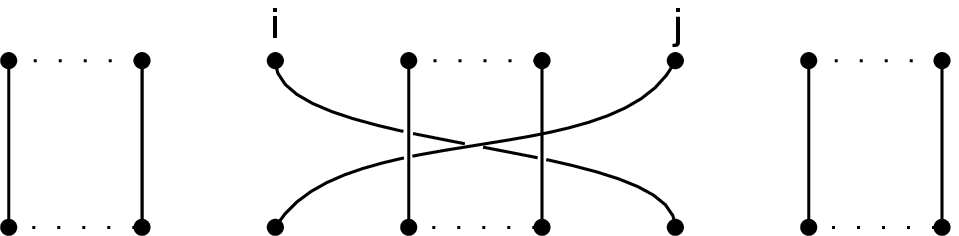}
 \caption{The braid $\beta_{(i,j)}$.}
 \label{Quasi-invariance}
\end{figure}

The proof of Proposition \ref{quasi-inv} is straightforward, but looking at the following equality in law: $(X_{2}^{-1}X_{1}X_{2}, X_{2}^{-1}X_{1}^{-1}X_{2} X_{1} X_{2}) = (X_{1}, X_{2}),$ where $(X_1,X_2)$ is an auto-invariant by conjugation couple of random variables, one can see that it gives identities which, at first glance, do not seem trivial. A last remark to be made about Proposition \ref{quasi-inv} is that there exist finite sequences of non-independent $G$-valued random variables which are purely invariant by braids.

\chapter{Planar Yang-Mills Fields}
\label{planaryangmills}

In this chapter, we construct a family of planar Markovian holonomy fields: the planar Yang-Mills fields. First, we construct the pure planar Yang-Mills fields. Then, in Section \ref{generalisationdeconstruction}, we generalize the construction in order to construct all planar Yang-Mills fields. 

\section{Construction of pure planar Yang-Mills fields}
\label{sec:constructionpure}

 In order to construct pure planar Yang Mills fields, given any L\'{e}vy process $Y$ which is invariant by conjugation by $G$, we define, in Proposition \ref{defplanarHF}, for any finite planar graph $\mathbb{G}$, an random holonomy field on $P(\mathbb{G})$ associated to $Y$. In Propositions \ref{creaYMpure} and \ref{creaYMpure2}, we show that these random holonomy fields allow us to define a family of random holonomy fields on $\mathbb{R}^{2}$ which is a strong planar Markovian holonomy field. In Section \ref{generalisationdeconstruction}, we weaken the condition on the LŽvy process by using our results about the extension of the structure group. First of all, let us recall the definition of L\'{e}vy processes.

\begin{definition}
\label{Levydefinition}
A L\'{e}vy process $(Z_t)_{t \geq 0}$ is a random  c\`{a}dl\`{a}g process from $\mathbb{R}^{+}$ to $G$, with independent and stationary right increments. This means: 
\begin{enumerate}
\item $\forall\ 0 \leq t_0 < ...< t_n, \big(Z_{t_{i-1}}^{-1}Z_{t_i}\big)_{i=1}^{n}$ are independent, 
\item $\forall\ 0\leq s<t$, $Z_s^{-1}Z_t$ has the same law as $Z_{t-s}$. 
\end{enumerate}
We say that $(Z_t)_{t \geq 0}$ is invariant by conjugation by $G$, or conjugation-invariant, if and only if for any $g \in G$, the process $(g^{-1}Z_tg)_{t \geq 0}$ has the same law as $(Z_t)_{t \geq 0}$. 
\end{definition}

There is a correspondence between continuous semi-groups of convolution of probability measures starting from the Dirac measure on the neutral element of $G$ and L\'{e}vy processes. Let us consider $Y = (Y_t)_{t \geq 0}$ a conjugation-invariant L\'{e}vy process on $G$ which is fixed until the end of Section \ref{sec:constructionpure}. 

\begin{proposition} 
\label{defplanarHF}
Let $\mathbb{G}$ be a finite planar graph, let $vol$ be a measure of area. There exists a unique random holonomy field $\mathbb{E}^{Y,\mathbb{G}}_{vol}$ on $P(\mathbb{G})$, whose weight is equal to $1$, such that for any rooted spanning tree $T$ of $\mathbb{G}$, any family $\left(c_{F}\right)_{F\in\mathbb{F}^{b}}$ of facial loops of $\mathbb{G}$, each oriented anti-clockwise, under $\mathbb{E}^{Y,\mathbb{G}}_{vol}$: 
\begin{enumerate}
\item the random variables $\big(h\left(\l_{c_{F},T}\right)\big)_{F\in \mathbb{F}^{b}}$ are independent, 
\item for any $F\in \mathbb{F}^{b}$, $h(\l_{c_{F},T})$ has the same law as $Y_{vol(F)}$. 
\end{enumerate}
The family $\left( \mathbb{E}^{Y,\mathbb{G}}_{vol}\right)_{\mathbb{G}, vol}$ is the {\em discrete planar Yang-Mills field} associated with $Y$. 
\end{proposition}

\begin{proof}
Let $\mathbb{G}$ be a finite planar graph, let $vol$ be a measure of area. For any positive real $t$, let us denote by $m_t$ the law of $Y_t$. For any rooted spanning tree $T$ and any family $\left(c_{F}\right)_{F\in\mathbb{F}^{b}}$ of facial loops oriented anti-clockwise, we define the measure $\mathbb{E}^{Y,\mathbb{G}}_{vol, T,\left(c_{F}\right)_{F\in\mathbb{F}^{b}}}$ on $\big(\mathcal{M}ult(P(\mathbb{G}), G), \mathcal{B}\big)$ as the unique gauge-invariant probability measure such that, under $\mathbb{E}^{Y,\mathbb{G}}_{vol, T,\left(c_{F}\right)_{F\in\mathbb{F}^{b}}}$: 
\begin{enumerate}
\item the random variables $\big(h\left(\l_{c_{F},T}\right)\big)_{F\in \mathbb{F}^{b}}$ are independent, 
\item for any $F\in \mathbb{F}^{b}$, $h(\l_{c_{F},T})$ has the same law as $Y_{vol(F)}$,  
\end{enumerate}
where we remind the reader that the loops $\l_{c_{F},T}$ where defined in Definition \ref{def:lassogenerat}. Since $\otimes_{F \in \mathbb{F}^{b}}m_{vol(F)}$ is invariant by diagonal conjugation, by applying the first part of Proposition \ref{constructfinal} we see that the definition makes sense. We will show that the probability measure $\mathbb{E}^{Y,\mathbb{G}}_{vol, T,\left(c_{F}\right)_{F\in\mathbb{F}^{b}}}$ neither depends on the choice of $T$, nor on the choice of $\left(c_{F}\right)_{F\in\mathbb{F}^{b}}$. Thanks to  the uniqueness property in this last definition, we have to prove that given another rooted spanning tree $T'$ and another family of facial loops $(c'_{F})_{F\in\mathbb{F}^{b}}$ oriented anti-clockwise, under $\mathbb{E}^{Y,\mathbb{G}}_{vol, T,\left(c_{F}\right)_{F\in\mathbb{F}^{b}}}$, $\big(h(\l_{c'_{F},T'})\big)_{F\in \mathbb{F}^{b}}$ has the same law as $\big(h\left(\l_{c_{F},T}\right)\big)_{F\in \mathbb{F}^{b}}$.

First of all, let us prove that one can suppose that $T$ and $T'$ are rooted at the same vertex $v$ of $\mathbb{G}$. Let $v$ be the root of $T$, let $v'$ be a vertex of $\mathbb{G}$ and let us define the rooted spanning tree $\tilde{T}$ as the tree $T$ rooted at $v'$. When we change the root of $T$ from $v$ to $v'$ we conjugate every of the $\l_{c_F,T}$ by the same path $[v',v]_T$. By Remark \ref{changementdepoint}, since $\mathbb{E}^{Y,\mathbb{G}}_{vol,T,\left(c_{F}\right)_{F\in\mathbb{F}^{b}}}$ is gauge-invariant, under $\mathbb{E}^{Y,\mathbb{G}}_{vol,T,\left(c_{F}\right)_{F\in\mathbb{F}^{b}}}$, $\big(h(\l_{c_F,\tilde{T}})\big)_{F \in \mathbb{F}^{b}}$ has the same law as $\big(h(\l_{c_F,T})\big)_{F \in \mathbb{F}^{b}}$, namely $\otimes_{F \in \mathbb{F}^{b}}m_{vol(F)}$: $\mathbb{E}^{Y,\mathbb{G}}_{vol, T,\left(c_{F}\right)_{F\in\mathbb{F}^{b}}} = \mathbb{E}^{Y,\mathbb{G}}_{vol, \tilde{T},\left(c_{F}\right)_{F\in\mathbb{F}^{b}}}$.

Now let us assume that $T$ and $T'$ are rooted at the same vertex. By Proposition \ref{tresseetgene}, there exists an enumeration $(F_i)_{i=1}^{\#\mathbb{F}^{b}}$ of the bounded faces of $\mathbb{G}$, a braid $\beta$ in $\mathcal{B}_{\#\mathbb{F}^{b}}$ such that: 
\begin{align*}
\beta \bullet \left(\l_{c_{F_i}, T}\right)_{i=1}^{\#\mathbb{F}^{b}} = \left(\l_{c'_{F_{\sigma_{\beta}(i)}}, T'}\right)_{i=1}^{\#\mathbb{F}^{b}}. 
\end{align*}
Using Lemma \ref{actionmieux}, $
h\left(\beta \bullet \left(\l_{c_{F_i}, T}\right)_{i=1}^{\#\mathbb{F}^{b}}\right) =\beta^{-1}\bullet  \left( h(\l_{c_{F_i}, T}) \right)_{i=1}^{\#\mathbb{F}^{b}}$, 
and thus: 
\begin{align*}
\beta^{-1}\bullet  \left( h(\l_{c_{F_i}, T}) \right)_{i=1}^{\#\mathbb{F}^{b}} = \sigma_{\beta^{-1}} \bullet \left( h(\l_{c'_{F_i}, T'}) \right)_{i=1}^{\#\mathbb{F}^{b}}.
\end{align*}
Applying the Proposition \ref{quasi-inv}, under $\mathbb{E}^{Y,\mathbb{G}}_{vol, T,\left(c_{F}\right)_{F\in\mathbb{F}^{b}}}$, the following equality in law holds: 
\begin{align*}
 \beta^{-1} \bullet \left(h(\l_{c_{{F}_{i}},T})\right)_{i=1}^{\# \mathbb{F}^{b}} = \sigma_{\beta^{-1}} \bullet \left(h(\l_{c_{{F}_{i}},T})\right)_{i=1}^{\# \mathbb{F}^{b}}.
\end{align*}
From this, we get the equality in law under $\mathbb{E}^{Y,\mathbb{G}}_{vol, T,\left(c_{F}\right)_{F\in\mathbb{F}^{b}}}$: $
\Big(h(\l_{c'_{{F}_{i}},T'})\Big)_{i=1}^{\# \mathbb{F}^{b}} = \Big(h(\l_{c_{{F}_{i}},T})\Big)_{i=1}^{\# \mathbb{F}^{b}}.$
\end{proof}

This proposition allows us not to have to choose a special rooted tree for each graph in order to construct planar Yang-Mills fields. More importantly, it will allow us to show the independence property and the area-preserving homeomorphism invariance of the family of random holonomy fields which we will construct thanks to Proposition \ref{constructfinal}.

\begin{proposition}
\label{creaYMpure}
There exists a unique family of gauge-invariant stochastically continuous random holonomy fields $\big(\mathbb{E}^{Y}_{vol}\big)_{vol}$, whose weight is equal to $1$, such that for any measure of area $vol$, for any finite planar graph $\mathbb{G}$, for any rooted spanning tree $T$ of $\mathbb{G}$ and any family of facial loops $\left(c_{F}\right)_{F\in\mathbb{F}^{b}}$ oriented anti-clockwise, under $\mathbb{E}^{Y}_{vol}$: 
\begin{enumerate}
\item the random variables $\left(h\left(\l_{c_{F},T}\right)\right)_{F\in \mathbb{F}^{b}}$ are independent, 
\item for any $F\in \mathbb{F}^{b}$, $h(\l_{c_{F},T})$ has the same law as $Y_{vol(F)}$. 
\end{enumerate}
The family $\big(\mathbb{E}^{Y}_{vol}\big)_{vol}$ is the {\em planar Yang-Mills field} associated with $(Y_t)_{t \geq 0}$.
\end{proposition}

In order to prove this result, we will need the following statement, from \cite{Levy}, which allows us to bound the distance of a L\'{e}vy process to the neutral element.  

\begin{proposition} \label{croissanceLevy}
There exists $K>0$ such that $\mathbb{E}\big[d_G\left(1,Y_t\right)\big]\leq K\sqrt{t}$ for any $t \geq 0$.
\end{proposition}

\begin{proof}[Proof of Proposition \ref{creaYMpure}] Let $vol$ be a measure of area on the plane. We will apply Proposition \ref{constructfinal} to the family of measures $\big(\mathbb{E}^{Y,\mathbb{G}}_{vol}\big)_{\mathbb{G} \in \mathcal{G}({\sf Aff}\left(\mathbb{R}^{2}\right))}$. Then we will study the restriction to general finite planar graphs of the random holonomy field that we will have defined. In order to do all this, we have to prove a compatibility condition and a uniform locally stochastically $\frac{1}{2}$-H\"{o}lder continuity property for the family $\big(\mathbb{E}^{Y,\mathbb{G}}_{vol}\big)_{\mathbb{G} \in \mathcal{G}({\sf Aff}\left(\mathbb{R}^{2}\right))}$. 

{\bf Compatibility condition: }Let $\mathbb{G}_1$ and $\mathbb{G}_2$ be two graphs in $\mathcal{G}\big({\sf Aff}(\mathbb{R}^{2})\big)$ such that $\mathbb{G}_1 \preccurlyeq \mathbb{G}_2$. Let us consider $m$ a vertex of $\mathbb{G}_1$ and $\mathbb{G}_2$. Using Proposition \ref{defplanarHF}, it is enough to show that $\mathbb{G}_1$ satisfies the following property: 
$$\hspace{+4pt}(\mathcal{H}) \left\{
   \hspace{-2pt} \begin{array}{ll}\text{
there exists a family of facial loops} \left(c_F\right)_{F \in \mathbb{F}_1^{b}} \text{ oriented anti-clockwise and}\\ \text{a spanning tree } T_1 \text{ of } \mathbb{G}_1 \text{ rooted at } m\text{, such that under } \mathbb{E}^{Y,\mathbb{G}_2}_{vol}: \\
\text{\ \ 1. the random variables $\big(h\left(\l_{c_{F},T}\right)\big)_{F\in \mathbb{F}^{b}_1}$ are independent, }\\
\text{\ \ 2. for any $F\in \mathbb{F}^{b}_1$, $h\left(\l_{c_{F},T}\right)$ has the same law as $Y_{vol(F)}$. }
\end{array}
\right.$$

We show this by an induction argument on the finite set $\big[\mathbb{G}_1, \mathbb{G}_2\big]$ which is equal to $\big\{\mathbb{G}, \mathbb{G}_1 \preccurlyeq \mathbb{G} \preccurlyeq \mathbb{G}_2 \big\}$, endowed with the partial order $\preccurlyeq$. It is clearly true that $(\mathcal{H})$ holds for $\mathbb{G} = \mathbb{G}_2$. Consider a finite planar graph $\mathbb{G}$ in $ \big[\mathbb{G}_1, \mathbb{G}_2\big]$ satisfying $(\mathcal{H})$, we will show that there exists $\mathbb{G}' \in \big[\mathbb{G}_1, \mathbb{G}\big[$ for which $(\mathcal{H})$ is still valid. Thanks to  Proposition \ref{defplanarHF}, property $(\mathcal{H})$ holds for $\mathbb{G}$ for any family of facial loops $\left(c_F\right)_{F \in \mathbb{F}^{b}}$ oriented anti-clockwise and any choice of spanning tree $T$ rooted at $m$. Since $\mathbb{G}_1 \preccurlyeq \mathbb{G}$, at least one of the following assertions is true: 
\begin{enumerate}
\item there exist an edge of $\mathbb{G}_1$, $e$, and a vertex $v$ of $\mathbb{G}$ of degree two such that $v \in e\big((0,1)\big)$, 
\item there exists a face $F_1$ of $\mathbb{G}_1$, bounded or not, such that the restriction of $\mathbb{G}$ to $F_1$ has a unique face $F_0$ and $\partial F_0$, oriented anti-clockwise, contains a sequence of the form $ee^{-1}$ with the interior of $e$ included in $F_1$, 
\item there exists a face $F_1$ of $\mathbb{G}_1$ which contains more than one face of $\mathbb{G}$. 
\end{enumerate}

Let us consider the three possibilities. 

(1) Let us consider a family of facial loops $\left(c_F\right)_{F \in \mathbb{F}^{b}}$ for $\mathbb{G}$, oriented anti-clockwise, none of which is based at $v$, and a choice of spanning tree $T$ of $\mathbb{G}$ rooted at $m$. Let $e_1$ and $e_2$ be the two edges of $\mathbb{G}$ such that $e = e_1e_2$ and $\overline{e_1} = v$. We consider $\mathbb{G}'$, the graph defined by: $$\left(\mathbb{V}', \mathbb{E}', \mathbb{F}'\right) = \big(\mathbb{V}\setminus \{v\}, \mathbb{E} \setminus\left\{e_1^{\pm1},e_2^{\pm1}\right\} \cup\left\{(e_1e_2)^{\pm1}\right\},\mathbb{F}\big).$$ By construction $\mathbb{G}' \in\big[\mathbb{G}_1, \mathbb{G}\big[$. Besides, $(c_{F})_{F \in \mathbb{F}'^{b}}$ is still a family of facial loops for $\mathbb{G}'$ oriented anti-clockwise and $T' = \big(T\setminus\left\{e_1^{\pm1},e_2^{\pm1}\right\}\big) \cup \left\{(e_1e_2)^{\pm1}\right\}$ is a spanning tree of $\mathbb{G}'$ rooted at $m$. It is now obvious that $\mathbb{G}'$ satisfies property $(\mathcal{H})$ with the choices of $\left(c_{F}\right)_{F \in \mathbb{F}'^{b}}$ and $T'$. 

(2) We will consider that $F_1$ is bounded, the unbounded case is similar. In this case, let $v$ be the vertex of $e$ of degree $1$ and define ${F}'_0 = F_0\cup e\big((0,1)\big) \cup \{v \}$. Consider any family of facial loops for $\mathbb{G}$ oriented anti-clockwise, $\left(c_F\right)_{F \in \mathbb{F}^{b}}$, such that $\underline{c_{F_0}} \ne v $. Let us choose any spanning tree of $\mathbb{G}$ rooted at $m$, $T$. We consider $\mathbb{G}'$, the graph defined by: $$\big(\mathbb{V}',\mathbb{E}',\mathbb{F}'\big) = \big(\mathbb{V}\setminus{v}, \mathbb{E}\setminus\{e, e^{-1}\}, (\mathbb{F} \setminus{F_0}) \cup F'_0\big).$$
The spanning tree $T$ of $\mathbb{G}$ must include the unoriented edge $\left\{e, e^{-1}\right\}$ in order to cover $v$, thus we can define $T' = T\setminus\{e, e^{-1}\}$. The facial loop $c_{F_0}$ contains the sequence $ee^{-1}$. We define ${c'}_{F'_0}$ from $c_{F_0}$ by removing this sequence. For any other face $F \in \mathbb{F}'$, we set ${c'}_{F} = c_{F}$. For any face $F \in \mathbb{F}'^{b}$, using the identification between $F_0$ and $F'_0$, $\l_{c_{F},T} = \l_{c'_{F},T'}$ in $RL_{m}(\mathbb{G})$, and by Remark \ref{equalityequiv}, $h(\l_{c_{F},T})= h(\l_{c'_{F},T'})$. The graph $\mathbb{G}'$ satisfies property $(\mathcal{H})$ with the choices of $\left({c}'_{F'}\right)_{F' \in \mathbb{F}'^{b}}$ and $T'$. 
 
(3) We will study this case under the hypothesis that $F_{1}$ is bounded, the unbounded case being easier. The key point will be the semigroup property satisfied by the marginal distributions of the L\'{e}vy process $Y$. Let $F_r$ and $F_l$ be two faces of $\mathbb{G}$ contained in $F_1$ and adjacent, sharing an edge $e$ on their boundaries. We can find a facial loop oriented anti-clockwise representing the boundary of $F_r$ (resp. $F_l$) of the form $c_{F_r} = e_1...e_n e$ (resp. $c_{F_l} = e^{-1}e'_1...e'_m$). Let $F_{r,l} = F_r \cup F_l \cup e\big((0,1)\big)$. We complete the family $\left(c_{F_r}, c_{F_l}\right)$ in order to have a family of facial loops $\left(c_F\right)_{F \in \mathbb{F}^{b}}$ oriented anti-clockwise for $\mathbb{G}$. Let us consider  $\mathbb{G}'$, the graph defined by: 
\begin{align*}
\left(\mathbb{V}', \mathbb{E}', \mathbb{F}'\right) = \Big(\mathbb{V}, \mathbb{E}\setminus\left\{e,e^{-1}\right\},( \mathbb{F}\setminus\left\{F_r,F_l\right\})\cup F_{r,l}\Big). 
\end{align*}
It is still a finite planar graph. Let us consider $T$ any spanning tree of $\mathbb{G}'$ rooted at $m$: it is also a spanning tree of $\mathbb{G}$ rooted at $m$. Let ${c}'_{F_{r,l}} = e_1...e_ne'_1...e'_m$. For any other face $F'$ of $\mathbb{G}'$ different from $F_{r,l}$, $F'$ is a face of $\mathbb{G}$ and we set ${c}'_{F'} = c_{F'}$. Once these choices made, it needs only a simple verification to check that the following equalities hold in $RL_{m}(\mathbb{G})$:  
\begin{align*}
\l_{c'_{F_{r,l}},T} &= \l_{c_{F_{r}},T}\l_{c_{F_{l}},T}, \\
\l_{c'_{F'},T} &= \l_{c_{F'},T}, \forall F' \in \mathbb{F}'^{b}, F' \ne F_{r,l}. 
\end{align*}
Using the multiplicativity of $h$: 
\begin{align*}
h(\l_{c'_{F_{r,l}},T}) &= h(\l_{c_{F_{l}},T}) h(\l_{c_{F_{r}},T}), \\
h(\l_{c'_{F'},T}) &= h(\l_{c_{F'},T}), \forall F' \in \mathbb{F}'^{b}, F' \ne F_{r,l}. 
\end{align*}
Let us recall that under $\mathbb{E}^{Y,\mathbb{G}_2}_{vol}$, 
\begin{enumerate}
\item the random variables $\left(h\left(\l_{c_{F},T}\right)\right)_{F\in \mathbb{F}^{b}}$ are independent, 
\item for any $F\in \mathbb{F}^{b}$, $h(\l_{c_{F},T})$ has the same law as $Y_{vol(F)}$. 
\end{enumerate}
Using the semigroup property of the marginal distributions of the process $Y$, we can conclude that $\mathbb{G}'$ satisfies $(\mathcal{H})$ with the choices of $({c'}_{F'})_{F' \in \mathbb{F}'^{b}}$ and~$T$. 

By descending induction, it follows that $\mathbb{G}_1$ satisfies property $(\mathcal{H})$. Let us prove the uniform $\frac{1}{2}$-H\"{o}lder continuity.

{\bf Uniform $\frac{1}{2}$-H\"{o}lder continuity: }Let $\mathbb{G}$ be a finite planar graph with piecewise affine edges. Let $l$ be a simple loop in $\mathbb{G}$ bounding a disk $D$. A consequence of what we have just seen is that the law of $h(l)$ under $\mathbb{E}^{Y,\mathbb{G}}_{vol}$ is the same as under $\mathbb{E}^{Y,\mathbb{G}(l)}_{vol}$, where $\mathbb{G}(l)$ is the graph containing only the edge $l$ (see Example \ref{G(l)}). Thus: 
\begin{align*}
\int_{\mathcal{M}ult(P(\mathbb{G}), G)} d_{G}\big(1,h(l)\big) \mathbb{E}^{Y,\mathbb{G}}_{vol}(dh) &= \int_{\mathcal{M}ult(P(\mathbb{G}(l)), G)} d_{G}\big(1,h(l)\big) \mathbb{E}^{Y,\mathbb{G}(l)}_{vol}(dh) \\ &= \mathbb{E}\left[d_G\left(1,Y_{vol(D)}\right)\right] \leq K \sqrt{vol(D)}, 
\end{align*}
where the last inequality comes from Proposition \ref{croissanceLevy} and where $K$ depends only on $G$. The family $(\mathbb{E}^{Y,\mathbb{G}}_{vol})_{\mathbb{G} \in \mathcal{G}({\sf Aff}\left(\mathbb{R}^{2}\right))}$ is uniformly locally stochastically $\frac{1}{2}$-H\"{o}lder continuous.

Thus, Proposition \ref{constructfinal} can be applied in order to construct a stochastically continuous random holonomy field $\mathbb{E}^{Y}_{vol}$ such that for any finite planar graph $\mathbb{G} \in \mathcal{G}({\sf Aff}\left(\mathbb{R}^{2}\right))$, for any rooted spanning tree $T$ of $\mathbb{G}$ and any family of facial loops $\left(c_F\right)_{F \in \mathbb{F}^{b}}$ oriented anti-clockwise, under $\mathbb{E}^{Y}_{vol}$: 
\begin{enumerate}
\item the random variables $\big(h\left(\l_{c_{F},T}\right)\big)_{F\in \mathbb{F}^{b}}$ are independent, 
\item for any $F\in \mathbb{F}^{b}$, $h\left(\l_{c_{F},T}\right)$ has the same law as $Y_{vol(F)}$. 
\end{enumerate}
It remains to prove that this property is true for any finite planar graph, not necessarily with piecewise affine edges. Let $(m_t)_{t \in \mathbb{R}^{+}}$ be the continuous semi-group of convolution associated with $(Y_t)_{t \geq 0}$. Let $\mathbb{G}=(\mathbb{V}, \mathbb{E}, \mathbb{F})$ be a finite planar graph, let $T$ be a rooted spanning tree and let $\big(c_F\big)_{F \in \mathbb{F}^{b}}$ be a family of facial loops oriented anti-clockwise. Let us consider a sequence of finite planar graphs $\big(\mathbb{G}_n=(\mathbb{V}_n, \mathbb{E}_n, \mathbb{F}_n)\big)_{n \in \mathbb{N}}$ in $\mathcal{G}\big({\sf Aff}(\mathbb{R}^{2})\big)$ and $(\psi_n)_{n \in \mathbb{N}}$ a sequence of orientation-preserving homeomorphisms which satisfy the conditions of Theorem~\ref{approxgraph}. For any integer $n$, $\left(\psi_n(c_F)\right)_{F \in \mathbb{F}^{b}}$ is a family of facial loops for $\mathbb{G}_n$ which is oriented anti-clockwise and $\psi_n(T)$ is a spanning tree of $\mathbb{G}_n$. Using the discussion we had before, the law of $\big(h(\l_{\psi_n(c_F), \psi_n(T)})\big)_{F \in \mathbb{F}^{b}}$ under $\mathbb{E}^{Y}_{vol}$ is $\bigotimes\limits_{F \in \mathbb{F}^{b}}m_{vol(\psi_n(F))}$. As for any edge $e \in \mathbb{E}$, $\left(\psi_n(e)\right)_{n \geq 0}$ converges to $e$ for the convergence with fixed endpoints, for any face $F\in \mathbb{F}^{b}$, one has $\l_{\psi_n(c_F), \psi_n(T)} \underset{n \to \infty}{\longrightarrow} \l_{c_F,T} $ for the fixed endpoints convergence. Besides, using condition $4$ of Theorem \ref{approxgraph} and the continuity of $(m_t)_{t \in \mathbb{R}^{+}}$, $\bigotimes\limits_{F \in \mathbb{F}^{b}}m_{vol(\psi_n(F))} \underset{n \to \infty}{\longrightarrow} \bigotimes\limits_{F \in \mathbb{F}^{b}}m_{vol(F)}.$ Since $\mathbb{E}^{Y}_{vol}$ is stochastically continuous, under $\mathbb{E}^{Y}_{vol}$, the law of $\big(h(\l_{c_F,T})\big)_{F \in \mathbb{F}^{b}}$ is $\bigotimes\limits_{F \in \mathbb{F}^{b}}m_{vol(F)}$. 
\end{proof}

Let us remark that, in the latest argument, we actually proved that the family $\big(\mathbb{E}^{Y,\mathbb{G}}_{vol}\big)_{\mathbb{G}, vol}$ is continuously area-dependent. For any L\'{e}vy process which is invariant by conjugation, we have constructed a family of gauge-invariant stochastically continuous random holonomy fields. In the following, we will show that this family is a strong planar Markovian holonomy field.

\begin{proposition}
\label{creaYMpure2}
The family of random holonomy fields $\big(\mathbb{E}^{Y}_{vol}\big)_{vol}$ is a constructible stochastically continuous strong planar Markovian holonomy field. 
\end{proposition}

\begin{proof}
We have already shown that the family  $\big(\mathbb{E}^{Y,\mathbb{G}}_{vol}\big)_{\mathbb{G}, vol}$ satisfies the Axiom $\mathbf{DP_4}$, is continuously area-dependent and locally stochastically $\frac{1}{2}$-H\"{o}lder continuous. By Theorem \ref{exten3}, it remains to check that $\big(\mathbb{E}^{Y,\mathbb{G}}_{vol}\big)_{\mathbb{G}, vol}$ satisfies the three Axioms $\mathbf{DP_1}$, $\mathbf{DP_2}$ and $\mathbf{DP_3}$ in Definition~\ref{discreteplanarHF}.

Besides, 
$\big(\mathbb{E}^{Y,\mathbb{G}}_{vol}\big)_{\mathbb{G}, vol}$ is stochastically continuous in law: using a continuity argument, it is enough to show that $\mathbf{wDP_2}$ holds instead of $\mathbf{DP_2}$. Let us give briefly the arguments which allow us to do so: first of all, using Theorem \ref{approxgraph}, we see that it is enough to consider graphs with piecewise affine edges. Let us suppose that the Axiom $\mathbf{wDP_2}$ holds, let us consider $\mathbb{G}$ a finite planar graph with piecewise affine edges and two loops $l_1$ and $l_2$ in $P(\mathbb{G})$ such that ${\sf Int}(l_1)\cap {\sf Int}(l_2) = \emptyset$. The only interesting case is when $\overline{{\sf Int}(l_1)}\cap \overline{{\sf Int}(l_2)} \ne \emptyset$: let us consider a point $v$ in this intersection. Using Remark \ref{changementdepoint}, it is enough to prove that for any family of loops $(l^{1}_{i})_{i=1}^{n}$ (resp. $(l^2_{i})_{i=1}^{m}$) in $\overline{{\sf Int}(l_1)}$ (resp. in $\overline{{\sf Int}(l_2)}$) based at $v$, $\left(h(l_i^{1})\right)_{i=1}^{n}$ is $\mathcal{I}$-independent of $\left(h(l_i^{2})\right)_{i=1}^{m}$ under the measure $\mathbb{E}^{Y,\mathbb{G}}_{vol}$. Let us consider such families of loops $(l^{1}_{i})_{i=1}^{n}$ and $(l^2_{i})_{i=1}^{m}$. One can always approximate $\mathbb{G}$ by a finite planar graph $\mathbb{G}'$ with piecewise affine edges such that there exist $l_1'$ and $l_2'$ two loops in $\mathbb{G}'$ and $p$ a path in $\mathbb{G}'$ such that: 
\begin{enumerate}
\item $\overline{{\sf Int}(l_1')} \cap \overline{{\sf Int}(l_2')} = \emptyset$,
\item for any loop $l$ in $\overline{{\sf Int}(l_1)}$ based at $v$, there exists a loop in  $\overline{{\sf Int}(l_1')}$ which approximates $l$ for the convergence with fixed endpoints, 
\item for any loop $l$ in $\overline{{\sf Int}(l_2)}$ based at $v$, there exist a loop in $\overline{{\sf Int}(l_2')}$, denoted by $l'$ such that $p l' p^{-1}$ approximates $l$  for the convergence with fixed endpoints.  
\end{enumerate}

Let us consider the graph $\mathbb{G}'$, the two loops $l_1'$ and $l_2'$ and the path $p$ given by the last assertion. We can approximate $(l^{1}_{i})_{i=1}^{n}$ and $(l^2_{i})_{i=1}^{m}$ by two families $(l'^{1}_{i})_{i=1}^{n}$ and $(pl'^2_{i}p^{-1})_{i=1}^{m}$ such that the first one is in $\overline{{\sf Int}(l'_1)}$ and $(l'^{2}_{i})_{i=1}^{m}$ is in $\overline{{\sf Int}(l'_2)}$. The $\mathcal{I}$-independence of $(l^{1}_{i})_{i=1}^{n}$ and $(l^2_{i})_{i=1}^{m}$ under the measure $\mathbb{E}^{Y,\mathbb{G}}_{vol}$ would be a consequence of the $\mathcal{I}$-independence of $(l'^{1}_{i})_{i=1}^{n}$ and $(pl'^2_{i}p^{-1})_{i=1}^{m}$ under  $\mathbb{E}^{Y,\mathbb{G}'}_{vol}$, which is equivalent to the the $\mathcal{I}$-independence of $(l'^{1}_{i})_{i=1}^{n}$ and $(l'^2_{i})_{i=1}^{m}$ under   $\mathbb{E}^{Y,\mathbb{G}'}_{vol}$. Using Remark  \ref{I-indep-et-indep-normale}, this is equivalent to the independence of $(l'^{1}_{i})_{i=1}^{n}$ and $(l'^2_{i})_{i=1}^{m}$ under   $\mathbb{E}^{Y,\mathbb{G}'}_{vol}$ which is granted since we supposed that the Axiom $\mathbf{wDP_2}$ holds.

Let us prove that $\big(\mathbb{E}^{Y,\mathbb{G}}_{vol}\big)_{\mathbb{G}, vol}$ satisfies the three Axioms $\mathbf{DP_1}$, $\mathbf{wDP_2}$ and $\mathbf{DP_3}$.

$\mathbf{\mathbf{DP_1}}$: Consider $vol$ and $vol'$ two measures of area on $\mathbb{R}^{2}$, $\mathbb{G}$ and $\mathbb{G}'$ two finite planar graphs and $\psi$ a homeomorphism which preserves the orientation. Let us suppose that $\psi(\mathbb{G}) = \mathbb{G}'$ and for any $F \in \mathbb{F}^{b}$, $vol(F)= vol'(\psi(F))$. Let $\left(c'_F\right)_{F \in \mathbb{F}'^{b}}$ be a family of facial loops oriented anti-clockwise for $\mathbb{G}'$ and let $T'$ be a rooted spanning tree of $\mathbb{G}'$. 
We consider $\big(c_F = \psi^{-1}\big(c'_{\psi(F)}\big)\big)_{F \in \mathbb{F}^{b}}$ and $T = \psi^{-1} (T')$. The family $\left(c_F\right)_{F \in \mathbb{F}^{b}}$ is a family of facial loops for $\mathbb{G}$ which are oriented anti-clockwise and $T$ is a rooted spanning tree of $\mathbb{G}$. Recall the notations used in the proof of Proposition \ref{defplanarHF}. By construction, we have the equality: 
\begin{align*}
\mathbb{E}^{Y,\mathbb{G}'}_{vol', T', \left(c'_F\right)_{F \in \mathbb{F}'^{b}}} \circ \psi^{-1} = \mathbb{E}^{Y,\mathbb{G}}_{vol, T, \left(c_F\right)_{F \in \mathbb{F}^{b}}}, 
\end{align*}
where we denoted also by $\psi$ the induced application from $\mathcal{M}ult(P(\mathbb{G}'), G)$ to $\mathcal{M}ult(P(\mathbb{G}), G)$ induced by the homeomorphism $\psi$. Using the Proposition \ref{defplanarHF}, we get $\mathbb{E}^{Y,\mathbb{G}'}_{vol'} \circ \psi^{-1} = \mathbb{E}^{Y,\mathbb{G}}_{vol}. $

$\mathbf{\mathbf{wDP_2}}$: Let $vol$ be a measure of area on $\mathbb{R}^{2}$, $\mathbb{G} = (\mathbb{V}, \mathbb{E}, \mathbb{F})$ be a finite planar graph in $\mathcal{G}\big({\sf Aff}\left(\mathbb{R}^{2}\right)\!\big)$, $m$ be a vertex of $\mathbb{G}$ and $L_1$ and $L_2$ be two simple loops in $P(\mathbb{G})$ whose closure of the interiors are disjoint. As an application of Proposition \ref{combi1}, we can consider $T$ a spanning tree rooted at $m$, such that for any family of facial loops $\left(c_F\right)_{F \in \mathbb{F}^{b}}$ oriented anti-clockwise: 
\begin{enumerate}
\item for every loop $l$ in $P(\mathbb{G})$ inside $\overline{{\sf Int(}L_1)}$, $\Big[[m,\underline{l}]_T l [m,\underline{l}]_T^{-1}\Big]_{\simeq}$ is a product in $RL_{m}(\mathbb{G})$ of elements of $ \left\{\l_{c_F,T}^{\pm 1} ; F\in \mathbb{F}^{b}, F \subset \overline{{\sf Int(}L_1)} \right\}$, 
\item for every loop $l$ in $P(\mathbb{G})$ inside $\overline{{\sf Int(}L_2)}$, $\Big[[m,\underline{l}]_T l [m,\underline{l}]_T^{-1}\Big]_{\simeq}$ is a product in $RL_{m}(\mathbb{G})$ of elements of $ \left\{\l_{c_F,T}^{\pm 1} ; F\in \mathbb{F}^{b}, F \subset \overline{{\sf Int(}L_2)}\right\}$.
\end{enumerate}

Let us consider $\left(p_i\right)_{i=1}^{n}$ (resp. $(p'_j)_{j=1}^{n'}$), some paths in $P(\mathbb{G})$ which are inside $\overline{{\sf Int(}L_1)}$ (resp. $\overline{{\sf Int(}L_2)}$). Recall the definition given by Equality (\ref{moyennisee}). For any continuous function $f_1$ (resp. $f_2$) defined on $G^{n}$ (resp. on $G^{n'}$), using the gauge-invariance, $\mathbb{E}^{Y,\mathbb{G}}_{vol}\left[f_1\left( \left( h(p_i)\right)_{i=1}^{n}\right) f_2 \left( \left( h(p'_j)\right)_{j=1}^{n'}\right)\right] $ is equal to: 
\begin{align*}
\mathbb{E}^{Y,\mathbb{G}}_{vol}\left[\hat{f_{1}}_{J_{p_1,...,p_n}} \left( h(p_1), ..., h(p_n)\right)    \hat{f_2}_{J_{p'_1,...,p'_{n'}}} \left( h(p'_1), ..., h(p'_{n'}) \right) \right]. 
\end{align*}
Thus, it is also equal to: 
\begin{align*}
\mathbb{E}^{Y,\mathbb{G}}_{vol}\left[\hat{f_{1}}_{J_{p_1,...,p_n}} \left( \left(h(\tilde{l}_i)\right)_{i=1}^{n}\right)    \hat{f_2}_{J_{p'_1,...,p'_{n'}}}\left( \left(h(\tilde{l}_i')\right)_{i=1}^{n'}\right) \right]. 
\end{align*}
where, for any $i \in \{1,...,n\}$, $\tilde{l}_i=[m,\underline{p_i}]_T p_i[m,\overline{p_i}]_T^{-1}$ and for any $i \in \{1,...,n'\}$, $\tilde{l}'_i = [m,\underline{p'_i}]_T p'_i[m,\overline{p'_i}]_T^{-1}$. Recall the form of $T$ given in the proof of Proposition \ref{combi1}: this implies that there exist $(l_i)_{i=1}^{n}$ some loops in $\overline{{\sf Int(}L_1)}$ and $(l'_i)_{i=1}^{n'}$ some loops in $\overline{{\sf Int(}L_2)}$ such that for any $i \in \{1,...,n\}$ and any $j \in \{1,...,n'\}$, 
\begin{align*}
\tilde{l}_i &= [m,\underline{l_i}]_{T} l_i [m,\underline{l_i}]_T^{-1},\\
\tilde{l}_j' &= [m,\underline{l_j'}]_{T} l_j' [m,\underline{l_j'}]_T^{-1}, 
\end{align*}
in $RL_{m}(\mathbb{G})$. Using the properties satisfied by $T$: 
\begin{align*}
\sigma \left( \left(h(\tilde{l}_i)\right)_{i=1}^{n} \right)&\subset \sigma\left(\left\{h(\l_{c_F,T}) ; F\in \mathbb{F}^{b}, F \subset \overline{{\sf Int(}L_1)} \right\}\right),\\
\sigma \left( \left(h(\tilde{l}'_i)\right)_{i=1}^{n'} \right) &\subset \sigma\left( \left\{h(\l_{c_F,T}); F\in \mathbb{F}^{b}, F \subset \overline{{\sf Int(}L_2)}\right\}\right).
\end{align*}
We recall that $\overline{{\sf Int(}L_1)} \cap \overline{{\sf Int(}L_2)} = \emptyset$, thus the two $\sigma$-fields: 
\begin{align*}
\sigma\left(\left\{h(\l_{c_F,T}) ; F\in \mathbb{F}^{b}, F \subset \overline{{\sf Int(}L_1)} \right\}\right)&, \\
\sigma\left( \left\{h(\l_{c_F,T}); F\in \mathbb{F}^{b}, F \subset \overline{{\sf Int(}L_2)}\right\}\right)&
\end{align*}
 are independent under $\mathbb{E}^{Y, \mathbb{G}}_{vol}$. Thus $\mathbb{E}^{Y,\mathbb{G}}_{vol}\left[f_1\left( \left( h(p_i)\right)_{i=1}^{n}\right) f_2 \left( \left( h(p'_j)\right)_{j=1}^{n'}\right)\right] $ is equal to: 
 \begin{align*}
\mathbb{E}^{Y,\mathbb{G}}_{vol}\left[\hat{f_{1}}_{J_{p_1,...,p_n}} \left( \left(h(\tilde{l}_i)\right)_{i=1}^{n}\right) \right] \mathbb{E}^{Y,\mathbb{G}}_{vol}\left[   \hat{f_2}_{J_{p'_1,...,p'_{n'}}}\left( \left(h(\tilde{l}_i')\right)_{i=1}^{n'}\right) \right], 
 \end{align*}
 which is equal to $\mathbb{E}^{Y,\mathbb{G}}_{vol}\left[f_1\left( \left( h(p_i)\right)_{i=1}^{n}\right) \right] \mathbb{E}^{Y,\mathbb{G}}_{vol} \left[f_2 \left( ( h(p'_j))_{j=1}^{n'}\right)\right]$:  the axiom $\mathbf{wDP_2}$ is satisfied. 

$\mathbf{DP_3}$: Let $l$ be a simple loop, let $vol$ and $vol'$ be two measures of area on $\mathbb{R}^{2}$ which are equal in the interior of $l$. Let $\mathbb{G}$ be a finite planar graph included in $\overline{Int(l)}$. The bounded faces of $\mathbb{G}$ are in the interior of $l$ thus for any bounded face $F$ of $\mathbb{G}$, $vol(F) = vol'(F)$. By definition, it is clear that $\mathbb{E}^{Y,\mathbb{G}}_{vol} = \mathbb{E}^{Y,\mathbb{G}}_{vol'}$.

We have proved all the conditions on $\big(\mathbb{E}^{Y,\mathbb{G}}_{vol}\big)_{\mathbb{G}, vol}$ we needed in order to apply Theorem \ref{exten3}. The family of holonomy fields $\big(\mathbb{E}^{Y}_{vol}\big)_{vol}$ is a constructible stochastically continuous strong planar Markovian holonomy field. \end{proof}

For this new construction, we used the loop paradigm which links the multiplicative functions on a set $P$ and the pre-multiplicative functions on its set of loops. The edge paradigm given by the Equation (\ref{lequationedgeparadigme}) can be used to give an explicit formula for discrete planar Yang-Mills fields associated with a conjugation invariant L\'{e}vy process with density. 

\begin{proposition}
\label{densitee}
Let us suppose that for any positive real $t$, $Y_t$ has a density $Q_t$ with respect to the Haar measure. 
Let $\big(\mathbb{E}^{Y,\mathbb{G}}_{vol}\big)_{\mathbb{G}, vol}$ be the discrete planar Yang-Mills field associated with $Y$. For any finite planar graph $\mathbb{G}$ and for any measure of area $vol$: 
\begin{align}
\mathbb{E}^{Y,\mathbb{G}}_{vol}(dh) = \prod_{F \in \mathbb{F}^{b}} Q_{vol(F)} \big(h(\partial F)\big) \bigotimes_{e \in \mathbb{E}^{+}} dh(e),
\end{align}
where $\partial F$ is the anti-clockwise oriented facial cycle associated with $F$, the notation $Q_{vol(F)} \big(h(\partial F )\big)$ means that we consider $Q_{vol(F)} \big(h(c)\big)$ where $c$ represents $\partial F$ (this does not depend on the choice of $c$ since $Q_{vol(F)}$ is invariant by conjugation) and $\bigotimes_{e \in \mathbb{E}^{+}} dh(e)$ is the push-forward of $\bigotimes_{e \in \mathbb{E}^{+}} dg_e$ on $\mathcal{M}ult(P, G)$ by the edge paradigm identification. It is independent of the choice of orientation $\mathbb{E}^{+}$.
\end{proposition}

Recall the definition of $L_{i,j}$ in Definition \ref{base}. In order to make the proof simple, we will use the upcoming Theorem \ref{caracterisation0} which roughly asserts that a stochastically continuous planar Markovian holonomy field is characterized by the law of the random sequence $(h(L_{n,0}))_{n \in \mathbb{N}}$.

\begin{proof}
A slight modification of Section 4.3 in \cite{Levy} shows that: 
\begin{align}
\label{densite}
\bigg(\prod_{F \in \mathbb{F}^{b}} Q_{vol(F)} ( h(\partial F)) \bigotimes_{e \in \mathbb{E}^{+}} dh(e)\bigg)_{\mathbb{G}, vol}
\end{align}
is a stochastically continuous in law discrete planar Markovian holonomy field. Let $\mathbb{G}=(\mathbb{V}, \mathbb{E}, \mathbb{F})$ be the planar graph $\mathbb{N}^{2}\cap\big(\mathbb{R}^{+} \times [0,1]\big)$. Let us consider $\mathbb{E}^{+}$ an orientation of $\mathbb{E}$. As an application of Theorem \ref{caracterisation0}, we only have to check that for any positive real $\alpha$, $\big(h(L_{n,0})\big)_{n \in \mathbb{N}}$ has the same law under $\mathbb{E}^{Y,\mathbb{G}}_{\alpha dx}$ as under $\prod\limits_{F \in \mathbb{F}^{b}} Q_{\alpha} \big(h(\partial F)\big) \bigotimes\limits_{e \in \mathbb{E}^{+}} dh(e)$. Since the value of $\alpha$ will not matter we suppose that $\alpha$ equals to $1$. 

By Proposition \ref{defplanarHF}, under $\mathbb{E}^{Y,\mathbb{G}}_{dx}$, $\big(h(L_{n,0})\big)_{n \in \mathbb{N}}$ are i.i.d. random variables which have the same law as $Y_1$. It remains to prove that this property is true under $\prod\limits_{F \in \mathbb{F}^{b}} Q_{1} \big(h(\partial F)\big) \bigotimes\limits_{e \in \mathbb{E}^{+}} dh(e)$.

Under the probability law $\bigotimes_{e \in \mathbb{E}^{+}} dh(e)$, $\big(h(e)\big)_{e \in \mathbb{E}^{+}}$ are i.i.d. and Haar distributed. Using the multiplicativity property of random holonomy fields, for any integer $n$, $h(L_{n,0})$ is a product of elements of $\left(h(e)\right)_{e \in \mathbb{E}} = (h(e))_{e \in \mathbb{E}^{+}} \cup (h(e)^{-1})_{e \in \mathbb{E}^{+}}$. An important remark is that for any integer $n$, the edge $e^{r}_{n,0}$, defined in the Notation \ref{not:er}, appears only once in the reduced decomposition of $L_{n,0}$ and in no other reduced decomposition of $L_{m,0}$ with $m\ne n$. Applying Lemma \ref{haar}, one has that under $ \bigotimes_{e \in \mathbb{E}^{+}} dh(e)$, $\big(h(L_{n,0})\big)_{n \in \mathbb{N}}$ are independent and each of them is a Haar random variable. Recall the notation $\partial c_{i,j}$ defined in Definition \ref{base}. Since $Y_{1}$ is invariant by conjugation, for any bounded face $F$ in $\mathbb{F}$, there exists an integer $n \in \mathbb{N}$ such that $Q_{1}\left( \partial F\right) = Q_{1}\big(h(\partial c_{n,0})\big) = Q_{1}\big(h(L_{n,0})\big)$. 
Let $f: G^{\mathbb{N}} \to \mathbb{R}$ be a measurable function, the following sequence of equality holds: 
\begin{align*}
\int_{\mathcal{M}ult\left(P(\mathbb{G}),G\right) } &f\left(\left(h(L_{n,0})\right)_{n\in \mathbb{N}}\right)\prod_{F \in \mathbb{F}^{b}} Q_1\left(h(\partial F)\right) \bigotimes_{e \in \mathbb{E}^{+}} dh(e)\\
&= \int_{\mathcal{M}ult\left(P(\mathbb{G}),G\right) } f\left(\left(h(L_{n,0})\right)_{n\in \mathbb{N}}\right)\prod_{n \in \mathbb{N}} Q_1\left(h(\partial c_{n,0})\right) \bigotimes_{e \in \mathbb{E}^{+}} dh(e) \\&= \int_{\mathcal{M}ult\left(P(\mathbb{G}),G\right) } f\left(\left(h(L_{n,0})\right)_{n \in \mathbb{N}}\right) \prod_{n \in\mathbb{N}} Q_1\left(h(L_{n,0})\right) \bigotimes_{e \in \mathbb{E}^{+}} dh(e)\\
&= \int_{G^{\mathbb{N}}} f\left(\left(g_n\right)_{n \in\mathbb{N}}\right) \bigotimes_{n \in \mathbb{N}} \left(Q_1(g_{n}) dg_{n}\right),
\end{align*}
which is the assertion we had to prove. \end{proof}

\begin{lemme}
\label{haar}
Let $(\alpha_{i})_{i=1}^{\infty}$ be a sequence of independent $G$-valued random variables which are Haar distributed. Let $(\beta_{j})_{j=1}^{\infty}$ be a sequence of $G$-valued random variables such that for every $j\in \mathbb{N}^{*}$, $\beta_{j}$ is a product of elements of $\{ \alpha_{i}, \alpha_{i}^{-1}, i\in \mathbb{N}^{*}\}$: $\beta_{j} = w_{j}\big( (\alpha_{i}, \alpha_{i}^{-1})_{i \in \mathbb{N}^{*}} \big),$ with $w_{j}$ being a finite word. 

Suppose that for any $j\in \mathbb{N}^{*}$, there exists an index $i_{j}$ such that $\alpha_{i_{j}}$ appears exactly once in $w_{j}$ and in no other word $(w_{j'})_{j'\ne j}$. Then $(\beta_{j})_{j=1}^{\infty}$ is a family of independent Haar distributed random variables.
\end{lemme}

\begin{proof}
Let $k$ be any positive integer and let $(j, j_1, ..., j_k)$ be a $k+1$-tuple of positive integers. Let $i_j \in \mathbb{N}^{*}$ such that $\alpha_{i_{j}}$ appears exactly once in $w_j$ and in no other word $(w_{j'})_{j' \ne j}$. Let $F: G^{k} \to \mathbb{R}$ and $f: G \to \mathbb{R}$ be two continuous functions. There exist $w_1$ and $w_2$ two words in $(\alpha_{i}, \alpha_i^{-1})_{i \ne i_j}$, $J$ a subset of $\mathbb{N}\setminus \{ i_j\}$ and $\tilde{F}$ a continuous function from $G^{\# J}$ to $\mathbb{R}$ such that a.s.: 
\begin{align*}
f(\beta_{j})F(\beta_{j_1}...\beta_{j_k}) =f( w_1 \alpha_{i_j} w_2) \tilde{F}\left((\alpha_{i})_{i \in J}\right). 
\end{align*}
Using the translation invariance of the Haar measure: 
\begin{align*}
\mathbb{E}\left[f(\beta_{j})F(\beta_{j_1}...\beta_{j_k}) \right] &= \mathbb{E}\left[ f( w_1 \alpha_{i_j} w_2) \tilde{F}\left((\alpha_{i})_{i \in J}\right)\right]\\
&= \int_{G} \mathbb{E}\left[f(w_1 x w_2) \tilde{F}\left((\alpha_{i})_{i \in J}\right)\right] dx\\
&=  \int_{G} \mathbb{E}\left[f(x) \tilde{F}\left((\alpha_{i})_{i \in J}\right)\right] dx\\
&= \left(\int_{G} f(x) dx\right) \mathbb{E}\left[\tilde{F}\left((\alpha_{i})_{i \in J}\right)\right].
\end{align*}
Thus for any $j \in \mathbb{N}^{*}$, $\beta_j$ is a Haar random variable which is independent of $\left(\beta_j, j \ne i\right)$. \end{proof}

\section{Construction of general planar Yang-Mills fields}
\label{generalisationdeconstruction}
In the last subsection, we considered only L\'{e}vy processes which were invariant by conjugation by $G$. Actually for any $G$-valued self-invariant by conjugation L\'{e}vy process $Y$, one can construct a planar Markovian holonomy field associated to $Y$. Let us recall the notion of support of a process. 
\begin{definition}
\label{supportlevy}
Let $Y=\big(Y_t\big)_{t \in \mathbb{R}^{+}}$ be a random process. The support of $Y$ is $H_Y = \overline{\big<\bigcup\limits_{t \in \mathbb{R}^{+}} H_{Y_t}\big>}.$ \end{definition}

An other formulation is to say that the support of a process, say $Y$, is the smallest closed group such that for any $t \in \mathbb{R}^{+}$, $\mathbb{P}(Y_t \in H_Y) = 1$. If $Y$ is a L\'{e}vy process, we can consider $Y$ as a process living in $H_{Y}$.

\begin{remarque}
Let $Y$ be a L\'{e}vy process and let us suppose that for any $t\geq 0$, $e$ is in ${\sf Supp}(Y_t)$. Then for any $t>0$, $H_Y=H_{Y_t}$. Indeed, using the property that $Y$ is a L\'{e}vy process, for any $0 \leq t<s$, ${\sf Supp}(Y_{s}) = \overline{{\sf Supp}(Y_{t}) {\sf Supp}(Y_{s-t})}$. Thus, using the condition on the support of $Y_{t}$, $H_{Y_t}$ is increasing in $t$. Yet, using the same argument, we see that $H_{Y_{2t}} = H_{Y_t}$, thus $H_{Y_t}$ does not depend on $t>0$. This remark explains why we impose that for any $t\geq 0$, $e \in {\sf Supp}(Y_t)$ in the upcoming Proposition \ref{unicity2-proc}. 
\end{remarque}

The notion of invariance by conjugation for a random process can be weakened. 

\begin{definition}
\label{self-invariance}
A $G$-valued process $(Y_t)_{t \geq 0}$ is self-invariant by conjugation if it is invariant by conjugation by $H_{Y}$. 
\end{definition}

Let $\eta$ be a finite Borel measure on $G^{n}$. For any $g$ in $G$, the measure $\eta^{g}$ on $G^{n}$ is the unique measure such that for any continuous function $f: G^{n} \to \mathbb{R}$: 
\begin{align}
\label{eq:diagconj}
\eta^g(f) = \int_{G} f(g^{-1}g_1g, ..., g^{-1}g_ng) \eta(dg_1,...,dg_n).
\end{align} 

We can now construct a planar Yang-Mills field associated with any self-invariant by conjugation L\'{e}vy process. 

\begin{theorem}
\label{defplanarHFgeneral} 
For every $G$-valued self-invariant by conjugation L\'{e}vy process~$Y$, there exists a unique stochastically continuous strong planar Markovian holonomy field $\big(\mathbb{E}^{Y}_{vol}\big)_{vol}$, called the planar Yang-Mills field associated with $Y$, such that for any measure of area $vol$, for any finite planar graph $\mathbb{G}$, for any rooted spanning tree $T$ of $\mathbb{G}$ and any family of facial loops $\left(c_{F}\right)_{F\in\mathbb{F}^{b}}$ oriented anti-clockwise, under $\mathbb{E}^{Y}_{vol}$, the law of $\big(h(\l_{c_{F},T})\big)_{F \in \mathbb{F}^{b}}$ is: $$\int_G \big( \otimes_{F \in \mathbb{F}^{b}} m_{vol(F)}\big)^{g} dg,$$
where $(m_t)_{t \geq 0}$ is the semi-group of convolution of measures associated with $Y$. 
\end{theorem}

Let us notice that on $\big(\mathbb{E}^{Y}_{vol}, \mathcal{B}\big)$, $\big(h(\l_{c_F, T})\big)_{F \in \mathbb{F}^{b}}$ is not, in general, a sequence of independent variables.

\begin{proof}
The unicity part uses the same arguments as usual. Let us prove the existence of the stochastically continuous strong planar Markovian holonomy field $\big(\mathbb{E}^{Y}_{vol}\big)_{vol}$. Let $Y=\big(Y_t\big)_{t \geq 0}$ be a $G$-valued self-invariant by conjugation L\'{e}vy process and $H_Y$ be the support of $Y$. Using the discussion after Definition \ref{supportlevy}, we can see the process $Y$ as a $H_Y$-valued L\'{e}vy process which is invariant by conjugation (by $H_Y$). Thus, applying Propositions \ref{creaYMpure} and \ref{creaYMpure2}, there exists a $H_Y$-valued stochastically continuous strong planar Markovian holonomy field such that for any measure of area $vol$, any finite planar graph $\mathbb{G}$, for any rooted spanning tree $T$ of $\mathbb{G}$ and any family of facial loops $\left(c_{F}\right)_{F\in\mathbb{F}^{b}}$ oriented anti-clockwise, under $\mathbb{E}^{Y}_{vol}$: 
\begin{enumerate}
\item the random variables $\left(h\left(\l_{c_{F},T}\right)\right)_{F\in \mathbb{F}^{b}}$ are independent, 
\item for any $F\in \mathbb{F}^{b}$, $h(\l_{c_{F},T})$ has the same law as $Y_{vol(F)}$. 
\end{enumerate}

Recall that Proposition \ref{extensiondiscret} can also be applied to planar (continuous) Markovian holonomy fields. Thus we can extend the group on which $\mathbb{E}^{Y}_{vol}$ is defined, from $H_Y$ to $G$: we will denote it $\mathbb{E}^{Y}_{vol}$. It is a $G$-valued stochastically continuous strong planar Markovian holonomy field and by definition, for any measure of area $vol$, for any finite planar graph $\mathbb{G}$, for any rooted spanning tree $T$ of $\mathbb{G}$ and any family of facial loops $\left(c_{F}\right)_{F\in\mathbb{F}^{b}}$ oriented anti-clockwise, under $\mathbb{E}^{Y}_{vol}$, the law of $\big(h(\l_{c_{F},T})\big)_{F \in \mathbb{F}^{b}}$ is: $$\int_G \left(\otimes_{F \in  \mathbb{F}^{b}}m_{vol(F)}\right)^{g} dg.$$
This ends the proof of the theorem. 
\end{proof}

\begin{definition}
By construction, the planar Yang-Mills field associated with a self-invariant by conjugation L\'{e}vy process $Y$ is constructible. Its restriction to multiplicative functions on finite planar graphs is called the {\em discrete planar Yang-Mills field} associated with $Y$, we will denote it $\big( \mathbb{E}^{Y,\mathbb{G}}_{vol}\big)_{\mathbb{G}, vol}$. 
\end{definition}

We are led to classify the planar Yang-Mills fields according to their degree of symmetry. In Section \ref{Charac}, we will prove equivalent conditions in order to classify planar Yang-Mills fields. 

\begin{definition}
\label{defclassification}
Let $\big( \mathbb{E}^{Y}_{vol} \big)_{vol}$ be a planar Yang-Mills field associated with a $G$-valued self-invariant by conjugation L\'{e}vy process $Y=\big(Y_t\big)_{t \in \mathbb{R}^{+}}$.  The planar Yang-Mills field $\big( \mathbb{E}^{Y}_{vol} \big)_{vol}$ and the L\'{e}vy process $Y$ are {\em pure} if $\big(Y_t\big)_{t\geq 0}$ is invariant by conjugation by $G$ and {\em mixed} if not pure. They are also {\em non-degenerate} if $H_Y = G$ and {\em degenerate} if $H_Y \ne G$. The same definition holds for the discrete planar Yang-Mills field associated with $Y$. 
\end{definition}

According to this definition, any planar Yang-Mills field is either pure non-degenerate, pure degenerate or mixed degenerate.

\part{Characterization of Planar Markovian Holonomy Fields}

\chapter[Braids and probabilities II]{Braids and probabilities II: a geometric point of view, infinite random sequences and random processes}
\label{Braids2}
In order to characterize planar Markovian holonomy fields, we will use intensively the invariance by area-preserving homeomorphisms. The braid group will appear again, as the diffeotopy group of the $n$-punctured disk which is studied in the following section. 

\section{Braids as the diffeotopy group of the $n$-punctured disk}

Let $\mathbb{D}$ be the disk of center $0$ and radius $1$ and $Q_{n}=\{ q_{k} = \frac{2k-1-n}{n},\ 1\leq k \leq n\}$. Let $\text{Diff}(\mathbb{D}, Q_{n}, \partial{\mathbb{D}})$ be the group of diffeomorphisms of $\mathbb{D}$ which fix the set $Q_{n}$ and fix pointwise a neighborhood of $\partial{\mathbb{D}}$. The class of isotopy of the identity mapping in $\text{Diff}(\mathbb{D}, Q_{n}, \partial{\mathbb{D}})$ is a normal subgroup called $\text{Diff}_{0}(\mathbb{D}, Q_{n}, \partial{\mathbb{D}})$. The diffeotopy group of the disk with $n$ points is $\mathcal{M}_{n}(\mathbb{D}) = \text{Diff}(\mathbb{D}, Q_{n}, \partial{\mathbb{D}}) \Big/ \text{Diff}_{0}(\mathbb{D}, Q_{n}, \partial{\mathbb{D}}).$ An important theorem is that: 
\begin{align}
\label{identification}
\mathcal{M}_{n} (\mathbb{D}) \simeq \mathcal{B}_{n}.
\end{align}

This isomorphism is constructed by sending some special elements, the half-twists or Dehn-twists, on the canonical family of generators of the braid groups denoted by $(\beta_i)_{i=1}^{n-1}$. A half-twist permutes the points $q_{k}$ and $q_{k+1}$ for some $k$ and does not move the other points $\left(q_i\right)_{i \notin \{k,k+1\}}$. For a precise definition, one considers for $1\leq k\leq n$, $t_k$ the isotopy class of the diffeomorphism $\tilde{t}_k$ equals to identity outside the disk of radius $\frac{2}{n}$ centered at $\frac{q_k+q_{k+1}}{2}$ and defined by $\tilde{t}_k(x) = \psi \circ t \circ \psi^{-1}$, where:
\begin{align*}
&\psi: x \mapsto \frac{n}{2}\left( x-\frac{q_k+q_{k+1}}{2}\right),\\
&t(re^{i \theta}) = r e^{i 2 \pi \big(\theta+\alpha( r)\big)}, 
\end{align*}
and $\alpha$ is a smooth bijection from $[0,1]$ to itself, which is equal to $0$ on a neighborhood of $1$ and to $\frac{1}{2}$ at $\frac{1}{2}$. 

This geometric construction of the braid group allows us to recover the action of the braid group which was given by Definition \ref{actionlibre}. Indeed, the group $\mathcal{M}_{n}(\mathbb{D})$ acts on the fundamental group of $\mathbb{D} \setminus Q_{n}$ which is isomorphic to $\mathbb{F}_{n}$ the free group of rank $n$. We will take $-\frac{i}{2}$ as the base point for the fundamental group of $\mathbb{D} \setminus Q_{n}$. Let $x_{k}$ be the homotopy class of the loop based at $-\frac{i}{2}$ which goes only around $q_{k}$ anti-clockwise. 
One can verify that the action of $\mathcal{M}_{n}(\mathbb{D})$ on $\mathbb{F}_{n}$, with the identification given in (\ref{identification}), is the action given by  Definition \ref{actionlibre}. 

Given a finite planar graph $\mathbb{G}$, the fundamental group of $\mathbb{G}$ is isomorphic to the fundamental group of the disk without one point in each of the bounded faces: $\pi_1(\mathbb{G}) \simeq \pi_{1}\left(\left(\mathbb{D}\setminus Q_{| \mathbb{F}^{b}|}\right)\right)$. Thus, we get a natural action of a braid group on the group $\pi_{1}(\mathbb{G})$ which is isomorphic to the reduced group of loops on $\mathbb{G}$ defined in Section \ref{sec:red}. A consequence of the existence of such action is the upcoming Proposition \ref{braidable}.

\section{A de-Finetti theorem for the braid group}
\label{infinite-sec}

Proposition \ref{quasi-inv} implies that every finite sequence of i.i.d. random variables which is auto-invariant by conjugation is invariant by braids. It is natural to wonder if one can characterize finite sequence of random variables which are invariant by braids. As for exchangeable sequences of random variables, it is easier to work with infinite sequence of random variables.

\begin{definition}
An infinite random sequence $\xi~= \big(\xi_{i}\big)_{i \in \mathbb{N}^{*}}$ in $G$ is braidable (or braid-invariant or invariant by braids) if for any integer $n$ greater that $1$ and any braid $\beta \in \mathcal{B}_{n}$, the following equality in law holds: 
$$\beta \bullet \big(\xi_{i}\big)_{1\leq i \leq n } = \big(\xi_{i}\big)_{1\leq i \leq n }.$$ 
\end{definition}

Let us recall that $\xi$ is spreadable if for any increasing sequence of positive integers $(k_i)_{i\geq 1}$, we have the equality in law: 
$$\big(\xi_{k_{i}}\big)_{1\leq i} = \big(\xi_{i}\big)_{1\leq i}.$$
These properties seem to be quite different, yet we are going to prove that one condition is weaker than the other.

\begin{lemme}
\label{braidspread}
Any braidable infinite family of random variables is spreadable. 
\end{lemme}

\begin{proof}
Let $\mathbf{k}~= \left(k_{1} < k_{2} < ... < k_{n}\right)$ be a finite strictly increasing sequence of integers. Let $\beta_{\mathbf{k}}$ be the braid: 
\begin{align*}
\beta_{\mathbf{k}} =\beta_{n}^{-1}...\beta_{k_n-1}^{-1}\beta_{2}^{-1}... \beta_{k_2-1}^{-1}\beta_{1}^{-1}... \beta_{k_1-1}^{-1}
\end{align*}
We have drawn in Figure \ref{betakspearfig} the braid $\beta_{\mathbf{k}}$ with $\mathbf{k} = (2,3,6)$. Since for $i=1,  ...,n$, the lines linking $(i,1)$ to $(k_{i},0)$ are behind in the diagram, this braid verifies that for any element $(g_{1},  ...,g_{k_{n}})$ of $G^{k_{n}}$, for any integer $i$ between $1$ and $n$, $\big(\beta_{\mathbf{k}}\bullet (g_{1},...,g_{k_{n}})\big)_{i} = g_{k_{i}}$. 
\begin{figure}
 \centering
  \includegraphics[width=200pt]{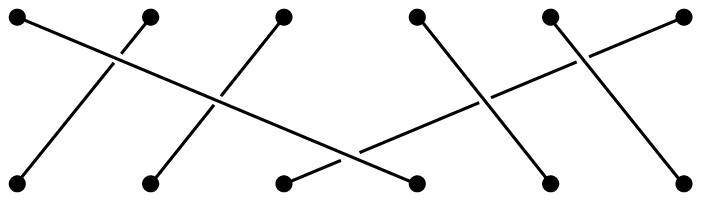}
 \caption{The braid $\beta_{\mathbf{k}}$, $\mathbf{k} = (2,3,6)$.}
 \label{betakspearfig}
\end{figure} 
Let $\xi = (\xi_{i})_{i \in \mathbb{N}^{*}}$ be a braidable random sequence. The following equality in law holds $\beta_{\mathbf{k}} \bullet \big(\xi_{i}\big)_{1\leq i \leq k_n} = \big(\xi_{i}\big)_{1\leq i\leq k_n}.$ By restricting it for $i$ between $0$ and $n-1$, we get the equality in law $(\xi_{k_{1}},... \xi_{k_{n}}) = (\xi_{1},... \xi_{n}),$ from which one can conclude that $\xi$ is spreadable. 
\end{proof}

Let $m$ be a probability measure on $G$. We denote by $m^{\otimes \infty}$ the measure on $G^{\mathbb{N}^{*}}$ such that the unidimensional projections are independent and identically distributed with law $m$. Let $\xi$ be an infinite random sequence in $G$.

\begin{definition}
 Let $\mathcal{A}$ be a $\sigma$-field. The sequence $\xi$ is i.i.d. conditionally to $\mathcal{A}$ if there exists a random measure $\eta$ on $G$, $\mathcal{A}$-measurable, such that the conditional distribution of $\xi$ given $\mathcal{A}$ is $\eta^{\otimes \infty}$: $\mathbb{P}\big[\xi \in . \mid \mathcal{A} \big] = \eta^{\otimes \infty}$. It is conditionally i.i.d. if there exists a $\sigma$-field $\mathcal{A}$ such that it is i.i.d. conditionally to~$\mathcal{A}$.
\end{definition}

If $\xi$ is i.i.d. conditionally to $\mathcal{A}$, its law is of the form: $$\int_{\mathcal{M}_{1}(G)} m^{\otimes \infty} d\nu(m),$$ where $\nu$ is the law of $\eta$. If we just want to keep in mind the form of the law of $\xi$, we will say that $\xi$ is a {\em mixture of i.i.d. random sequences.} Let us state an extension of de Finetti-Ryll-Nardzewski's theorem for the braid group. 

\begin{theorem}
\label{definettiseq}
The sequence $\xi$ is braidable if and only if it is i.i.d. conditionally to $\mathcal{T}_{\xi}=\cap_{n \in \mathbb{N}^{*}}\sigma\left(\xi_k, k \geq n\right)$ and conditionally to $\mathcal{T}_{\xi}$, almost surely the law of $\xi_1$ is invariant by conjugation by its own support. 
\end{theorem}

\begin{proof}
An application of Proposition \ref{quasi-inv} to any subsequence of the form $(\xi_{n})_{n=1}^{N}$ shows that the second condition implies the fact that $\xi$ is braidable. 

Now, let us suppose that $\xi$ is braidable. As a consequence of Lemma \ref{braidspread}, the infinite sequence $\xi$ is spreadable. Using the de Finetti-Ryll-Nardzewski's theorem (Theorem $1.1$ of \cite{Kall}), $\xi$ is  i.i.d. conditionally to $\mathcal{T}_\xi$. Besides, conditionally to $\mathcal{T}_{\xi}$, $\xi$ is still braidable: an application of Proposition \ref{quasi-inv} shows that conditionally to $\mathcal{T}_{\xi}$, $(\xi_n)_{n\in \mathbb{N}^{*}}$ is an i.i.d. sequence of random variables invariant by conjugation by their own support: the second condition holds.  
\end{proof}

If $\xi$ is braidable, the law of $\xi$ is of the form $\int_{\mathcal{M}_{1}(G)} m^{\otimes \infty} d\nu(m),$ where $\nu$-a.s., $m$ is almost surely invariant by conjugation by its own support. In the next theorem, we give a condition under which one can characterize the mixture which appears in the last theorem. In order to do so, we consider the diagonal conjugation of $G$ on $G^{\mathbb{N}^{*}}$ defined for any $g\in G$ and $(x_n)_{n \in \mathbb{N}^{*}} \in G^{\mathbb{N}^{*}}$ by $g.(x_n)_{n \in \mathbb{N}^{*}}= (g^{-1}x_ng)_{n \in \mathbb{N}^{*}}$. Recall the notion of $\mathcal{I}$-independence defined in Definition \ref{Iindependence}. 

\begin{definition}\label{proprietepdesincrements}
The sequence $\xi$ satisfies the property $(\mathcal{P})$ if for any integer $n \in \mathbb{N}^{*}$, $(\xi_k)_{k \leq n}$ and $(\xi_k)_{k > n}$ are $\mathcal{I}$-independent: for any $n,m\geq 0$, any continuous functions which are invariant by diagonal conjugation $f: G^{n} \to \mathbb{R}$ and $g:G^{m} \to \mathbb{R}$, $f(\xi_1,...,\xi_n)$ and $g(\xi_{n+1}, ..., \xi_{n+m})$ are independent. 
\end{definition}

\begin{theorem}
\label{finalsuite}
Let $\xi$ be an sequence of random variables in $G$ such that: 
\begin{enumerate}
\item $\xi$ is braidable, 
\item it is invariant (in law) by diagonal conjugation: for any $g \in G$, $g.\xi$ has the same law as $\xi$,
\item it satisfies the property $(\mathcal{P})$. 
\end{enumerate}
There exists $m_0$ a probability measure on $G$, invariant by conjugation by its own support, such that the law of $\xi$ is: 
\begin{align*}
\int_{G}  \left(m_0^{\otimes \infty}\right)^{g} dg, 
\end{align*}
where $\left(m_0^{\otimes \infty}\right)^{g}$ is defined using a similar equation as (\ref{eq:diagconj}).
\end{theorem}

\begin{proof}
Let $\xi$ be an infinite braidable sequence of $G$-valued random variables which is invariant by diagonal conjugation and satisfies the property $(\mathcal{P})$. As a consequence of Theorem \ref{definettiseq}, there exists a random measure $\eta$ on $G$, which is almost surely invariant by conjugation by its own support, such that the conditional distribution of $\xi$ given $\mathcal{T}_{\xi}$ is $\eta^{\otimes \infty}$. Let $\nu$ be the law of $\eta$, the law of $\xi$ is $\int_{\mathcal{M}^{1}(G)} m^{\otimes\infty} d\nu(m).$ Since $\xi$ is invariant by diagonal conjugation by $G$, we only have to show that there exists a probability measure $m_0$ such that the law of $\xi$ on the invariant $\sigma$-field $\mathcal{I}$ is equal to $m_0^{\otimes \infty}$. Let us remark that this would imply also that $m_0$ is invariant by conjugation by any element of its own support. Let $k$ be a positive integer and $f: G^{k} \to \mathbb{R}$ be a continuous function invariant by diagonal conjugation: as $G^{k}$ is compact, $f$ is bounded. As the sequence $\xi$ satisfies the property $(\mathcal{P})$, $\left(f\left(\xi_{ik+1}, ..., \xi_{ik+k}\right)\right)_{i \geq 0}$ is an i.i.d. sequence of bounded random variables. Thus, by the law of large numbers, there exists a real $l^k(f)$ such that: 
\begin{align*}
\frac{1}{n} \sum_{0\leq i< n} f\left(\xi_{ik+1}, ..., \xi_{ik+k}\right)  \underset{n \to \infty}{\to} l^{k}(f) \text{ a.s.}. 
\end{align*}
Yet, by disintegration and the law of large numbers, we get also that: 
 \begin{align*}
\frac{1}{n} \sum_{0\leq i< n} f\left(\xi_{ik+1}, ..., \xi_{ik+k}\right) \underset{n \to \infty}{\to} \eta^{\otimes k}(f), \text{ a.s.}.  
\end{align*}
The random variable $\eta^{\otimes k}(f)$ is thus almost surely constant. Let us define the set of measures: 
\begin{align*}
\Omega_{f,k} = \left\{m \in \mathcal{M}_{1}(G), m^{\otimes k}(f) = l^{k}(f)\right\}. 
\end{align*}
We just proved that for any positive integer $k$, for any continuous function $f: G^{k} \to \mathbb{R}$, invariant by diagonal conjugation, 
\begin{align*}
\nu(\Omega_{f,k}) = 1. 
\end{align*}
Let us consider $\mathcal{F}_{k}$ a dense countable set of continuous functions which are invariant by diagonal conjugation of $G$ on $G^{k}$. Our previous discussion allows us to write the following equality: 
\begin{align*}
\nu\bigg(\bigcap_{k \in \mathbb{N}^{*}}\bigcap_{f \in \mathcal{F}_k}\Omega_{f,k}\bigg) = 1. 
\end{align*}

Let us take a measure $m_0 \in \bigcap_{k \in \mathbb{N}^{*}}\bigcap_{f \in \mathcal{F}_k}\Omega_{f,k}$: for any positive integer $k$ and any continuous function $f$ on $G^{k}$ invariant by diagonal conjugation, $$m^{\otimes k}(f)=m_{0}^{\otimes k}(f) \text{, } \nu(dm) \text{ a.s. }$$
We have just managed to prove that on the invariant $\sigma$-field, the law of $\xi$ is $m_0^{\otimes \infty}$: the law of $\xi$ is thus $\int_{G} (m_{0}^{\otimes \infty})^{g} dg$ and $m_0$ is invariant by conjugation by its own support. 
\end{proof}

We will give conditions which ensure the fact that $\xi$ is actually a sequence of i.i.d. random variables. Let us remark that it is enough to show that the measure $m_0$, given by Theorem \ref{finalsuite}, is invariant by conjugation by $G$.

\section{Degeneracy of the mixture}
\label{degen-sec}
Since the measure $m_0$ is invariant by conjugation by its own support, one way in order to show that $m_0$ is invariant by conjugation by $G$ is to prove that ${\sf Supp}(m_0) = G$.

\subsection{Non-degeneracy case} 
\label{sec:non-degeneracy}
Let $G$ be a finite group whose neutral element is $e$. 
\begin{proposition}
\label{unicity1}
Let $\xi$ be an infinite sequence of $G$-valued random variables which is braidable, invariant by diagonal conjugation, satisfies the property $(\mathcal{P})$ and such that $e \in {\sf Supp}(\xi_1)$. The following assertions are equivalent: 
\begin{enumerate}
\item $\xi$ is a sequence of i.i.d. random variables which support generates $G$, 
\item there exists an integer $k$ such that ${\sf Supp}(\prod_{i=1}^{k}\xi_i) = G$. 
\end{enumerate}
\end{proposition}

In order to prove this proposition, we will need the two following facts which hold only when $G$ is finite. The first assertion is that for any measure $m$ on $G$ such that $e\in {\sf Supp}(m)$, there exists an integer $k$ such that for any $k'\geq k$, ${\sf Supp}(m^{*k'})$ is equal to $H_{m}$, where $H_m$ was defined after Definition \ref{def:purelyinv}. The second asserts that no subgroup of $G$ can intersect every conjugacy classes of $G$: it is Jordan's theorem given below.

\begin{theorem} \label{Jordan}
Let $H$ be a subgroup of $G$. Let us suppose that: $$G = \bigcup\limits_{g \in G} g^{-1}Hg,$$ then $G=H$. 
\end{theorem}

\begin{proof}
The proof is taken from Serre's lesson \cite{Serre} and can be summarized in a simple calculation: 
$$\#G \leq \#\bigg(\bigcup_{g \in G} g^{-1}Hg\bigg)=  \#\bigg(\bigcup\limits_{g \in G} \left(g^{-1}Hg\setminus{e} \right) \cup \{e\} \bigg)\leq \frac{\#G}{\#H} ( \# H-1)+1, $$ which can hold if and only if $H=G$.
\end{proof}

We can now handle the proof of Proposition \ref{unicity1}. 

\begin{proof}[Proof of Proposition \ref{unicity1}]
It is quite obvious that the assertion $1$ implies the assertion $2$. It remains to prove the other implication.
As a consequence of Theorem \ref{finalsuite}, there exists a probability measure $m_0$ invariant by conjugation by its own support such that the law of $\xi$ is $\int_{G} (m_0^{\otimes \infty})^{g} dg$. Let $m_0$ be such a probability measure. As $e\in {\sf Supp}(\xi_1)$, $e$ is in the support of $m_0$: the support of $m_0^{*k}$ is increasing in $k$ and so is the support of $\prod_{i=1}^{k}\xi_i$. Let $k$ be an integer such that ${\sf Supp}\big(\prod_{i=1}^{k}\xi_{i}\big)=G$: for any $k'\geq k$, ${\sf Supp}\big(\prod_{i=1}^{k'}\xi_{i}\big)=G$. Let $N \geq k$ such that ${\sf Supp}(m_0^{*N})=H_{m_0^{*N}}$. Since $N$ is greater that $k$, ${\sf Supp}\big(\prod_{i=1}^{N}\xi_i\big) = G$. On the other side, since the law of $\prod_{i=1}^{N}\xi_i$ is $\int_G (m_0^{*N})^{g} dg$, its support is equal to $\cup_{g \in G} g^{-1} H_{m_0^{*N}} g$. Thus, one has the equality: 
$$G= \bigcup_{g \in G} g^{-1} H_{m_0^{*N}}g$$
which implies, by Jordan's theorem (Theorem \ref{Jordan}), that $H_{m_0^{*N}}=G$ and then $H_{m_0}=G$: the support of $m_0$ generates the group $G$. We recall that $m_0$ was invariant by conjugation by its support, hence by $H_{m_0}$: $m_0$ is invariant by conjugation by $G$ and the law of $\xi$ is thus $m_0^{\otimes \infty}$.
\end{proof}

For an arbitrary compact Lie group, it is not true in general that for any measure $m$ on $G$ such that $e \in {\sf Supp}(m)$, there exists $k$ such that ${\sf Supp}(m^{*k}) = H_{m}$. Thus, in order to deal with any compact Lie group, we will substitute this fact by the It\^{o}-Kawada's theorem (Theorem \ref{Itokawada}). 

As for Jordan's theorem, it does not hold when $G$ is infinite, as in every compact Lie group, any maximal torus intersects all the conjugacy classes. Thus, we have to impose that the subgroup $H$ intersects every conjugacy class ``as much as'' $G$ does, which is the meaning of the condition imposed in the upcoming Proposition \ref{Jordangen}. Doing so, we will be able to prove the following proposition which holds for any arbitrary compact Lie group. From now on, let $G$ be a compact Lie group.

\begin{proposition}
\label{unicity2}
 Let $\xi$ be an infinite sequence of $G$-valued random variables which is braidable, invariant by diagonal conjugation and satisfies the property $(\mathcal{P})$. Let us suppose that $e\in {\sf Supp}(\xi_1)$. The following assertions are equivalent: 
\begin{enumerate}
\item $\xi$ is a sequence of i.i.d. random variables which support generates $G$. 
\item the random variables $\prod_{k=1}^{n}\xi_k$ converge in law to a Haar random variable as $n$ goes to infinity.
\end{enumerate}
\end{proposition}

In order to prove this proposition, let us introduce It\^{o}-Kawada's theorem and a measurable version of the theorem of Jordan which holds for any compact Lie group. 

\begin{definition}
Let $m$ be a probability measure on $G$. It is:
\begin{itemize}
\item {\em aperiodic} if its support ${\sf Supp}(m)$ is not contained in a left or right proper coset of a proper closed subgroup of $G$,
\item {\em non-degenerate} if $H_{m} = G$.
\end{itemize}
\end{definition}

\begin{remarque}
\label{aperiodique}
 It is obvious that $m$ is non-degenerate if it is seen as a measure on $H_{m}$. Besides, if ${e} \in {\sf Supp}(m)$ then $m$ is aperiodic. 
\end{remarque}

Under the condition of aperiodicity and non-degeneracy, It\^{o}-Kawada's theorem (Theorem $3.3.5.$ of \cite{Itokawa1}, first proved in \cite{Itokawa2}) explains the behavior of $m^{*n}$ when $n$ goes to infinity. 

\begin{theorem}[It\^{o}-Kawada's theorem] \label{Itokawada}
Let $\mu$ be a non-degenerate and aperiodic probability measure on $G$. The sequence $\mu^{*n}$ converges in distribution to the normalized Haar measure on $G$ as $n$ goes to infinity. 
\end{theorem}

Let us state our generalization of Jordan's theorem, Theorem \ref{Jordan}, valid for any compact Lie group. Recall that for any compact Lie group $K$, we denote by $\lambda_K$ the normalized Haar measure on $K$.

\begin{proposition}
 \label{Jordangen}
Let $H$ be a closed subgroup of $G$. If 
\begin{align*}
\int_G \lambda_{g^{-1}Hg} dg = \lambda_G, 
\end{align*}
then $G=H$. 
\end{proposition}

\begin{proof}
We want to prove that  $\int_G \lambda_{g^{-1}Hg}  \neq \lambda_G$: it is enough to construct a continuous function $\phi: G \to \mathbb{R}$, invariant by conjugation,  such that $\lambda_H(\phi) \neq \lambda_G(\phi)$. The space $H\setminus G$ of right cosets of $H$ is a nice topological space: it is a differentiable manifold and there exists $\tilde{f}$ a non-constant real continuous function on $H\setminus G$. Let $p: G \to H \setminus G$ be the canonical projection, the function $f = \tilde{f} \circ p$ is a real non-constant square-integrable function $f$ on $G$ invariant by left multiplication by $H$: for any $g \in G$, for any $h\in H$, $f(g)= f(hg).$ One can also assume that $f$ is of zero mean on $G$. 

Let $E=\big\{ \phi \in L^{2}(G),\ \int_{G}\phi(g)d\lambda_{G}(g)=0\big\}$ be the space of square-integrable zero mean functions on $G$. The group $G$ acts on $E$, by left multiplication on the argument and this representation has no non-zero fixed point. On the other hand, the restriction of this representation on $H$ has at least one fixed point, namely $f$. 
We can decompose $E$ as a sum of finite dimensional irreducible representations of~$G$: $$E = \bigoplus\limits_{i=1}^{\infty} E_{i}.$$ None of the $E_{i}$ is the trivial representation of $G$ as we have restricted the action of $G$ to zero mean functions. The action of $H$ on $E$ admits a fixed point $f$. We can decompose $f$ on $\bigoplus\limits_{i=1}^{\infty} E_{i}$. As, for any integer $i$, the space $E_i$ is invariant under the action of $H$, there exists at least an integer $i_{0}$ such that $E_{i_{0}}$ seen as a $H$-module is not irreducible. We denote by $\chi_{i_{0}}$ the character of the $G$-module $E_{i_{0}}$. By the classical theory of character, 
\begin{align*}
\int_{G} \chi_{i_{0}} d\lambda_{G} &= \dim(E_{i_{0}}^{G}) = 0, 
\end{align*}
whereas: 
\begin{align*}
\int_{G} \chi_{i_{0}} d\lambda_{H} &= \int_{H} \chi_{i_{0}} d\lambda_{H} = \dim(E_{i_{0}}^{H}) \geq 1,
\end{align*}
where $E_{i_0}^{G}$ and $E_{i_0}^{H}$ are the vector spaces of fixed points in $E_{i_0}$ under the actions of $G$ and $H$. Thus, we just found a central function $\chi_{i_0}$ such that: 
\begin{align*}\int_{G} \chi_{i_0} d\lambda_{G} \neq \int_{G} \chi_{i_0} d\lambda_{H}.\end{align*}
This ends the proof. \end{proof}

%

We have now all the tools in order to prove Proposition \ref{unicity2}.

\begin{proof}[Proof of Proposition \ref{unicity2}]
As a consequence of Theorem \ref{Itokawada}, Remark \ref{aperiodique} and the fact that $e \in {\sf Supp}(\xi_1)$, it is easy to see that the condition $1$ implies condition $2$. Let us prove the other implication. 

Let $\xi$ be an infinite sequence of $G$-valued random variables which is braidable, invariant by diagonal conjugation and which satisfies the property $(\mathcal{P})$. Let us suppose that $e \in {\sf Supp}(\xi_1)$. As a consequence of Theorem \ref{finalsuite}, there exists a probability measure $m_0$ invariant by conjugation by its own support such that the law of $\xi$ is $\int_{G} (m_0^{\otimes \infty})^{g} dg$. Let $m_0$ be such a probability measure. Using the hypothesis on $\xi_1$, $e \in  {\sf Supp}(m_0)$. 

Let us suppose that $\prod_{k=1}^{n} \xi_k$ converges in law to a Haar random variable. As we have seen in Remark \ref{aperiodique}, the measure $m_0$ is aperiodic and non-degenerate if seen as a measure on $H_{m_0}$. Using It\^{o}-Kawada's theorem, when $n$ goes to infinity, $m_0^{*n}$ converges to the Haar probability measure on $H_{m_0}$ which we denote by $\lambda_{H_{m_0}}$. For any integer $n$, the law of $\prod\limits_{i=1}^{n}\xi_i$ is $\int_{G} \left(m_0^{*n}\right)^{g} dg$ and thus, using the hypothesis on the law of $\prod\limits_{i=1}^{n}\xi_i$ and our previous discussion, one gets the equality: 
\begin{align*}
 \lambda_G = \int_G \lambda_{g^{-1}H_{m_0}g}dg.
\end{align*}
By Proposition \ref{Jordangen}, it follows that $H_{m_0} = G$. Since the measure $m_0$ is invariant by conjugation by $H_{m_0}$, it is invariant by conjugation by $G$: the law of $\xi$ is $m_0^{\otimes \infty}$.
\end{proof}

\subsection{ General case }
\label{sec:general}

The following theorem gives weaker conditions on $\xi$ in order to understand the case when $\xi$ is an infinite sequence of i.i.d. random variables such that $H_{\xi_1} \ne G$. We recall that $G$ is a compact Lie group. 

\begin{theorem} 
\label{unicity3}
Let $\xi$ be an infinite sequence of $G$-valued random variables which is braidable, invariant by diagonal conjugation and satisfies the property $(\mathcal{P})$. The following assertions are equivalent: 
\begin{enumerate}
\item the sequence $\xi$ is a sequence of i.i.d. random variables invariant by conjugation by $G$, 
\item there exists $\nu$ a probability measure on $G$ such that for any positive integer $n$, the law of $\prod_{k=1}^{n}\xi_k$ is $\nu^{*n}$.
\end{enumerate}
\end{theorem}

Before proving this theorem, let us recall some basic, yet crucial, results about representations and integration. First of all, Peter-Weyl's theorem asserts that the set of matrix elements of irreducible representations 
\begin{align*}
\left\{ g \mapsto v(\pi(g)w), \pi \in \hat{G}, v \in V_{\pi}^{*}, w \in V_{\pi} \right\}, 
\end{align*} 
where $\hat{G}$ is the set of irreducible representations of $G$, is dense for the uniform norm in the set of continuous functions on $G$. Thus, any measure $m$ on $G$ is fully characterized by its Fourier coefficients defined as: 
\begin{align*}
\forall \pi \in \hat{G}\text{, } \pi(m) = \int_{G} \pi(g) m(dg). 
\end{align*}
Secondly, let $\pi: G \to Gl(V)$ be an irreducible representation of dimension $d_\pi$. Let $A$ be a matrix acting on $V$. By the Schur's lemma, $$\int_G \pi(g) A \pi(g)^{-1}dg = \frac{Tr(A)}{d_\pi} Id.$$

Let $m$ be a probability measure on $G$. 
\begin{definition}
The measure $m$ {\em quasi-invariant by conjugation} if there exists $\nu$ a probability measure on $G$ such that \text{for any }$n \in \mathbb{N}$ 
$$\int_G (m^{g})^{*n} dg = \nu^{*n}.$$
\end{definition}

The main characterization of quasi-invariant by conjugation probability measures is given by the following proposition. 

\begin{proposition}
\label{onemargprop}
The measure $m$ is quasi-invariant by conjugation if and only if for any irreducible representation $\pi$ of $G$, the matrix $\pi(m)$ has only one eigenvalue. 
\end{proposition}

\begin{proof}
Let $\pi$ be an irreducible representation of $G$ and let $n$ be a positive integer. Let us compute $\pi\left(\int_G(m^{*n})^{g}dg\right)$ and $\pi\left(\left(\int_Gm^{g}dg\right)^{*n} \right)$: 
\begin{align*}
\pi\left(\int_G(m^{*n})^{g}dg\right) &= \int_{G^{2}} \pi(gg'g^{-1}) m^{*n}(dg') dg \\&= \int_{G} \pi(g) \left(\int_G\pi(g') m^{*n}(dg')\right) \pi(g)^{-1} dg
\\&=  \frac{1}{d_{\pi}} Tr\left( \pi(m)^{n}\right) Id, \\
\pi\left(\left(\int_Gm^{g}dg\right)^{*n} \right) &= \pi\left(\left(\int_Gm^{g}dg\right) \right)^{n} \\
&= \left(\frac{1}{d_\pi} Tr(\pi(m))\right)^{n}Id.
\end{align*}

Proposition \ref{onemargprop} is equivalent to the following assertion: for any positive integer~$n$, 
\begin{align}
\label{eq:quasi}
\int_{g} (m^{g})^{*n }dg = \left( \int_{g} m^{g} dg \right)^{*n}
\end{align}
if and only if for any irreducible representation $\pi$ of $G$, the matrix $\pi(m)$ has only one eigenvalue. Yet, using the remark about Peter-Weyl's theorem, Equality (\ref{eq:quasi}) holds for any positive integer $n$ if and only if for any irreducible representation $\pi$, for any positive integer $n$: 
\begin{align*}
\pi\left(\int_{g} (m^{g})^{*n }dg\right) = \pi\left(\left( \int_{g} m^{g} dg \right)^{*n}\right),
\end{align*}
hence if and only if for any irreducible representation $\pi$, for any positive integer $n$: 
\begin{align*}
Tr\left(\pi(m)^{n}\right) = Tr\left( \left(\frac{Tr(\pi(m))}{d_\pi} Id \right)^{n}\right).
\end{align*}

The proposition is a consequence of the link between the traces of the positive powers of a finite matrix and the set of its eigenvalues and the fact that the matrix $\frac{Tr(\pi(\nu))}{d_\pi} Id$ has only one eigenvalue.
 \end{proof}

It is natural to wonder if a quasi-invariant by conjugation probability measure is invariant by conjugation. The answer is no and we will construct a counter-example in the symmetric group $\mathfrak{S}_{3}$. 

\begin{lemma}
\label{contre-ex}
Let $\left(\mu_t\right)_{t\geq 0}$ (reps. $\left(\eta_t\right)_{t\geq 0}$) be the continuous semi-group of convolution of measures starting from $\delta_{id}$ on the symmetric group $\mathfrak{S}_3$, associated with the jump measure $m$ (resp. $m_0$): 
\begin{align*}
&m((12))\ =0, &&m_0((12))\ =1,\\ 
&m ((13))\ = 1, &&m_0((13))\ =1,\\
&m((23))\ = 2, && m_0((23))\ = 1, \\ 
&m((123))=2, &&m_0((123))=1,\\
&m((132)) =0, &&m_0((132))=1. 
\end{align*}
The measure $\mu_1$ is quasi-invariant by conjugation and for any $n \in \mathbb{N}$, 
\begin{align*}
\int_{G} (\mu_1^{g})^{*n} dg = \eta_1^{*n}.
\end{align*} 
\end{lemma}

\begin{proof}
We have to check that the condition of Proposition \ref{onemargprop} is fulfilled by $\mu_1$. Actually we only have to show that for any $\pi\in \widehat{\mathfrak{S}_3}$, $\pi(m)$ has only one eigenvalue, which is equal to the one of $\pi(m_0)$. The group $\mathfrak{S}_3$ has only three irreducible representations two of which have dimension one. It remains to compute $\pi(m)$ where $\pi$ is the representation of $\mathfrak{S}_3$ on $\{(a, b, c) \in \mathbb{R}^{3}, a+b+c=0 \}.$ We leave this calculation as an exercise. 
\end{proof}

The quasi-invariance by conjugation property does not imply the invariance by conjugation property, yet, we have the following theorem. 

\begin{theorem}
\label{onemarg}
Let $m$ be a probability measure on $G$. Let us suppose that $m$ is invariant by conjugation by its own support. Then, $m$ is quasi-invariant by conjugation if and only if $m$ is invariant by conjugation by $G$. 
\end{theorem}

\begin{proof}
Almost by definition, any probability measure which is invariant by conjugation is quasi-invariant by conjugation. It remains to prove the ``only if" part of the theorem. Let $m$ be a quasi-invariant by conjugation probability measure. Using Proposition \ref{onemargprop}, for any irreducible representation $\pi$ of $G$, $\pi(m)$ has only one eigenvalue. Any irreducible representation $\pi \in \hat{G}$ defines by restriction a representation of $H_{m}$. Since $H_{m}$ is a closed subgroup of $G$, it is a compact Lie group, thus we can apply Peter-Weyl's theorem which allows us to decompose any representation of $H_{m}$ as a direct sum of irreducible representations: $$\pi = \bigoplus_{i=1}^{n} \pi_i,$$ with $\pi_i \in \hat{H}_{m}$. Since $m$ is invariant by conjugation by $H_{m}$, thanks to Schur's lemma, for any $i \in \{1,...,n\}$, $\pi_i(m)$ is a scalar matrix, hence $\pi(m)$ is diagonal. As it has only one eigenvalue, it is a multiple of the identity. Let $\pi \in \hat{G}$ acting on $V$, let $w \in V$, $v \in V^{*}$ and let $h \in G$: 
\begin{align*}
\int_{G} v(\pi(hgh^{-1})w) m(dg) &= v\left( \pi(h) \left( \int_G\pi(g) m(dg) \right) \pi(h)^{-1} w\right) \\
&=v\left( \left( \int_G\pi(g) m(dg) \right) w\right) \\
&= \int_G v(\pi(g) w) m(dg).
\end{align*}
Thus, using Peter-Weyl's theorem, $m$ is invariant by conjugation by $G$. 
\end{proof}

We have now all the tools in order to prove Theorem \ref{unicity3}.
\begin{proof}[Proof of Theorem \ref{unicity3}]
Let $\xi$ be an infinite sequence of $G$-valued random variables which is braidable, invariant by diagonal conjugation and which satisfies the property $(\mathcal{P})$. As a consequence of Theorem \ref{finalsuite}, there exists a probability measure $m_0$ invariant by conjugation by its own support such that the law of $\xi$ is $\int_{G} (m_0^{\otimes \infty})^{g} dg$. 

Let us suppose that $\xi$ is a sequence of i.i.d. random variables, then one can take as $m_0$ the law of $\xi_1$: the law of $\xi$ is equal to $m_0^{\otimes \infty}$ and thus, for any $n$, the law of $\prod_{k=1}^{n} \xi_{k}$ is $m_0^{*n}$. 

Instead of assuming that $\xi$ is a sequence of i.i.d. random variables, let us suppose that there exists $\nu$ a probability measure on $G$ such that for any $n$, the law of $\prod_{k=1}^{n} \xi_{k}$ is $\nu^{*n}$. The law of $\prod_{k=1}^{n} \xi_{k}$ is $\int_{G} (m_0^{*k})^{g} dg$: it shows that the probability measure $m_0$ is quasi-invariant by conjugation. Yet, it is also invariant by conjugation by its own support. By Theorem \ref{onemarg}, $m_0$ is invariant by conjugation by $G$ and thus $\xi$ is a sequence of i.i.d. random variables. 
\end{proof}

\section{Processes and the braid group}

Sections \ref{infinite-sec} and \ref{degen-sec} can be generalized in order to get similar results for $G$-valued processes indexed by $\mathbb{R}^{+}$. 

\subsection{Definitions}

We define the rational increments of a process as Kallemberg does in \cite{Kall}. Let $(X_t)_{t \in \mathbb{R}^{+}}$ be a $G$-valued random process indexed by $\mathbb{R}^{+}$.
\begin{definition}
\label{increment}
We define the (rational) increments of $(X_t)_{t \in \mathbb{R}^{+}}$ for $n$ in $\mathbb{N}^{*}\cup (\mathbb{N}^{*})^{-1} $ and $j \geq 1 $ as: 
\begin{align}
X_{n, j} = X_{\frac{j-1}{n}}^{-1}X_{\frac{j}{n}} .
\end{align}
The process $(X_t)_{t \in \mathbb{R}^{+}}$ has spreadable (resp. braidable) increments if for every $n\in \mathbb{N}^{*}\cup (\mathbb{N}^{*})^{-1}$, the sequence $\left(X_{n, i}\right)_{0< i}$ is spreadable (resp. braidable). 
\end{definition}

In the following we will need a weak version of the notion of independence of increments which will replace the property $(\mathcal{P})$ in the study of processes with braidable increments. 
\begin{definition}
The process $(X_t)_{t \in \mathbb{R}^{+}}$ has $\mathcal{I}$-independent increments if for any increasing sequence of real $0=t_0<t_1<...<t_n<...$ the sequence $(X_{t_{n-1}}^{-1}X_{t_n})_{n\in \mathbb{N}^{*}}$ satisfies the property $(\mathcal{P})$ defined in Definition \ref{proprietepdesincrements}.
\end{definition}

Let us consider a L\'{e}vy process $(X_t)_{t \in \mathbb{R}^{+}}$.
\begin{definition}
The process $(X_t)_{t \in \mathbb{R}^{+}}$ has auto-invariant by conjugation increments if for any $0=t_0\leq t_1\leq t_2...\leq t_k$, the sequence of increments $(X_{t_{i-1}}^{-1}X_{t_i})_{i=1}^{k}$ is auto-invariant by conjugation . 
\end{definition}

Recall the notion of self-invariance defined in Definition \ref{self-invariance}. In the next proposition, we link the different notions of invariance by conjugation that can be applied to a L\'{e}vy process. 

\begin{proposition}
\label{prop-support}
The three following conditions are equivalent: 
\begin{enumerate}
\item $X$ has auto-invariant by conjugation increments,
\item for any $t \in \mathbb{R}^{+}$, $X_{t}$ is invariant by conjugation by its own support,
\item $X$ is invariant by conjugation by $H_X$, thus self-invariant by conjugation. 
\end{enumerate}
\end{proposition}

\begin{proof}
We will show that $1$ implies $2$, $3$ implies $1$ and at last $2$ implies $3$. 

Let us assume that $X$ has auto-invariant by conjugation increments. Let $t \in \mathbb{R}^{+}$, then $(X_t, X_{t}^{-1}X_{2t})$ is auto-invariant by conjugation: $X_t$ is invariant by conjugation by the support of $X_{t}^{-1}X_{2t}$. As $X$ is a L\'{e}vy process, $X_{t}^{-1}X_{2t}$ and $X_t$ has the same law, thus the same support: $X_t$ is invariant by conjugation by its own support. 
 
Now, let us show that $3$ implies $1$. Let us assume that $X$ is invariant by conjugation by $H_X$. By definition for any $t \in \mathbb{R}^{+}$, ${\sf Supp}(X_t) \subset H_X$. Let $t_1 < t_2$ and $t_3 < t_4$ be four non negative reals. As the process $X$ is invariant by conjugation by $H_X$, $X_{t_1}^{-1}X_{t_2}$ is invariant by ${\sf Supp}(X_{t_{4}-t_3})$  and thus by ${\sf Supp}(X_{t_{3}}^{-1}X_{t_{4}})$. This implies easily that $X$ has auto-invariant by conjugation increments. 

It remains to prove that $2$ implies $3$. Let us assume that for any $t\in \mathbb{R}^{+}$, $X_t$ is invariant by conjugation by its own support. Let us first remark that, if $U$ and $V$ are two random independent variables in $G$, ${\sf Supp}(UV) = \overline{{\sf Supp}(U) .{\sf Supp}(V)}.$
Besides, if they are both invariant in law by conjugation by a set $S$, $UV$ is also invariant by conjugation by $S$. Moreover, if $U$ is invariant by conjugation by any element of $S$, it is invariant by conjugation by any element of the closure of the semi-group generated by $S$: $\overline{\bigcup_{k=1}^{\infty}S^{k}}$, which, in the case where $G$ is compact, is a group. Let $n$ be a positive integer and let $t$ be a positive real. Using the hypothesis on $X$, $X_{\frac{t}{n}}$ is invariant by conjugation by ${\sf Supp}(X_{\frac{t}{n}})$. Taking $n$ independent copies of $X_{\frac{t}{n}}$ and applying the remarks above, we find that $X_{t}$ is still invariant by conjugation by ${\sf Supp}(X_{\frac{t}{n}})$ and thus also for any integer $k\geq 1$, by $\overline{{\sf Supp}(X_{\frac{t}{n}})^{k}} = {\sf Supp}(X_{\frac{k}{n}t})$, or by the semi-group generated by ${\sf Supp}(X_{\frac{k}{n}t})$, which is nothing but $H_{X_{\frac{k}{n}t}}$. Thus, $X_t$ is invariant by conjugation by: 
$$\overline{\bigcup_{q \in \mathbb{Q}^{+}} H_{X_{qt}}}.$$

Since the laws of $\big(X_t\big)_{t \geq 0}$ form a continuous semi-group of convolution of measures, for any $s \geq 0$, $X_s \in \overline{\bigcup\limits_{q \in \mathbb{Q}^{+}} H_{X_{q.t}}}$ a.s. and thus $H_{X_s} \subset \overline{\bigcup_{q \in \mathbb{Q}^{+}} H_{X_{qt}}},$ hence the equality: $$H_X = \overline{\bigcup_{q \in \mathbb{Q}^{+}} H_{X_{qt}}}.$$
Thus, for any positive real $t$, $X_t$ is invariant by conjugation by $H_X$. Using the fact that $X$ has independent and stationary increments, it implies that the L\'{e}vy process $X$ is self-invariant by conjugation. \end{proof}

\subsection{Generalized B\"ulhmann's theorem}
We can now state the generalization of B\"uhlmann's theorem (Theorem 1.19 of \cite{Kall}) for the braid group. 
\begin{theorem} 
\label{definetti}
Let $X$ be a  $G$-valued stochastically continuous process indexed by $\mathbb{R}_{+}$ with $X_{0}=e$. The following conditions are equivalent: 
\begin{enumerate}
\item X has braidable increments,
\item X is a mixture of self-invariant by conjugation L\'{e}vy processes. 
\end{enumerate}
The $\sigma$-field which makes the rational increments, as defined in Definition \ref{increment}, conditionally i.i.d. is the $\sigma$-field $\mathcal{T} = \cap_{t\in \mathbb{Q}^{+}} \sigma(X_{t}^{-1}X_{s}, s>t) $. 
Besides, the following conditions are equivalent: 
\begin{enumerate}
\item X is invariant by conjugation by $G$ and has braidable and $\mathcal{I}$-independent increments, 
\item there exists a self-invariant by conjugation L\'{e}vy process $Y$, such that the law of $X$ is $UYU^{-1}$, where $U$ is a Haar variable on $G$ independent of $Y$. 
\end{enumerate}
\end{theorem}

\begin{proof}
Let us consider the first part of the theorem. Let us show that the condition $2$ implies the first one: it is enough to show that any self-invariant by conjugation L\'{e}vy process has braidable increments. Let $Z$ be a self-invariant by conjugation L\'{e}vy process. By Proposition \ref{prop-support}, for any $n \in \mathbb{N}^{*}\cup (\mathbb{N}^{*})^{-1}$, the sequence of increments $(Z_{n,j})_{j}$, defined in Definition \ref{increment}, is a sequence of i.i.d. random variables which are invariant by conjugation by their own support. Hence, by Theorem \ref{definettiseq}, it is braidable: the process $Z$ has braidable increments. 

Now, let us consider $X$ a $G$-valued stochastically continuous process indexed by $\mathbb{R}_{+}$ with $X_{0}=e$. Let us suppose that $X$ has braidable increments. Following the proof of Theorem $1.19$ of \cite{Kall}, we introduce the processes: 
\begin{align*}
Y_{n}^{k}(t) = X\left(\frac{k-1}{n}\right)^{-1}X\left(t+\frac{k-1}{n}\right),  \ \ \ t \in \mathbb{Q} \cap [0,n^{-1}],\ k \in \mathbb{N}^{*},\ n \in \mathbb{N}^{*}
\end{align*}
Let $n$ be a positive integer. Using the same arguments used in Lemma \ref{braidspread}, notice that the sequence $(Y_{n}^{k})_{k \in \mathbb{N}^{*}}$ is spreadable. Applying to these sequences the deFinetti-Ryll-Nardzewski's theorem (Theorem $1.1$ and Corollary $1.6$ in \cite{Kall}) which is valid for sequences in Polish spaces, we conclude that for any $n \in \mathbb{N}^{*}$, conditionally to the tail $\sigma$-field $\mathcal{T}_n$, the sequence $(Y_{n}^{k})_{k \in \mathbb{N}^{*}}$ is a sequence of i.i.d. random variables. We considered $ t \in \mathbb{Q}^{+} \cap [0,n^{-1}]$ in the definition of $Y_n^{k}$ as the product of a countable family of Polish spaces is still a Polish space. 
The $\sigma$-field $\mathcal{T}_n$ we are conditioning on is a.s. independent of $n$: we call it $\mathcal{T}$. Given $\mathcal{T}$, $X$ has conditionally stationary independent (rational) increments. For any $t\in \mathbb{Q}^{+}$, let  $m_t$ be the law of $X_{t}$ conditionally to $\mathcal{T}$: for any $t \in \mathbb{Q}^{+}$ and any $s\in \mathbb{Q}^{+}$, almost surely $m_t *m_s = m_{t+s}$. Besides, using the stochastic continuity of $X$ and the fact that $X_0 = e$, one has that almost surely $(m_{t})_{t \in \mathbb{Q}^{+}}$ is uniformly continuous. We can extend the semi-group $(m_{t})_{t \in \mathbb{Q}^{+}}$ in order to get a semi-group $(m_t)_{t \in \mathbb{R}^{+}}$: by stochastic continuity the process $X$ is then a mixture of L\'{e}vy processes. Let us consider a positive rational number $q$. Using a similar argument as in the proof of Theorem \ref{definettiseq}, applied to the sequence $(X_{nq})_{n \in \mathbb{N}}$, conditionally on $\mathcal{T}$, the random variable $X_{q}$ is invariant by conjugation by its own support. Using a continuity argument allows us to extend the result for any $q \in \mathbb{R}^{+}$. 
The result follows from Proposition \ref{prop-support}:  $X$ is a mixture of self-invariant by conjugation L\'{e}vy process. 

The second part of the theorem is deduced from Theorem \ref{finalsuite} applied to the increments of $X$. 
\end{proof}

\subsection{Degeneracy of the mixture}
\label{deg-proc}
In this subsection, Sections \ref{sec:non-degeneracy} and \ref{sec:general} are generalized in the setting of processes. The proofs will be omitted since the results follow directly from their counterpart in the setting of sequences and from Theorem \ref{definetti}. Recall Definition \ref{defclassification}, where the notions of pure/mixed, degenerate/non-degenerate L\'{e}vy processes were defined. Let us state the consequence of Proposition \ref{unicity1}. 

\begin{proposition}
\label{unicity1-proc}
Let $G$ be a finite group. Let $X$ be a $G$-valued stochastically continuous process invariant by conjugation by $G$ such that $X_0=e$ and which has braidable and $\mathcal{I}$-independent increments. The following assertions are equivalent: 
\begin{enumerate}
\item $X$ is a pure non-degenerate L\'{e}vy process, 
\item there exists $t\in \mathbb{R}^{+}$ such that ${\sf Supp}(X_t) = G$. 
\end{enumerate}
If one of the two conditions holds then for any $t\in \mathbb{R}^{+}$, ${\sf Supp}(X_t) = G$. 
\end{proposition}

Let us remark that, in order to prove the last proposition, we have to replace the property of ${\sf Supp}(m^{*k})$ we used in the proof of Proposition \ref{unicity1} by the following straightforward lemma.

\begin{lemme}
\label{support}
Let $G$ be a finite group and let $\big(Y_t\big)_{t \geq 0}$ be a L\'{e}vy process on $G$. For any real $t \geq 0$, ${\sf Supp}(Y_t)= H_{Y}$. \end{lemme}

Let us state the consequence of Proposition \ref{unicity2}. From now on, let $G$ be a compact Lie group.

\begin{proposition}
\label{unicity2-proc}
 Let $X$ be a $G$-valued stochastically continuous process invariant by conjugation by $G$ such that $X_0=e$ and which has braidable and $\mathcal{I}$-independent increments.  Let us suppose that $e$ is in ${\sf Supp}(X_t)$ for any $t \in \mathbb{R}^{+}$. The following conditions are equivalent: 
\begin{enumerate}
\item the process $X$ is a pure non-degenerate L\'{e}vy process, 
\item the random variable $X_t$ converges in law to a Haar random variable on $G$ when $t$ goes to infinity. 
\end{enumerate}
\end{proposition}

In order to conclude this section, it remains to state the consequence of Theorem~\ref{unicity3}. 

\begin{theorem}
\label{unicity3-proc}
Let $X$ be a $G$-valued stochastically continuous process invariant by conjugation by $G$ such that $X_0=e$ and which has braidable and $\mathcal{I}$-independent increments. The following assertions are equivalent: 
\begin{enumerate}
\item the process $X$ is a pure L\'{e}vy process, 
\item there exists a $G$-valued L\'{e}vy process $Z$ such that for any $t \in \mathbb{R}^{+}$, $X_t$ has the same law as $Z_t$. 
\end{enumerate}
\end{theorem}

\begin{remarque}
If the increments of $X$ are exchangeable then the theorem is no more valid. Let us consider the L\'{e}vy process $Y$ (respectively $Z$) associated with $(\mu_{t})_{t \geq 0}$ (respectively $(\eta_t)_{t \geq 0}$) defined in Lemma \ref{contre-ex}. Let $U$ be a random Haar variable which is independent from $Y$ and $Z$ and let $X = UYU^{-1}$. The process $X$ is stochastically continuous, invariant by conjugation by $G$, has exchangeable and $\mathcal{I}$-independent increments and for any $t \geq 0$, $X_t$ has the same law as $Z_t$. Yet, the process $X$ is not a pure LŽvy process since $Y$ is not invariant by conjugation by $G$.  
\end{remarque}

\chapter[Characterization of Planar Markovian Holonomy Fields]{Characterization of Stochastically Continuous in Law Weak Discrete Planar Markovian Holonomy Fields}
\label{carsec}

The characterization of stochastically continuous in law weak discrete planar Markovian holonomy fields is given by the following theorem whose proof will be the main goal of this section. 

\begin{theorem}
\label{caracterisation}
Let $\left(\mathbb{E}^{\mathbb{G}}_{vol}\right)_{\mathbb{G}, vol}$ be a $G$-valued stochastically continuous in law weak discrete planar Markovian holonomy field. There exists a $G$-valued L\'{e}vy process, $\left(Y_t\right)_{t \geq 0}$, self-invariant by conjugation such that $\left(\mathbb{E}^{\mathbb{G}}_{vol}\right)_{\mathbb{G}, vol}$ is the weak discrete planar Yang-Mills field associated with $\left(Y_t\right)_{t \geq 0}$. This means that for any measure of area $vol$ and any graph $\mathbb{G} \in \mathcal{G}\left({\sf Aff}(\mathbb{R}^{2})\right)$,  $\mathbb{E}^{\mathbb{G}}_{vol}$ is equal to  $\mathbb{E}^{Y,\mathbb{G}}_{vol},$ where $\left(\mathbb{E}^{Y,\mathbb{G}}_{vol}\right)_{\mathbb{G}, vol}$ is the discrete planar Yang-Mills field associated with $\left(Y_t\right)_{t \geq 0}$.

If $G$ is Abelian, the L\'{e}vy process is unique and is characterized by the fact that for any simple loop $l$ in ${\sf Aff}(\mathbb{R}^{2})$, under $\mathbb{E}^{\mathbb{G}(l)}_{vol}$ (see Example \ref{G(l)}), $h(l)$ has the same law as $Y_{vol({\sf Int(}l))}$. 
\end{theorem}

\begin{remarque}
If $G$ is a non Abelian group, the L\'{e}vy process $\left(Y_t \right)_{t \geq 0}$ is not unique: it is unique up to an equivalence. Two L\'{e}vy processes are equivalent if they have the same law when we restrict their law to the invariant $\sigma$-field on $G^{\mathbb{R}^{+}}$. Theorem \ref{caracterisation} asserts that there exists a one-to-one correspondence between the set of equivalence classes of $G$-valued self-invariant by conjugation L\'{e}vy processes and the set of $G$-valued stochastically continuous in law weak discrete planar Markovian holonomy fields.
\end{remarque}

\section[Proof of the characterization theorem]{Proof of Theorem \ref{caracterisation}}

In Section \ref{1}, we show that the two-dimensional time objects which are the planar Markovian holonomy fields are characterized by a one-dimensional time process. In Section \ref{2}, it is shown that this one-dimensional time process has $\mathcal{I}$-independent increments: this allows us to prove Theorem \ref{caracterisation} when the group $G$ is abelian. In general, the result follows from the braidability of the one-dimensional time process which is proved in Section \ref{3}.

\subsection{First correspondence for stochastically continuous in law weak discrete planar Markovian holonomy fields.}
\label{1}
We can go further than Proposition \ref{unic} when one considers stochastically continuous in law weak discrete planar Markovian holonomy fields. 
\begin{definition}
 For any $0 \leq s \leq t$ we define $\partial c_{s}^{t} = (s,0) \to (t,0) \to (t,1) \to (s,1) \to (s,0),$ and $p_{s} = (0,0) \to (s,0)$. These paths can be seen as paths on the same finite planar graph. Recall the notion of reduced loops defined in the beginning of Section \ref{sec:red}. The reduced loop $L_s^{t}$ is: 
 \begin{align*}
 L_{s}^{t} &= [p_{s} \partial c_{s}^{t} p_s^{-1}]_{\simeq},.
 \end{align*} 
 \end{definition}
 
 \begin{remarque}
 \label{levyloop}
These loops satisfy the following equalities: 
\begin{align*}
L_{r}^{t}\ \ \  &= L_{s}^{t}L_{r}^{s},\ \ \  \forall\ 0\leq r\leq s\leq t, \\
L_{i}^{i+1} &= L_{i,0}\ \ ,\ \ \  \forall\  i \in \mathbb{N},
\end{align*}
where we considered the reduced product in the first equality and where the family $(L_{i,j})_{i,j}$ was defined in Definition \ref{base}. 
\end{remarque}

The following lemma is a straightforward application of Theorem \ref{defplanarHFgeneral}. 

\begin{lemme}
\label{casyangmills}
Let $Y$ be a $G$-valued self-invariant by conjugation L\'{e}vy process. Let $U$ be a Haar variable on $G$ which is independent from $Y$. Let $\left(\mathbb{E}^{Y}_{vol}\right)_{vol}$ be the planar Yang-Mills field associated with $Y$.
Under the measure $\mathbb{E}^{Y}_{dx}$, $\left(h(L_{0}^{t})\right)_{t \in \mathbb{R}^{+}}$ has the same law as $\left(UY_{t}U^{-1}\right)_{t \in \mathbb{R}^{+}}$.  
\end{lemme}

The following theorem asserts that the process $\left(h(L_{0}^{t})\right)_{t \in \mathbb{R}^{+}}$ characterizes any stochastically continuous in law weak discrete planar Markovian holonomy field.

\begin{theorem}
\label{caracterisation0}
Let $\left(\mathbb{E}^{\mathbb{G}}_{vol}\right)_{\mathbb{G}, vol}$ and $\left(\tilde{\mathbb{E}}^{\mathbb{G}}_{vol}\right)_{\mathbb{G}, vol}$ be two stochastically continuous in law weak discrete planar Markovian holonomy fields. The three following assertions are equivalent: 
\begin{enumerate}
\item $\left(\mathbb{E}^{\mathbb{G}}_{vol}\right)_{\mathbb{G}, vol}$ and $\left(\tilde{\mathbb{E}}^{\mathbb{G}}_{vol}\right)_{\mathbb{G}, vol}$ are equal, 
\item $\left(h(L_0^{t})\right)_{t \in \mathbb{R}^{+}}$ has the same law under $\mathbb{E}_{dx}^{\sf Aff}$ as under $\tilde{\mathbb{E}}_{dx}^{\sf Aff}$,
\item for any positive real $\alpha$, $\left(h(L_{n,0})\right)_{n \in \mathbb{N}}$ has the same law under $\mathbb{E}^{\mathbb{N}^{2}}_{\alpha dx}$ as under $\tilde{\mathbb{E}}^{\mathbb{N}^{2}}_{\alpha dx}$. 
\end{enumerate}
\end{theorem}
We remind the reader that $\mathbb{E}_{dx}^{\sf Aff}$ and $\mathbb{E}_{dx}^{\mathbb{N}^{2}}$ were defined in Remark \ref{graphaff}. The proof will consist in proving the equivalence between the conditions $1$ and $2$, then between $2$ and~$3$.

\begin{proof}
Since condition $1$ clearly implies condition $2$, let us show that condition $2$ implies condition $1$. Let us suppose that $\left(h(L_0^{t})\right)_{t \in \mathbb{R}^{+}}$ has the same law under $\mathbb{E}_{dx}^{\sf Aff}$ as under $\tilde{\mathbb{E}}_{dx}^{\sf Aff}$. Let $vol$ be a measure of area and let $\mathbb{G}$ be a finite planar graph in $\mathcal{G}\left({\sf Aff}\left(\mathbb{R}^2\right)\right)$. We have to show that $\mathbb{E}^{\mathbb{G}}_{vol} = \tilde{\mathbb{E}}^{\mathbb{G}}_{vol}$. The proof will consist in a sequence of simplifications, changing the graph and the measure of area little by little. By Proposition \ref{unicite1}, these measures are characterized by the way they integrate functions of the form: $h \mapsto f(h(l_1), ..., h(l_m))$, where $f$ is a continuous function invariant by conjugation and $l_1$, ..., $l_n$ are elements of $L_v(\mathbb{G})$, where $v$ is any given vertex of $\mathbb{G}$. Thus, we have to show that $\left(\mathbb{E}^{\mathbb{G}}_{vol}\right)_{\mid L_{v}(\mathbb{G})} = \left(\tilde{\mathbb{E}}^{\mathbb{G}}_{vol}\right)_{\mid L_{v}(\mathbb{G})}.$ Since these two measures are defined on multiplicative functions on loops, we can suppose that $\mathbb{G}$ is simple. Let us consider a sequence of generic and simple graphs $(\mathbb{G}_{n})_{n\geq 0}$ which approximate the graph $\mathbb{G}$ in the sense of  Lemma \ref{gene}. Using the stochastic continuity in law of $\left(\mathbb{E}^{\mathbb{G}}_{vol}\right)_{\mathbb{G}, vol}$ \big(resp. $\left(\tilde{\mathbb{E}}^{\mathbb{G}}_{vol}\right)_{\mathbb{G}, vol}$\big), the measure $\left(\mathbb{E}^{\mathbb{G}}_{vol}\right)_{\mid L_{v}(\mathbb{G})}$ \big(resp. $\left(\tilde{\mathbb{E}}^{\mathbb{G}}_{vol}\right)_{\mid L_{v}(\mathbb{G})}$\big) can be recovered using the sequence of measures $\left(\left(\mathbb{E}^{\mathbb{G}_n}_{vol}\right)_{\mid L_{v}(\mathbb{G}_n)}\right)_{n\in \mathbb{N}}$ \big(resp.  $\left(\left(\tilde{\mathbb{E}}^{\mathbb{G}_n}_{vol}\right)_{\mid L_{v}(\mathbb{G}_n)}\right)_{n\in \mathbb{N}}$ \big) . This allows us to suppose that $\mathbb{G}$ is also generic. 

Using Corollary \ref{injectiondansn}, there exists $\mathbb{G}'$ a subgraph of the $\mathbb{N}^{2}$ graph such that the set of $\mathbb{G}-\mathbb{G}'$ piecewise diffeomorphisms is not empty: let $\psi$ be such a homeomorphism. Let $vol'$ be a measure of area on $\mathbb{R}^{2}$ such that for any bounded face $F$ of $\mathbb{G}$, $vol'(\psi(F)) = vol(F)$. Using the Axiom $\mathbf{wDP_{1}}$, we know that $\mathbb{E}^{\mathbb{G}'}_{vol'}\circ \psi^{-1} = \mathbb{E}^{\mathbb{G}}_{vol}. $ By definition of $\mathbb{E}^{\mathbb{N}^{2}}_{vol'}$, the measure $\mathbb{E}^{\mathbb{G}'}_{vol'}$ is equal to $\left(\mathbb{E}^{\mathbb{N}^{2}}_{vol'}\right)_{\mid \mathcal{M}ult(P(\mathbb{G}'), G)}$. The same discussion holds for $\left(\tilde{\mathbb{E}}^{\mathbb{G}}_{vol}\right)_{\mathbb{G}, vol}$. Thus, if we show that, for any measure of area $vol'$, $\mathbb{E}^{\mathbb{N}^{2}}_{vol'} = \tilde{\mathbb{E}}^{\mathbb{N}^{2}}_{vol'},$ we will get that $\left(\mathbb{E}^{\mathbb{G}}_{vol}\right)_{\mathbb{G}, vol}$ and $\left(\tilde{\mathbb{E}}^{\mathbb{G}}_{vol}\right)_{\mathbb{G}, vol}$ are equal.

Let $vol'$ be a measure of area on $\mathbb{R}^{2}$. Since $\{L_{i,j}, i,j \in \mathbb{N}^2\}$ is a family which generates $RL_{0}(\mathbb{N}^{2})$ and since we are considering gauge-invariant measures, we only have to prove that $(h(L_{i,j}))_{i,j}$ has the same law under $\mathbb{E}^{\mathbb{N}^{2}}_{vol'}$ as under $\tilde{\mathbb{E}}^{\mathbb{N}^{2}}_{vol'}$. Let us show that it is actually enough to know that $(h(L_{n,0}))_{n\in \mathbb{N}}$ has the same law under $\mathbb{E}^{\mathbb{N}^{2}}_{vol'}$ as under $\tilde{\mathbb{E}}^{\mathbb{N}^{2}}_{vol'}$. 

Let us consider the two finite planar graphs $\mathbb{G}$ and $\tilde{\mathbb{G}}$ drawn in Figure \ref{Graphes1}. Let $vol''$ be a measure of area such that $vol''_{\mid F_{\infty}} = vol'_{\mid \tilde{F}_{\infty}}$, where $F_{\infty}$ (resp. $\tilde{F}_{\infty}$) is the unbounded face of $\mathbb{G}$ (resp.~$\tilde{\mathbb{G}}$). Besides, we impose that the following condition holds for $vol''$: 
\begin{align}
\label{volegal}
\forall i \in \{1, ..., 5\}, vol''(F'_i)=vol'(F_{i}). 
\end{align} 

\begin{figure}
 \centering
  \includegraphics[width=355pt]{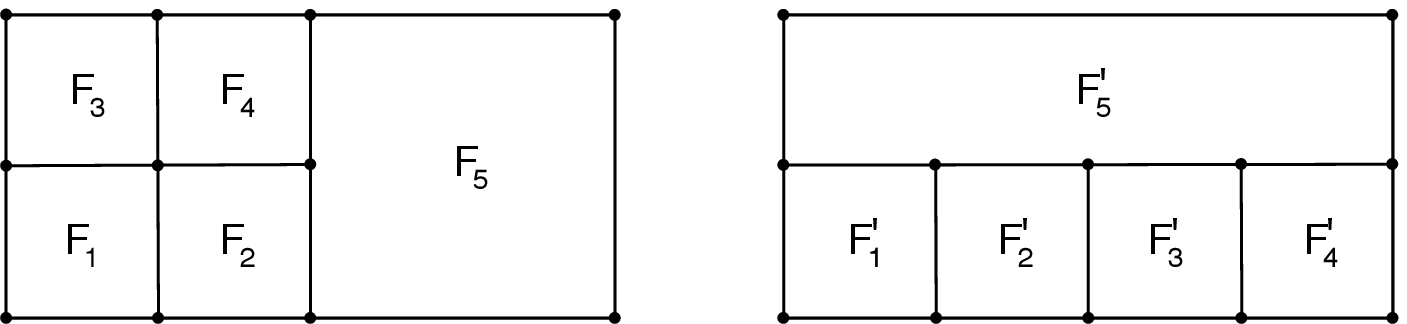}
 \caption{Graphs $\mathbb{G}$ and $\tilde{\mathbb{G}}$.}
 \label{Graphes1}
\end{figure}
\begin{figure}
 \centering
  \includegraphics[width=355pt]{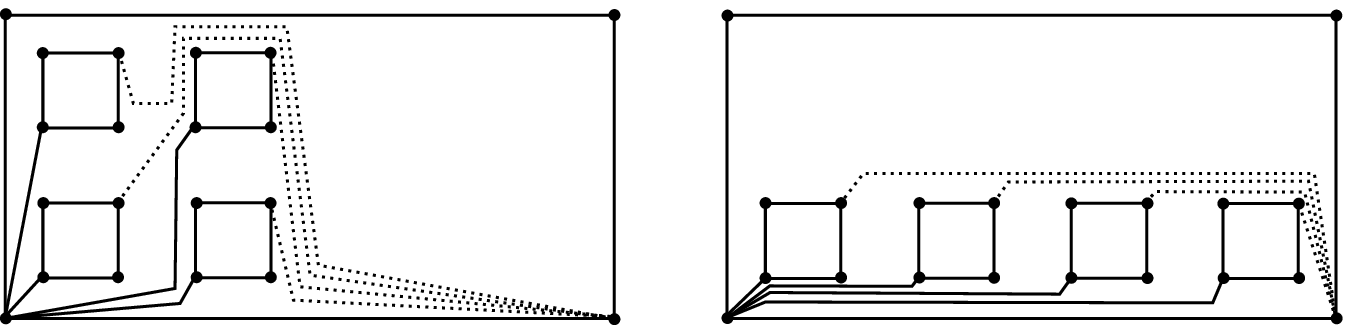}
 \caption{Graphs $\mathbb{G}_1$ and $\mathbb{G}_2$.}
 \label{Graphes2}
\end{figure}

The loops $(L_{i,j})_{i,j \in \{0,1\}}$ belong to $L_{(0,0)} (\mathbb{G})$ and the loops $(L_{i,0})_{i \in [\mid 0,3\mid]}$ are in $L_{(0,0)} (\tilde{\mathbb{G}})$. Let us approximate the loops $\left(L_{i,j}\right)_{i,j\in \{0,1\}}$ by loops whose intersection is reduced to the base point. Such loops are drawn in bold in the left part of Figure \ref{Graphes2}. The two graphs $\mathbb{G}_1$ and $\mathbb{G}_2$ drawn in Figure \ref{Graphes2} satisfy the hypothesis of Theorem \ref{homeograph} and are in $\mathcal{G}\left({\sf Aff}(\mathbb{R}^{2})\right)$: they are homeomorphic. Thus, by Proposition \ref{triang}, there exists an orientation-preserving $\mathbb{G}_1-\mathbb{G}_2$ piecewise diffeomorphism which we denote by $\psi$. We can suppose, up to a modification of $\mathbb{G}_1$ and $\mathbb{G}_2$ which will not change the general form of both graphs and thus will not invalidate the discussion, that $vol'(F) = vol''(\psi(F))$ for any bounded face $F$ of $\mathbb{G}_1$. This last assertion is essentially due to the condition (\ref{volegal}) on $vol''$. Using Axiom $\mathbf{wDP_1}$ and using the stochastic continuity in law property, we conclude that under $\mathbb{E}^{\mathbb{N}^{2}}_{vol'}$ (resp. $\tilde{\mathbb{E}}^{\mathbb{N}^{2}}_{vol'}$), $(h(L_{0,1}),h(L_{0,0}), h(L_{1,1}), h(L_{1,0}))$ has the same law as $(h(L_{i,0}))_{i\in [\mid 0, 3\mid]}$ under $\mathbb{E}^{\mathbb{N}^{2}}_{vol''}$ (resp. $\tilde{\mathbb{E}}^{\mathbb{N}^{2}}_{vol''}$). A slight generalization of these arguments allows us to show that for any integer $n$ there exists a measure of area $vol''$ such that under $\mathbb{E}^{\mathbb{N}^{2}}_{vol}$ (resp. $\tilde{\mathbb{E}}^{\mathbb{N}^{2}}_{vol}$), $(h(L_{0,n-1}), ..., h(L_{0,0}), ..., h(L_{n-1,n-1}), ..., h(L_{n-1,0}))$ has the same law as $(h(L_{i,0}))_{i\in [\mid 0, n^{2}-1\mid]}$ under $\mathbb{E}^{\mathbb{N}^{2}}_{vol''}$ (resp. $\tilde{\mathbb{E}}^{\mathbb{N}^{2}}_{vol''}$). 

Thus it is now enough to show that for any measure of area $vol''$,  $(h(L_{i,0}))_{i\in \mathbb{N}}$ has the same law under $\mathbb{E}^{\mathbb{N}^{2}}_{vol''}$ as under $\tilde{\mathbb{E}}^{\mathbb{N}^{2}}_{vol''}$. Let $n \in \mathbb{N}$, let $\mathbb{S}$ be the graph defined as the intersection  of the $\mathbb{N}^{2}$ graph and $[0,n+1]\times [0,1]$. Let us consider an orientation-preserving homeomorphism $\psi$ such that its restriction on $\mathbb{R}^{+}\times \mathbb{R}$ is given by: 
\begin{align*}
\psi: \ \ \ \ \mathbb{R}^{+}\times \mathbb{R} &\to \mathbb{R}^{2}\\
(x,y) &\mapsto \left(vol''\left([0,x] \times [0,1]\right) , y\right). 
\end{align*}
The image of $\mathbb{S}$ by $\psi$, $\psi(\mathbb{S})$ is a simple graph in $\mathcal{G}({\sf Aff}(\mathbb{R}^{2}))$ and for any bounded face $F$ of $\mathbb{S}$, $vol''(F) = dx(\psi(F))$. Let us define for any $i \in \{0,...,n+1\}$, $t_i = vol''(L_{0}^{i})$. We can apply the Axiom $\mathbf{wDP_1}$ to the two graphs $\mathbb{S}$ and $\psi(\mathbb{S})$, to the two measures of area $vol''$ and $dx$ and to $\psi$. It shows that under $\mathbb{E}^{\mathbb{N}^{2}}_{vol''}$ (resp. $\tilde{\mathbb{E}}^{\mathbb{N}^{2}}_{vol''}$), $(h(L_{i,0}))_{i=0}^{n}$ has the same law as $(h(L_{t_i}^{t_{i+1}}))_{i=0}^{n}$ under $\mathbb{E}_{dx}^{\sf Aff}$ (resp. $\tilde{\mathbb{E}}_{dx}^{\sf Aff}$). It remains to show that for any integer $n$, any sequence of positive reals $t_0< ...< t_n$, $(h(L_{t_i}^{t_{i+1}}))_{i=0}^{n}$ has the same law under $\mathbb{E}_{dx}^{\sf Aff}$ as under $\tilde{\mathbb{E}}_{dx}^{\sf Aff}$. Yet, we started with the fact that $(h(L_0^{t}))_{t \in \mathbb{R}^{+}}$ has the same law under $\mathbb{E}_{dx}^{\sf Aff}$ as under $\tilde{\mathbb{E}}_{dx}^{\sf Aff}$. This allows us to conclude. 

Let us prove the equivalence between the conditions $2.$ and $3.$ Suppose that for any positive real $\alpha$, $(h(L_{n,0}))_{n \in \mathbb{N}}$ has the same law under $\mathbb{E}^{\mathbb{N}^{2}}_{\alpha dx}$ as under $\tilde{\mathbb{E}}^{\mathbb{N}^{2}}_{\alpha dx}$. By the Axiom $\mathbf{wDP_4}$, we can change $\mathbb{E}^{\mathbb{N}^{2}}_{\alpha dx}$ (resp. $\tilde{\mathbb{E}}^{\mathbb{N}^{2}}_{\alpha dx}$) by $\mathbb{E}^{{\sf Aff}}_{\alpha dx}$ (resp. $\tilde{\mathbb{E}}^{{\sf Aff}}_{\alpha dx}$). As an application of the Axiom $\mathbf{wDP_1}$, with $\psi$ given by $\psi: (x,y) \mapsto (\alpha x, y),$ the random vector $\left(h(L_{\alpha.n}^{\alpha(n+1)})\right)_{n \in \mathbb{N}}$ has the same law under $\mathbb{E}_{dx}^{\sf Aff}$ as under $\tilde{\mathbb{E}}_{dx}^{\sf Aff}$. Using the fact that for any positive integers $p$ and $q$: $h\bigg(L_{0}^{\frac{p}{q}}\bigg) = \prod\limits_{i=0}^{p-1} h\left(L_{\frac{i}{q}}^{\frac{i+1}{q}}\right),$ we can conclude that $(h(L_t))_{t \in \mathbb{Q}^{+}}$ has the same law under $\mathbb{E}_{dx}^{\sf Aff}$ as under $\tilde{\mathbb{E}}_{dx}^{\sf Aff}$ and by stochastically continuity the same assertion holds for $(h(L_t))_{t \in \mathbb{R}^{+}}$. 
Thus, condition $3.$ implies condition $2$. The other implication can be proved using the same arguments.\end{proof}

\subsection{Law of the conjugacy classes and the Abelian case}
\label{2}
Let $\left(\mathbb{E}^{\mathbb{G}}_{vol}\right)_{\mathbb{G},vol}$ be a stochastically continuous in law weak discrete planar Markovian holonomy field and $\mathbb{E}_{dx}^{\sf Aff}$ be the usual expectation associated with it. 

\begin{definition}
\label{process}
Until the end of this section, we set, for any $0 \leq s \leq t$, $Z_s^{t} = h(L_s^{t})$ and $Z_t = Z_0^{t}.$ 
\end{definition}

\begin{remarque}
\label{debutlevy}
Using the multiplicativity property of random holonomy fields and Remark \ref{levyloop}, for any $0 \leq r \leq s \leq t $, $Z_r^{t} =  Z_r^{s}Z_s^{t}$, hence for any $0\leq s\leq t$, $$Z_{s}^{t}=(Z_s)^{-1} Z_t.$$
Since $Z_s^{t} = h\!\left(p_s \partial c_s^{t}p_s^{-1}\right)$, by Remark \ref{changementdepoint}, under $\mathbb{E}_{dx}^{\sf Aff}$, $Z_s^{t}$ has the same law as $h(\partial c_s^{t})$. Besides the left translation by $s$ sends $dx$ on itself and $\partial c_s^{t}$ on $\partial c_0^{t-s}$: applying Axiom $\mathbf{wDP_1}$ (Definition \ref{discreteplanarHF}), we get that under $\mathbb{E}_{dx}^{\sf Aff}$, 
\begin{align*}
Z_s^{t} \text{ has the same law as } Z_{t-s}.
\end{align*}
Moreover, using the stochastic continuity property, under $\mathbb{E}_{dx}^{\sf Aff}$, the process $(Z_t)_{t\geq0}$ is stochastically continuous and $Z_0$ is equal to the neutral element of $G$. 
\end{remarque}

A simple but important lemma is the following.

\begin{lemme}
 \label{conjugclass}
Under $\mathbb{E}_{dx}^{\sf Aff}$, for any $t_0>0$, for any finite subset $T$ of $[0,t_0]$ and any finite subset $T'$ of $[t_0,\infty[$, $(Z_{t})_{t \in T}$ and $(Z_{t_0}^{-1}Z_{t})_{t \in T'}$ are $\mathcal{I}$-independent. This means that for any continuous functions $f: G^{T} \to \mathbb{R}$ and $f': G^{T'} \to \mathbb{R}$ invariant by diagonal conjugation by $G$, 
\begin{align*}
\mathbb{E}_{dx}^{\sf Aff}\! \left[f\left(\left(Z_t\right)_{t \in T}\right)\! f'\left(\left(Z_{t_0}^{-1}Z_t\right)_{t \in T'} \right)\!\right]\! \!=\! \mathbb{E}_{dx}^{\sf Aff}\! \left[f\left(\left(Z_t\right)_{t \in T}\right)\right] \mathbb{E}_{dx}^{\sf Aff}\! \left[f'\left(\left(Z_{t_0}^{-1}Z_t\right)_{t \in T'}\right)\right]\!\!. 
\end{align*} 
\end{lemme}

\begin{proof}
Let $t_0>0$, $T$ be a finite subset of $[0,t_0]$ and $T'$ be a finite subset of $[t_0,\infty[$. Obviously we can suppose that $T' \subset ]t_0, \infty[$. Let $t_0'$ be any real strictly greater that $t_0$ such that $T' \subset [t_0', \infty[$. We remind the reader that for any $t \in T'$, 
$$Z_{t_{0}'}^{-1} Z_t  = h\big(p_{t_0'} \partial c_{t_0'}^{t} p_{t_0'}^{-1}\big),$$
thus for any continuous functions $f: G^{T} \to \mathbb{R}$ and $f': G^{T'} \to \mathbb{R}$ invariant by diagonal conjugation by $G$: 

$$\mathbb{E}_{dx}^{\sf Aff} \Big[f\big(\left(Z_t\right)_{t \in T}\big) f'\left(\left(Z_{t_0'}^{-1}Z_t\right)_{t \in T'} \right)\Big] = 
\mathbb{E}_{dx}^{\sf Aff} \Big[f\big( (h(L_0^{t}))_{t \in T}\big) f' \big((h(\partial c_{t_0'}^{t}))_{t \in T'} \big)\Big]. $$
Let us denote by $t_1$ the maximum of $T'$. The loop $L_{0}^{t}$ is in $\overline{{\sf Int(}L_0^{t_0})}$ for any $t\in T$ and the loop $\partial c_{t_0'}^{t}$ is in $\overline{{\sf Int(}\partial c_{t_0'}^{t_1})}$ for any $t \in T'$. Besides $\overline{{\sf Int(}L_{0}^{t_0})} \cap \overline{{\sf Int(}\partial c_{t'_0}^{t_1})} = \emptyset$. Using the Axiom $\mathbf{wDP_2}$: 
\begin{align*}
\mathbb{E}_{dx}^{\sf Aff}\! \left[f\left(\left(Z_t\right)_{t \in T}\right)\! f'\left(\left(Z_{t_0'}^{-1}Z_t\right)_{t \in T'} \right)\!\right]\!\! =\! \mathbb{E}_{dx}^{\sf Aff}\! \left[f\left(\left(Z_t\right)_{t \in T}\right)\right] \mathbb{E}_{dx}^{\sf Aff}\! \left[f'\left(\left(Z_{t_0'}^{-1}Z_t\right)_{t \in T'}\right)\right]\!\!. 
\end{align*} 
The stochastic continuity of $\mathbb{E}_{dx}^{\sf Aff}$ and taking the limit $t_0' \to t_0$ allows us to conclude the proof. 
\end{proof}

When $G$ is Abelian, for any $n$-tuple $(g_1, ..., g_n)$ of elements of $G$, the diagonal conjugacy class of $(g_1,..., g_n)$ is reduced to $\big\{(g_1, ..., g_n)\big\}$. Thus, the last lemma asserts that $\sigma \big(\{Z_t, t \leq t_0\} \big)$ is independent of $\sigma\big(\{Z_{t_0}^{-1}Z_t,t\geq t_0\}\big)$. Using Remark \ref{debutlevy}, this implies that $(Z_t)_{t \in \mathbb{R}^{+}}$ is a L\'{e}vy process. Applying Theorem \ref{caracterisation0} and Lemma \ref{casyangmills}, we deduce that $\big(\mathbb{E}^{\mathbb{G}}_{vol}\big)_{\mathbb{G},vol}$ is the planar Yang-Mills field associated with the L\'{e}vy process $(Z_t)_{t\geq0}$. The Abelian part of Theorem \ref{caracterisation} is thus proved.  

When $G$ is not Abelian, we have to get rid of the conjugacy classes in Lemma \ref{conjugclass}: it is what we intend to do in the following subsection. 

\subsection{Braidability and the non-Abelian case}
\label{3}

Recall the notions and notations set in Definition \ref{increment}.

\begin{proposition}
\label{braidable} 
Under $\mathbb{E}_{dx}^{\sf Aff}$, the process $(Z_{t})_{t \in \mathbb{R}^{+} }$ has braidable increments. 
\end{proposition}

\begin{proof}
The proof will be essentially graphical. The braid group with $m$ strands is generated by the elementary braids $(\beta_i)_{i=1}^{m-1}$ defined in Section \ref{Braids}. This allows us to reduce the braidability condition to the fact that for any $n \in \mathbb{N}^{*}\cup (\mathbb{N}^{*})^{-1} $, any positive integers $m$ and $i$ such that $i<m$, the following equality in law holds: 
\begin{align*}
\beta_i \bullet \big(Z_{n,1},...,Z_{n,m}\big) = \big(Z_{n,1},...,Z_{n,m}\big).
\end{align*}
The proof does not depend on the value of $n$, we will suppose it is equal to $1$. We remind the reader that $Z_{1,i} = h\left(p_{i-1} \partial c_{i-1}^{i} p_{i-1}^{-1}\right)$: we have to understand the law of the random variables associated with $m$ lassos. Using the stochastic continuity of $\mathbb{E}_{dx}^{\sf Aff}$, we can ``shrink" the meander of these lassos and we can suppose that their intersection is reduced to the base point as we did in the proof of Theorem \ref{caracterisation0}. Let $i$ be a positive integer, we will focus only on what happens in the interior of $\partial c_{i-1}^{i+1}$. Let us consider the graphs $\mathbb{G}_1$ and $\mathbb{G}_2$ drawn in Figure \ref{Graphes3}. They represent what is happening in $\partial c_{i-1}^{i+1}$: the loops in bold represent the part of the $i^{th}$ and $i+1^{th}$ lassos inside $\partial c_{i-1}^{i+1}$ and we added to it two paths in dots in order to consider simple graphs. The two graphs satisfy the hypothesis of Theorem \ref{homeograph} thus, they are homeomorphic. Let us consider an orientation-preserving homeomorphism $\phi$ between $\mathbb{G}_1$ and $\mathbb{G}_2$ which sends $F_i$ on $F'_i$ for any $i \in \{1, ..., 5\}$. By Proposition \ref{triang}, there exists an orientation-preserving $\mathbb{G}_1-\mathbb{G}_2$ piecewise diffeomorphism $\psi$ which is equivalent to $\phi$ on $\mathbb{G}_1$: it sends $F_i$ on $F'_i$ for any $i \in \{1, ..., 5\}$. It is possible to take it such that  $\psi$ is the identity on the unbounded face of $\mathbb{G}_1$. Besides, one can remark that $\mathbb{G}_2$ is the horizontal flip of $\mathbb{G}_1$: for any integer $i \in \{1,...,5\}$, $dx(F_i) = dx(F'_i)$. Thus for any bounded face $F$ of $\mathbb{G}_1$, $dx(\psi(F))=dx(F)$. Using the area-preserving homeomorphism invariance, namely Axiom $\mathbf{wDP_1}$, $\mathbb{E}^{\mathbb{G}_2}_{dx} \circ \psi^{-1} = \mathbb{E}^{\mathbb{G}_1}_{dx}$. Letting the shrinking parameter to zero in this equality allows us to recover the following equality in law: under $\mathbb{E}_{dx}^{\sf Aff}$, $\beta_i \bullet \big(Z_{1,1},...,Z_{1,m}\big) = \big(Z_{1,1},...,Z_{1,m}\big)$. 
\end{proof}

\begin{figure}
 \centering
 \includegraphics[width=355pt]{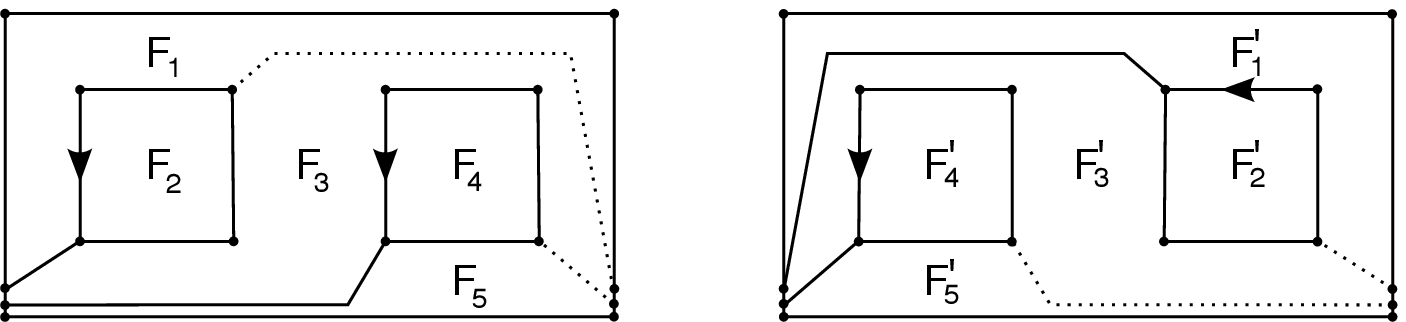}
 \caption{The graphs $\mathbb{G}_1$ and $\mathbb{G}_2$.}
 \label{Graphes3}
\end{figure}

Recall that we are working under $\mathbb{E}_{dx}^{\sf Aff}$. Using the results of Section \ref{2} and \ref{3}, we already know that the process $Z=(Z_t)_{t \in \mathbb{R}^{+}}$ is invariant by conjugation by $G$ and has braidable and $\mathcal{I}$-independent increments. By Theorem \ref{definetti}, there exists a self-invariant by conjugation L\'{e}vy process $Y$, such that the law of $Z$ is $UYU^{-1}$, where $U$ is a Haar variable on $G$ independent of $Y$. Lemma \ref{casyangmills}, combined with Theorem \ref{caracterisation0} allows us to finish the proof of Theorem \ref{caracterisation}.

\section[Consequences of the characterization theorem]{Consequences of Theorem \ref{caracterisation}}

\begin{theorem}
For a discrete planar Markovian holonomy field, the following conditions are equivalent: 
\begin{itemize}
\item it is stochastically continuous in law, 
\item it is regular. 
\end{itemize} 
If the discrete planar Markovian holonomy field is a weak one, then one can replace the regularity condition by the locally stochastically $\frac{1}{2}$-H\"{o}lder continuity.
\end{theorem}

\begin{proof}
We already saw in Corollary \ref{coroll:regularity} that, depending if we are working with weak or strong discrete planar Markovian holonomy fields, the regularity or the locally stochastically $\frac{1}{2}$-H\"{o}lder continuity implies the stochastically continuity in law of the discrete planar Markovian holonomy field. 

Besides if a discrete planar Markovian holonomy field $\left(\mathbb{E}^{\mathbb{G}}_{vol}\right)_{\mathbb{G}, vol}$ is stochastically continuous in law, its restriction to the piecewise affine graphs is a stochastically continuous in law weak discrete planar Markovian holonomy field. By Theorem \ref{caracterisation}, there exists a planar Yang-Mills field $\left(\mathbb{E}^{Y}_{vol}\right)_{vol}$ such that $$\left(\mathbb{E}^{\mathbb{G}}_{vol}\right)_{\mathbb{G} \in \mathcal{G}({\sf Aff}(\mathbb{R}^{2})), vol} = \left((\mathbb{E}^{Y}_{vol})_{|\mathcal{M}ult(P(\mathbb{G}), G)}\right)_{\mathbb{G} \in \mathcal{G}({\sf Aff}(\mathbb{R}^{2})), vol}.$$ Using the stochastic continuity in law of both of the fields, this equality holds without the restriction on the graphs. Using the proof of Proposition \ref{creaYMpure}, up to a slight modification since $Y$ is only self-invariant by conjugation, $\left((\mathbb{E}^{Y}_{vol})_{|\mathcal{M}ult(P(\mathbb{G}), G)}\right)_{\mathbb{G}, vol}$ is locally stochastically $\frac{1}{2}$-H\"{o}lder continuous and continuously area-dependent, thus $\left(\mathbb{E}^{\mathbb{G}}_{vol}\right)_{\mathbb{G}, vol}$ is also locally stochastically $\frac{1}{2}$-H\"{o}lder continuous and continuously area-dependent. 
\end{proof}

In Section \ref{planarmarkovianholonomyfields} we have defined four different notions of planar Markovian holonomy fields. By now, we know that, by restriction, a strong planar Markovian holonomy field defines a weak continuous one. Using the results of Section \ref{weakconstructibility}, a weak planar Markovian holonomy field defines, when restricted, a weak discrete planar Markovian holonomy field. Theorem \ref{caracterisation} now allows us to show that the four different notions are in some sense equivalent when one considers stochastically continuous objects. Indeed, a stochastically continuous in law weak discrete planar Markovian holonomy field is the restriction of a planar Yang-Mills field, which by the results of Section \ref{planaryangmills} was shown to be a stochastically continuous strong planar Markovian holonomy field. Besides, by construction, any planar Yang-Mills field is constructible. This discussion allows us to state the following theorems.

\begin{theorem}
\label{maintheoremdecaract}
Let $\left(\mathbb{E}_{vol}\right)_{vol}$ be a family of stochastically continuous random holonomy fields. We have equivalence between: 
\begin{enumerate}
\item $\left(\mathbb{E}_{vol}\right)_{vol}$ is a strong planar Markovian holonomy field, 
\item $\left(\left(\mathbb{E}_{vol}\right)_{| \mathcal{M}ult({\sf Aff}(\mathbb{R}^{2}), G)}\right)_{vol}$ is a weak planar Markovian holonomy field,
\item $\left(\left(\mathbb{E}_{vol}\right)_{| \mathcal{M}ult(P(\mathbb{G}), G)}\right)_{\mathbb{G}, vol}$ is a strong discrete planar Markovian holonomy field, 
\item $\left(\left(\mathbb{E}_{vol}\right)_{| \mathcal{M}ult(P(\mathbb{G}), G)}\right)_{\mathbb{G} \in \mathcal{G}({\sf Aff}(\mathbb{R}^{2})), vol}$ is a weak discrete planar Markovian holonomy field, 
\item $\left(\mathbb{E}_{vol}\right)_{vol}$ is a planar Yang-Mills field associated with a L\'{e}vy process which is self-invariant by conjugation. 
\end{enumerate}
Thus, any stochastically continuous strong planar Markovian holonomy field is constructible. 
\end{theorem}

\begin{theorem}
\label{strongconst}
Any $G$-valued stochastically continuous in law weak discrete planar Markovian holonomy field is the restriction of a unique $G$-valued stochastically continuous strong planar Markovian holonomy field. 
\end{theorem}

We encourage the reader to have a look at the diagram page \pageref{diagram} where we drawn the different links between all the notions introduced or used in this paper. The last consequence of Theorem \ref{caracterisation} is the Proposition \ref{restrictiondiscret}. Before giving the proof of this proposition, we will construct an explicit $G$-valued stochastically continuous in law discrete planar Markovian holonomy field $(\mathbb{E}_{vol}^{\mathbb{G}})_{\mathbb{G}, vol}$ which satisfies the hypothesis of Proposition \ref{restrictiondiscret} but for which the natural restriction defined by the Equation (\ref{restrict}) is not a discrete planar Markovian holonomy field. 

\label{ex:restriction}For this, we consider the symmetrical group $G = \mathfrak{S}_3$. Let $H$ be the subgroup of $G$ isomorphic to $\mathbb{Z}/3\mathbb{Z}$ which contains the neutral element $e$ and the two $3$-cycles $c$ and $c^{2}$. Let $X$ be a $H$-valued L\'{e}vy process which jumps only by multiplication by $c$. As $H$ is abelian, $X$ is a self-invariant by conjugation $G$-valued L\'{e}vy process and, because of the condition on the jumps, for any positive time $t$: 
\begin{align}
\label{ineg}
\mathbb{P}[X_t=c] \ne \mathbb{P}\left[X_t=c^{2}\right]. 
\end{align}
Let $\left(\mathbb{E}^{X}_{vol}\right)_{vol}$ be the $G$-valued planar Yang-Mills field associated with $X$ and let us consider the restriction $\big(\mathbb{E}^{X,\mathbb{G}}_{vol}\big)_{\mathbb{G}, vol}$ of $\left(\mathbb{E}^{X}_{vol}\right)_{vol}$ to the finite planar graphs: it is a $G$-valued stochastically continuous in law discrete planar Markovian holonomy field. Since $H$ is normal in $G$, $\big(\mathbb{E}^{X,\mathbb{G}}_{vol}\big)_{\mathbb{G}, vol}$ satisfies the conditions stated in Proposition \ref{restrictiondiscret}. The natural restriction of $\big(\mathbb{E}^{X,\mathbb{G}}_{vol}\big)_{\mathbb{G}, vol}$, as defined by the Equation (\ref{restrict}), and denoted by $\big(\tilde{\mathbb{E}}^{X,\mathbb{G}}_{vol}\big)_{\mathbb{G}, vol}$, is neither a strong nor a weak $H$-valued discrete planar Markovian holonomy field since it does not satisfy the weak independence property. Indeed, let us consider two loops $l$, $l'$ and a path $p$ as drawn in Figure \ref{restrictionfig} and let us suppose that they are drawn on a finite planar graph $\mathbb{G}$. Let us suppose that $vol({\sf Int}(l)) = vol({\sf Int}(l')) = 1$. If $\big(\tilde{\mathbb{E}}^{X,\mathbb{G}}_{vol}\big)_{\mathbb{G}, vol}$ satisfied the weak independence property, under $\tilde{\mathbb{E}}^{X,\mathbb{G}}_{vol}$, $h(l)$ and $h(l')$ would be independent. Recall Equation (\ref{eq:egalitenoninde}), where $f_1$ and $f_2$ are, in this case, any functions on $H$ since $H$ is abelian. Under $\mathbb{E}^{X,\mathbb{G}}_{vol}$, the random couple $(h(l), h(pl'p^{-1}))$ has the same law as $(UAU^{-1}, UBU^{-1})$ where $A$, $B$ and $U$ are three independent random variables such that $U$ is a Haar variable on $G$ and the two random variables $A$ and $B$ have the same law as $X_1$. The two random variables $UAU^{-1}$ and $UBU^{-1}$ are not independent since the Equality (\ref{ineg}) implies that: $$\mathbb{P}\big[ (UAU^{-1}, UBU^{-1}) = (c,c)\big]\ne  \mathbb{P}[UAU^{-1} = c] \mathbb{P}[UBU^{-1}=c].$$ This proves that under  $\tilde{\mathbb{E}}^{X,\mathbb{G}}_{vol}$, $h(l)$ and $h(l')$ are not independent. 

Using Theorem \ref{caracterisation}, we can now prove Proposition \ref{restrictiondiscret}.

\label{restrictiondiscretpage}
\begin{proof}[Proof of Proposition \ref{restrictiondiscret}]
As a consequence of Theorem \ref{caracterisation}, $\big(\mathbb{E}_{vol}^{\mathbb{G}}\big)_{\mathbb{G}, vol}$ is a discrete planar Yang-Mills field associated with a self-invariant by conjugation L\'{e}vy process $(Y_t)_{t \geq 0}$. Besides, the law of a simple loop $l$ under $\mathbb{E}^{\mathbb{G}(l)}_{vol}$ is equal to $\int_{G}m_{vol({\sf Int(}l))}^{g} dg$ where $m_{vol({\sf Int(}l))}$ is the law of $Y_{vol({\sf Int(}l))}$. Hence, under the hypothesis of Proposition \ref{restrictiondiscret}, for any positive real $t$, $Y_t$ is almost surely in $H$: we can find a modification of $(Y_t)_{t \geq 0}$ which is $H$-valued. Using Theorem \ref{defplanarHFgeneral}, we can define a $H$-valued planar Yang-Mills field, associated with the $H$-valued L\'{e}vy process $(Y_t)_{t \geq 0}$, whose restriction to planar graphs is a $H$-valued stochastically continuous in law discrete planar Markovian holonomy field which satisfies the required conditions. 
\end{proof}

%
%
%

\chapter[Classification of Planar Markovian Holonomy Fields]{Classification of Stochastically Continuous Strong Planar Markovian Holonomy Fields}
\label{Charac}

Let $\big(\mathbb{E}_{vol}\big)_{vol}$ be a stochastically continuous strong planar Markovian holonomy field: it is a planar Yang-Mills field to which is associated a L\'{e}vy process $Y=(Y_t)_{t \geq 0}$. In Definition \ref{defclassification}, the notions of pure non-degenerate/pure degenerate/mixed degenerate planar Yang-Mills field were defined according to the degree of symmetry and the support of $Y$.  In this section, we will see equivalent conditions for $\big(\mathbb{E}_{vol}\big)_{vol}$  to be in each of these categories. The theorems explained below are straightforward applications of Theorem \ref{caracterisation} and Section \ref{deg-proc}. Indeed, by definition, $\big(\mathbb{E}_{vol}\big)_{vol}$ is a pure non-degenerate (resp. pure) planar Yang-Mills field if and only if $Y$ is pure non-degenerate (resp. pure). Applying Proposition \ref{unicity1-proc} (resp. Proposition \ref{unicity2-proc}, resp. Theorem \ref{unicity3-proc}) to the process $Z_t$ defined in Definition \ref{process} allows us to prove Theorem \ref{unicity1-f} (resp. Theorem \ref{unicity2-f}, resp. Theorem \ref{unicity3-f}). 

The first theorem gives an equivalent condition, when $G$ is a finite group, for $\big(\mathbb{E}_{vol}\big)_{vol}$ to be a pure non-degenerate planar Yang-Mills field.

\begin{theorem}
\label{unicity1-f}
Let $G$ be a finite group, let $\big(\mathbb{E}_{vol}\big)_{vol}$ be a stochastically continuous strong planar Markovian holonomy field. It is a pure non-degenerate planar Yang-Mills field if and only if for any simple loop $l$, for any measure of area $vol$, the support of $h(l)$ under $\mathbb{E}_{vol}$ is $G$. 
\end{theorem}

Let $G$ be a compact Lie group and $\big(\mathbb{E}_{vol}\big)_{vol}$ be a $G$-valued stochastically continuous strong planar Markovian holonomy field. The second theorem gives an equivalent condition for $\big(\mathbb{E}_{vol}\big)_{vol}$ to be a pure non-degenerate planar Yang-Mills field.

\begin{theorem}
\label{unicity2-f}
Let us suppose that for any loop $l$ and any measure of area $vol$, under $\big(\mathbb{E}_{vol}\big)_{vol}$, $e$ is in the support of $h(l)$. The planar Markovian holonomy field $\big(\mathbb{E}_{vol}\big)_{vol}$ is a pure non-degenerate planar Yang-Mills field if and only if for any sequence of simple loops $\big(l_{n}\big)_{n\in\mathbb{N}}$ in $\mathbb{R}^{2}$ satisfying $vol \big({\sf Int(}l_{n})\big) \underset{n \to \infty}{\longrightarrow} \infty$, one has: $$\mathbb{E}_{vol} \circ h(l_{n})^{-1}  \underset{n \to \infty}{\longrightarrow} \lambda_{G},$$ where $\lambda_G$ is the Haar measure on $G$.
\end{theorem}

One could remove the condition on the support of $h(l)$ if one could understand the support of any L\'{e}vy process which is invariant by conjugation. The third theorem gives an equivalent condition for $\big(\mathbb{E}_{vol}\big)_{vol}$ to be a pure planar Yang-Mills field.

\begin{theorem}
\label{unicity3-f}
The planar Markovian holonomy field $\big(\mathbb{E}_{vol}\big)_{vol}$ is a pure planar Yang-Mills field if and only there exists a L\'{e}vy process $(X_t)_{t\geq0}$ such that for any simple loop $l$, for any measure of area $vol$, the law of $h(l)$ under $\big(\mathbb{E}_{vol}\big)_{vol}$ is the law of $X_{vol({\sf Int(}l))}$.
If this condition holds, then $(X_t)_{t\geq0}$ is invariant by conjugation and it is the unique L\'{e}vy process associated with $\big(\mathbb{E}_{vol}\big)_{vol}$. 
\end{theorem}

\part{Markovian Holonomy Fields}

\chapter{Markovian Holonomy Fields}
\label{Markov}
In this chapter, some definitions and results about Markovian holonomy fields, taken from \cite{Levy}, are recalled. In the next chapter, the free boundary expectation on the plane associated with any Markovian holonomy field will be defined. This is a planar Markovian holonomy field which will allow us to apply the results obtain previously and to characterize the spherical part of regular Markovian holonomy fields (Theorem \ref{theo}). 

\section{Measured marked surfaces with $G$-constraints}
Until the end of the paper, $M$ will be an oriented smooth compact surface with boundary and, from now on, as we only consider such surfaces, we will call them simply surfaces. 

\begin{definition}
To any connected component of the boundary of $M$ one can associate a non-oriented cycle (Definition ~\ref{cycle}). The union of these non-oriented cycles is denoted by $\mathcal{B}(M)$.

A collection of marks $\mathcal{C}$ on $M$ is a finite union of disjoint simple smooth non-oriented cycles in the interior of $M$. The couple $\left(M,\mathcal{C}\right)$ is called a {\em marked surface} and any element of $\mathcal{C}$ is called a {\em mark}.  

The orientation of $M$ induces an orientation on each connected component of the boundary: we denote by $\mathcal{B}^{+}(M)$ the subset of $\mathcal{B}(M)$ of positively oriented representative of each non-oriented cycle. The non-oriented cycles included in $\mathcal{C}$ does not carry a canonical orientation.
 \end{definition}

Let us recall that a non-oriented cycle is a set $\{c,c^{-1}\}$ where $c$ and $c^{-1}$ are oriented cycle, thus by definition a mark is an oriented cycle. Besides if $M$ has only one boundary, we will denote by $\partial M$ the positively oriented cycle associated with the unique boundary of $M$. Let $\left(M,\mathcal{C}\right)$ be a marked surface. 
\begin{definition}
A graph on $\left(M,\mathcal{C}\right)$ is a graph on $M$ such that each oriented cycle in $\mathcal{C}$ is represented by a loop in $L(\mathbb{G})$. 
\end{definition}

The Proposition 1.3.10 in \cite{Levy} asserts that for any graph $\mathbb{G}$ on $M$, any cycle of $\mathcal{B}(M)$ is represented by a loop in $L(\mathbb{G})$. Let $G$ be a Lie group fixed once for all. Let ${\sf Conj}(G)$ be the set of conjugacy classes of $G$. 

\begin{definition}
A set of $G$-constraints on $\left(M,\mathcal{C}\right)$ is a mapping $C$ from $\mathcal{C} \cup \mathcal{B}(M)$ to ${\sf Conj}(G)$ such that $C(l^{-1}) = C(l)^{-1}$ for any $l\in \mathcal{C} \cup \mathcal{B}(M)$. The family of sets of $G$-constraints on $\left(M,\mathcal{C}\right)$ is denoted by ${\sf Conj}_{G}\left(M,\mathcal{C}\right)$.
\end{definition}

\begin{notation}
\label{modcontr}
Let $C$ be a set of $G$-constraints, let $c$ be an oriented cycle in $\mathcal{C}\cup \mathcal{B}(M)$ and let $x$ be an element of $G$. We will denote by $C_{c \to x}$ the unique set of $G$-constraints such that: 
\begin{enumerate}
\item for any oriented cycle $c' \in \mathcal{C}\cup \mathcal{B}(M)$ different of $c$ and $c^{-1}$, $C_{l \to x}(c') = C(c'),$ 
\item $C_{c \to x}(c) = \left[x\right]$ and $C_{c \to x}(c^{-1}) =\left[x^{-1}\right],$ 
\end{enumerate}
where we recall that $[x]$ is the conjugacy class of $x$. Besides, we will denote by $c \to [x]$ the set of $G$-constraints defined on $\{c, c^{-1}\}$ which sends $c$ on $[x]$ and $c^{-1}$ on $[x^{-1}]$. 
\end{notation}

\begin{definition}
A measured marked surface with $G$-constraints is a quadruple $\left(M, vol,\mathcal{C}, C\right)$ where $\left(M,\mathcal{C}\right)$ is a marked surface, $vol$ is a measure of area on $M$ and $C$ is a set of $G$-constraints on $\left(M,\mathcal{C}\right)$. 

The isomorphism notion on the set of measured marked surfaces with $G$-constraints is the following: $\left(M, vol,\mathcal{C}, C\right)$ and $\left(M',vol',\mathcal{C}', C'\right)$ are isomorphic if and only if there exists a diffeomorphism $\psi: M \to M'$ such that: 
\begin{itemize}
\item $vol \circ \psi^{-1} = vol'$,
\item $\psi$ sends $\mathcal{C}$ on $\mathcal{C}'$,
\item $\forall\ l \in \mathcal{C}\cup \mathcal{B}(M), C'(\psi(l)) = C(l)$. 
\end{itemize}
\end{definition}

\section{Splitting of a surface}
An important notion for the definition of Markovian holonomy fields is the operation of splitting. We will not define this notion rigorously in this paper but instead we refer the reader to Section $1.1.2$ in \cite{Levy} for a rigorous definition.

Let $M$ be a surface, a splitting of $M$ is the data of a surface $M'$ and a gluing: $M' \to M$, which is an application which glues two boundary components of $M'$. The set which consists of the image of the boundary glued and its inverse is the joint of the gluing: we split according to this non-oriented cycle. 
Thus, we will consider a splitting as the inverse of the gluing: a splitting is the action to split a surface according to a non-oriented cycle drawn on it. We will also say that we split a surface according to a mark $l$ and by this, we mean that we split according to the non-oriented cycle $\{l,l^{-1}\}$. There is uniqueness (modulo isomorphism) of the splitting according to a mark on a surface $M$: the split surface of $M$ according to $l$ is denoted by $Spl_{l}(M)$.

Let $\left(M,\mathcal{C},vol, C\right)$ be a measured marked surface with $G$-constraints. Let $l$ be a mark in $\mathcal{C}$ and $f_{l}: Spl_{l}(M) \to M$ be the gluing associated with the splitting $Spl_{l}(M)$. Thanks to  the empty intersection of the marks on $M$, we can transport the marks on $Spl_{l}(M)$. We will denote them $Spl_{l}(\mathcal{C})$. Since outside a negligible subset, a gluing is a diffeomorphism, it is possible to transport the measure of area on $Spl_{l}(M)$ by setting $Spl_{l}(vol) = vol \circ f_l$. In order to transport the $G$-constraints on $Spl_{l}(M)$, we set $Spl_{l}(C)(l') = C(f_{l}(l'))$ for any $l' \in Spl_{l}(\mathcal{C}) \cup \mathcal{B}(Spl_{l}(M))$.

\section{Markovian holonomy fields}
The definition of a Markovian holonomy field was first stated in Definition $3.1.2$ of \cite{Levy}. In this paper, we only consider oriented Markovian holonomy fields: for sake of simplicity, we will call them Markovian holonomy fields. In the following definition, we add the condition that the measures are non-degenerate, which means that their weight is strictly positive. Besides, we change Axiom ${\mathbf{A_{4}}}$: in order to undestand this change in the definition, one can read Remark \ref{rem:erreur}. 

\begin{definition} 
\label{Markovcont}
A $G$-valued Markovian holonomy field, ${\sf HF}$, is the data, for each measured marked surface with $G$-constraints $(M, vol,\mathcal{C},C)$ of a non-degenerate finite measure ${\sf HF}_{(M, vol,\mathcal{C},C)}$ on $\left(\mathcal{M}ult\left(P(M),G\right), \mathcal{I}\right)$ such that: 
\begin{description}
\item[$\mathbf{A_{1}}$ ] For any $(M, vol,\mathcal{C},C)$, ${\sf HF}_{(M, vol,\mathcal{C},C)} \left(\exists l \in \mathcal{C} \cup \mathcal{B}(M), h(l) \notin C(l) \right) = 0$.
\item[$\mathbf{A_{2}}$ ] For any $\left(M, vol, \mathcal{C}\right)$ and any event $\Gamma$ in $\mathcal{I}$, the function which sends $C$ on ${\sf HF}_{(M, vol,\mathcal{C},C)}(\Gamma)$ is a measurable function on ${\sf Conj}_{G}\left(M,\mathcal{C}\right)$.
\item[$\mathbf{A_{3}}$ ] For any $(M, vol, \mathcal{C}, C)$ and any $l \in \mathcal{C}$, 
$${\sf HF}_{\left(M, vol, \mathcal{C}\setminus \{ l, l^{-1}\}, C_{\mid \mathcal{B}(M) \cup \mathcal{C}\setminus \{ l, l^{-1}\}}\right)} = \int_{G} {\sf HF}_{\left(M, vol,\mathcal{C},C_{l \to [x]}\right)} dx,$$
where $C_{l \to [x]}$ is defined in Notation \ref{modcontr}.
\item[$\mathbf{A_{4}}$ ] Let $\psi: (M, vol,\mathcal{C},C) \to (M',vol',\mathcal{C}',C')$ be a bi-Lipschitz homeomorphism which preserves the orientation such that $vol\circ \psi^{-1} = vol'$, $\psi(\mathcal{C}) = \mathcal{C}'$ and $C=C'\circ\psi$. The mapping from $\mathcal{M}ult(P(M'),G)$ to $\mathcal{M}ult\left(P(M),G\right) $ induced by $\psi$, also denoted $\psi$, satisfies: 
$${\sf HF}_{(M',vol',\mathcal{C}',C')} \circ \psi^{-1} = {\sf HF}_{(M, vol,\mathcal{C},C)}.$$
Moreover, let $\mathbb{G}$ (resp. $\mathbb{G}'$) be a graph on $(M,\mathcal{C})$ (resp. on $(M,\mathcal{C}')$), let $\phi$ be a homeomorphism from $(M, vol,\mathcal{C},C)$ to $(M',vol',\mathcal{C}',C')$ which sends $\mathbb{G}$ on $\mathbb{G}'$, which preserves the orientation and such that $vol\circ \phi^{-1} = vol'$, $\phi(\mathcal{C}) = \mathcal{C}'$ and $C=C'\circ\phi$. The mapping from $\mathcal{M}ult(P(\mathbb{G}'),G)$ to $\mathcal{M}ult\left(P(\mathbb{G}),G\right) $ induced by $\phi$, also denoted $\phi$, satisfies: 
$$\left({\sf HF}_{(M',vol',\mathcal{C}',C')}\right)_{| \mathcal{M}ult(P(\mathbb{G}'),G)} \circ \phi^{-1} =\left( {\sf HF}_{(M, vol,\mathcal{C},C)}\right)_{| \mathcal{M}ult(P(\mathbb{G}),G)}.$$

\item[$\mathbf{A_{5}}$ ] For any $(M_{1},vol_{1},\mathcal{C}_{1},C_{1})$ and $(M_{2},vol_{2},\mathcal{C}_{2},C_{2})$, one has the identity: 
$${\sf HF}_{\left(M_{1} \sqcup M_{2},vol_{1} \sqcup vol_{2},\mathcal{C}_{1} \sqcup \mathcal{C}_{2},C_{1} \sqcup C_{2} \right)} = {\sf HF}_{\left(M_{1},vol_{1},\mathcal{C}_{1},C_{1}\right)} \otimes {\sf HF}_{\left(M_{2},vol_{2},\mathcal{C}_{2},C_{2}\right)}.$$
\item[$\mathbf{A_{6}}$ ] For any $(M, vol,\mathcal{C},C)$, any $l \in \mathcal{C}$ and any gluing along $l$, $\psi: Spl_{l}(M) \to M$, one has: 
$${\sf HF}_{\left(Spl_{l}(M),Spl_{l}(vol),Spl_{l}(\mathcal{C}),Spl_{l}(C) \right)} = {\sf HF}_{\left(M, vol,\mathcal{C},C\right)} \circ \psi^{-1}.$$
\item[$\mathbf{A_{7}}$ ] For any $(M, vol,\emptyset, C)$ and for any $l$ in $\mathcal{B}(M)$, 
$$\int_{G} {\sf HF}_{\left(M, vol,\emptyset, C_{l \to x} \right)} (\mathbbm{1}) dx = 1.$$
\end{description}
\end{definition}

The Markovian holonomy fields are easier to understand them when they are exposed in a less formal way. A Markovian holonomy field is a {\em family of measures}. For each surface $M$ with marks, set of $G$-constraints and measure of area, we are given a {\em gauge-invariant random holonomy field on $M$} which  {\em satisfies the set of $G$-constraints} $(\mathbf{A_{1}})$. Moreover, the family of measures given by a Markovian holonomy field is {\em invariant under a class of area-preserving homeomorphisms}, $\mathbf{A_{4}}$, and satisfies a kind of {\em Markov property}, $\mathbf{A_{5}}$ and $\mathbf{A_{6}}$. 
The measures associated with $\left(M, vol, \mathcal{C}, C \right)$, seen as a function of the $G$-constraints, provide a {\em regular disintegration of ${\sf HF}_{\left(M, vol, \emptyset, C_{\mid \mathcal{B}(M)}\right)}$} (Axioms $\mathbf{A_{1}}$, $\mathbf{A_{2}}$ and $\mathbf{A_{3}}$). 
The last assumption is a {\em normalization} axiom. 

As for the planar Markovian holonomy fields, if not specified, all the Markovian holonomy fields will be $G$-valued, thus we will omit to specify it. In the definition of a Markovian holonomy field, we didn't specify any regularity condition on the field. In what follows, we will focus only on regular Markovian holonomy field in the following sense.

\begin{definition}
Let ${\sf HF}$ be a Markovian holonomy field. \begin{enumerate}
\item ${\sf HF}$ is stochastically continuous if, for any $(M, vol,\mathcal{C},C)$, ${\sf HF}_{(M, vol,\mathcal{C},C)}$ is stochastically continuous (Definition \ref{cont}). 
\item ${\sf HF}$ is Fellerian if, for any $\left(M, vol,\mathcal{C}\right)$, the function 
$$(t, C) \mapsto {\sf HF}_{(M, vol,\mathcal{C},C)}(\mathbbm{1}),$$ 
defined on $\mathbb{R}^{*}_{+} \times {\sf Conj}_{G}\left(M,\mathcal{C}\right)$ is continuous. 
\item ${\sf HF}$ is regular if it is both stochastically continuous and Fellerian. 
\end{enumerate}
\end{definition}

\section{Partition functions for oriented surfaces}

Given an even positive integer $g$  and $p$ a positive integer, let $\Sigma_{p, g}^{+}$ be the connected sum of $\frac{g}{2}$ tori with $p$ holes. For $g=0$ we define $\Sigma_{p,0}^{+}$ to be the sphere with $p$ holes. The classification of surfaces asserts that any connected oriented compact surface is diffeomorphic to one and exactly one of $\left\{\Sigma_{p, g}^{+}, p\in \mathbb{N}, g\in 2\mathbb{N}\right\}$. Besides, as a consequence of a theorem of Moser and as explained in Proposition 4.1.1 of \cite{Levy}, if $M$ and $M'$ are oriented, if $(M, vol,\emptyset, C)$ and $(M',vol',\emptyset, C')$ are two measured marked surfaces with $G$-constraints, then they are isomorphic if and only if: 
\begin{enumerate}
\item $M$ and $M'$ are diffeomorphic,
\item $vol(M)=vol'(M'),$
\item there exists a bijection $\psi = \mathcal{B}^{+}(M)\to \mathcal{B}^{+}(M')$ such that $C=C' \circ \psi$ on $\mathcal{B}^{+}(M).$
\end{enumerate}
Let ${\sf HF}$ be a Markovian holonomy field, we define the partition functions of ${\sf HF}$. 
\begin{definition}
Let $g $ be an even positive integer, $p$ be a positive integer and $t$ be a positive real. Let $vol$ be a measure of area on $\Sigma_{p, g}^{+}$ of total mass $t$. Let $\left\{b_{1}, b_{2},  ...,b_{p} \right\}$ be an enumeration of $\mathcal{B}^{+}(\Sigma_{p, g}^{+})$. We define the mapping: 
\begin{align*}
 Z^{+}_{p, g, t}(x_{1},  ...,x_{p}):\ \ \ \ \ \ G^{p} \ \ \ \ &\longrightarrow \mathbb{R}_{+}^{*}\\
(x_{1},  ...,x_{p}) &\ \mapsto \ Z^{+}_{p, g, t}(x_{1},  ...,x_{p})~= {\sf HF}_{\left(\Sigma_{p, g}^{+}, vol, \emptyset, (b_{i}\mapsto [x_{i}])_{i=1}^{p}\right)}(\mathbbm{1}),
\end{align*}
 It is called the partition function of the surface of genus $g$ with $p$ holes. Using the diffeomorphism invariance given by Axiom $\mathbf{A_4}$ and using Moser's theorem, it depends neither on the choice of $vol$ nor on the choice of the enumeration: $Z^{+}_{p, g, t}$ is a symmetric function. 
\end{definition}

\begin{remarque}
If $p=0$, we define $Z^{+}_{0,g, t}$ as the positive number which is equal to the mass of ${\sf HF}_{\left(\Sigma_{0,g}^{+}, vol, \emptyset, \emptyset\right)}$.
\end{remarque}

The discussion on the notion of isomorphism between measured marked surfaces with $G$-constraints implies that if $(M, vol,\emptyset, C)$ is a measured marked surface with $G$-constraints then there exist $p$ and $g$ such that $M$ is diffeomorphic to $\Sigma_{p, g}^{+}$:  
$${\sf HF}_{(M, vol,\emptyset, C)}(\mathbbm{1}) = Z^{+}_{p, g, vol(M)} (x_{1}, ...,x_{p}),$$
where $x_{1}$,  ...,$x_{p}$ are representatives of the $p$ constraints put on $\mathcal{B}^{+}(M)$.

The Fellerian condition satisfied by regular Markovian holonomy fields implies that their partition functions are continuous in $(t, x_{1}, ...,x_{p})$. Besides we can reformulate the axiom of normalization $\mathbf{A_{7}}$ (Definition \ref{Markovcont}) in terms of partition functions. If ${\sf HF}$ is a Markovian holonomy field, for any $t>0$, $$\int_{G}Z_{1,0,t}^{+} (g) dg =1,$$ that is to say: $Z_{1,0,t}^{+} dg$ is a probability measure on $G$. In one of the main theorems proved in Chapter $4$ of \cite{Levy}, L\'{e}vy characterized the family of probability measures $\left(Z_{1,0,t}^{+} dg\right)_{t>0}$. 

\begin{theorem}
\label{partitionlevy}
Let ${\sf HF}$ be a regular Markovian holonomy field. The probability measures $\left(Z_{1,0,t}^{+} dg\right)_{t>0}$ on $G$ are the one dimensional distributions of a unique conjugation-invariant L\'{e}vy process $(Y_{t})_{t\geq 0}$. Moreover, this L\'{e}vy process characterizes completely the partition functions of ${\sf HF}$. 
\end{theorem}

We say that $Y=(Y_{t})_{t \geq 0}$ (resp. ${\sf HF}$) is the L\'{e}vy process (resp. a regular Markovian holonomy field) associated with ${\sf HF}$ (resp. to $Y$). Given this theorem, it is natural to wonder if every L\'{e}vy process which is conjugation-invariant is associated with a regular Markovian holonomy field. Of course, some other conditions must hold such as the existence of a conjugation-invariant square-integrable density. Indeed, as the constraints on the boundary are given by specifying a conjugacy class, $Z^{+}_{1, 0, t}(x)$ is a function of $\left[x\right]$. Besides, by definition of regularity, it must be continuous in $x$ thus square-integrable. To finish, let us remark that $Z_{1,0,t}^{+}(x)$ is strictly positive since we supposed that the measures $\left({\sf H}_{(M, vol,\mathcal{C},C)} \right)_{(M, vol,\mathcal{C},C)}$ are non-degenerate. Hence the natural following definition: 
 
\begin{definition}
Let $(Y_{t})_{t\geq 0}$ be a L\'{e}vy process on $G$. It is admissible if: 
\begin{itemize}
\item it is invariant by conjugation by $G$,
\item the distribution of $Y_{t}$ admits a strictly positive square-integrable density $Q_{t}$ with respect to the Haar measure on $G$ for any $t>0$.
\end{itemize}
\end{definition}

The discussion we just had allows us to write the following proposition.
\begin{proposition}
\label{Levyass}
Let ${\sf HF}$ be a regular Markovian holonomy field, the L\'{e}vy process $(Y_{t})_{t\geq 0}$ associated with ${\sf HF}$ is an admissible L\'{e}vy process. 
\end{proposition}

In fact, we get all the admissible L\'{e}vy processes by studying the L\'{e}vy processes which are associated with regular Markovian holonomy fields: this is given by Theorem $4.3.1$ in \cite{Levy}. 
\begin{theorem} 
\label{YMcrea}
Every admissible L\'{e}vy process $Y$ is associated with a regular Markovian holonomy field.
\end{theorem}
The proof of this assertion consists in constructing, just as we did for planar Markovian holonomy fields, for every admissible L\'{e}vy, a special Markovian holonomy field ${\sf YM}$ which will be called a Yang-Mills field. For this, L\'{e}vy used the edge paradigm for random holonomy fields. A Yang-Mills field is a kind of deformation of a uniform measure.

\section{Uniform measure and Yang-Mills fields}
Let $(M, vol,\mathcal{C},C)$ be a measured marked surface with $G$-constraints, endowed with a graph $\mathbb{G}= (\mathbb{V}, \mathbb{E}, \mathbb{F})$. The uniform measure on $\mathcal{M}ult\left(P(\mathbb{G}),G\right)$ is almost a product of Haar measures as for any orientation $\mathbb{E}^{+}$ of $\mathbb{G}$, $\mathcal{M}ult\left(P(\mathbb{G}),G\right) \simeq {G}^{\mathbb{E}^{+}}$. But one has to be careful: since $(M,\mathcal{C},C)$ is an oriented marked surface with $G$-constraints, the elements in $\mathcal{M}ult\left(P(\mathbb{G}),G\right)$ that we have to consider have to obey the constraints.

\begin{notation}
For any conjugacy class $\mathcal{O} \subset G$ and any integer $n\geq 1$, we denote by $\delta_{\mathcal{O}(n)}$ the natural extension to $G^{n}$ of the unique $G^{n}$-invariant probability measure on $$\mathcal{O}(n)~= \left\{(x_{1}, ...,x_{n}) \in G^{n}: x_{1}...x_{n} \in \mathcal{O}\right\}$$ under the $G^{n}$ action $(g_{1},  ...,g_{n})\bullet(x_{1}, ...,x_{n}) = (g_{1}x_{1}g_{2}^{-1},  ...,g_{n}x_{n}g_{1}^{-1})$. 
\end{notation}

Let $l_{1}, ...,l_{q}$ be $q$ disjoint simple loops in $L(\mathbb{G})$ such that $ \mathcal{C}\cup \mathcal{B}(M)$ is equal to $\{l_{1}, l_{1}^{-1},...,l_{q},l_{q}^{-1} \}$. For any $i\in \{1,...,q\}$, we can decompose $l_{i}=e_{i,1}...e_{i,n_{i}}$ with $e_{i,j} \in \mathbb{E}$ for any $i$ and $j$. Let $\mathbb{E}^{+}$ be an orientation of $\mathbb{G}$, such that for any $i \in \{1,...,q\}$ and $j \in \{1,...,n_{i} \},$ $e_{i, j}\in \mathbb{E}^{+}$. We label $e_{1}$,  ...,$e_{m}$ the other edges of $\mathbb{E}^{+}$. Recall that any measure constructed on ${G}^{\mathbb{E}^{+}}$ defines canonically a unique measure on $\mathcal{M}ult\left(P(\mathbb{G}),G\right)$. 
\begin{definition}
The uniform measure ${\sf U}^{\mathbb{G}}_{M,\mathcal{C}, C
}$ is the measure provided by the following measure on ${G}^{\mathbb{E}^{+}}$: 
\begin{align*}
dg_{1} \otimes ... \otimes dg_{m} \otimes \delta_{C(l_{1})(n_{1})} (dg_{1,n_{1}} ... dg_{1,1})\otimes ... \otimes \delta_{C(l_{q})(n_{q})} (dg_{q, n_{q}}...dg_{q,1}). 
\end{align*}
This probability measure on $\mathcal{M}ult\left(P(\mathbb{G}),G\right) $ does not depend on any of the choices we made. 
\end{definition}

 \begin{notation}
 \label{uniforme}
 We define also a similar measure without constraints for any surface $M$ (resp. $\mathbb{R}^{2}$) endowed with a graph (resp. a planar graph) $\mathbb{G}$. Let $\mathbb{E}^{+}$ be an orientation of $\mathbb{G}$. The measure on $\mathcal{M}ult\left(P(\mathbb{G}),G\right) $ seen on $G^{\mathbb{E}^{+}}$ as $\bigotimes_{e \in \mathbb{E}^{+}} dg_{e}$ is denoted by ${\sf U}^{\mathbb{G}}$. 
 \end{notation}

Yang-Mills fields can now be defined. 

\begin{definition}
\label{YMfields}
Let $(Y_{t} )_{t\geq 0}$ be an admissible L\'{e}vy process on $G$. For any positive real $t$, let $Q_t$ be the density of $Y_t$. A regular Markovian holonomy field ${\sf YM}$ is called a  Yang-Mills field (or sometimes a  Yang-Mills measure) associated with $(Y_t)_{t \geq 0}$ if for any measured marked surface with $G$-constraints $(M, vol,\mathcal{C},C)$ and any graph $\mathbb{G}$ on $\left(M,\mathcal{C}\right)$, 
\begin{align*}
\left({\sf YM}_{(M, vol,\mathcal{C},C)}\right)_{\mid \mathcal{M}ult(P(\mathbb{G}),G) } = \prod_{F \in \mathbb{F}} Q_{vol(F)} \left(h(\partial F )\right) {\sf U}^{\mathbb{G}}_{M,\mathcal{C}, C} (dh), 
\end{align*}
where $\partial F$ is the oriented facial cycle associated with $F$, defined in Definition $1.3.13$ of \cite{Levy} and the notation $Q_{vol(F)} \left(h(\partial F )\right)$ means that we consider $Q_{vol(F)} \left(h(c)\right)$ where $c$ represents $\partial F$: this does not depend on the choice of $c$ since $Q_{vol(F)}$ is invariant by conjugation. 
\end{definition}

A Yang-Mills field associated with $(Y_t)_{t \geq 0}$ is a regular Markovian holonomy field associated with $(Y_{t})_{t\geq 0}$: Theorem \ref{YMcrea} is a consequence of the following proposition.

\begin{proposition}
\label{LevyadmissibleYM}
For any $G$-valued admissible L\'{e}vy process $(Y_t )_{t \geq 0}$ there exists a unique Yang-Mills field ${\sf YM}$ associated with $( Y_t)_{t \geq 0}$. 
\end{proposition}

For this proposition one has to introduce, as we did for planar Markovian holonomy fields, a discrete analog of Markovian holonomy fields: the discrete Markovian holonomy fields. The definition of discrete Markovian holonomy fields can be found in Section $3.2$ of \cite{Levy}. Then one can show that the family of measures: $$\left(\left({\sf YM}_{(M, vol,\mathcal{C},C)}\right)_{\mid \mathcal{M}ult(P(\mathbb{G}),G) }\right)_{\left(M, vol,\mathcal{C},C, \mathbb{G}\right)}$$ is a Fellerian continuously area-dependent (Proposition $4.3.11$ in \cite{Levy}) and locally stochastically $\frac{1}{2}$-H\"{o}lder continuous (Proposition $4.3.15$ in \cite{Levy}) discrete Markovian holonomy field (Proposition 4.3.10 in \cite{Levy}) associated with $Y$. Then it is shown, in Theorem 3.2.9 of \cite{Levy}, that under these regularity conditions, every discrete Markovian holonomy field can be extended to a regular Markovian holonomy field. It has to be noticed that we changed the definition of Markovian holonomy fields (Axiom $\mathbf{A}_4$): this allows us to correct the arguments used in the proof of Axiom $\mathbf{A}_4$ in Theorem 3.2.9 of \cite{Levy} by using the one explained before Theorem \ref{exten3}. This allows to conclude for the proof of Proposition \ref{LevyadmissibleYM}. 

The definition of discrete Markovian holonomy field follows closely the definition of a Markovian holonomy field except for the invariance by homeomorphisms which becomes almost a combinatorial condition. It is the same difference between the Axioms $\mathbf{P_{1}}$ and $\mathbf{DP_{1}}$ of Definitions \ref{planarHF} and \ref{discreteplanarHF} in Section \ref{secdefplanarHF}.

\begin{remarque}
The difference between the assumption $\mathbf{A_{4}}$ in Definition 3.1.2 in \cite{Levy} and $\mathbf{D_{4}}$ in Definition 3.2.1 in \cite{Levy} makes the proof of Lemma 3.2.2. in the same book incomplete. Thus, it is not clear that any Markovian holonomy field defines by restriction a discrete Markovian holonomy field. 
\end{remarque}

This remark leads us to the following definition. 

\begin{definition}
Let ${\sf HF}$ be a Markovian holonomy field. It is {\em constructible} if the family of measures $\left(\left({\sf HF}_{(M, vol,\mathcal{C},C)}\right)_{\mid \mathcal{M}ult(P(\mathbb{G}),G) }\right)_{\left(M, vol,\mathcal{C},C, \mathbb{G}\right)}$ is a discrete Markovian holonomy field. 
\end{definition}

It is still an open question to know if any Markovian holonomy field is constructible. 

\section{Conjecture and main theorem }
We can resume the results of Proposition \ref{Levyass} and Theorem \ref{YMcrea} by the following diagram.

\begin{align*}
 \xymatrix{ \text{Regular Markovian holonomy fields } \ar@/^2.5pc/[r]^{\text{Partition function}} & \text{Admissible L\'{e}vy processes} \ar@/^2.5pc/[l]^{\text{Yang-Mills fields}}}
 \end{align*}

Besides, it was shown that the left arrow goes into the constructible regular Markovian holonomy fields and the composition of the two arrows is equal to the identity on the set of admissible L\'{e}vy processes. It is natural to wonder if the two arrows are each other inverse: this leads us to the following conjecture. 
\begin{conjecture}
\label{conjectureHF}
Every regular Markovian holonomy field is a Yang-Mills field. 
\end{conjecture}
From this conjecture, it would be true that every regular Markovian holonomy field is constructible. In order to state the main result concerning this conjecture, we need the notion of planar mark. Let $M$ be an oriented smooth compact surface with boundary. 
\begin{definition}
A planar mark is a mark $l$ on $M$ such that $l$ cuts $M$ in two parts, one of which is of genus $0$. 
\end{definition}

\begin{theorem}
\label{theo}
Let $\left({\sf HF}_{(M, vol,\mathcal{C},C)}\right)_{(M, vol,\mathcal{C},C)}$ be a regular Markovian holonomy field and $(Y_{t})_{t\in\mathbb{R}^{+}}$ its associated $G$-valued L\'{e}vy process. Let us consider $\left({\sf YM}_{(M, vol,\mathcal{C},C)}\right)_{M, vol,\mathcal{C},C}$ the Yang-Mills field associated with $(Y_{t})_{t\in\mathbb{R}^{+}}$. 

Let $\left(M, vol,\emptyset, C\right)$ be a measured marked surface with $G$-constraints, let $l$ be a planar mark on $M$, let $M_1$ be a part of $M$ of genius $0$ defined by $l$ and let $m$ be a point in $M_1$.  The following equality holds: 
$$\left({\sf HF}_{(M, vol,\emptyset, C)}\right)_{\mid \mathcal{M}ult(L_m(M_1),G)} = \left({\sf YM}_{(M, vol,\emptyset, C)}\right)_{\mid \mathcal{M}ult(L_{m}(M_1),G)}.$$
Let $\mathcal{C}$ be a collection of marks on $M$ which do not intersect the mark $l$. Let us choose an orientation of $\mathcal{C}$ denoted by $\mathcal{C}^{+}$. Let $C$ be a set of $G$-constraints on $\mathcal{B}(M)$. We endow the set of $G$-constraints on $\mathcal{C}\cup \mathcal{B}(M)$ with the measure $d\lambda_{C_{\mid \mathcal{B}(M)}}$ coming from: $$\bigotimes\limits_{c \in \mathcal{C}^{+}} dg_{c} \otimes \bigotimes\limits_{b\in \mathcal{B}(M)^{+}} \delta_{C(b)}.$$
By disintegration, for any set of constraints on $\mathcal{B}(M)$, $d\lambda_{C_{\mid \mathcal{B}(M)}}$ almost surely: 
\begin{align*}
\left({\sf HF}_{(M, vol,\mathcal{C}, C)} \right)_{\mid \mathcal{M}ult(L_{m}(M_1),G)} = \left({\sf YM}_{(M, vol,\mathcal{C}, C)}\right)_{\mid \mathcal{M}ult(L_{m}(M_1),G)}.
\end{align*}

\end{theorem}

In order to prove Conjecture \ref{conjectureHF}, one would have to generalize Theorem \ref{theo} in order to include all the remaining loops, including the generators of the fundamental group of the surface.

\chapter{The Free Boundary Condition on The Plane }

Let ${\sf HF}$ be a regular Markovian holonomy field which will be fixed until the end of the chapter. The measure ${\sf HF}_{\left(M, vol,\mathcal{C}, C\right)}$ is not in general a probability measure. One way to deal with probability measure would be to normalize it by their mass. Yet, a better way to get a probability measure is to consider the free boundary condition measure.

\section{Free boundary condition on a surface}
Let $M$ be a surface homomorphic to a disk $\Sigma^{+}_{0,1}$ endowed with a measure of area $vol$. 
\begin{definition}
\label{freesurface}
The free boundary condition expectation on $M$ associated with ${\sf HF}$ is the probability measure on $\big(\mathcal{M}ult\big(P(M),G\big), \mathcal{B}\big)$ such that for any positive integer $n$, any measurable function $f: G^{n} \to \mathbb{R}^{+}$ and any finite family $c_{1}$, ... $c_{n}$ of elements of $P(M)$: 
\begin{align*}
\mathbb{E}^{{\sf HF}}_{M, vol}\Big( f\big(h(c_{1}),  ...,h(c_{n})\big) \Big) = \int_{G} \widehat{{\sf HF}}_{\left(M, vol, \emptyset, \partial M\to [x]\right)} \Big(f\big(h(c_{1}),  ...,h(c_{n})\big) \Big) dx, 
\end{align*}
where $\widehat{{\sf HF}}_{M, vol, \emptyset, \partial M\to [x]}$ is the extension of ${{\sf HF}}_{M, vol, \emptyset, \partial M\to [x]}$ to the Borel $\sigma$-field given by Proposition \ref{exten}. 
\end{definition}

\begin{remarque}
In this definition we have extended the $\sigma$-field to the Borel $\sigma$-field, in a way such that the new measure becomes invariant by the gauge group. In order for the definition of $\mathbb{E}^{{\sf HF}}_{M, vol}$ to be consistent with the way we named it, one has to verify that it is indeed a probability measure. Since the constant function $\mathbbm{1}$ is gauge-invariant, $\hat{\mathbbm{1}}_{J_{c_{1},  ...,c_{n}}} = \mathbbm{1}$, thus $\mathbb{E}^{{\sf HF}}_{M, vol}(\mathbbm{1}) = \int_{G} {\sf HF}_{\left(M, vol, \emptyset, \partial M\to [x]\right)} (\mathbbm{1}) dx = 1,$ the last equality coming from the normalization Axiom $\mathbf{A_{7}}$ in Definition \ref{Markovcont}.
\end{remarque}

The free boundary condition expectation on $M$ of a Yang-Mills field is computed in the following lemma. 

\begin{lemme}
\label{YMfreegraph}
Let ${\sf YM}$ be the Yang-Mills field associated with a $G$-valued admissible L\'{e}vy process $(Y_{t})_{t\in \mathbb{R}^{+}}$. For any positive real $t$, let $Q_{t}$ be the density of $Y_{t}$. Let $\mathbb{G}$ be a graph on $M$: 
\begin{align*}
{\big(\mathbb{E}^{{\sf YM}}_{M, vol}\big)}_{\mid \mathcal{M}ult\left(P(\mathbb{G}),G\right) } = \prod_{F \in \mathbb{F}} Q_{vol(F)}\left(h(\partial F )\right) {\sf U}^{\mathbb{G}} (dh), 
\end{align*}
where ${\sf U}^{\mathbb{G}}$ was defined in Notation \ref{uniforme} and where we used the notation $Q_{vol(F)}\left(h(\partial F )\right)$ already used in Proposition \ref{densitee}. 
\end{lemme}

\begin{proof}
This follows from the fact that $\int_{G} {\sf U}^{\mathbb{G}}_{M,\emptyset,\partial M \to x} dx = {\sf U}^{\mathbb{G}} $ which is a consequence~of: $\int_{G} \bigg[ \int_{G^{n}} f \delta_{[y](n)}(dx_{1},  ...,dx_{n}) \bigg] dy = \int_{G^{n}} f dx_{1}... dx_{n},$ given by the Equality $(26)$ of Lemma $2.3.4$ in \cite{Levy}.
\end{proof}

\section{Free boundary condition on the plane}
\label{freeboundaryconditionontheplane}

\begin{proposition}
Let $(M, vol)$ and $(M', vol')$ be two measured compact two-dimensional sub-manifolds of $\mathbb{R}^{2}$ which are homeomorphic to the unit disk. Let us suppose that ${M}$ is included in the interior of $M'$, denoted by ${\sf Int}(M')$. Let us assume that $vol'_{\mid M} = vol$. The free boundary condition expectations on $M$ and $M'$ are related by: $\mathbb{E}^{{\sf HF}}_{M, vol} = \mathbb{E}^{{\sf HF}}_{M', vol'} \circ \rho^{-1}_{M, M'},$ where we remind the reader that $\rho_{M, M'}$ was defined in Notation \ref{restriction}. Thus, for any measure of area $vol$ on $\mathbb{R}^{2}$, the family: 
\begin{align*}
\left\{\left( \mathcal{M}ult \left( P\left( M\right), G\right), \mathcal{B}, \mathbb{E}^{{\sf HF}}_{M, vol_{\mid M}} \right)_{M \subset \mathbb{R}^{2}}, \left(\rho_{M, M'}\right)_{M \subset {\sf Int}(M')} \right\},
\end{align*} 
is a projective family of probability spaces. 
\end{proposition}

\begin{proof}
Since $\mathbb{E}^{{\sf HF}}_{M, vol}$ and $\mathbb{E}^{{\sf HF}}_{M', vol'} \circ \rho^{-1}_{M, M'}$ are gauge-invariant, it is enough, by Proposition \ref{unicite1}, to show that, for any positive integer $n$, for any continuous conjugation-invariant function $f$ on $G^{n}$ and any $n$-tuple of loops $l_{1},  ...,l_{n}$ in $M$ based at a fixed point $m$ of $M$: $\mathbb{E}^{{\sf HF}}_{M',vol'}\big[ f\big(h\left(l_{1}\right),  ...,h\left(l_{n}\right)\big)\big]\! =\! \mathbb{E}^{{\sf HF}}_{M, vol}\big[ f\big(h\left(l_{1}\right),  ...,h\left(l_{n}\right)\big)\big]\!.$ The l.h.s. is equal to: 
\begin{align*}
&\int_{G} \int_{\mathcal{M}(P(M'),G)} f\big(h(l_{1}),  ...,h(l_{n})\big) {\sf HF}_{\left(M', vol', \emptyset, \partial M' \to [x]\right)}(dh) dx \\
&= \int_{G} \int_{G} \int_{} f\big(h(l_{1}),  ...,h(l_{n})\big) {\sf HF}_{\left(M', vol', \partial M, \{\partial M' \to [x], \partial M \to [y]\} \right)}(dh) dy dx \\
&= \int_{G} \int_{G} \int_{\mathcal{M}(P((M' \setminus {\sf Int}(M)) \sqcup M),G)} f\big(h(l_{1}),  ...,h(l_{n})\big) 
\\& \ \ \ \ \ \ \ \ \ \ \ \ \ \ \ \ \ {\sf HF}_{\left((M' \setminus  {\sf Int}(M)) \sqcup {M}, vol'_{\mid (M' \setminus {\sf Int}(M))} \sqcup vol, \emptyset, \{{\partial M' \to [x], \partial M \to [y]} \}\right)}(dh) dy dx \\
&= \int_{G} \int_{G} \int_{\mathcal{M}(P(M' \setminus {\sf Int}(M)),G)} \int_{\mathcal{M}(P(M),G)} f\big(h(l_{1}),  ...,h(l_{n})\big) {\sf HF}_{ \left({M}, vol, \emptyset, \{{ \partial M \to [y]} \}\right)}(dh)\\ &{\ \ \ \ \ \ \ \ \ \ \ \ \ \ \ \ \ } {\sf HF}_{ \left(M' \setminus {\sf Int}(M), vol'_{\mid M' \setminus {\sf Int}(M)}, \emptyset, \{{\partial M' \to [x], \partial M \to [y]} \}\right)}(dh') dy dx \\
&= \int_{G} \int_{G} \int_{\mathcal{M}(P(M),G)} f\big(h(l_{1}),  ...,h(l_{n})\big) {\sf HF}_{ \left({M}, vol, \emptyset, \{{ \partial M \to [y]} \}\right)}(dh)\\ &{\ \ \ \ \ \ \ \ \ \ \ \ \ \ \ \ \ } {\sf HF}_{\left( M' \setminus {\sf Int}(M), vol'_{\mid M' \setminus {\sf Int}(M)}, \emptyset, \{{\partial M' \to [x], \partial M \to [y]} \}\right)}(\mathbbm{1}) dy dx\\
&= \int_{G} \int_{\mathcal{M}(P(M),G)} f\big(h(l_{1}),  ...,h(l_{n})\big) {\sf HF}_{ \left({M}, vol, \emptyset, \{{ \partial M \to [y]} \}\right)}(dh)\\ &{\ \ \ \ \ \ \ \ \ \ \ \ \ \ \ \ } \Big(\int_{G} {\sf HF}_{ \left(M' \setminus {\sf Int}({M}), vol'_{\mid M' \setminus {\sf Int}({M})}, \emptyset, \{{\partial M' \to [x], \partial M \to [y]} \}\right)}(\mathbbm{1}) dx \Big) dy\\
&=\int_{G} \int_{\mathcal{M}(P(M),G)} f\big(h(l_{1}),  ...,h(l_{n})\big) {\sf HF}_{ \left({M}, vol, \emptyset, \{{ \partial M \to [y]} \}\right)}(dh) dy \\
&=\mathbb{E}^{{\sf HF}}_{M, vol} \Big[f\big(h(l_{1}),  ...,h(l_{n})\big)\Big], 
\end{align*}
where we applied successively the definition of $\mathbb{E}^{{\sf HF}}_{M',vol'}$, the Axioms $\mathbf{A_{3}}$, $\mathbf{A_{6}}$ and $\mathbf{A_{5}}$. Then after a change of notation and a Fubini exchange of integrals, the normalization Axiom $\mathbf{A_{7}}$ with the definition of $\mathbb{E}^{{\sf HF}}_{M, vol}$ lead us to the result. 
\end{proof}

The free boundary expectation on the plane is the projective limit of this family of measured spaces. Let $vol$ be a measure of area on $\mathbb{R}^{2}$.

\begin{definition}
\label{freebound}
 The free boundary condition expectation on $\mathbb{R}^{2}$ associated with ${\sf HF}$, denoted by $\mathbb{E}_{vol}^{{\sf HF}}$, and defined on $\big(\mathcal{M}ult(P(\mathbb{R}^{2}),G), \mathcal{B}\big)$ is the projective limit of: 
\begin{align*}
\left\{\left( \mathcal{M}ult \left( P\left( M\right), G\right), \mathcal{B}, \mathbb{E}^{{\sf HF}}_{M, vol_{\mid M}} \right)_{M \subset \mathbb{R}^{2}}, \left(\rho_{M, M'}\right)_{M \subset {\sf Int}(M')} \right\}.
\end{align*}
This random holonomy field is gauge-invariant. 
\end{definition}

Lemma \ref{YMfreegraph} gives for any embedded planar graph $\mathbb{G}$ an explicit formula for the restriction on $\mathcal{M}ult(P(\mathbb{G}), G)$ of the free boundary condition expectation on the plane associated with a Yang-Mills field. Proposition \ref{embfinit} asserts that any finite planar graph $\mathbb{G}'$ can be seen as a subgraph of an embedded planar graph. It is thus possible to give an explicit formula for the restriction on $\mathcal{M}ult(P(\mathbb{G}'), G)$ of the free boundary condition expectation on the plane associated with a Yang-Mills field.

\begin{proposition}
\label{graphconditionlibre}
Suppose that $\mathbb{R}^{2}$ is endowed with a measure of area $vol$. Let $\mathbb{G} = (\mathbb{V}, \mathbb{E}, \mathbb{F})$ be a finite planar graph. Let $Y=(Y_{t})_{t \geq 0}$ be a $G$-valued admissible L\'{e}vy process with associated semigroup of densities $(Q_{t})_{t\geq 0}$.  Let ${\sf YM}$ be the Yang-Mills field associated with $Y$. The free boundary condition expectation on $\mathbb{R}^{2}$ of ${\sf YM}$ satisfies: 
\begin{align*}
\big(\mathbb{E}^{{\sf YM}}_{vol}\big)_{\mid \mathcal{M}ult(P(\mathbb{G}),G) } (dh) = \prod_{F \in \mathbb{F}^{b}} Q_{vol(F)} \big(h(\partial F)\big) {\sf U}^{\mathbb{G}} (dh),
\end{align*}
where $\partial F$ is the anti-clockwise oriented facial cycle associated with $F$ and where we used the same convention as before for $Q_{vol(F)} \big(h(\partial F)\big)$.
\end{proposition}

In order to simplify the proof, we will use the upcoming Theorem \ref{freecondplanarHF}. 
\begin{proof}
We have already seen in the proof of Proposition \ref{densitee} that 
\begin{align*}
 \left(\prod_{F \in \mathbb{F}^{b}} Q_{vol(F)} \big(h(\partial F)\big) {\sf U}^{\mathbb{G}} (dh)\right)_{\mathbb{G}, vol}
\end{align*}
is a stochastically continuous in law weak discrete planar Markovian holonomy field. Using Theorem \ref{freecondplanarHF} and the constructibility result of Section \ref{weakconstructibility}, the family $\left(\big(\mathbb{E}^{{\sf YM}}_{vol}\big)_{\mid \mathcal{M}ult\left(P(\mathbb{G}),G\right)}\right)_{\mathbb{G},vol}$ is a stochastically continuous in law weak discrete planar Markovian holonomy field. Recall the definition of $L_{i,j}$ in Definition \ref{base}. Using Theorem \ref{caracterisation0}, we only have to check that for any positive real $\alpha$, $\big(h(L_{n,0}) \big)_{n\in\mathbb{N}}$ has the same law under $\mathbb{E}^{{\sf YM}}_{\alpha dx}$ as under $\prod\limits_{F \in \mathbb{F}^{b}} Q_{\alpha} \big(h(\partial F) \big) {\sf U}^{\mathbb{N}^{2}} (dh)$, where $\mathbb{F}^{b}$ is the set of bounded faces of the $\mathbb{N}^{2}$ graph. The value of $\alpha$ will not be important, so we will suppose that $\alpha=1$. In fact, we will prove that for any positive integer $n$, $\left(\mathbb{E}^{{\sf YM}}_{dx}\right)_{\mid \mathcal{M}ult(P(\mathbb{G}_n),G)}(dh)= \prod\limits_{F \in \mathbb{F}_n^{b}} Q_{1}(h(\partial F)) {\sf U}^{\mathbb{G}_n}(dh)$, where $\mathbb{G}_n = (\mathbb{V}_n, \mathbb{E}_n, \mathbb{F}_n) = \mathbb{N}^{2} \cap\big( [0,n] \times [0,1]\big)$.

Let $\partial D(0,n+1)$ be the loop based at $(n+1,0)$ turning anti-clockwise, representing the cycle bounding the disk of radius $n+1$ centered at $(0,0)$. Let $\mathbb{G}_n' = (\mathbb{V}_n', \mathbb{E}_{n}',\mathbb{F}_n')$ be the graph defined by: 
\begin{itemize}
\item $\mathbb{E}_n' = \mathbb{E}_n \cup \big\{ \partial D(0,n+1), \partial D(0,n+1)^{-1}, e^{r}_{n,0}, (e^{r}_{n,0})^{-1} \big\}$,
\item $\mathbb{V}_n' = \mathbb{V}_n\cup\big\{ (n+1,0)\big\}$. 
\end{itemize}
The finite planar graph $\mathbb{G}_n'$ is an embedded graph and $\mathbb{G}_n$ is a subgraph of $\mathbb{G}_n'$. Using Lemma \ref{YMfreegraph}: 
$$\big(\mathbb{E}^{{\sf YM}}_{dx}\big)_{\mid \mathcal{M}ult(P(\mathbb{G}_n'),G)} (dh) = \prod_{F \in \mathbb{F}_n'^{b}} Q_{dx(F)} \big(h(\partial F)\big) {\sf U}^{\mathbb{G}'_n} (dh).$$ Since $ \left(\prod_{F \in \mathbb{F}^{b}} Q_{vol(F)} \big(h(\partial F)\big) {\sf U}^{\mathbb{G}} (dh)\right)_{\mathbb{G}, vol}$ is a weak discrete planar Markovian holonomy field, the restriction of $\prod_{F \in \mathbb{F}_n'^{b}} Q_{dx(F)} \big(h(\partial F)\big) {\sf U}^{\mathbb{G}'_n} (dh)$ to $\mathcal{M}ult(P(\mathbb{G}_n), G)$ is $\prod\limits_{F \in \mathbb{F}_n^{b}} Q_{1}(h(\partial F)) {\sf U}^{\mathbb{G}_n}(dh)$: $\left(\mathbb{E}^{{\sf YM}}_{dx}\right)_{\mid \mathcal{M}ult(P(\mathbb{G}_n),G)}(dh)= \prod\limits_{F \in \mathbb{F}_n^{b}} Q_{1}(h(\partial F)) {\sf U}^{\mathbb{G}_n}(dh).$
\end{proof}

This last proposition and Proposition \ref{densitee} show that the free boundary condition expectation on $\mathbb{R}^{2}$ of a Yang-Mills field associated with an admissible L\'{e}vy process $Y$ is the planar Yang-Mills field associated with $Y$. This implies the following result.

\begin{proposition}
\label{egaliteentrelesdeuxconstructions}
Let ${\sf YM}$ be the Yang-Mills field associated with an admissible L\'{e}vy process $Y=(Y_t)_{t \in \mathbb{R}^{+}}$. For any planar graph $\mathbb{G} = (\mathbb{V}, \mathbb{E}, \mathbb{F})$, any measure of area $vol$, any family of facial loops $(c_{F})_{F \in \mathbb{F}^{b}}$ oriented anti-clockwise and any rooted spanning tree $T$ of $\mathbb{G}$, under the free boundary condition on the plane $\mathbb{E}^{{\sf YM}}_{vol}$, the random variables $\big(h\left(\l_{c_{F},T}\right)\big)_{F\in \mathbb{F}^{b}}$ are independent and for any $F\in \mathbb{F}^{b}$, $h(\l_{c_{F},T})$ has the same law as $Y_{vol(F)}$. 
\end{proposition}

\section[Building a bridge]{Building a bridge between general and planar Markovian holonomy fields}

The free boundary condition expectation on $\mathbb{R}^{2}$ allows us to link the theory of Markovian holonomy fields with the one of planar Markovian holonomy fields. Consider ${\sf HF}$ a regular Markovian holonomy field and let $\big(\mathbb{E}^{{\sf HF}}_{vol}\big)_{vol}$ be the free boundary condition expectation on the plane associated with ${\sf HF}$.

\begin{theorem}
\label{freecondplanarHF}
The family $\big(\mathbb{E}^{{\sf HF}}_{vol}\big)_{vol}$ is a stochastically continuous strong planar Markovian holonomy field.
\end{theorem}

Using the theory of planar Markovian holonomy fields, it is enough to show that for any $vol$, $\mathbb{E}^{{\sf HF}}_{vol}$ is stochastically continuous and that its restriction to ${\sf Aff}(\mathbb{R}^{2})$ is a stochastically continuous weak planar Markovian holonomy field. As we have already checked the weight condition and as we have noticed the gauge-invariance of the free boundary condition expectation in Definition \ref{freebound}, it remains to show that it is stochastically continuous and that the Axioms $\mathbf{wP_1}$, $\mathbf{wP_2}$ and $\mathbf{wP_3}$ in Definition \ref{weakplanarHF} hold. These are proved in the following Lemmas \ref{stochcondfree}, \ref{area-preserv-exp}, \ref{DP2condfree} and \ref{DP3condfree}.

\begin{lemme}
\label{stochcondfree}
For any measure of area $vol$, $\mathbb{E}_{vol}^{{\sf HF}}$ is a stochastically continuous random holonomy field. 
\end{lemme}

\begin{proof}
Let $vol$ be a measure of area, let $p_n$ be a sequence of paths which converges, as $n$ goes to infinity, to a path $p$ for the convergence with fixed endpoints. Let $\mathbb{D}$ be a disk centered at $(0,0)$ such that for any integer $n$, $p_n \in \mathbb{D}$. We remind the reader that $\widehat{{\sf HF}}_{\left(\mathbb{D}, vol_{\mid \mathbb{D}}, \emptyset, \partial \mathbb{D} \to [x]\right)}$ is the extension given by Proposition \ref{exten} of ${{\sf HF}}_{\left(\mathbb{D}, vol_{\mid \mathbb{D}}, \emptyset, \partial \mathbb{D} \to [x]\right)} $ on the Borel $\sigma$-field. By definition, 
\begin{align*}
\mathbb{E}_{vol}^{{\sf HF}}\big[d_G(h(p_n), h(p )\big] &= \int_G \widehat{{\sf HF}}_{\left(\mathbb{D}, vol_{\mid \mathbb{D}}, \emptyset, \partial \mathbb{D} \to [x]\right) } \big[d_G(h(p_n), h( p)) \big]dx\\
&= \int_G {\sf HF}_{\left(\mathbb{D}, vol_{\mid \mathbb{D}}, \emptyset, \partial \mathbb{D} \to [x]\right) } \Big[\widehat{(d_G)}_{J_{\{p_n,p\}}}(h(p_n), h( p)) \Big]dx.
\end{align*}
Since $p_n$ and $p$ have the same endpoints, $J_{\{p_n,p\}}$ is equal to $G^{2}$ and its action on $G^{2}$ is given by: $(k_1,k_2) \bullet (g_1, g_2) = (k_2^{-1}g_1k_1, k_2^{-1}g_2k_1)$. The invariance of $d_G$, by right and left translations, implies that $\widehat{(d_G)}_{J_{\{p_n,p\}}}=d_G$. This leads to: 
\begin{align*}
\mathbb{E}_{vol}^{{\sf HF}}\big[d_G(h(p_n), h( p)\big] &= \int_G {\sf HF}_{\left(\mathbb{D}, vol_{\mid \mathbb{D}}, \emptyset, \partial \mathbb{D} \to [x]\right) } \big[d_G(h(p_n), h( p)) \big]dx.
\end{align*}
Since ${\sf HF}$ is regular, it is stochastically continuous, thus we have: $${\sf HF}_{\left(\mathbb{D}, vol_{\mid \mathbb{D}}, \emptyset, \partial \mathbb{D} \to [x] \right)} \big[d_G(h(p_n), h( p)) \big] \underset{n \to \infty}{\longrightarrow} 0.$$ Thus, with an argument of dominated convergence, $\mathbb{E}_{vol}^{{\sf HF}}\big[d_G(h(p_n), h(p ))\big]$ converges to zero as $n$ goes to infinity. \end{proof}

\begin{lemme}
\label{area-preserv-exp}
The family of random holonomy fields $\big(\mathbb{E}^{{\sf HF}}_{vol}\big)_{vol}$ satisfies the area-preserving diffeomorphisms at infinity invariance property $\mathbf{wP_1}$. 
\end{lemme}

\begin{proof}
Consider $vol$ and $vol'$ two measures of area on $\mathbb{R}^{2}$. Let $\psi$ be a diffeomorphism at infinity which preserves the orientation and let $R$ be a positive real such that: 
\begin{enumerate}
\item $vol' = vol \circ \psi^{-1}$, 
\item $\psi: \mathbb{D}(0,R)^{c} \to \psi\big( \mathbb{D}(0,R)^{c}\big)$ is a diffeomorphism. 
\end{enumerate}
Using the gauge-invariance of $\mathbb{E}_{vol}^{{\sf HF}}$, it is enough to consider piecewise affine loops based at the same point. Let $l_1, ..., l_n$ be loops in ${\sf Aff}\left(\mathbb{R}^{2}\right)$ based at the same point such that for any $i \in \{1, ..., n\}$, $l'_i=\psi(l_i)$ is in ${\sf Aff}\left(\mathbb{R}^{2}\right)$. Let $R'$ be a positive real such that $R'$ is greater than $R$ and such that for any $i\in\{1, ..., n\}$, $l_i$ is in $M_{R'}=\mathbb{D}(0,R')$. The set $M'=\psi\left(M_{R'}\right)$ is a connected compact two-dimensional sub-manifold of $\mathbb{R}^{2}$. Let us consider $f: G^{n} \to \mathbb{R}$, a continuous function invariant by diagonal conjugation: $\mathbb{E}^{{\sf HF}}_{vol}\Big[f\big((h(l_i)_{i=1}^{n})\big)\Big]$ is equal to: 
\begin{align*}
\mathbb{E}^{{\sf HF}}_{M_{R'},vol_{\mid M_{R'}}}& \Big[f\big((h(l_i)_{i=1}^{n})\big)\Big] \\
&= \int_{G} \int_{\mathcal{M}ult(P(M_R),G)} f\Big(\big(h(l_i)_{i=1}^{n}\big)\Big) {\sf HF}_{\left(M_{R'}, vol_{\mid M_{R'}}, \emptyset, \partial{M_R} \to [x] \right)}(dh) dx \\
&= \int_{G} \int_{\mathcal{M}ult(P(M'),G)} f\Big(\big(h(l'_i)_{i=1}^{n}\big)\Big) {\sf HF}_{\left(M', vol'_{\mid M'}, \emptyset, \partial{M'} \to [x] \right)} (dh)dx\\
&= \mathbb{E}^{{\sf HF}}_{M',vol'_{\mid M'}} \Big[f\big((h(l'_i)_{i=1}^{n})\big)\Big] = \mathbb{E}^{{\sf HF}}_{vol'}\Big[f\big((h(l'_i)_{i=1}^{n})\big)\Big].
\end{align*}
The Axiom $\mathbf{wP_1}$ is satisfied by $\big(\mathbb{E}^{{\sf HF}}_{vol}\big)_{vol}$.
\end{proof}
\begin{lemme}
\label{DP2condfree}
The family of random holonomy fields $\big(\mathbb{E}^{{\sf HF}}_{vol}\big)_{vol}$ satisfies the weak independence property $\mathbf{wP_2}$.
\end{lemme}

\begin{proof}
Let $vol$ be a measure of area on $\mathbb{R}^{2}$. Let $l$ and $l'$ be two loops in ${\sf Aff}(\mathbb{R}^{2})$ such that $\overline{{\sf Int}(l)} \cap \overline{{\sf Int}(l')} = \emptyset$. We can always consider $\tilde{l}$ and $\tilde{l'}$ two smooth simple loops in $\mathbb{R}^{2}$ such that the closure of their interiors are also disjoint and such that $l \subset {\sf Int}(\tilde{l})$ and $l' \subset {\sf Int}(\tilde{l'})$. Using this remark, we can suppose that $l$ and $l'$ are smooth. Using the gauge-invariance of $\mathbb{E}_{vol}^{{\sf HF}}$, as we did in order to show the Axiom $\mathbf{wDP_2}$ in the proof of Proposition \ref{creaYMpure2}, we can work with loops. Let us consider $l_1, ..., l_n$ some loops in $\overline{{\sf Int(}l)}$ and $l'_1, ..., l'_m$ some loops in $\overline{{\sf Int(}l')}$. The aim is to prove that for any continuous functions $f$ and $g$, from $G^{n}$, respectively $G^{m}$, to $\mathbb{R}$, we have: 
\begin{align*}
\mathbb{E}_{vol}^{{\sf HF}} \Big[f\big((h(l_{i}))_{i=1}^{n}\big) g\big((h(l'_{i}))_{i=1}^{m}\big) \Big] &= \mathbb{E}_{vol}^{{\sf HF}} \Big[f\big((h(l_{i}))_{i=1}^{n}\big) \Big]\mathbb{E}_{vol}^{{\sf HF}} \Big[ g\big((h(l'_{i}))_{i=1}^{m}\big) \Big].
\end{align*}

We will use the notations and results stated in Remark \ref{I-indep-et-indep-normale}. Let $\mathcal{L}_{0}$ be a smooth loop such that $\mathcal{L}_{0}$ surrounds $l$ and $l'$. Depending on the context $\mathcal{L}_{0}$ we either stand for $\overline{{\sf Int(}\mathcal{L}_{0})}$ or for the oriented cycle represented by $\mathcal{L}_{0}$. Besides, we will suppose that the orientation of $\mathcal{L}_0$ was chosen such that $\mathcal{L}_0 = \partial \overline{{\sf Int(}\mathcal{L}_{0})}$.  The same notations will hold for $l$ and $l'$. Using the different axioms in Definition \ref{Markovcont}, we have: 
\begin{align*}
&\mathbb{E}_{dx}^{{\sf HF}} \Big[f\big((h(l_{i}))_{i=1}^{n}\big) g\big((h(l'_{i}))_{i=1}^{m}\big)\Big] 
\\&= \mathbb{E}_{\mathcal{L}_{0},dx}^{{\sf HF}} \Big[f\big((h(l_{i}))_{i=1}^{n}\big) g ((h(l'_{i}))_{i=1}^{m})\Big] \\
&= \int_{G} \int_{\mathcal{M}ult(P(\mathcal{L}_{0}),G)}\!\!\!\!\!\!\!\!\!\! \widehat{(f \otimes g)}_{J_{(l_{i})_{i=1}^{n}, (l'_{i})_{i=1}^{m}}}\big((h(l_{i}))_{i=1}^{n},(h(l'_{i}))_{i=1}^{m}\big) {\sf HF}_{\left(\mathcal{L}_{0},dx,\emptyset, \mathcal{L}_{0}\to [y]\right)}(dh) dy \\
&= \int_{G} \int_{\mathcal{M}ult(P(\mathcal{L}_{0}),G)} \!\!\!\!\!\!\!\!\!\! \widehat{f}_{J_{(l_{i})_{i=1}^{n}}}\big((h(l_{i})\big)_{i=1}^{n}\big)\ \widehat{g}_{J_{(l'_{i})_{i=1}^{m}}}\big((h(l'_{i})\big)_{i=1}^{m}\big) {\sf HF}_{\left(\mathcal{L}_{0},dx,\emptyset, \mathcal{L}_{0}\to [y]\right)}(dh) dy \\
&= \int_{G^{3}} \int_{\mathcal{M}ult(P(\mathcal{L}_{0}),G)}\!\!\!\!\!\!\!\!\!\! \widehat{f}_{J_{(l_{i})_{i=1}^{n}}}\big((h(l_{i})\big)_{i=1}^{n}\big)\ \widehat{g}_{J_{(l'_{i})_{i=1}^{m}}}\big((h(l'_{i}))_{i=1}^{m}\big) \\
& \ \ \ \ \ \ \ \ \ \ \ \ \ \ \ \ \ \ \ \ \ \ \ \ \ \ \ \ \ \ \ \ \ \ \ \ \ \ \ \ \ \ \ \ \ \ \ \ \ \ \ {\sf HF}_{\left(\mathcal{L}_{0},dx,\emptyset, \{\mathcal{L}_{0}\to [y], l' \to [z] , l \to [w]\}\right)}(dh) dy dz dw\\
&= \int_{G^{3}} \int_{\mathcal{M}ult(P(l),G)}\!\!\!\!\!\!\!\!\!\! \widehat{f}_{J_{(l_{i})_{i=1}^{n}}}\big((h(l_{i}))_{i=1}^{n}\big) {\sf HF}_{\left(l,dx,\emptyset, l\to[w]\right)} (dh)
\\& \ \ \ \ \ \ \ \ \ \ \ \ \int_{\mathcal{M}ult(P(l'),G)}\!\!\!\!\!\!\!\!\!\!\widehat{g}_{J_{(l'_{i})_{i=1}^{m}}} \big((h(l'_{i}))_{i=1}^{n}\big) {\sf HF}_{\left(l', dx,\emptyset, l' \to [z]\right)}(dh) \\
& \ \ \ \ \ \ \ \ \ \ \ \ \ \ \ \ \ \ \ \ \int_{\mathcal{M}ult(P(\overline{\mathcal{L}_{0} \setminus (l \cup l')}),G)}\!\!\!\!\!\!\!\!\!\!{\sf HF}_{\left(\mathcal{L}_{0},dx,\emptyset, \{\mathcal{L}_{0}\to [y], l' \to [z] , l \to [w]\}\right)}(dh) dy dz dw.
\end{align*}

Since $\int_{G} \int_{\mathcal{M}ult(P(\overline{\mathcal{L}_{0} \setminus (l \cup l')}),G)}{\sf HF}_{\left(\mathcal{L}_{0},dx,\emptyset, \{\mathcal{L}_{0}\to [y], l' \to [z] , l \to [w]\}\right)}(dh) dy$ is equal to $1$, 
$\mathbb{E}_{dx}^{{\sf HF}} \Big[f\big((h(l_i))_{i=1}^{n}\big)g\big((h(l'_{i}))_{i=1}^{m}\big)\Big]$ is equal to  
\begin{align*}
&\int_{G} \int_{\mathcal{M}ult(P(l),G)} \hat{f}_{J_{(l_{i})_{i=1}^{n}}}\big((h(l_i))_{i=1}^{n}\big) {\sf HF}_{\left(l,dx,\emptyset, l\to[w]\right)} (dh) dw
 \\&\ \ \ \ \ \ \ \ \ \ \ \ \ \ \ \ \ \ \ \ \ \ \ \ \ \ \ \int_{G}\int_{\mathcal{M}ult(P(l'),G)}\hat{g}_{J_{(l'_{i})_{i=1}^{m}}} \big((h(l'_{i}))_{i=1}^{m}\big) {\sf HF}_{\left(l', dx,\emptyset, l' \to [z]\right)}(dh) dz,\end{align*}
 which is equal to $\mathbb{E}_{dx}^{{\sf HF}} \Big[f\big((h(l_{i}))_{i=1}^{n}\big) \Big]\mathbb{E}_{dx}^{{\sf HF}} \Big[ g\big((h(l'_{i}))_{i=1}^{m}\big) \Big].$
\end{proof}

\begin{lemme}
\label{DP3condfree}
The family of random holonomy fields $\big(\mathbb{E}^{{\sf HF}}_{vol}\big)_{vol}$ satisfies the locality property $\mathbf{wP_3}$.
\end{lemme}

\begin{proof}
Let $l$ be a simple loop, let $vol$ and $vol'$ be two measures of area whose restrictions to the closure of the interior of $l$ are equal. The random holonomy fields $\mathbb{E}^{{\sf HF}}_{vol}$ and $\mathbb{E}^{{\sf HF}}_{vol'}$ being gauge invariant and stochastically continuous, by Proposition \ref{unicite1}, we only have to prove, for any loops $l_1, ..., l_n$ in ${\sf Int}(l)$ based at the same point and for any continuous function $f: G^{n} \to \mathbb{R}$ invariant under the diagonal action of $G$, that: 
\begin{align*}
\mathbb{E}^{{\sf HF}}_{vol}\Big[f\big(h(l_1), ..., h(l_n)\big)\Big] = \mathbb{E}^{{\sf HF}}_{vol'}\Big[f\big(h(l_1), ..., h(l_n)\big)\Big]. 
\end{align*}
Using Riemann's uniformization theorem, we can find a smooth curve $\tilde{l}$ in the interior of $l$ such that $l_1, ..., l_n$ are in the interior of $\tilde{l}$. Let $M$ be the closure of the interior of $\tilde{l}$:  
\begin{align*}
\mathbb{E}^{{\sf HF}}_{vol}\Big[f\big(h(l_1), ..., h(l_n)\big)\Big] &= \mathbb{E}^{{\sf HF}}_{M,vol_{\mid M}} \Big[f\big(h(l_1), ..., h(l_n)\big)\Big] \\
&=\mathbb{E}^{{\sf HF}}_{M,vol'_{\mid M}} \Big[f\big(h(l_1), ..., h(l_n)\big)\Big] \\
&= \mathbb{E}^{{\sf HF}}_{vol'}\Big[f\big(h(l_1), ..., h(l_n)\big)\Big]. 
\end{align*}
This allows us to conclude. 
\end{proof}

\begin{remarque}
\label{loiduneboucle}
Using the same kind of calculations as the one explained in this subsection and using Theorem \ref{partitionlevy}, it is easy to see that for any simple loop $l$, the law of $h(l)$ under $\mathbb{E}^{{\sf HF}}_{vol}$ is the law of $Y_{vol(\overline{{\sf Int(}l)})} $ where $(Y_{t})_{t\in\mathbb{R}^{+}}$ is the L\'{e}vy process associated with ${\sf HF}$.
\end{remarque}

\chapter[Spherical Part of Regular Markovian Holonomy Fields]{Characterization of the Spherical Part of Regular Markovian Holonomy Fields}
We have now all the tools in order to prove Theorem \ref{theo}. 

\begin{proof}[Proof of Theorem \ref{theo}]
Let us remark that the second part about marks is a consequence of the first part by conditioning: we will prove the first assertion. Let $\left({\sf HF}_{(M, vol,\mathcal{C},C)}\right)_{(M, vol,\mathcal{C},C)}$ be a regular Markovian holonomy field and $(Y_{t})_{t\in\mathbb{R}^{+}}$ its associated $G$-valued L\'{e}vy process. Let $\left({\sf YM}_{(M, vol,\mathcal{C},C)}\right)_{M, vol,\mathcal{C},C}$ be the Yang-Mills field associated with $(Y_{t})_{t\in\mathbb{R}^{+}}$. 

Let $\left(M, vol,\emptyset, C\right)$ be a measured marked surface with $G$-constraints, let $l$ be a planar mark on $M$, let $M_1$ be a part of $M$ of genius $0$ defined by $l$ and let $m$ be a point in $M_1$.  We want to prove that: 
$$\left({\sf HF}_{(M, vol,\emptyset, C)}\right)_{\mid \mathcal{M}ult(L_m(M_1),G)} = \left({\sf YM}_{(M, vol,\emptyset, C)}\right)_{\mid \mathcal{M}ult(L_{m}(M_1),G)}.$$
 
For any loops $l_1,...,l_n$ in $M_1$ based at $m$ and any continuous function $f$ invariant by diagonal conjugation, $\int f(h(l_1),...,h(l_n)) {\sf HF}_{(M, vol,\emptyset, C)}(dh)$ is equal to:
\begin{align*}
\int \int f(h(l_1),...,h(l_n)) {\sf HF}_{(M_1, vol_{\mid M_1},\emptyset, C_{\mid \partial M_1 \setminus \{l,l^{-1}\}}\cup \{ l \to [x]\})}&(dh) \\
&\!\!\!\!\!\!\!\!\!\!\!\!\!\!\!\!\!\!\!\!\!\!\!\!\!\!\!\!\!\!\!\!\!\!\!\!{\sf HF}_{M_2, vol_{\mid M_2}, \emptyset, C_{\mid \partial M_2\setminus \{l,l^{-1}\}} \cup \{ l \to [x]\}}(\mathbbm{1}) dx
\end{align*}
where $M_2$ is the second part of $M$ defined by $l$. Using Theorem \ref{partitionlevy}, ${\sf HF}$ and ${\sf YM}$ have the same partition functions. Thus, since $l$ is a planar mark, it is enough to show that for any measure marked surface with $G$-constraints $(M, vol, \emptyset, C)$ such that $M$ is homeomorphic to a sphere with a positive number $p$ of holes, 
\begin{align*}
{\sf HF}_{(M, vol,\emptyset, C)} = {\sf YM}_{(M, vol,\emptyset, C)}. 
\end{align*}

The proof can be made by induction on the number of holes: we will only prove the case where $p=1$ since the arguments for the induction are similar. Let $\left(\mathbb{E}^{{\sf HF}}_{vol}\right)_{vol}$ be the free boundary condition expectation on the plane, defined in Definition \ref{freebound}, associated with ${\sf HF}$. It is a stochastically continuous strong planar Markovian holonomy field as shown in Theorem \ref{freecondplanarHF}. Hence, by Theorem \ref{weakconst}, it induces a stochastically continuous in law weak discrete planar Markovian holonomy field $\left(\mathbb{E}^{{\sf HF},\mathbb{G}}_{vol}\right)_{\mathbb{G},vol}$. The Remark \ref{loiduneboucle} ensures that the condition in order to apply Theorem \ref{unicity3-f} is satisfied by $\left(\mathbb{E}^{{\sf HF},\mathbb{G}}_{vol}\right)_{\mathbb{G},vol}$. It is equal to the pure discrete planar Yang-Mills field, denoted by $\left(\mathbb{E}^{Y,\mathbb{G}}_{vol}\right)_{\mathbb{G},vol}$, associated with the L\'{e}vy process $(Y_t)_{t \in \mathbb{R}^{+}}$. By stochastic continuity, for any measure of area $vol$, $\mathbb{E}^{{\sf HF}}_{vol} =\mathbb{E}^{Y}_{vol}$, where $\mathbb{E}^{Y}_{vol}$ is the pure continuous planar Yang-Mills field associated with $(Y_t)_{t \in \mathbb{R}^{+}}$.

Let $\left(\mathbb{E}^{{\sf YM}}_{vol}\right)_{vol}$ be the associated free boundary condition expectation on the plane associated with ${\sf YM}$. Using Proposition \ref{egaliteentrelesdeuxconstructions}, for any measure of area $vol$, $\mathbb{E}^{{\sf YM}}_{vol} = \mathbb{E}^{Y}_{vol}$. Recall the notation for the free boundary condition on a surface and let us consider a disk-shaped suface $M$ endowed with a measure of area $vol$. The last two equalities imply that $\mathbb{E}^{{\sf HF}}_{M, vol} = \mathbb{E}^{{\sf YM}}_{M, vol}.$ Using Definition \ref{freesurface}: 
\begin{align*}
\mathbb{E}_{M,vol}^{{\sf HF}} = \int_{G} \widehat{{\sf HF}}_{\left(M,vol,\emptyset,\{\partial M \to [x]\}\right)} dx, 
\end{align*}
and a similar equation holds for ${\sf YM}$. Let $t$ be equal to $vol(M)$ and let us define $Z_t(x)={\sf HF}_{\left(M,vol,\emptyset,\{\partial M \to [x]\}\right)}(\mathbbm{1})$ which is also equal to ${\sf YM}_{\left(M,vol,\emptyset,\{\partial M \to [x]\}\right)}(\mathbbm{1})$ and which is strictly positive. Then: 
\begin{align*}
\mathbb{E}_{M,vol}^{{\sf HF}} = \int_{G} \frac{\widehat{{\sf HF}}_{\left(M,vol,\emptyset,\{\partial M \to [x]\}\right)}}{Z_t(x)} Z_t(x) dx.
\end{align*}
Besides, the law of $h(\partial M)$ is $Z_t(g) dg$: it implies that $\frac{\widehat{{\sf HF}}_{\left(M,vol,\emptyset,\{\partial M \to [x]\}\right)}}{Z_t(x)}$ is a disintegration of $\mathbb{E}_{M,vol}^{{\sf HF}}$ with respect to $h(\partial M)$. The same discussion holds for ${\sf YM}$. By almost sure uniqueness of the disintegration we have: 
\begin{align*}
 \frac{\widehat{{\sf HF}}_{\left(M,vol,\emptyset,\{\partial M \to [x]\}\right)}}{Z_t(x)}= \frac{\widehat{{\sf YM}}_{\left(M,vol,\emptyset,\{\partial M \to [x]\} 
 \right)}}{Z_t(x)}, \text{ a.s. in } x, 
\end{align*}
thus: 
\begin{align}
\label{eq:egaps} \widehat{{\sf HF}}_{\left(M,vol,\emptyset,\{\partial M \to [x]\}\right)}= \widehat{{\sf YM}}_{\left(M,vol,\emptyset,\{\partial M \to [x]\}\right)}, \text{ a.s. in } x. 
\end{align}

It remains to remove the $a.s.$ part. Using Proposition \ref{unicite1} and Lemma \ref{dense}, we need to show that, for any continuous function $f$ invariant by diagonal conjugation from $G^{n}$ to $G$ and any piecewise affine loops, for any Riemannian metric, $l_{1}, ...,l_{n}$ in the interior of $M$, based at the same point:  
\begin{align*}
\widehat{{\sf HF}}_{\left(M,vol,\emptyset, \{\partial M \to [x]\}\right)} \left(f\left(h(l_{1}),  ...,h(l_{n})\right)\right)\!\! =\!\! \widehat{{\sf YM}}_{\left(M,vol,\emptyset,\{\partial M \to [x]\}\right)} \left(f\left(h(l_{1}),  ...,h(l_{n})\right)\right). 
\end{align*}
Yet, given such $n$-tuple, we can always find a mark $l$ such that $l_1,...,l_n$ is in the interior of~$l$. Thus, it is enough to show that for any mark $l$, for any $x \in G$, once we restrain the measures on $\mathcal{M}ult(P(\overline{{\sf Int}(l)}), G)$, we have the equality: 
\begin{align*}
 \widehat{{\sf HF}}_{\left(M,vol,\emptyset,\{\partial M \to [x]\}\right)}= \widehat{{\sf YM}}_{\left(M,vol,\emptyset,\{\partial M \to [x]\}\right)}.
\end{align*}

Let $l$ be a mark on $M$ and let us denote by $M'$ the closure of the interior of~$l$. Let us suppose that the orientation of $l$ is such that $l = \partial M'$. Let us recall that $Z^{+}_{2,0,s}$ was the notation for the partition function of the regular Markovian holonomy field ${\sf HF}$ associated to the planar annulus of total volume which is equal to $s$. Applying the Axioms  $\mathbf{A_3}$, $\mathbf{A_6}$ and $\mathbf{A_5}$, we get that for any continuous function $f$ invariant by diagonal conjugation from $G^{n}$ to $G$ and any loops $l_1,...,l_n$ in $M'$: 
\begin{align*}
&\widehat{{\sf HF}}_{\left(M,vol,\emptyset, \{\partial M \to [x]\}\right)} \left(f\left(h(l_{1}),  ...,h(l_{n})\right)\right) \\
&\ = \int f\left(h(l_{1}),  ...,h(l_{n})\right){\sf HF}_{\left(M,vol,\emptyset, \{\partial M \to [x]\}\right)} (dh) \\
&\ = \int_{G} \int f\left(h(l_{1}),  ...,h(l_{n})\right) {\sf HF}_{\left(M,vol,l, \{\partial M \to [x], l\to [y]\}\right)} (dh) dy\\
&\ = \int_{G} \int f\left(h(l_{1}),  ...,h(l_{n})\right) \\&\ \ \ \ \ \ {\sf HF}_{\left(M',vol_{\mid M'},\emptyset, \{\partial M' \to [y]\}\right)} \otimes {\sf HF}_{\left(M\setminus {\sf Int}({M'}),vol_{\mid M\setminus {\sf Int}({M'})},\emptyset, \{\partial M \to [x], \partial M' \to [y]\}\right)} (dh) dy \\
&\ = \int_{G} \int\ f\left(h(l_{1}),  ...,h(l_{n})\right) \\& \ \ \  \ \ \ {\sf HF}_{\left(M',vol_{\mid M'},\emptyset, \{\partial M' \to [y]\}\right)} (dh) Z^{+}_{2,0,vol(M\setminus {\sf Int}({M'}))}(x,y^{-1}) dy. 
\end{align*}
Thus, if we only consider the restriction on $\mathcal{M}ult\left(P(M'),G\right)$: 
$$\widehat{{\sf HF}}_{\left(M,vol,\emptyset, \{\partial M \to [x]\}\right)} = \int_{G}\widehat{{\sf HF}}_{\left(M',vol_{\mid M'},\emptyset, \{\partial M' \to [y]\}\right)}  Z^{+}_{2,0,vol(M\setminus {\sf Int}({M'}))}(x,y^{-1})dy.$$
Recall that the partition functions of ${\sf HF}$ and ${\sf YM}$ are equal. Thus, again if we only consider the restriction on $\mathcal{M}ult\left(P(M'),G\right)$:  
$$\widehat{{\sf YM}}_{\left(M,vol,\emptyset, \{\partial M \to [x]\}\right)} = \int_{G}\widehat{{\sf YM}}_{\left(M',vol_{\mid M'},\emptyset, \{\partial M' \to [y]\}\right)}  Z^{+}_{2,0,vol(M\setminus {\sf Int}({M'}))}(x,y^{-1})dy.$$
Once we restrain the measures on $\mathcal{M}ult\left(P(M'),G)\right)$, using Equation (\ref{eq:egaps}), for any $x \in G$: 
\begin{align*}
\widehat{{\sf HF}}_{\left(M,vol,\emptyset, \{\partial M \to [x]\}\right)} &= \int_{G}\widehat{{\sf HF}}_{\left(M',vol_{\mid M'},\emptyset, \{\partial M' \to [y]\}\right)}  Z^{+}_{2,0,vol(M\setminus {\sf Int}({M'}))}(x,y^{-1})dy \\&= \int_{G}\widehat{{\sf YM}}_{\left(M',vol_{\mid M'},\emptyset, \{\partial M' \to [y]\}\right)}  Z^{+}_{2,0,vol(M\setminus {\sf Int}({M'}))}(x,y^{-1})dy \\&= \widehat{{\sf YM}}_{\left(M,vol,\emptyset, \{\partial M \to [x]\}\right)}. \end{align*}
This proves that the Equality (\ref{eq:egaps}) now holds for any $x$ in $G$. \end{proof}

\begin{landscape}
\label{diagram}
\begin{align*}
 \xymatrix{ & \textbf{Planar Objects} && &&& \textbf{General Objects} \\
 		 & && &&& \\
 		 \text{Strong } \ar@/^1.5pc/[r]^{\text{Restriction}} & \text{Weak} \ar@/^.5pc/[d]^{\text{Thm. \ref{weakconst}}} && \text{Free boundary expectation} \ar@/^1.5pc/[ll]^{\text{Thm. \ref{freecondplanarHF}}} &&& \text{Regular Markovian H.F.} \ar@/^1.5pc/[dd]^{\text{Conj. \ref{conjectureHF}}} \ar@/^-2pc/[ddddlll]_{\text{Partition function, Thm. \ref{partitionlevy}\ \ \ \ }} \ar@/^1.5pc/[lll]^{\text{Def. \ref{freebound}}} \\
 		 & \text{Weak discrete} \ar@/^0.5pc/[d]^{\text{Thm. \ref{caracterisation}}} && &&& \\
		 \text{Strong discrete} \ar@/^1.5pc/[uu]^{\text{Thm. \ref{exten3}}}& \text{Planar Y.M.} \ar@/^1.5pc/[l]^{\text{Thm. \ref{defplanarHFgeneral}}} && &&& \text{Y.M. measures}\ar@/^1.5pc/[uu]^{\text{Def. \ref{YMfields}}} \\
		 & && &&& \\
		 \textbf{L\'{e}vy processes} & \text{Self-invariant by conjugation} \ar@/^0.5pc/[uu]^{\text{Thm. \ref{defplanarHFgeneral}}}&& \text{Admissible}\ar@/^1.5pc/[ll]^{\text{is}} \ar@/^-1.5pc/[uurrr]^{\text{Prop.  \ref{LevyadmissibleYM}}} }
 \end{align*}
 \end{landscape}


\backmatter
\bibliographystyle{amsalpha}

\bibliography{Gabriel-biblio}


\end{document}